%% file: Thesis.tex
\documentclass[11pt,twosides,a4paper]{report}
\usepackage{graphicx}
\usepackage{graphics}
\graphicspath{{thesis_images/}}
\usepackage[margin=2cm]{geometry}
\usepackage{fancyhdr}
\usepackage{tgpagella}
\usepackage{setspace}
\usepackage[Lenny]{fncychap}
\ChNameVar{\fontsize{14}{16}\usefont{OT1}{phv}{m}{n}\selectfont}
\ChNumVar{\fontsize{60}{62}\usefont{OT1}{ptm}{m}{n}\selectfont}
\ChTitleVar{\Huge\bfseries\rm}
\ChRuleWidth{1pt}
\usepackage{empheq}
 
\usepackage[most]{tcolorbox}

\newtcbox{\mymath}[1][]{%
    nobeforeafter, math upper, tcbox raise base,
    enhanced, colframe=blue!30!black,
    colback=blue!30, boxrule=1pt,
    #1}
\usepackage{enumerate}
\usepackage{mathtools}
\usepackage{rotating}
\usepackage[square, numbers, comma, sort&compress]{natbib}
\usepackage{cancel}
\usepackage{bigints}
\usepackage{hyperref}
\hypersetup{
    colorlinks=true,
    linkcolor=blue,
    filecolor=magenta,      
    urlcolor=teal,
    citecolor=red}
\usepackage{afterpage}
\newcommand\blankpage{%
    \null
    \thispagestyle{empty}%
    \newpage}   
\usepackage{relsize}
\usepackage[toc,page]{appendix}
\usepackage{indentfirst}
\usepackage{amssymb}
\usepackage{fixmath}
\usepackage{subfigure}
\usepackage{bookmark}
\usepackage{amsthm}
\usepackage{wrapfig}
\usepackage{amsmath}
\newtheorem*{theorem}{Theorem}
\newtheorem*{lemma}{Lemma}
\newcommand{\newc}{\newcommand}
\def\eq$#1${\begin{equation}#1\end{equation}}
\def\eqarr$#1${\begin{eqnarray}#1\end{eqnarray}}
\newc{\pa}{\partial}
\newc{\alp}{\alpha}
\newc{\gam}{\gamma}
\newc{\Gam}{\Gamma}
\newc{\del}{\delta}
\newc{\eps}{\epsilon}
\newc{\lam}{\lambda}
\newc{\sig}{\sigma}
\newc{\ups}{\upsilon}
\newc{\ome}{\omega}
\newc{\pphi}{\varphi}
\newc{\nonum}{\nonumber}
\newc{\hami}{\text{\textbf{\lat{H}}}}
\newc{\gren}{\mathcal{G}}
\newc{\lagr}{\mathcal{L}}
\newc{\timor}{\mathcal{T}}
\newc{\prop}{\mathcal{K}}
\newc{\zcal}{\mathcal{Z}}
\newc{\operx}{\text{\textbf{\lat{x}}}}
\newc{\opera}{\text{\textbf{\lat{a}}}}
\newc{\operp}{\text{\textbf{\lat{p}}}}
\newc{\operl}{\text{\textbf{\lat{L}}}}
\newc{\gfv}{g^{(5)}}
\newc{\kfv}{\kappa_{(5)}}
\newc{\dist}{\displaystyle}
\newc{\ra}{\rightarrow}
\newc{\Ra}{\Rightarrow}
 \numberwithin{equation}{chapter}
\def\lat#1{\textlatin{#1}}

\newcommand{\mathsym}[1]{{}}
\newcommand{\unicode}[1]{{}}

\AtBeginDocument{\renewcommand{\bibname}{References}}
\usepackage{verbatim}  
\usepackage{vector}  

\begin{document}
\begin{titlepage}
\begin{center}

\large{\textbf{UNIVERSITY OF IOANNINA}}

\vspace{10em}

\textbf{\Huge{Searching for Localized Black-Hole\vspace{0.2em} solutions in Brane-World models}}

\vspace{8em}

\large{by}\\
\textbf{\Large{Theodoros Nakas}}

\vspace{8em}

\large{A thesis submitted in partial fulfilment for the \\
\textbf{Master's degree}}

\vspace{3em}

in the

\vspace{1em}

\includegraphics[scale=0.4]{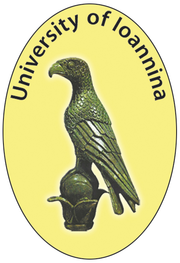}

\vspace{2em}

\large{\textbf{PHYSICS DEPARTMENT\\ SCHOOL OF NATURAL SCIENCES\\ UNIVERSITY OF IOANNINA\\ GREECE\\ JULY 2017}}

\end{center}
\end{titlepage}

\blankpage

\pagenumbering{roman}

\newpage
\blankpage

\newpage
\blankpage

\newpage
\begin{center}
\vspace*{18em}
\large{\emph{``Every atom in your body came from a star that exploded. And, the atoms in your left hand probably came from a different star than your right hand. It really is the most poetic thing I know about physics: You are all stardust. You couldn't be here if stars hadn't exploded, because the elements - the carbon, nitrogen, oxygen, iron, all the things that matter for evolution and for life - weren't created at the beginning of time. They were created in the nuclear furnaces of stars, and the only way for them to get into your body is if those stars were kind enough to explode. So, forget Jesus. The stars died so that you could be here today."}}\\ 
\vspace{1em}
\textbf{\large{Lawrence M. Krauss}}
\vspace*{\fill}
\end{center}
\blankpage

\newpage
\blankpage

\input{Chapters/Abstract}

\input{Chapters/Acknowledgements}

\large{\tableofcontents}
\addcontentsline{toc}{chapter}{\numberline{}Contents}

\blankpage

\pagenumbering{arabic}

\input{Chapters/Basic_Notation}

\newpage
\blankpage
\newpage

\input{Chapters/Chapter_Introduction}

\newpage

\input{Chapters/Chapter_RS_model}

\newpage

\input{Chapters/Chapter_Localization}

\newpage

\input{Chapters/Chapter_Sol_the_Eqs}

\newpage

\input{Chapters/Chapter_Conclusions}

\appendix

\input{Appendices/Appendix_BHS}

\newpage

\input{Appendices/Appendix_GNC}

\newpage

\input{Appendices/Appendix_LG}

\newpage

\input{Appendices/Appendix_Vaidya}	

\newpage

\input{Appendices/Appendix_5Dgeometry}

\newpage

\input{Appendices/Appendix_Variation}

\newpage

\input{Appendices/Appendix_ener_mom}

\newpage

\nocite{*}
\bibliographystyle{unsrt}
\bibliography{Bibliography}
\addcontentsline{toc}{chapter}{\numberline{}References}

\end{document}

%% file: Chapters/Abstract.tex
\begin{center}
\vspace*{5em}
\Huge{\textbf{Abstract}}
\end{center}
\addcontentsline{toc}{chapter}{\numberline{}Abstract}
\vspace{2em}

\large{\par In the context of this thesis, the question that is going to occupy us, is the existence of a 5-dimensional braneworld black hole solution that is localized close to the 3-brane and has the properties of a regular 4-dimensional one. For this purpose, the 4-dimensional part of the complete 5-dimensional spacetime is considered to be a generalized Vaidya metric, in the context of which, the mass parameter $m$ is allowed to vary with respect to time, while it is also allowed to have both $y$ and $r$ dependence. The dependence on the $r$-coordinate essentially means that our black hole solution can deviate from the conventional Schwarzschild solution. Additionally, the dependence on the $y$-coordinate leads to a non-trivial profile of the black hole along the extra dimension. In order to justify physically the existence of such general mass parameter, we consider the case of two scalar fields $\phi(v,r,y)$, $\chi(v,r,y)$ which interact with each other and they are also non-minimally coupled to gravity via a general coupling function $f(\phi,\chi)$. In all the cases that were investigated in the context of this particular scenario, the result for the existence of a viable 5-dimensional localized black hole solution was negative, a result that causes concern about the compatibility of brane-world models with basic predictions of General Theory of Relativity.}

\thispagestyle{plain}
\newpage
\blankpage

%% file: Chapters/Acknowledgements.tex
\begin{center}
\vspace*{5em}
\Huge{\textbf{Acknowledgements}}
\end{center}
\addcontentsline{toc}{chapter}{\numberline{}Acknowledgements}\vspace{2em}

\large{\par First and foremost, I would like to thank my parents; without their love and support this thesis would never have been started.
\par I would also like to thank my supervisor, Panagiota Kanti, not only for the guidance, help and support, but also for providing me the opportunity to work with her and learn more about the enchanting field of General Relativity.
\par I am very thankful to Charalambos Kolasis as well. His passion for mathematics and his clear proofs in every theorem or every problem that he mentioned, helped me understand the way of reasoning as an undergraduate student.
\par Last but not least, I would like to say a big thank you to Ilias for letting me evaluate my mathematica codes on his ``monster" computer and Georgia, Konstantinos and my friends for helping me maintain my sanity. This journey would not have been so enjoyable without you (clich\'{e} but also true).}

\thispagestyle{plain}
\newpage
\blankpage

%% file: Chapters/Basic_Notation.tex
\chapter*{Basic Notation}
\addcontentsline{toc}{chapter}{\numberline{}Basic Notation}

\begin{enumerate}[(i)]
\item \normalsize{The} signature of the metric tensor $(g_{\mu\nu})$ that is going to be used in the context of this thesis is chosen to be: $(-,+,+,\ldots,+)$. Therefore, a flat and 4-D (four-dimensional) space-time has the following line-element.
$$ds^2=-(cdt)^2+(dx)^2+(dy)^2+(dz)^2$$
\item Upper-case Latin indices $M,N,\ldots$ will denote bulk coordinates, thus for a 6-D space-time they will take the values $0,1,2,\ldots,5$. Greek indices $\mu,\nu,\ldots$ will be used for brane coordinates, hence they will take the values $0,1,2,3$. Finally, lower-case Latin indices $a,b,\ldots$ will denote the three spatial coordinates taking the values $1,2,3$.
\item Most of the times, we are going to use natural units or Planck units in order mathematical expressions to be less complicated, i.e. $c=\hbar=1$ or $c=\hbar=G=1$ respectively.
\item The partial derivative of a function $f=f(x^0,x^1,x^2,x^3)$ with respect to $x^\mu$ can be expressed as
$$\frac{\pa f}{\pa x^\mu}=\pa_\mu f=f_{,\mu}$$
\item The covariant derivative of a tensor quantity with respect to $x^\alp$ is denoted as
$$T^{\mu\nu}{}_{;\alp}\ \ \ \ or \ \ \ \ \nabla_\alp T^{\mu \nu}$$
where we arbitrarily chose the contravariant tensor $T^{\mu\nu}$ to present the notation. Of course, the same notation holds for any kind of tensor.
\end{enumerate}

%% file: Chapters/Chapter_Introduction.tex
\chapter{Introduction}
\par \normalsize{The} fundamental similarity of people throughout the history of humankind is curiosity. Since ancient times, this special characteristic has made humans wonder about nature, the origin of the universe and consequently the origin of life. The questions that immediately come to mind are: How did the universe begin? Which are the fundamental building blocks of the universe? Are there extra dimensions?  These questions are still open and they may remain open for many more years.
\par Our understanding about the universe is limited. It is worth mentioning that everything that is possible to be observed in the night sky with a naked eye or by using telescopes adds up to about $4.9\%$ of the entire universe, the other $95.1\%$ of the universe is made of \emph{dark matter} \cite{Bertone:2016nfn} ($26.8\%$) and \emph{dark energy} \cite{RevModPhys.75.559} ($68.3\%$)\footnote{The name \emph{dark} is related to the fact that they do not interact electromagnetically. Therefore, they are invisible to the entire electromagnetic spectrum.}. Almost nothing is known about dark matter and dark energy, but it is certainly known that both of them interact gravitationally. Dark matter behaves like ordinary matter (in terms of gravitation) but it is not luminous, its existence and properties are deduced by astrophysical and cosmological measurements i.e. galaxy rotation curves, gravitational lensing, cosmic microwave background (CMB), etc. On the other hand, dark energy is associated with the accelerated expansion of the universe, hence it is needed to act repulsively. The most popular method and historically the first one that was formulated in order to describe dark energy, is via the \emph{cosmological constant}, which can be identified to a perfect fluid of constant energy density and negative pressure that fills the entire universe homogeneously\footnote{An alternative method to describe dark energy is via \emph{quintessence} \cite{PhysRevD.37.3406,PhysRevLett.80.1582}, a scalar field that it is allowed to be both space and time varying.}. The evidence of the existence of dark energy is indicated also by astrophysical and cosmological observations, i.e. high redshifted observations from supernovae, CMB, Large-scale structures, observational Hubble constant data (OHD), etc.
\par So far we have discussed about what we do not know, thus, it is now reasonable to discuss about what we actually know about the universe. Everything that is known about the universe nowadays can be provided by two theories. The first one in chronological order is General Theory of Relativity (GTR or GR) \cite{Groen:2007zz,Carroll:2004st} and the second one is the Standard Model (SM) \cite{Schwartz:2013pla} of particle physics.
\par GR constitutes the modern theory of gravity and it was formulated by Albert Einstein in 1915 \cite{einstein1,einstein2,einstein3}. The mathematical framework of the theory is Differential Geometry and, in this context, gravity emerges as a geometric property of spacetime\footnote{More precisely, the curvature of the spacetime and the gravitational field are the two sides of the same coin.}, which is a four dimensional (4-D) manifold. GR provides a clear description and explanation for a number of astronomical and cosmological observations. However, only a year after its publication, Karl Schwarzschild showed that GR breaks down in high energies by proving that the spherically symmetric vacuum solution of the gravitational field equations\footnote{In 1923, George David Birkhoff proved that any spherically symmetric solution of the vacuum gravitational field equations must be static and asymptotically flat, which means that the spacetime outside of a spherical, non-rotating body must be described by the Schwarzschild metric.} leads to a singularity at $r=0$ \cite{schwarzschild}. For obvious reasons, infinities (or singularities) are undesirable features for a physical theory. Consequently, GR cannot be considered as the ultimate theory of gravity. When quantum effects become important GR fails to describe them.
\begin{wrapfigure}{R}{0.4\textwidth}
\centering
\includegraphics[width=0.4\textwidth]{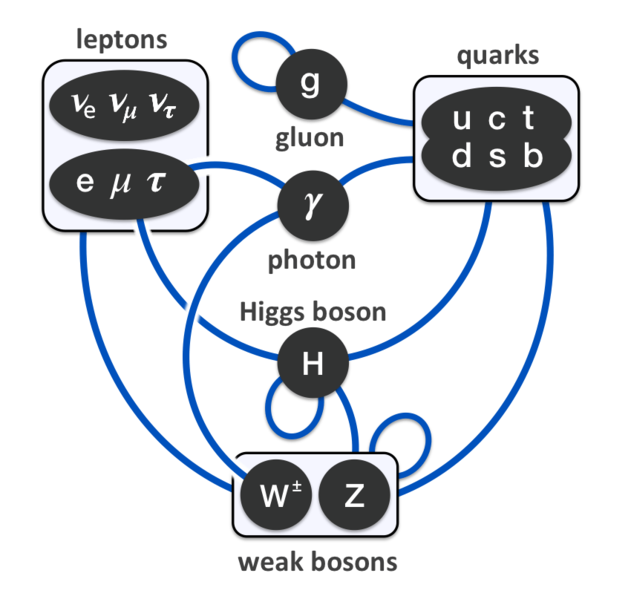}
\caption{\label{SM}Particle interactions as described by Standard Model.}
\end{wrapfigure}
\par SM on the other hand, constitutes the theory of the quantum world. It successfully describes the remaining three of the four known fundamental forces of nature (electromagnetism, weak interaction and strong interaction) and additionally accommodates all the known elementary particles (see Figure 1.1). SM is the most experimentally tested theory and its predictions are confirmed with extreme accuracy. Although its phenomenal success to explain the vast majority of processes in particle physics, SM is still not capable of being the fundamental theory that describes all four fundamental forces (including gravity). Moreover, SM does not provide any insight for the nature of dark energy and dark matter, which as we already mentioned are the biggest mysteries nowadays.
\par It is crystal clear from the aforementioned problems that a new and more fundamental theory is needed, or a deeper understanding of the already established ideas. Such a theory should encompass both GR and SM, namely a theory of everything, but it is not necessary to be based on either of them. However, this hypothetical theory should reproduce both GR and SM in the appropriate limit. Currently, the most promising theory of unification is \emph{string theory}, which replaces elementary particles with one-dimensional strings, instead of being accounted for as point-like objects. The fatal flaw of string theory is its complexity; it requires at least 10 dimensions in order to be consistent and also it is beyond experimental verification because of the extremely high energies that are needed for such a purpose.
\begin{wrapfigure}{R}{0.4\textwidth}
\centering
\includegraphics[width=0.4\textwidth]{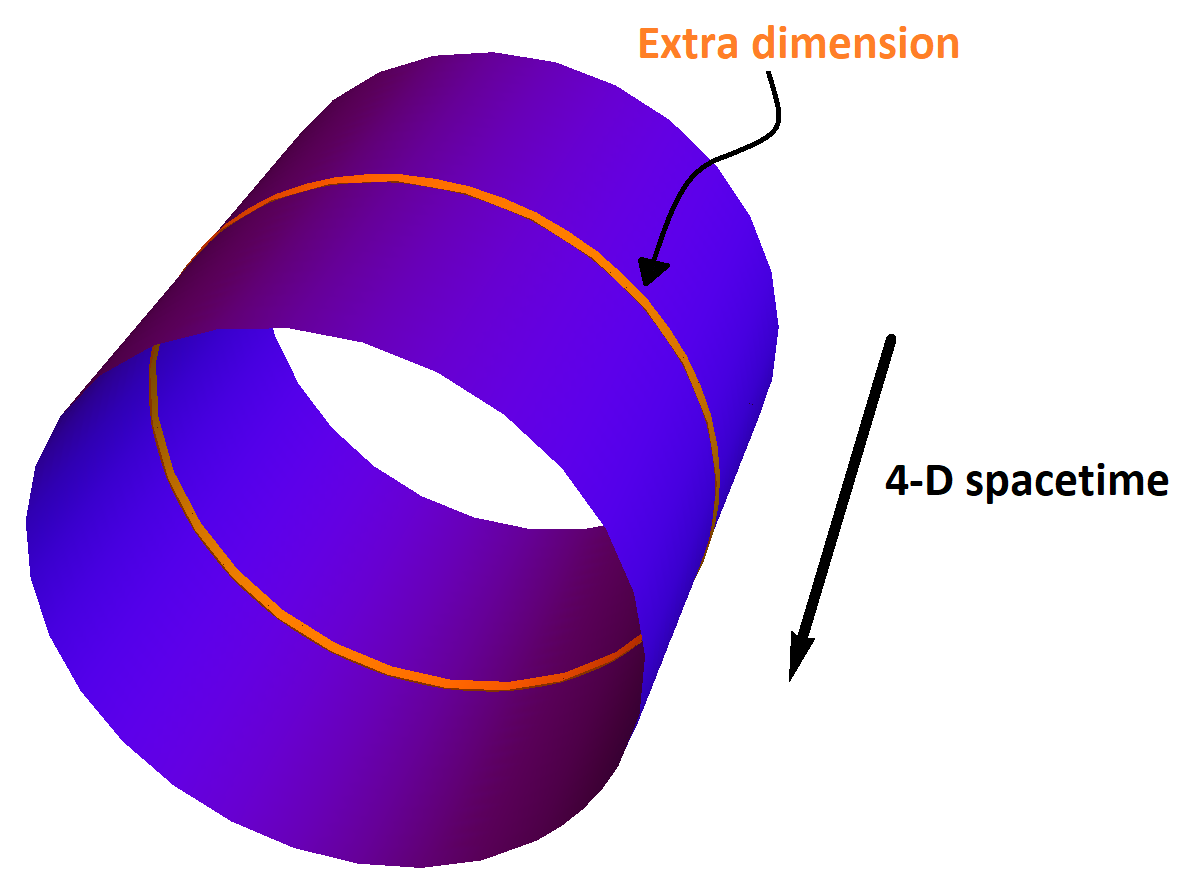}
\caption{\label{KK}The 5-D spacetime in Kaluza Klein theory.}
\end{wrapfigure}
\par Historically, the first theory that introduced extra dimensions was \emph{Kaluza-Klein theory}. The primary idea came from Theodor Kaluza (1921) \cite{Kaluza:1921tu}, who extended general theory of relativity by allowing the existence of another spatial dimension in addition to the four dimensional spacetime fabric of GR. Furthermore, he constrained the five-dimensional metric tensor by demanding none of its components to be dependent on the extra dimension. This condition is also known as the ``cylinder condition". His motivation for a higher dimensional spacetime was the unification of electromagnetism and gravity. The contribution of Oscar Klein (1926) \cite{Klein1926,Klein:1926fj} was his quantum interpretation of Kaluza's theory. He hypothesized that the extra dimension is curled up and tiny (see Figure 1.2) in order to explain the cylinder condition and also evaluated the scale of the extra dimension by taking into consideration the quantum nature of the electric charge.

\thispagestyle{plain}

\par In the last twenty years, two additional extra dimensional theories was added to the literature, \emph{ADD model} (1998) \cite{ARKANIHAMED1998263,ANTONIADIS1998257,PhysRevD.59.086004}  and \emph{RS models} (1999) \cite{PhysRevLett.83.3370,PhysRevLett.83.4690}. Both models \cite{ARKANIHAMED1998263,PhysRevLett.83.3370} were motivated by the \emph{Hierarchy Problem}\footnote{Hierarchy problem is the vast discrepancy between the electroweak scale $m_{EW}\sim 1\  TeV$ and Planck scale $m_{P}\sim 10^{19} GeV$. ($m_{EW}/m_P\sim 10^{-16}$. $m_P=\sqrt{\frac{\hbar c}{G}}$, while $M_P=\frac{m_P}{\sqrt{8\pi}}\sim 10^{18}\ GeV$).} and also both of them were based on the concept of \emph{braneworld}, which simply indicates that Standard Model particles are confined on our 4-D spacetime (\emph{3-brane}) while gravity can freely propagate in the \emph{bulk}, i.e the entire $(4+n)$-dimensional spacetime, where $n$ is the number of the extra dimensions.
\par Particularly, the ADD model allows the existence of $n$ compact extra spatial dimensions of the same radius $R$. Thus, the total number of dimensions that gravity encounters are $(4+n)$. Subsequently, the fundamental Planck scale corresponds to the higher dimensional one, namely $M_{P(4+n)}$, while the Planck scale $M_P$ which we experience is an effective one. Hence, it is possible the fundamental Planck scale $M_{P(4+n)}$ to be equal to the electroweak scale $m_{EW}$ and simultaneously the effective Planck scale to maintain its huge value. As it is indicated in \cite{ARKANIHAMED1998263}, the gravitational potential between two masses $m_1$ and $m_2$ which are separated by a distance $r\ll R$ in $(4+n)$-dimensions is given by the following expression
\eq$\label{in1}
V(r)\sim \frac{m_1 m_2}{M_{P(4+n)}^{n+2}}\frac{1}{r^{n+1}}$
Equation \eqref{in1} can be derived by the gravitational Gauss's law in $(4+n)$ dimensions \cite{SatheeshKumar:2006it,Kehagias:1999my}. Assuming now that $r\gg R$ we obtain
\eq$\label{in2}
V(r)\sim\frac{m_1m_2}{M^{n+2}_{P(4+n)}R^n}\frac{1}{r}$
Therefore, the effective Planck scale should fulfil the following equation 
\eq$\label{in3}
M_P^2\sim M_{P(4+n)}^{n+2}R^n$
Solving the above equation with respect to $R$ and setting $M_{P(4+n)}=m_{EW}$, we obtain
$$R\sim 10^{\frac{30}{n}-19}\left(\frac{1\ TeV}{m_{EW}}\right)^{1+\frac{2}{n}} m\ \xRightarrow{m_{EW}\sim 1\ TeV}$$
\eq$\label{in4}
R\sim 10^{\frac{30}{n}-19}\ m$
Equation \eqref{in4} associates the number of the extra dimensions $n$ with the size of their radius $R$. It is obvious that for $n=1$ and $R\sim 10^{11}\ m$ Newton's gravitational law would differ from the law that we are all used to. Deviations from the conventional Newton's law have not been measured, thus, $n\geq 2$. The upper bound for the size of the extra dimensions is $R\sim 10^{-4}\ m$ \cite{PhysRevLett.86.1418} and it results for $n=2$. The figures below depict schematically the two different types of topologies of the extra dimensions that was taken into consideration in \cite{ARKANIHAMED1998263}.
\begin{figure}[htp]
\centering
\subfigure[toroidal compactification]{\label{add1}
\includegraphics[width=3.25in]{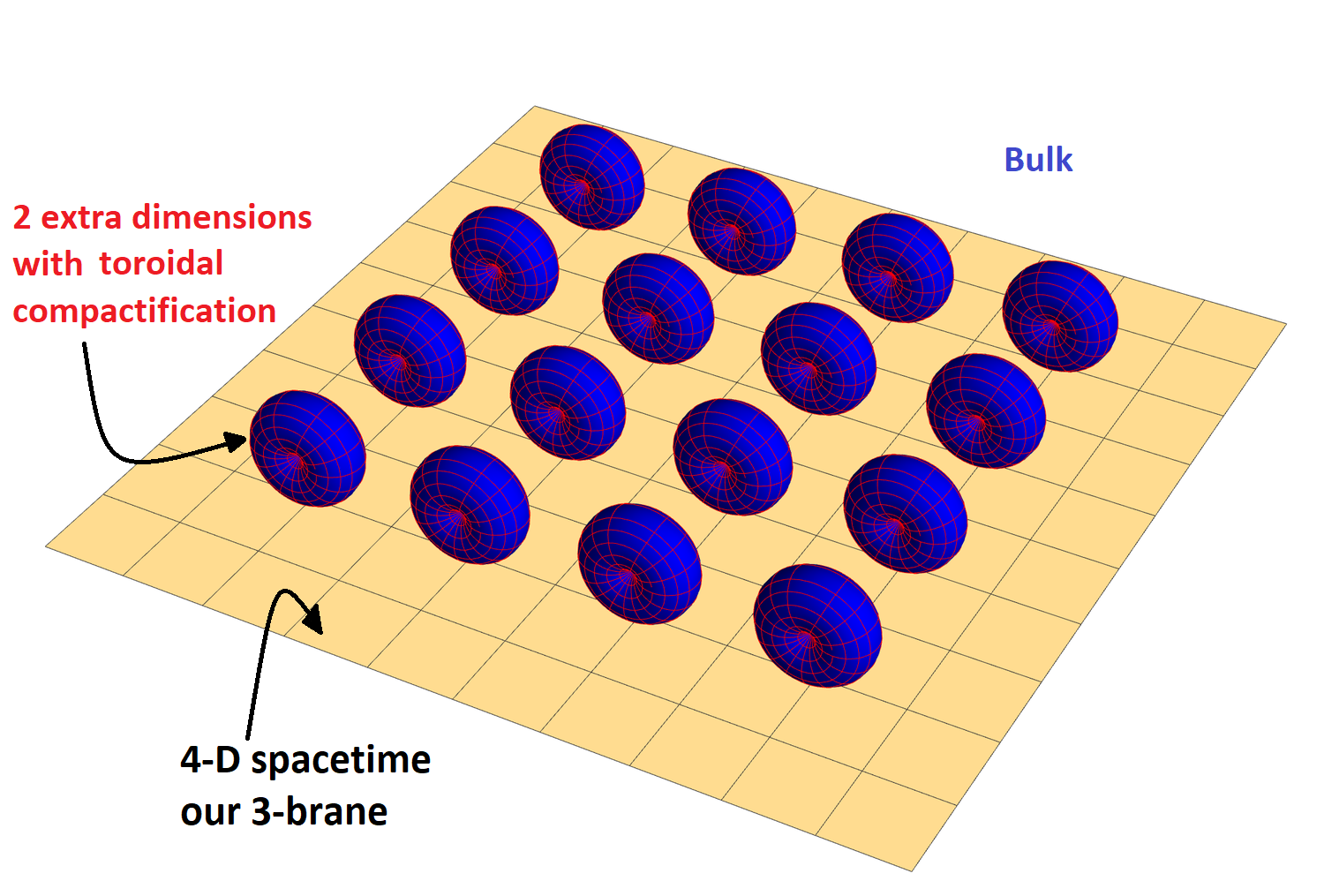}}
\subfigure[spherical compactification]{\label{add2}
\includegraphics[width=3.25in]{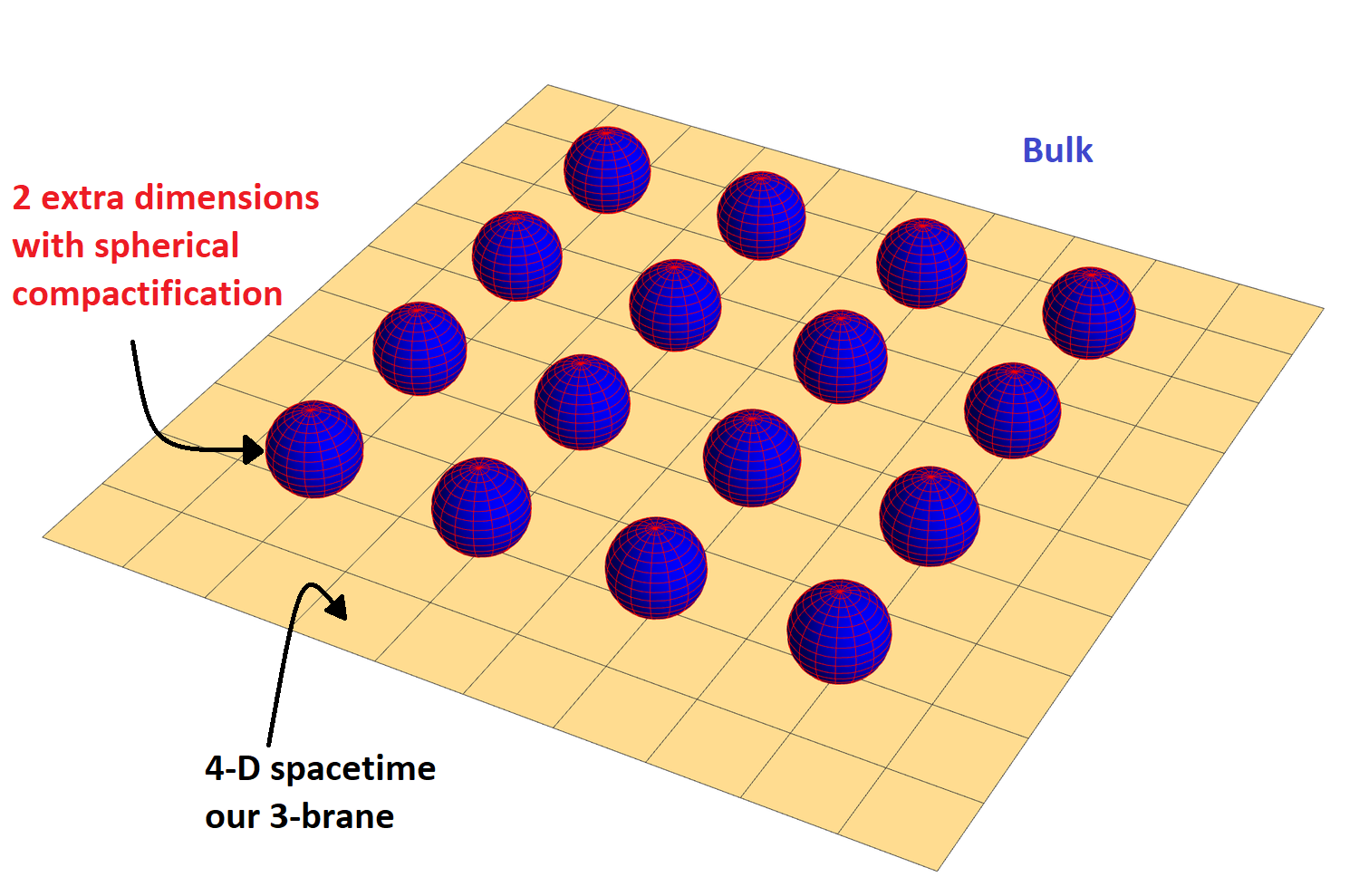}}
\caption{The 6-D spacetime in ADD model for n=2.}
\end{figure}

\thispagestyle{plain}

\par ADD model is elegant from a mathematical point of view, but only seemingly solves the Hierarchy Problem. In particular, ADD model simply replaces the problem of hierarchy with that of the size of the extra dimensions $R$, because now it is requisite to explain the reason why $R$ can be so much larger than the length $10^{-19}\ m$, which is associated with the fundamental Planck scale $M_{P(4+n)}=m_{EW}\sim 1\ TeV$. Although it is difficult to measure gravitational deviations of the Newton's law in sub-millimeter distances, the fact that the fundamental gravitational scale $M_{P(4+n)}$ can be equal to $1\ TeV$ gives us the opportunity to detect (indirectly) the existence of extra dimensions in collider experiments through the formation of tiny black holes from highly energetic particles \cite{Kowalczyk03constraintson, Banks:1999gd,PhysRevD.65.056010,ARGYRES199896,PhysRevLett.87.161602,doi:10.1142/S0217751X04018324,Kanti2009}.

\thispagestyle{plain}

\par On the other extreme, the RS models \cite{PhysRevLett.83.3370,PhysRevLett.83.4690} -which were published only one year after the ADD model- examine the case of only but one curved extra dimension in the bulk. The special characteristic of these models is that now the 3-brane itself possesses tension, therefore it interacts gravitationally with the bulk. Particularly, the first RS model (RS1) assumes the existence of two 3-branes in the bulk and achieves to generate the electroweak scale from the Planck scale through an exponential hierarchy, which arises purely from the geometry of the 5-D spacetime. RS2 model assumes the existence of only one 3-brane embedded in an infinitely long extra dimension but nevertheless manages to reproduce 4-dimensional gravity on the brane. Both RS models will be discussed in Chapter 2, thus, any further detail will be postponed for later.

\section{Motivation and Thesis outline}
\par In the context of GR, all the information that is needed to describe a black hole comes only from three classical parameters: mass $(M)$, electric charge $(Q)$ and angular momentum $(J)$. It is impossible to distinguish two black holes if they are characterized by the same aforementioned parameters and these parameters have the same value. This is an incredible characteristic of the 4-D black holes and it is rarely encountered in other objects in nature. Given these parameters, there are only four different black hole solutions\footnote{Schwarzschild metric $(M\neq 0,\ Q=0,\ J=0)$ \cite{schwarzschild}, Reissner-Nordstr\"{o}m metric $(M\neq 0,\ Q\neq 0,\ J=0)$ \cite{ANDP:ANDP19163550905}, Kerr metric $(M\neq 0,\ Q=0,\ J\neq 0)$ \cite{PhysRevLett.11.237}, Kerr-Newman metric $(M\neq 0,\ Q\neq 0,\ J\neq 0)$ \cite{doi:10.1063/1.1704350} (See also \hyperref[bhs]{Appendix A}).} \cite{schwarzschild,ANDP:ANDP19163550905,PhysRevLett.11.237,doi:10.1063/1.1704350} that can be derived by the gravitational field equations of GR; this particular statement is also known as the \emph{no-hair theorem} \cite{PhysRev.164.1776,Israel1968,PhysRevLett.26.331}. On the other hand, braneworld black holes are not so easily manageable and most importantly there is not a corresponding no-hair theorem for higher dimensional spacetimes. Therefore, the ``families" of higher dimensional black holes solutions are not yet known. This thesis is entirely focused on braneworld black holes and more specifically on the existence of localized black hole solutions in an RS2-type braneworld model. 
\par The outline of the thesis is as follows. Chapter 2 -as it was mentioned previously- constitutes a detailed analysis of the RS models. The existence of black string solutions in the context of RS models and the difficulties in finding a localized 5-D black hole solution on our 3-brane are reviewed in Chapter 3. In the same Chapter, the theoretical framework of the thesis, namely the geometrical background that the thesis is based on, and the scalar field theory (two non-minimally-coupled and interacting scalar fields with a general coupling to the Ricci scalar) are also presented. This scalar field theory constitutes a more general theory compared to the scalar field theory that is presented in \cite{0264-9381-33-1-015003}. This extra degree of freedom -that the additional scalar field offers- may lead to a solution to the problem of localizing a 5-dimensional black hole on a 3-brane. Finally, in Chapter 3, the field equations of this particular ansatz are derived. Subsequently, the various cases\footnote{Different cases vary on the spacetime coordinates from which the scalar fields and the coupling function depend on.} which are studied in this particular scenario are presented in a series of paragraphs of increasing complexity in Chapter 4. Finally, our results and our conclusions are discussed in Chapter 5.
\thispagestyle{plain}

%% file: Chapters/Chapter_RS_model.tex
\chapter{Randall-Sundrum Brane-World Models}
\par In their first model \cite{PhysRevLett.83.3370} (RS1), Lisa Randall and Raman Sundrum introduced a compact extra dimension, which is finite and bounded by two 3-branes. This specific type of spacetime enriched with some additional properties, which are going to be discussed in the following section, manages to address the hierarchy problem. Amazingly, they proved that the 4-D gravity can also be recovered on the brane even if the extra dimension has infinite size, this constitutes the RS2 model \cite{PhysRevLett.83.4690}. Let us now proceed to the analysis of these models.

\section{RS1 Model}
\par The RS model postulates the existence of one extra dimension which is compactified on a circle $S^1$ (one-dimensional sphere) and also possesses a $\mathbb{Z}_2$ symmetry, which simply means that the points $(x^\mu,y)$ and $(x^\mu,-y)$ are identified. Hence, the extra dimension is an $S^1/\mathbb{Z}_2$ orbifold. As it is illustrated in the figure below, this type of compactification contains two fixed points, at $y=0$ and $y=\pi r_c\equiv L$. The range of $y$ is from $-L$ to $L$, but the metric is completely specified by the values in the range $0\leq y\leq L$. These fixed points of the extra dimension ($y=0$ and $y=L$) host the two 3-branes, which essentially are two separated 4-D worlds.\vspace{2em}
\begin{figure}[htp]
\centering
\includegraphics[width=6in]{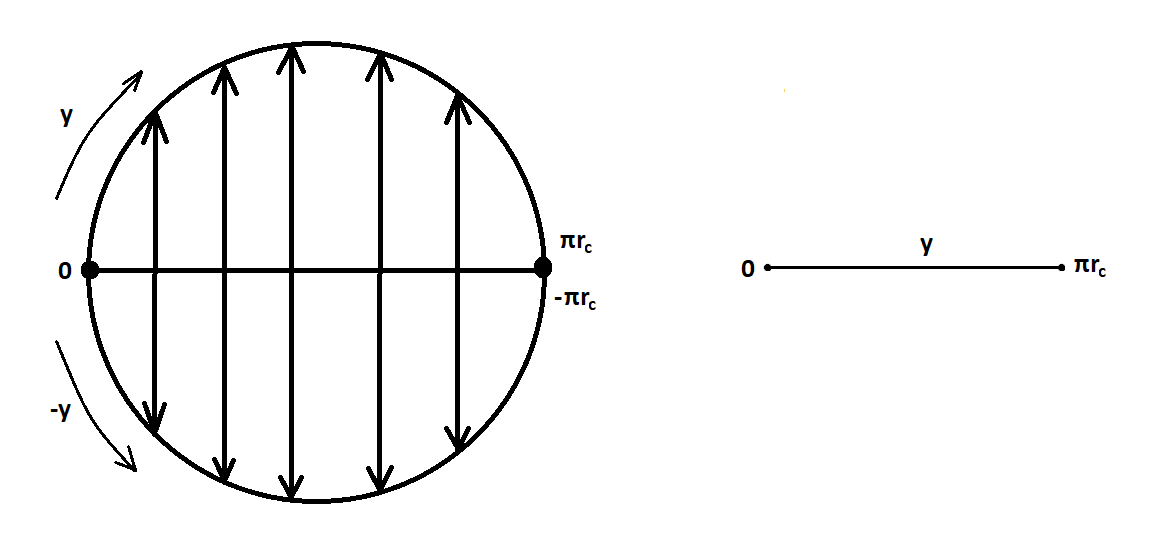}
\caption{$S^1/\mathbb{Z}_2$ orbifold}
\end{figure}

\par The action of the model is as follows:
\begin{gather}
\label{rs1}
S=S_{grav}+S_{1}+S_{2}\\ \nonum\\
\label{rs2}
S_{grav}=\int d^4x\int_{-L}^{L}dy\ \sqrt{-g^{(5)}}\left(\frac{R}{2\kappa_{(5)}}-\Lambda_5\right)\\ \nonum\\
\label{rs3}
S_1=\int d^4x\sqrt{-g_1}\ (\lagr_1-\sig_1)\\ \nonum\\
\label{rs4}
S_2=\int d^4x\sqrt{-g_2}\ (\lagr_2-\sig_2)
\end{gather}
where
\begin{gather}
\label{rs5}
\gfv=\det\left[g_{MN}(x^\lam,y)\right]\\ \nonum\\
\label{rs6}
g_1=\det\left[g^1_{\mu\nu}(x^\lam)\right]\\ \nonum\\
\label{rs7}
g_2=\det\left[g^2_{\mu\nu}(x^\lam)\right]
\end{gather}
The determinants $g_1$ and $g_2$ are derived by the metrics $g^1_{\mu\nu}$ and $g_{\mu\nu}^2$ of the two 3-branes which are located at $y=0$ and $y=L$ respectively. $R$ is the 5-D Ricci scalar, $\Lambda_5$ is the higher dimensional cosmological constant and $\kappa_{(5)}=8\pi G_{(5)}=M_{P(5)}^{-3}$ where $M_{P(5)}$ is the fundamental Planck scale of the 5-D spacetime. The quantities $\sig_1$ and $\sig_2$ represent the energy densities of the 3-branes, while $\lagr_1$ and $\lagr_2$ are the Lagrangians on each 3-brane. 
\par The variation of Eq.\eqref{rs1} with respect to the components of the 5-D metric tensor $g_{MN}$ provides us with the field equations (see \hyperref[var]{Appendix F} for more details about the variation of a general action). The field equations, which are depicted below, were deduced under the assumption that $\lagr_1=\lagr_2=0$. This particular assumption was made in order to determine the geometrical background of the model. The field equations have the following form:
\begin{align}\label{rs8}
\sqrt{-\gfv}\ G_{MN}=-\kappa_{(5)}\left[\sqrt{-\gfv}\ g_{MN}\  \Lambda_5\right.&+\sqrt{-g_1}\ \sig_1\ g^1_{\mu\nu}\ \del^\mu_{M}\ \del^\nu_N\ \del(y)\nonum\\ 
&\left.+\sqrt{-g_2}\ \sig_2\ g^2_{\mu\nu}\ \del^\mu_{M}\ \del^\nu_N\ \del(y-L)\right]
\end{align}
\par In order to continue the analysis of the model it is necessary to introduce an appropriate metric for such a setup. A property that we need to impose on the 5-D metric is to respect the Poincar\'{e} symmetry on the two 3-branes. This property comes naturally from the fact that the 4-D induced metrics on the 3-branes should describe the real world and therefore they should respect the same symmetries as the physical world. The general form of a five-dimensional line-element that includes an extra dimension with $S^1/\mathbb{Z}_2$ compactification and also respects Poincar\'{e} symmetry has the following form\vspace{1em}
\eq$\label{rs9}
ds^2=e^{2A(y)}\eta_{\mu\nu}dx^\mu dx^\nu+dy^2$
The function $A(y)$ is called \emph{warp factor} and it will be evaluated subsequently by the gravitational field equations of this theory. It is obvious that the metric tensor that results from this particular line-element is of the form
\begin{gather}
\label{rs10}
(g_{MN})=\left(\begin{array}{ccccc}
-e^{2A(y)} & 0 & 0 & 0 & 0\\
0 & e^{2A(y)} & 0 & 0 & 0\\
0 & 0 & e^{2A(y)} & 0 & 0\\
0 & 0 & 0 & e^{2A(y)} & 0\\
0 & 0 & 0 & 0 & 1
\end{array}\right)\\ \nonum\\
\label{rs11}
g_{MN}=\underbrace{e^{2A(y)}\eta_{\mu\nu}}_{g_{\mu\nu}}\del^{\mu}_M\del^{\nu}_N+\del^4_M\del^4_N=g_{\mu\nu}\del^{\mu}_M\del^{\nu}_N+\del^4_M\del^4_N\\ \nonum\\
\label{rs12}
(g^{MN})=\left(
\begin{array}{ccccc}
 -e^{-2 A(y)} & 0 & 0 & 0 & 0 \\
 0 & e^{-2 A(y)} & 0 & 0 & 0 \\
 0 & 0 & e^{-2 A(y)} & 0 & 0 \\
 0 & 0 & 0 & e^{-2 A(y)} & 0 \\
 0 & 0 & 0 & 0 & 1 \\
\end{array}
\right)\\ \nonum\\
\label{rs13}
g^{MN}=\underbrace{e^{-2A(y)}\eta^{\mu\nu}}_{g^{\mu\nu}}\del_{\mu}^M\del_{\nu}^N+\del_4^M\del_4^N=g^{\mu\nu}\del_{\mu}^M\del_{\nu}^N+\del_4^M\del_4^N
\end{gather}
\par Consequently, the induced metric tensors on the 3-branes at the points $y=0$ and $y=L$ are given by
\begin{gather}
\label{rs14}
g^1_{\mu\nu}=e^{2A(0)}\eta_{\mu\nu}\\ \nonum\\
\label{rs15}
g^2_{\mu\nu}=e^{2A(L)}\eta_{\mu\nu}
\end{gather}
\par The non-zero and two-times covariant components of the Einstein tensor $G_{MN}$ of the aforementioned ansatz are depicted below:\\
\eq$\label{rs16}
\left\{\begin{array}{c}
G_{00}=G_{tt}=-3e^{2A(y)}\left[2A'^2(y)+A''(y)\right]\\ \\
G_{11}=G_{22}=G_{33}=3e^{2A(y)}\left[2A'^2(y)+A''(y)\right]\\ \\
G_{44}=G_{yy}=6A'^2(y)
\end{array}\right\}\Ra \left\{\begin{array}{c}
G_{\mu\nu}=3\left[2A'^2(y)+A''(y)\right]g_{\mu\nu}\\ \\
G_{44}=6A'^2(y)
\end{array}\right\}$\\
\par Combining equations \eqref{rs8}, \eqref{rs11} and \eqref{rs16} for $M=N=4$ we get:\\
\eq$\label{rs17}
6A'^2(y)=-\kappa_{(5)}\Lambda_5\Ra A'^2(y)=-\frac{\kfv\Lambda_5}{6}$\\
while for $M=\mu$ and $N=\nu$ we obtain:
\begin{gather}
6A'^2(y)+3A''(y)=-\kfv\Lambda_5-\kfv\ \sig_1\ \del(y)-\kfv\ \sig_2\ \del(y-L)\xRightarrow{\eqref{rs17}}\nonum\\ \nonum\\
3A''(y)=-\kfv\left[\sig_1\ \del(y)+\sig_2\ \del(y-L)\right]\Ra\nonum\\ \nonum\\
\label{rs18}
A''(y)=-\frac{\kfv}{3}\left[\sig_1\ \del(y)+\sig_2\ \del(y-L)\right]
\end{gather}
\par The five-dimensional cosmological constant $\Lambda_5$ must be negative in order to have a real solution for the function $A(y)$. A negative 5-D cosmological constant affects decisively the geometry of the spacetime between the two 3-branes. Particularly, the bulk spacetime of this scenario leads to an anti-de Sitter spacetime between the two branes, which is also denoted as $AdS_5$. Subsequently, we define the following constant
\eq$\label{rs19}
\boxed{k^2\equiv -\frac{\kfv\Lambda_5}{6}}$
The specific function of $A(y)$ can be easily deduced from Eq.\eqref{rs17}. Combining equations \eqref{rs17} and \eqref{rs19}, we get
\eq$\label{rs20}
A'^2(y)=k^2\Ra A'(y)=\pm k\Ra A(y)=\pm ky\Ra\boxed{A(y)=-k|y|}$
The reason that we kept the minus sign in Eq.\eqref{rs20} will be understood later. In addition, it is necessary to preserve the orbifold symmetry $S_1/\mathbb{Z}_2$ for the extra dimension $y$. Hence, we are obliged to express the function $A(y)$ in terms of $|y|$. Substituting Eq.\eqref{rs20} into Eq.\eqref{rs9} we obtain the geometrical background of the RS model:
\eq$\label{rs21}
\boxed{ds^2=e^{-2k|y|}\eta_{\mu\nu}dx^\mu dx^\nu+dy^2}$
\par We can relate the constants $\sig_1$ and $\sig_2$ with $k$ by evaluating the second derivative of $A(y)$ as given by Eq.\eqref{rs20} and then equate the result with Eq.\eqref{rs18}.
\eq$\label{rs22}
A'(y)=-k(|y|)'=-k\ sgn(y)=-k\left[H(y)-H(-y)\right]$
where
\eq$\label{rs23}
H(y)=\left\{\begin{array}{cr}
0, & -L\leq y< 0\\ \\
1, & 0\leq y \leq L
\end{array}\right\}\hspace{0.5em},\hspace{1.5em} H(-y)=\left\{\begin{array}{cr}
1, & -L\leq y< 0\\ \\
0, & 0\leq y \leq L
\end{array}\right\}$\vspace{2em}
\eq$\label{rs24}
A''(y)=-k\left[H'(y)-H'(-y)\right]=-2k\left[\del(y)-\del(y-L)\right]$
where\\
\eq$\label{rs25}
\left\{\begin{gathered}
H'(y)=\del(y)-\del(y-L)\\ \\
H'(-y)=\del[-(y-L)]-\del(-y)=-\del(y)+\del(y-L)
\end{gathered}\right\}$\\
\begin{gather}
\eqref{rs18}\xRightarrow{\eqref{rs25}}-\frac{\kfv}{3}\left[\sig_1\ \del(y)+\sig_2\ \del(y-L)\right]=-2k\left[\del(y)-\del(y-L)\right]\Ra\nonum\\
\label{rs26}
\boxed{\sig_1=-\sig_2=\frac{6k}{\kfv}}
\end{gather}\\
\par Let us now focus on the 3-brane at $y=L$ and include the Higgs field in the 4-dimensional action. We are going to evaluate the Vacuum Expectation Value (VEV) of the Higgs field on the brane which determines the physical masses in the Standard Model. The action will have the form:
$$S_H=\int d^4x\int_{-L}^{L} dy\sqrt{-\gfv}\left[g^{MN}D_M H^\dagger D_N H-\lam(H^\dagger H-v_0^2)^2\right]\del(y-L)\Ra$$
$$S_H=\int d^4x\sqrt{-g_2}\left[g_2^{\mu\nu}
D_\mu H^\dagger D_\nu H-\lam(|H|^2-v_0^2)^2\right]\xRightarrow[\sqrt{-g_2}=e^{4A(L)}=e^{-4kL}]{g_2^{\mu\nu}=e^{-2A(L)}\eta^{\mu\nu}=e^{2kL}\eta^{\mu\nu}}$$
$$S_H=\int d^4x\ e^{-4kL}\left[e^{2kL}\eta^{\mu\nu}D_\mu H^\dagger D_\nu H-\lam(|H|^2-v_0^2)^2\right]\Ra$$
\eq$\label{rs27}
S_H=\int d^4x\left[\eta^{\mu\nu}D_\mu\tilde{H}^\dagger D_\nu\tilde{H}-\lam(|\tilde{H}|^2-v^2)^2\right]$
In order to obtain a canonically normalized action we defined
\begin{gather}
\label{rs28}
\tilde{H}\equiv e^{-kL}H\\ \nonum\\
\label{rs29}
v\equiv e^{-kL}v_0
\end{gather}
The action of Eq.\eqref{rs27} depicts the ordinary action of the Higgs field. The corresponding VEV of the renormalized Higgs scalar $\tilde{H}$ is $v$ and is given by Eq.\eqref{rs29}, while the Higgs scalar $H$ is the bare Higgs with VEV $v_0$. As we mentioned earlier, the VEV of the Higgs field determines all the mass parameters in the context of SM, thus we can safely conclude that
\eq$\label{rs30}
\boxed{m=e^{-kL}m_0}$
where $m$ constitutes the physical mass as it is measured on the 3-brane at $y=L$. Eq.\eqref{rs30} is a simple and ``powerful" result. It can be considered as a necessary condition for the solution of the hierarchy problem, because it does not demand a huge discrepancy between parameters $k$ and $L$ in order to achieve that. It is easy to verify the previous statement by setting $m_0$ equal to the Planck mass $M_P\sim 10^{18}\ GeV$ and $m$ equal to the electroweak scale $m_{EW}\sim 1\ TeV$. Then:
\begin{gather}
10^{3}\ GeV=e^{-kL}\ 10^{18}\ GeV\Ra e^{-kL}=10^{-15}\Ra kL=15\ln(10)\Ra\nonum\\
\label{rs31}
\boxed{kL\approx 35}
\end{gather}
\par In order to confidently state that the type of exponential suppression of Eq.\eqref{rs30} successfully addresses the hierarchy problem, it is also necessary to examine the dependence of the effective scale (4-D scale) of gravity on the size of the extra dimension $y$. For this purpose, we need to perturb the four-dimensional part of the five-dimensional line-element of Eq.\eqref{rs21} and then extract the 4-D gravitational action from the original 5-D action which is shown at Eq.\eqref{rs2}. The form of the perturbative line-element is given by the following relation:
\begin{align}\label{rs32}
ds^2&=e^{-2k|y|}[\eta_{\mu\nu}+h_{\mu\nu}(x)]dx^\mu dx^\nu+dy^2\nonum\\
&=e^{-2k|y|}g^{per}_{\mu\nu}dx^\mu dx^\nu+dy^2
\end{align}
where $|h_{\mu\nu}|\ll 1$. Using the metric that follows from Eq.\eqref{rs32} into Eq.\eqref{rs2} and focusing on the term which includes the Ricci scalar, we obtain
\eq$\label{rs33}
S_{eff}= \int d^4x \int_{-L}^{L}dy\ \frac{e^{-2k|y|}}{2\kfv}\sqrt{-g_{per}}\ R^{4D}(g^{per}_{\mu\nu})$
The $S_{eff}$ should be also equal to
\eq$\label{rs34}
S_{eff}=\int d^4x \sqrt{-g_{per}}\ \frac{R^{4D}}{2\kappa}$
where $\kappa=8\pi G=M_P^{-2}$. Consequently, by equating the last two relations we obtain the following expression for the effective scale of gravity.
$$\frac{1}{\kappa}=\frac{1}{\kfv}\int_{-L}^L dy\ e^{-2k|y|}\Ra M_P^2=M_{P(5)}^3\left(\int_{-L}^0 dy\ e^{2ky}+\int_{0}^{L}dy\ e^{-2ky}\right)\Ra$$
$$M_P^2=M_{P(5)}^3\left(\frac{1}{2k}-\frac{e^{-2kL}}{2k}-\frac{e^{-2kL}}{2k}+\frac{1}{2k}\right)\Ra M_P^2=\frac{M_{P(5)}^3}{k}\left(1-e^{-2kL}\right)\Ra$$
\eq$\label{rs35}
\boxed{M_P^2=\frac{M_{P(5)}^3}{k}\left(1-e^{-2kL}\right)}$
Substituting now the value of the product $kL$ given by Eq.\eqref{rs31} into Eq.\eqref{rs35} we get
\eq$\label{rs36}
M_P^2=\frac{M_{P(5)}^3}{k}(1-e^{-70})\simeq \frac{M_{P(5)}^3}{k}$
\par It is clear from Eq.\eqref{rs35} and Eq.\eqref{rs36} that gravity is essentially independent of the size of the extra dimension. Surprisingly, even if we infinitely extend the length $L$ of the extra dimension in Eq.\eqref{rs35}, the four-dimensional Planck scale $M_P$ remains finite. This particular observation was the central point of the RS2 model which is going to be analyzed in the section that follows.
\par Conclusively, it was shown that in the context of the RS model the hierarchy problem has an extremely simple and clear solution. Simultaneously, the RS model does not introduce new huge hierarchies (in contrast with the ADD model) between its fundamental parameters ($k$, $\kfv$ or $M_{P(5)}$, $v_0$, $L$ or $r_c$). The only constraint that is required by the model is that $kL\approx 35$. Of course, a stabilizing mechanism (\emph{Goldberger-Wise mechanism} \cite{PhysRevLett.83.4922,GOLDBERGER2000275}) must be included in the model for this purpose but this is beyond the context of the present analysis.

\section{RS2 Model}
\par In their second paper, Lisa Randall and Raman Sundrum proved that is in fact possible to extend the length $L$ of the extra dimension to an infinite value and nevertheless get an effectively four-dimensional gravity. The verification of the previous statement is derived from the fact that the 5-D graviton is localized near the 3-brane at $y=0$. The setup of the RS2 model is similar to the RS1 model with the difference that the 3-brane at $y=L$ is practically removed from the original picture, by taking $L$ to infinity. Thus, the action of the RS2 model is
\begin{gather}\label{rs37}
S=S_{grav}+S_1\\ \nonum\\
\label{rs38}
S_{grav}=\int d^4x\int dy\ \sqrt{-\gfv}\left(\frac{R}{2\kfv}-\Lambda_5\right)\\ \nonum\\
\label{rs39}
S_1=\int d^4x\sqrt{-g_1}\left(\lagr_1-\sig_1\right)
\end{gather}
The line-element of the RS2 model is given by Eq.\eqref{rs11} as well. Let us now proceed to the derivation of the graviton modes.
\par Gravitons correspond to small fluctuations in the spacetime ``fabric". Therefore, in the context of RS2 model we have
\eq$\label{rs40}
ds^2=e^{-2k|y|}\left[\eta_{\mu\nu}+h_{\mu\nu}(x,y)\right]dx^\mu dx^\nu+dy^2$
In the last equation it was chosen $h_{M4}=0$. It is always possible to find a set of coordinates with this particular property in the region of the 3-brane. We now seek a new variable $z$ for the description of the extra dimension in order to construct a metric tensor which will be more convenient for future calculations. For this purpose, we demand from the new variable $z$ to satisfy the following relation:
\eq$\label{rs41}
dy^2\equiv e^{-2k|y|}dz^2$
Performing the integration of the previous equation we find
\eq$\label{rs42}
k|z|=e^{k|y|}-1$
where we chose the integration constant appropriately in order to get $z=0$ when $y=0$. Eq.\eqref{rs42} leads to
\eq$\label{rs43}
e^{-2k|y|}=\frac{1}{(k|z|+1)^2}$
\par Combining equations \eqref{rs40} and \eqref{rs43} we obtain
\begin{gather}
ds^2=\frac{1}{(k|z|+1)^2}\left\{\left[\eta_{\mu\nu}+h_{\mu\nu}(x,z)\right]dx^\mu dx^\nu+dz^2\right\}\xRightarrow{h_{4M}=0}\nonum\\ \nonum\\
ds^2=e^{2A(z)}\left[\eta_{MN}+h_{MN}(x,z)\right]dx^Mdx^N\Ra\nonum\\ \nonum\\
\label{rs44}
\boxed{ds^2=e^{2A(z)}\bar{g}_{MN}(x,z)dx^Mdx^N}
\end{gather}
where
\begin{gather}\label{rs45}
\boxed{e^{2A(z)}\equiv \frac{1}{(k|z|+1)^2}}\\ \nonum\\
\label{rs46}
\boxed{\bar{g}_{MN}(x,z)\equiv \eta_{MN}+h_{MN}(x,z)}
\end{gather}
From Eq.\eqref{rs45} we get
\begin{gather}\label{rs47}
A(z)=-\ln(k|z|+1)\\ \nonum\\
\label{rs48}
A'(z)=-\frac{k\ sgn(z)}{k|z|+1}=-\frac{k[H(z)-H(-z)]}{k|z|+1}\\ \nonum\\
\label{rs49}
A''(z)=-\frac{2k\del(z)}{k|z|+1}+\frac{k^2}{(k|z|+1)^2}
\end{gather}
\par Under the above transformation of the extra dimension (from $y$ to $z$), Eq.\eqref{rs38} is also going to be transformed into
\eq$\label{rs50}
S_{grav}=\int d^4x\int_{-L_z}^{L_z} dz\sqrt{-\gfv}\left(\frac{R}{2\kfv}-\Lambda_5\right)$
where $L_z=(e^{kL}-1)/k$. The field equations which are yielded by the variation of Eq.\eqref{rs37} with respect to the components of the 5-D metric are depicted below ($\lagr_1$ in Eq.\eqref{rs39} is set to be 0). 
\eq$\label{rs51}
\sqrt{-\gfv}\ G_{MN}=-\kfv\left[\sqrt{-\gfv}\ g_{MN}\ \Lambda_5+\sqrt{-g_1}\ g^1_{\mu\nu}\ \del^\mu_M\ \del^\nu_N\ \del(z)\ \sig_1\right]$
where
\begin{gather}\label{rs52}
g^1_{\mu\nu}\ \del(z)=e^{2A(z)}(\eta_{\mu\nu}+h_{\mu\nu})\ \del(z)\\ \nonum\\
\label{rs53}
g_1\ \del(z)=\gfv\ e^{-2A(z)}\ \del(z)
\end{gather}
Hence, the combination of equations \eqref{rs51}-\eqref{rs53} results to
\eq$\label{rs54}
 G_{MN}=-\kfv\left[g_{MN}\ \Lambda_5+e^{A(z)}(\eta_{\mu\nu}+h_{\mu\nu})\del^\mu_M\ \del^\nu_N\ \del(z)\ \sig_1\right]$
\par The calculation of the Einstein tensor components $G_{MN}$ using ``brute force" is a difficult task in this model. Hence, we are going to use a conformal transformation\footnote{For more details about conformal transformations see \cite{Dabrowski:2008kx}.} in order to obtain the components of the Einstein tensor more easily. Particularly, we mention (without proof) that if $\tilde{g}_{MN}$ is the conformally transformed metric of the metric $g_{MN}$ and the two metrics are connected through the relation
\eq$\label{rs55}
\tilde{g}_{MN}=\Omega^2(x,z)\ g_{MN}$
then the corresponding components of the Einstein tensor are connected as follows:
\eq$\label{rs56}
\tilde{G}_{MN}=G_{MN}+\frac{D-2}{2\Omega^2}\left[4\Omega_{,M}\Omega_{,N}+(D-5)\Omega_{,K}\Omega^{,K}\ g_{MN}\right]-\frac{D-2}{\Omega}\left(\Omega_{;MN}-g_{MN}\square \Omega\right)$
where $D$ indicates the total number of spacetime's dimensions. The adjustment of equations \eqref{rs55} and \eqref{rs56} in our case is straightforward. We simply execute the following substitutions: \{$\tilde{g}_{MN}\ra g_{MN}$, $\Omega(x,z)\ra e^{A(z)}$, $g_{MN}\ra \bar{g}_{MN}$ and $D=5$\}. Hence, we obtain the expression:
$$G_{MN}=\bar{G}_{MN}+3\left[\pa_M A\ \pa_N A-\nabla_M\nabla_NA+\bar{g}_{MN}(\square A+\pa_LA\ \pa^LA)\right]\Ra$$
\begin{align}\label{rs57}
G_{MN}=\bar{G}_{MN}+3[\pa_M A\ \pa_N A&-\pa_M\pa_NA+\bar{\Gam}^L{}_{MN}\pa_LA\nonum\\
&+\bar{g}_{MN}(\pa_L\pa^LA+\bar{\Gam}^L{}_{LK}\pa^KA+\pa_LA\ \pa^LA)]
\end{align}
For all the subsequent calculations in this section, the terms which contain fluctuations $(h_{MN})$ of order higher than the first will be neglected. Additionally, it is possible and extremely convenient to perform appropriate coordinate transformations\footnote{A complete analysis of the legitimacy and derivation of this particular gauge is presented in \cite{Ivanov:1999mt,Myung:2000hu}. See \hyperref[gnc]{Appendix B} as well.} in which the fluctuations satisfy the following properties:
\eq$\label{rs58}
\left\{\begin{array}{c}
h_{4M}=0\\ \\
h=h^{\mu}{}_\mu=0\\ \\
\pa_\mu h^\mu{}_\nu=0\end{array}\right\}$
In this gauge and with the use of equations \eqref{ln3} and \eqref{ln8} it is quite trivial to show that the Christoffel symbols $\bar{\Gam}^L{}_{MN}$ and the components of Einstein tensor $\bar{G}_{MN}$ are expressed through the following equations: 
\eq$\label{rs59}
\bar{\Gam}^L{}_{MN}=\frac{1}{2}\left(\pa_Nh^L{}_M+\pa_Mh^L{}_N-\pa^Lh_{MN}\right)$
\eq$\label{rs60}
\bar{G}_{MN}=-\frac{1}{2}\pa_L\pa^Lh_{MN}$
Using equations \eqref{rs57}-\eqref{rs60} we are led to
\eq$\label{rs61}
\left\{\begin{array}{c}
G_{44}=6A'^2\\ \\
G_{4\mu}=0\\ \\
G_{\mu\nu}=-\frac{1}{2}\pa_L\pa^Lh_{\mu\nu}-\frac{3}{2}\ \pa_4h_{\mu\nu}\ A'+3(\eta_{\mu\nu}+h_{\mu\nu})(A''+A'^2)
\end{array}\right\}$\\
The combination of equations \eqref{rs54} and \eqref{rs61} for $M=N=4$ results to Eq.\eqref{rs47}, namely
$$6A'^2=-\kfv\Lambda_5\ g_{44}=-\kfv\Lambda_5\ e^{2A(z)}(\underbrace{\eta_{44}}_1+\underbrace{h_{44}}_0)\xRightarrow{\eqref{rs19}} A'^2(z)=k^2e^{2A(z)}\Ra$$
$$A(z)=-\ln(k|z|+1)$$
where we fixed again the integration constant appropriately (the reason has already been discussed in the sentence before equation \eqref{rs43}). Correspondingly, for $M=\mu$ and $N=\nu$ we have:
\eq$\label{rs62}
-\frac{1}{2}\pa_L\pa^Lh_{\mu\nu}-\frac{3}{2}\ \pa_4h_{\mu\nu}\ A'+3(\eta_{\mu\nu}+h_{\mu\nu})(A''+A'^2)=-\kfv(\eta_{\mu\nu}+h_{\mu\nu})[e^{2A}\Lambda_5+e^{A}\del(z)\ \sig_1]$
It is not hard to combine equations \eqref{rs19}, \eqref{rs26}, \eqref{rs45}, \eqref{rs48} and \eqref{rs49} in order to prove the following expressions (these are going to simplify Eq.\eqref{rs62} afterwards):
\eq$\label{rs63}
-\kfv\Lambda_5\ e^{2A(z)}=-\frac{\kfv\Lambda_5}{6}\ 6\ e^{2A(z)}=6\ k^2\ e^{2A(z)}=6A'^2(z)$
\eq$\label{rs64}
-\kfv\ \sig_1\ e^{A(z)}\ \del(z)=-6\ k\ e^{A(z)}\ \del(z)=3[A''(z)-A'^2(z)]$
Consequently, from equations \eqref{rs62}-\eqref{rs64} we get
\eq$\label{rs65}
\boxed{\pa_L\pa^Lh_{\mu\nu}+3\ \pa_4h_{\mu\nu}\ A'=0}$
\par We now rescale the fluctuations $h_{\mu\nu}$ as follows.
\eq$\label{rs66}
h_{\mu\nu}\ra e^{\lam A(z)}h_{\mu\nu}$
Then, Eq.\eqref{rs65} takes the form:
\eq$\label{rs67}
h_{\mu\nu}\left[(3\lam+\lam^2)A'^2+\lam\ A''\right]+(2\lam+3)A'\ \pa_4h_{\mu\nu}+\pa_L\pa^Lh_{\mu\nu}=0$
It is now obvious that we can nullify the term which contains the quantity $\pa_4h_{\mu\nu}$ by choosing $\lam=-3/2$. Furthermore, we perform a Kaluza-Klein decomposition on the fluctuations $h_{\mu\nu}(x,z)$:
\eq$\label{rs68}
h_{\mu\nu}(x,z)=\sum_{n=0}^\infty h_{\mu\nu}^n(x)\psi_n(z)$
where
\eq$\label{rs69}
\left\{\begin{array}{c}
h^n_{\mu\nu}(x)=e^{i p_n  x}=e^{i p^n_\lam x^\lam}\\ \\
\pa_\sig\pa^\sig h_{\mu\nu}^n(x)=m_n^2\ h_{\mu\nu}^n(x)\end{array}\right\}$\\
Substituting equations \eqref{rs68} and \eqref{rs69} into Eq.\eqref{rs67} and setting $\lam=-3/2$, we obtain
$$\pa_L\pa^L\left[\sum_{n=0}^\infty h^n_{\mu\nu}(x)\psi_n(z)\right]-\left(\frac{3}{2}A''+\frac{9}{4}A'^2\right)\sum_{n=0}^\infty h^n_{\mu\nu}(x)\psi_n(z)=0\Ra$$
$$\pa_4\pa^4\left[\sum_{n=0}^\infty h^n_{\mu\nu}(x)\psi_n(z)\right]+\pa_\sig\pa^\sig\left[\sum_{n=0}^\infty h^n_{\mu\nu}(x)\psi_n(z)\right]-\left(\frac{3}{2}A''+\frac{9}{4}A'^2\right)\sum_{n=0}^\infty h^n_{\mu\nu}(x)\psi_n(z)=0\Ra$$
$$\sum_{n=0}^\infty\left\{h^n_{\mu\nu}(x)\left[\psi_n''(z)+m_n^2\ \psi_n(z)-\left(\frac{3}{2}A''+\frac{9}{4}A'^2\right)\psi_n(z)\right]\right\}=0\Ra$$
\eq$\label{rs70}
\boxed{-\psi_n''(z)+V(z)\ \psi_n(z)=m_n^2\ \psi_n(z)}$
where\\
\eq$\label{rs71}
\boxed{V(z)=\frac{3}{2}A''+\frac{9}{4}A'^2=-\frac{3k\del(z)}{k|z|+1}+\frac{15k^2}{4(k|z|+1)^2}=\frac{15k^2}{4(k|z|+1)^2}-3k\del(z)}$
\par The boundary condition at $z=0$ can be found as follows:
$$\eqref{rs70}\Ra \int_{0^-}^{0^+}dz\left[-\psi_n''(z)+V(z)\psi(z)\right]=\int_{0^-}^{0^+}dz\ m_n^2\ \psi_n(z)\Ra$$
$$-\psi_n'(0^+)+\psi_n'(0^-)-3k\psi_n(0)=0\xRightarrow[\psi'_n(z)=-\psi'_n(-z)]{\psi_n(z)=\psi_n(-z)}$$
\eq$\label{rs72}
\boxed{\psi'_n(0^+)=-\frac{3}{2}\ k\ \psi_n(0)}$
\par Although we have focused on the RS2 model, we may develop a unified analysis for the study of gravitons in both RS models. Thus, if we reintroduce the second brane at $y=L=\pi r_c$ or at $z=(e^{k\pi r_c}-1)/k\equiv L_z$ (as it is indicated by Eq.\eqref{rs42}), then the potential of Eq.\eqref{rs71} will be modified as follows and a new boundary condition at $z=L_z$ will be added as well:
\eq$\label{rs73}
V_{new}(z)=\frac{15k^2}{4(k|z|+1)^2}-\frac{3k[\del(z)-\del(z-L_z)]}{k|z|+1}$
The boundary condition at $z=L_z$ can be deduced with the use of equations \eqref{rs70} and \eqref{rs73}.
$$\int_{L_z^-}^{L_z^+}dz[-\psi_n''(z)+V_{new}(z)\psi_n(z)]=\int_{L_z^-}^{L_z^+}dz\ m_n^2\ \psi_n(z)\Ra$$
$$-\psi_n'(L_z^+)+\psi_n(L_z^-)+\frac{3k\psi_n(L_z)}{k|L_z|+1}=0\xRightarrow{\psi_n(L_z^+)=-\psi_n(L_z^-)}$$
\eq$\label{rs74}
\boxed{\psi_n'(L_z)=-\frac{3k\psi_n(L_z)}{2(kL_z+1)}}$
\begin{figure}[htp]
\centering
\includegraphics[width=5.9in]{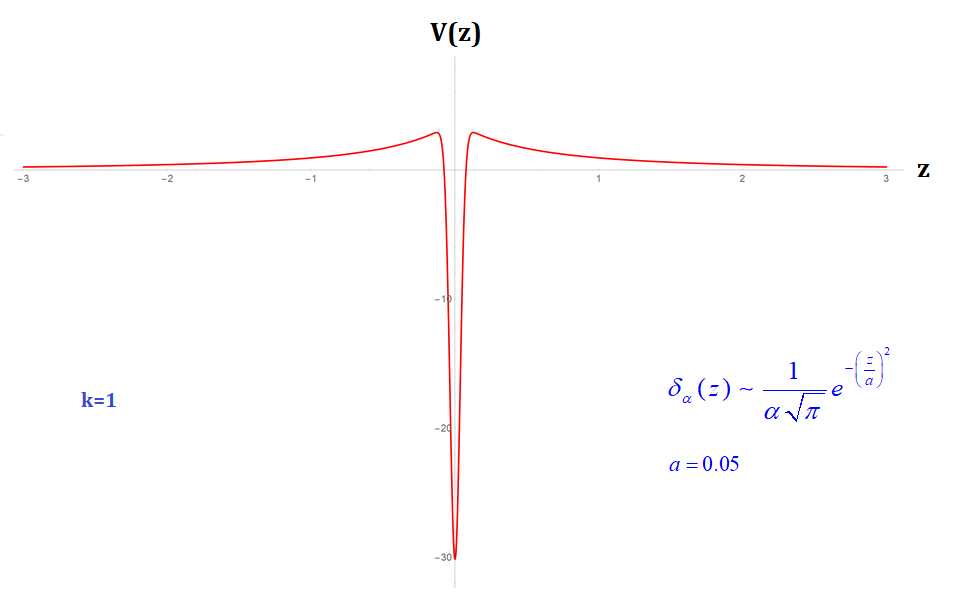}
\caption{Qualitative graph of the function $V(z)$.}
\end{figure}

\subsection{Kaluza-Klein Modes}
\par \vspace{1em}The \textbf{zero-mode} $\psi_0(z)$ for $m_0=0$ can be easily evaluated from Eq.\eqref{rs70}. Thus, it is
$$-\psi''_0(z)+\left[\frac{3}{2}A''(z)+\frac{9}{4}A'^2(z)\right]\psi_0(z)=0\Ra\psi_0(z)=N_0\ e^{\frac{3}{2}A(z)}\xRightarrow{\eqref{rs47}}$$
\eq$\label{rs75}
\psi_0(z)=N_0(k|z|+1)^{-3/2}$
where $N_0$ is a normalization constant which is going to be evaluated followingly.
$$\int_{-L_z}^{L_z}dz\ |\psi_0(z)|^2=1\Ra (N_0)^2\int^{L_z}_{-L_z}dz(k|z|+1)^{-3}=2(N_0)^2\int_0^{L_z}dz(kz+1)^{-3}=1\Ra$$
\eq$\label{rs75a}
(N_0)^2\left[-\frac{1}{k(kL_z+1)^2}+\frac{1}{k}\right]=1\xRightarrow{kL_z=e^{kL}-1\gg 1}(N_0)^2\frac{1}{k}=1\Ra N_0=\sqrt{k}$
From equations \eqref{rs75} and \eqref{rs75a} we are led to
\eq$\label{75b}
\boxed{\psi_0(z)=\frac{1}{k}\left(|z|+\frac{1}{k}\right)^{-3/2}}$
\par The figure above depicts the plot of the potential $V(z)$ as it is given by equation \eqref{rs71}, where it was used the approximation $\del(z)\sim(1/a\sqrt{\pi})e^{-(z/a)^2}$ with $a=0.05$. This particular approximation helps us to visualize better the behaviour of the potential.
\par The boundary conditions are both satisfied. The figure in the next page depicts the plot of the graviton zero-mode at the region of the 3-brane which is located at $z=0$.
\begin{figure}[htp]
\centering
\includegraphics[width=4.7in]{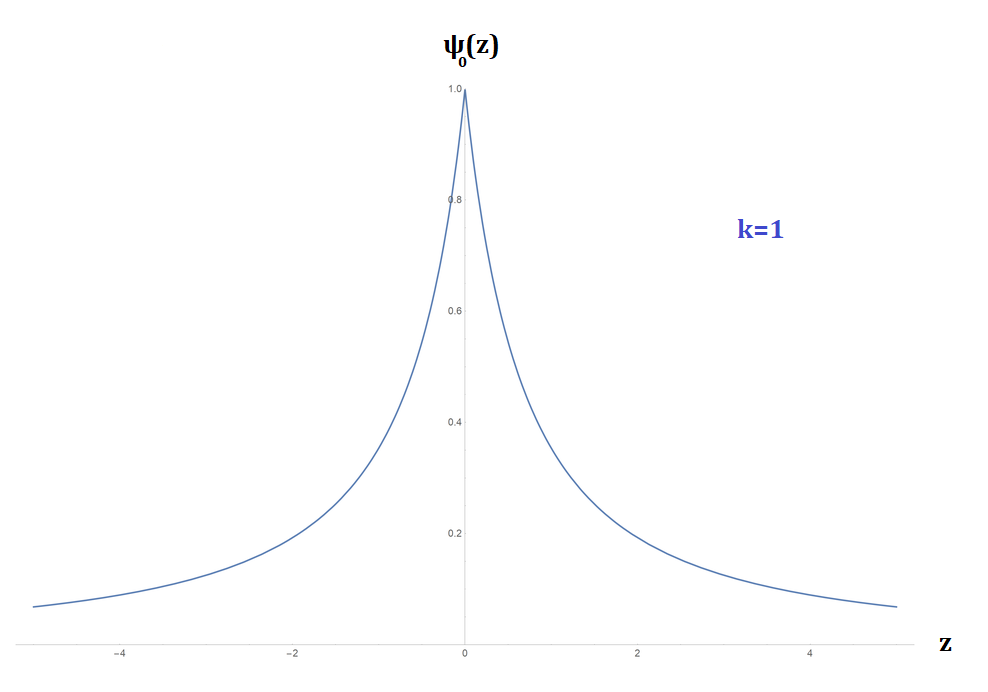}
\caption{Qualitative graph of the zero-mode around the 3-brane at $z=0$.}
\end{figure}
\par The \textbf{Kaluza-Klein (KK) modes} for $n>0$ can be provided from the general equation \eqref{rs70}.
\eq$\label{rs76}
\psi_n''(z)+\left[m_n^2-\frac{15k^2}{4(k|z|+1)^2}\right]\psi_n(z)=0$
while for $z=0$ and $z=L_z$ equations \eqref{rs72} and \eqref{rs74} should be satisfied respectively. The above differential equation has the following general solution:
\eq$\label{rs77}
\psi_n(z)=\left(|z|+\frac{1}{k}\right)^{1/2}\left\{a_n J_2\left[m_n\left(|z|+\frac{1}{k}\right)\right]+b_nY_2\left[m_n\left(|z|+\frac{1}{k}\right)\right]\right\}$
where $a_n$, $b_n$ are constant coefficients and $J_2(z)$, $Y_2(z)$ are the Bessel functions of the first and second kind respectively. The derivative of the function $\psi_n(z)$ is given by
\begin{align}\label{rs78}
\psi_n'(z)=&m_n\left(|z|+\frac{1}{k}\right)^{1/2}\left\{a_n J_1\left[m_n\left(|z|+\frac{1}{k}\right)\right]+b_nY_1\left[m_n\left(|z|+\frac{1}{k}\right)\right]\right\}\nonum\\
&-\frac{3}{2}\left(|z|+\frac{1}{k}\right)^{-1/2}\left\{a_n J_2\left[m_n\left(|z|+\frac{1}{k}\right)\right]+b_nY_2\left[m_n\left(|z|+\frac{1}{k}\right)\right]\right\}
\end{align}
We can relate the constants $a_n$ and $b_n$ by simply applying the boundary condition at $z=0$. Thus, from equations \eqref{rs72}, \eqref{rs77} and \eqref{rs78} we get:
\eq$\label{rs78a}
a_n=-b_n\frac{Y_1\left(\frac{m_n}{k}\right)}{J_1\left(\frac{m_n}{k}\right)}$
Hence, it is:
\eq$\label{rs78b}
\boxed{\psi_n(z)=N_n\left(|z|+\frac{1}{k}\right)^{1/2}\left\{-Y_1\left(\frac{m_n}{k}\right)J_2\left[m_n\left(|z|+\frac{1}{k}\right)\right]+J_1\left(\frac{m_n}{k}\right)Y_2\left[m_n\left(|z|+\frac{1}{k}\right)\right]\right\}}$
where we defined	
\eq$\label{rs78c}N_n=\frac{b_n}{J_1\left(\frac{m_n}{k}\right)}$
\par Since the massless graviton mode
is accompanied by a tower of massive Kaluza-Klein states, one may wonder about the gravitational potential that an observer on the visible brane will feel. In order to answer this, we will now consider each RS model separately.\vspace{1em}
\begin{center}
\underline{\textbf{The limit $\mathbold{m_n/k\gg 1}$}}\vspace{1em}
\end{center}
\par Combining equations \eqref{rs72}, \eqref{rs74}, \eqref{rs77} and \eqref{rs78} we get
\begin{gather}
\label{rs79}
a_n J_1\left(\frac{m_n}{k}\right)+b_nY_1\left(\frac{m_n}{k}\right)=0\\ \nonum\\
\label{rs80}
a_n J_1\left[m_n\left(L_z+\frac{1}{k}\right)\right]+b_nY_1\left[m_n\left(L_z+\frac{1}{k}\right)\right]=0
\end{gather}
It is well-known that a homogeneous system of equations in order to have a non-trivial solution it is necessary to have a vanishing determinant. Thus, for the previous system of equations we demand that\\
$$J_1\left(\frac{m_n}{k}\right)Y_1\left[m_n\left(L_z+\frac{1}{k}\right)\right]-J_1\left[m_n\left(L_z+\frac{1}{k}\right)\right]Y_1\left(\frac{m_n}{k}\right)=0\xRightarrow{kL_z=e^{kL}-1}$$
\eq$\label{rs81}
J_1\left(\frac{m_n}{k}\right)Y_1\left(\frac{m_n}{k}e^{kL}\right)-J_1\left(\frac{m_n}{k}e^{kL}\right)Y_1\left(\frac{m_n}{k}\right)=0$

\newpage
In this limit, the large value of the quantity $m_n/k$ results to the following expressions for the Bessel functions:
\eq$\label{rs94a}
\left\{\begin{array}{cc}
J_1\left(\frac{m_n}{k}\right)\sim\sqrt{\frac{2k}{\pi m_n}}\cos\left(\frac{m_n}{k}-\frac{3\pi}{4}\right),&\hspace{1em}
J_1\left(\frac{m_n}{k}e^{kL}\right)\sim\sqrt{\frac{2k}{\pi m_n}e^{-kL}}\cos\left(\frac{m_n}{k}e^{kL}-\frac{3\pi}{4}\right)\\ \\
Y_1\left(\frac{m_n}{k}\right)\sim\sqrt{\frac{2k}{\pi m_n}}\sin\left(\frac{m_n}{k}-\frac{3\pi}{4}\right),& \hspace{1em}
Y_1\left(\frac{m_n}{k}e^{kL}\right)\sim\sqrt{\frac{2k}{\pi m_n}e^{-kL}}\sin\left(\frac{m_n}{k}e^{kL}-\frac{3\pi}{4}\right)\end{array}\right\}$
Therefore, Eq.\eqref{rs81} takes the form
$$\cos\left(\frac{m_n}{k}-\frac{3\pi}{4}\right)\sin\left(\frac{m_n}{k}e^{kL}-\frac{3\pi}{4}\right)-\cos\left(\frac{m_n}{k}e^{kL}-\frac{3\pi}{4}\right)\sin\left(\frac{m_n}{k}-\frac{3\pi}{4}\right)=0\Ra$$
$$\sin\left[\frac{m_n}{k}\left(e^{kL}-1\right)\right]=0\xRightarrow{kL_z=e^{kL}-1}\sin\left(m_n L_z\right)=0\Ra$$
\eq$\label{rs94b}
\boxed{m_n=\frac{s\pi}{L_z},\hspace{1em}s=1,2,\ldots}$
Correspondingly, Eq.\eqref{rs78b} is given by
\begin{align}
\psi_n(z)\sim N_n\sqrt{\frac{2k}{\pi m_n}}\sqrt{|z|+\frac{1}{k}}&\left\{-\sin\left(\frac{m_n}{k}-\frac{3\pi}{4}\right)J_2\left[m_n\left(|z|+\frac{1}{k}\right)\right]\right.\nonum\\
&\hspace{0.5em}\left.+\cos\left(\frac{m_n}{k}-\frac{3\pi}{4}\right)Y_2\left[m_n\left(|z|+\frac{1}{k}\right)\right]\right\}\xRightarrow{m_n(|z|+1/k)\gg 1}\nonum
\end{align}
\begin{align}
\psi_n(z)\sim N_n\frac{2\sqrt{k}}{\pi m_n}&\left\{-\sin\left(\frac{m_n}{k}-\frac{3\pi}{4}\right)\cos\left[m_n\left(|z|+\frac{1}{k}\right)-\frac{5\pi}{4}\right]\right.\nonum\\
&\hspace{0.5em}\left.+\cos\left(\frac{m_n}{k}-\frac{3\pi}{4}\right)\sin\left[m_n\left(|z|+\frac{1}{k}\right)-\frac{5\pi}{4}\right]\right\}\Ra\nonum
\end{align}
\eq$\label{rs94c}
\psi_n(z)\sim N_n\frac{2\sqrt{k}}{\pi m_n}\sin\left(m_n|z|-\frac{\pi}{2}\right)$
where we have used the following expressions for the Bessel function at the limit $m_n(|z|+1/k)\gg 1$.
\eq$\label{rs91}
\left(|z|+\frac{1}{k}\right)^{1/2}J_2\left[m_n\left(|z|+\frac{1}{k}\right)\right]\sim \sqrt{\frac{2}{\pi m_n}}\cos\left[m_n\left(|z|+\frac{1}{k}\right)-\frac{5\pi}{4}\right]$
\eq$\label{rs92}
\left(|z|+\frac{1}{k}\right)^{1/2}Y_2\left[m_n\left(|z|+\frac{1}{k}\right)\right]\sim \sqrt{\frac{2}{\pi m_n}}\sin\left[m_n\left(|z|+\frac{1}{k}\right)-\frac{5\pi}{4}\right]$
It is now trivial to evaluate the normalization constant $N_n$.
$$\int_{-L_z}^{L_z}dz\ \psi_n^2(z)=1\Ra N_n^2\frac{4k}{\pi^2m_n^2}\int_{-L_z}^{L_z}dz\ \sin^2\left(m_n|z|-\frac{\pi}{2}\right)=1\Ra$$
\eq$\label{rs94d}
N_n^2\frac{4k}{\pi^2m_n^2}\ 2\ \underbrace{\int_{0}^{L_z}dz\ \sin^2\left(m_nz-\frac{\pi}{2}\right)}_{L_z/2\ (for\ L_z\gg 1)}=1\Ra \boxed{N_n=\frac{\pi m_n}{2\sqrt{kL_z}}}$\\
\eq$\label{rs94e}
\eqref{rs94c}\xRightarrow{\eqref{rs94d}}\boxed{\psi_n(z)\sim \frac{1}{\sqrt{L_z}}\sin\left(m_n|z|-\frac{\pi}{2}\right)}$
For the higher graviton modes to be very heavy, as assumed above, we should have a very small inter-brane distance according to Eq. \eqref{rs94b}. Then, this case applies to the RS1 model. Each graviton mode will contribute to the gravitational potential on the brane through a Yukawa-like potential, thus
$$U(r) \propto \sum_n \frac{G_{(4)} m_1 m_2 e^{-m_n r}}{r}$$
But since all the KK masses are very heavy, it is only the zero graviton that will create a $1/r$ potential while the contributions of the higher modes will be very much suppressed.

\begin{center}
\underline{\textbf{The limit $\mathbold{m_n/k\ll 1}$}}\vspace{1em}
\end{center}
\par At the approximation of $m_n/k\ll 1$, it is
\eq$\label{rs81a}
J_1\left(\frac{m_n}{k}\right)\sim\frac{m_n}{2k},\hspace{3em}
Y_1\left(\frac{m_n}{k}\right)\sim -\frac{2k}{\pi m_n}$
Hence, we can safely state that $-Y_1\left(\frac{m_n}{k}\right)\gg J_1\left(\frac{m_n}{k}\right)$. Subsequently, Eq.\eqref{rs81} results to the following simple requirement in order to have a vanishing determinant:
\eq$\label{rs82}
J_1\left(\frac{m_n}{k}\ e^{kL}\right)=0\Ra\boxed{m_n=k e^{-kL} R^{J_1}_n\simeq\frac{R^{J_1}_n}{L_z}}$
where $R^{J_1}_n$ are the roots of the function $J_1(x)$, namely $J_1(R^{J_1}_n)=0$. The masses $m_n$ of the massive graviton modes constitute the graviton spectrum (or Kaluza Klein spectrum). The substitution of the approximate relations from equation \eqref{rs81a} into \eqref{rs78b} leads to
$$\psi_n(z)\sim N_n\left(|z|+\frac{1}{k}\right)^{1/2}\left\{\frac{2k}{\pi m_n}J_2\left[m_n\left(|z|+\frac{1}{k}\right)\right]+\frac{m_n}{2k}Y_2\left[m_n\left(|z|+\frac{1}{k}\right)\right]\right\}\xRightarrow{N_n\ra N_n\frac{m_n}{2k}}$$
\eq$\label{rs85}
\boxed{\psi_n(z)\sim N_n\left(|z|+\frac{1}{k}\right)^{1/2}\left\{\frac{4k^2}{\pi m_n^2}J_2\left[m_n\left(|z|+\frac{1}{k}\right)\right]+Y_2\left[m_n\left(|z|+\frac{1}{k}\right)\right]\right\}}$\\
The normalization constant $N_n$ can be determined by the integral
$$\int_{-L_z}^{L_z}dz\ \psi_n^2(z)=1$$
The fact that $L_z\ra \infty$ allows us to perform the calculation of $N_n$ in the approximation of $m_n|z|\gg 1$. Using also the fact that $\frac{m_n}{k}\ll 1\Ra\frac{k}{m_n}\gg 1$, we can ignore the second term of Eq.\eqref{rs85} as negligible. Thus, the substitution of Eq.\eqref{rs91} into Eq.\eqref{rs85} leads to
\eq$\label{rs93}
\psi_n(z)\sim N_n\frac{4k^2}{\pi m_n^2}\sqrt{\frac{2}{\pi m_n}}\cos\left(m_n|z|-\frac{5\pi}{4}\right)$
Thus, from Eq.\eqref{rs93} we have:
$$\int_{-L_z}^{L_z}dz\ \psi_n^2(z)=1\Ra N^2_n\ \frac{32k^4}{\pi^3m^5_n}\int^{L_z}_{-L_z}\cos^2\left(m_n|z|-\frac{5\pi}{4}\right)=1\Ra$$
\eq$\label{rs94}
N^2_n\ \frac{32k^4}{\pi^3m^5_n}\ 2\  \underbrace{\int^{L_z}_{0}\cos^2\left(m_n\  z-\frac{5\pi}{4}\right)}_{L_z/2\ \ (for\ L_z\gg 1)}=1\Ra \boxed{N_n=\frac{\pi^{3/2}m_n^{5/2}}{4k^2\sqrt{2L_z}}}$
Consequently, the graviton modes $\psi_n(z)$ for $n>0$ and $m_n|z|\gg 1$ are given by 
\eq$
\boxed{\psi_n(z)=\frac{1}{\sqrt{L_z}}\ \cos\left(m_n|z|-\frac{5\pi}{4}\right)}$\vspace{1em}
\par In this case, the KK masses are assumed to be small and this can be realised only for a large inter-brane distance. Therefore, this analysis applies to the RS1 model, where resolving the hierarchy problem is not an objective any more, or to the RS2 model. Especially, in the latter case where $L\ra \infty$, even the higher graviton modes become massless. It is thus even more important to verify if the Newtonian limit of gravity is indeed recovered on the visible brane. Here, the sum over the KK graviton states changes to an integral that must be carefully evaluated - we refer the interested reader to \cite{PhysRevLett.84.2778,1126-6708-2000-03-023} for more details on this. At the end, we present the result from \cite{PhysRevLett.84.2778} which is related to the gravitational potential that is generated by a massive object with mass $M$ on the 3-brane at $z=0$. The energy-momentum tensor for a point-mass at rest on the brane at $\vec{r}=0$ is the following:
$$T_{\mu\nu}=M\del^{(3)}(\vec{r})\ \del^0_\mu\ \del^0_\nu$$
For $r=|\vec{r}|\gg 1/k$ we get:
\eq$\label{rs96}
h_{\mu\nu}=\frac{2G_{(5)}kM}{r}\left[\left(1+\frac{1}{3k^2r^2}\right)\eta_{\mu\nu}+\left(2+\frac{1}{k^2r^2}\right)\del^0_\mu \del^0_\nu \right]$
Consequently, from equations \eqref{rs36} and \eqref{rs96} we obtain
\begin{gather}\label{rs97}
h_{00}=\frac{2G_{(4)}M}{r}\left(1+\frac{2}{3k^2r^2}\right)\\ \nonum\\
\label{rs98}
h_{ij}=\frac{2G_{(4)}M}{r}\left(1+\frac{1}{3k^2r^2}\right)\eta_{ij}
\end{gather}
The gravitational potential $V(r)$ is given by $h_{00}$ as follows:
\eq$\label{rs99}
\boxed{V(r)=\frac{1}{2}h_{00}=\frac{G_{(4)}M}{r}\left(1+\frac{2}{3k^2r^2}\right)\xrightarrow{kr\gg 1}V(r)\sim \frac{G_{(4)}M}{r}}$
Obviously, the desirable behaviour of an effective 4-dimensional gravity at the brane is produced by the model as $kr\gg 1$. On the other hand, if we consider the case of $kr\ll 1$, then the gravitational potential is proportional to $1/r^2$, i.e. $V(r)\propto 1/r^2$. This simply means that at short distances -compared to the scale $k$ which is also associated with the AdS curvature- gravity becomes 5-dimensional. The possibility to observe deviations from the Newton's law of gravitation depends entirely on the value of $k$.

\newpage
\blankpage

%% file: Chapters/Chapter_Localization.tex
\chapter{Localization of a 5-Dimensional Brane-World Black Hole}

\par The preceding analysis of RS models made clear that there is indeed a possibility to live at the boundaries of a higher dimensional cosmos (specifically a five-dimensional one) and the effective gravitational interaction between massive objects on our four-dimensional brane can be the same (in some limit) as the gravitational interaction which is provided by the established and well-tested gravitational theory of General Relativity. Hence, it is absolutely natural to wonder about the behaviour of 5-dimensional black holes in the context of RS models. 
\par The first attempt to find a black hole solution in the aforementioned scenario was in \cite{PhysRevD.61.065007}. The line-element which was considered is of the following form:
\eq$\label{loc1}
ds^2=e^{2A(y)}\left[-\left(1-\frac{2M}{r}\right)dt^2+\left(1-\frac{2M}{r}\right)^{-1}dr^2+r^2\left(d\theta^2+\sin^2\theta\ d\varphi^2\right)\right]+dy^2$
where $M$ is a constant quantity and represents the mass of the black hole. As stated in the previous chapter, the coordinate $y$ is used for the extra dimension and $A(y)$ denotes the warp factor. For $A(y)=-k|y|$ we get the warp factor of the RS model. It is quite obvious that for $y=0$ the induced 4-dimensional line-element is identified with the well-known Schwarzschild spacetime geometry. However, a 5-dimensional observer will not be able to associate the complete 5-dimensional line-element of Eq.\eqref{loc1} with a regular black hole. The previous statement can be easily verified by the Kretschmann scalar $K\equiv R_{MNKL}R^{MNKL}$ that emanates from  the line-element \eqref{loc1}.
\eq$\label{loc1a}
K\equiv R_{MNKL}R^{MNKL}=8 \left[2 A''(y)^2+5 A'(y)^4+4 A'(y)^2 A''(y)+\frac{6 M^2 e^{-4 A(y)}}{r^6}\right]$
The last term of Eq.\eqref{loc1a} reveals the existence of a singularity at $r=0$ which extends from $-\infty$ to $+\infty$ along the extra dimension $y$. Hence, the line-element of Eq.\eqref{loc1} generates a \emph{black string} rather than a black hole. Of course, for any descending function $A(y)$ the last term of Eq.\eqref{loc1a} becomes even more problematic, because when $y\ra \infty$ the quantity $e^{-4A(y)}$ goes to infinity as well. This particular behaviour is in complete contrast to the purpose of the RS model, which intends to keep gravity localized near the brane at $y=0$.
\par There is a plethora of attempts in the literature that try to derive a 5-dimensional black hole which is localized close to the brane (we mention some of them: \cite{DADHICH20001,PhysRevD.66.104028,PhysRevD.65.084010,PhysRevD.68.124001,PhysRevD.65.084040,PhysRevD.70.064007,PhysRevD.73.064014,Creek:2006je,CASADIO2012251,PhysRevD.90.047902,Kuerten:2016cho,0264-9381-33-1-015003}. Additionally, there are numerical solutions which describe both mini and large brane-world black holes in the context of RS models and beyond (i.e. 6-dimensional black holes) \cite{PhysRevD.68.024035,PhysRevD.69.104019,PhysRevLett.107.081101,Abdolrahimi:2012pb}. However, an analytical solution which can be written in closed form has not been found, thus, the investigation of such a solution is still incomplete. Finally, it is important to refer that in \cite{Emparan:1999wa,Emparan:1999fd,Anber:2008qu} black hole solutions were found which are localized on a 2-brane embedded in a (3+1)-dimensional bulk. Therefore, the motivation to search for a corresponding solution for the higher dimensional problem (3-brane in a (4+1)-dimensional bulk) is completely justified.
\par This Chapter and the rest of the thesis is entirely focused on the examination of the existence of a localized black hole solution in the context of a brane-world model which is similar to the RS2 model. The word ``similar" was used because the warp factor $A(y)$ is allowed to differ from the expression $-k|y|$ (but it is necessary to be a descending function, as it is in the RS models) in order to increase the chances of finding a localized 5-dimensional black hole solution. In the following section we present the geometrical background that is going to be used in the framework of this thesis.

\section{The Geometry of the 5-Dimensional Spacetime}
\par We assume that the 4-dimensional part of the general 5-dimensional geometrical background is a Vaidya metric\footnote{See \hyperref[vaid]{Appendix D} for more information about Vaidya metric.}. We also consider that mass is a function of three variables $m=m(v,r,y)$ (motivated by \cite{0264-9381-33-1-015003}). The non-trivial dependence of the mass from the extra dimension $y$ is needed in order to counter the anomalous effect of the quantity $e^{-4A(y)}$ in the last term of equation \eqref{loc1a}. An appropriate mass function that decreases faster than $e^{-4A(y)}$ increases and has also the necessary r-dependence in order to describe a black hole, would solve the localization problem. The drawback of this non-trivial assumption of the mass parameter is that it demands the existence of a bulk matter distribution in order to get consistent field equations. The bulk matter distribution that it is going to be considered in the context of this thesis will be presented later in this Chapter. Thus, the line element is of the form
\eq$\label{loc.1}
ds^2=e^{2A(y)}\left[-\left(1-\frac{2m(v,r,y)}{r}\right)dv^2+2dvdr+r^2(d\theta^2+\sin^2\theta\  d\pphi^2)\right]+dy^2$
It is easily deducible that the covariant components of the metric tensor in matrix form are:
\eq$\label{loc.2}
(g_{MN})=\left(
\begin{array}{ccccc}
 -e^{2 A(y)} \left(1-\frac{2 m(u,r,y)}{r}\right) & e^{2 A(y)} & 0 & 0 & 0 \\
 e^{2 A(y)} & 0 & 0 & 0 & 0 \\
 0 & 0 & e^{2 A(y)} r^2 & 0 & 0 \\
 0 & 0 & 0 & e^{2 A(y)} r^2 \sin ^2(\theta ) & 0 \\
 0 & 0 & 0 & 0 & 1 \\
\end{array}
\right)$
Thus, the contravariant components of the metric tensor are of the following form:
\eq$\label{loc.3}
(g^{MN})=\left(
\begin{array}{ccccc}
 0 & e^{-2 A(y)} & 0 & 0 & 0 \\
 e^{-2 A(y)} & e^{-2 A(y)}\left(1-\frac{2 m(u,r,y)}{r}\right) & 0 & 0 & 0 \\
 0 & 0 & \frac{e^{-2 A(y)}}{r^2} & 0 & 0 \\
 0 & 0 & 0 & \frac{e^{-2 A(y)}}{r^2\sin^2\theta} & 0 \\
 0 & 0 & 0 & 0 & 1 \\
\end{array}
\right)$ 
\par For the metric ansatz \eqref{loc.1}, the quantities that are invariant and also contain the complete information of the 5-dimensional curvature, are $R=g_{MN}R^{MN}$, $R_{MN}R^{MN}$ and $R_{ABCD}R^{ABCD}$. As it is mentioned earlier, an acceptable solution for the mass function $m=m(v,r,y)$ should have an appropriate dependence on the y-coordinate in order to expunge any singularities in the bulk. Hence, the desirable solution should yield to a regular Anti-de Sitter space-time at a finite distance from the brane and a 5-dimensional black hole localized on the brane, otherwise the solution is rejected.
\par The Christoffel symbols can be evaluated with the use of equations \eqref{loc.2}, \eqref{loc.3} and the following relation:
\eq$\label{loc3a}
\Gam^K{}_{MN}=\frac{1}{2}g^{KL}\left(g_{ML,N}+g_{NL,M}-g_{MN,L}\right)$
In the table that follows are depicted the non-zero Christoffel symbols.
\begin{empheq}[box=\mymath]{equation}\label{loc3b}
\begin{array}{ccc}
\Gam^0{}_{00}=\frac{m-r \pa_r m}{r^2} & \Gam^0{}_{04}=\Gam^0{}_{40}=A' & \Gam^0{}_{22}=-r \\ \\ \Gam^0{}_{33}=-r \sin ^2\theta & \Gam^1{}_{00}=\frac{r^2 \pa_v m-(r-2 m) \left(r \pa_r m-m\right)}{r^3} & \Gam^1{}_{01}=\Gam^1{}_{10}=\frac{r \pa_rm-m}{r^2}\\ \\
\Gam^1{}_{04}=\Gam^1{}_{40}=\frac{\pa_ym}{r} & \Gam^1{}_{14}=\Gam^1{}_{41}=A' & \Gam^1{}_{22}=2 m-r \\ \\
\Gam^1{}_{33}=(2 m-r)\sin ^2\theta & \Gam^2{}_{12}=\Gam^2{}_{21}=\frac{1}{r} & \Gam^2{}_{24}=\Gam^2{}_{42}=A'\\ \\
\Gam^2{}_{33}=-\sin \theta \cos \theta & \Gam^3{}_{13}=\Gam^3{}_{31}=\frac{1}{r} & \Gam^3{}_{23}=\Gam^3{}_{32}=\cot \theta \\ \\
\Gam^3{}_{34}=\Gam^3{}_{43}=A' & \Gam^4{}_{00}=\frac{e^{2 A} \left(A' (r-2 m)-\pa_ym\right)}{r} & \Gam^4{}_{01}=\Gam^4{}_{10}=-e^{2 A} A'\\ \\
\Gam^4{}_{22}=-r^2 e^{2 A} A' & \Gam^4{}_{33}=-r^2A' e^{2 A} \sin ^2\theta & \end{array}
\end{empheq}
Subsequently, we present the 5-dimensional curvature invariant quantities.
\begin{empheq}[box=\mymath]{equation}\label{loc.4}
R=R_{MN}g^{MN}=-20 A'^2-8 A''+\frac{2 e^{-2 A}}{r} \left(\pa_r^2m+\frac{2\pa_rm}{r}\right)
\end{empheq}
\begin{empheq}[box=\mymath]{equation}\label{loc.5}
\begin{array}{l}R_{MN}R^{MN}=80A'^2+64A'^2 A''+20A''^2-\displaystyle{\frac{4 e^{-2 A}}{r}}\left(\pa_r^2m+\frac{2\pa_rm}{r}\right)(4A'^2+A'')\\
\hspace{5em}+\displaystyle{\frac{2e^{-4A}}{r^2}}\left[(\pa_r^2m)^2+\frac{4 (\pa_rm)^2}{r^2}\right]
\end{array}
\end{empheq}
\begin{empheq}[box=\mymath]{equation}\label{loc.6}
\begin{array}{c}\scriptstyle{R^{ABCD}R_{ABCD}}=\scriptstyle{40A'^4+32A'^2A''+16A''^2+}\scriptstyle{\frac{48e^{-4A}m^2}{r^6}}-\frac{8e^{-2A}A'^2}{r}\left(\pa_r^2m+\frac{2\pa_rm}{r}\right)\\ \\
\hspace{6em}\scriptstyle{+\frac{4e^{-4A}}{r^2}}\left[(\pa_r^2m)^2+\frac{4m}{r^2}\left(\pa_r^2m-\frac{4\pa_rm}{r}\right)-\frac{4\pa_rm\pa_r^2m}{r}+\frac{8(\pa_r m)^2}{r}\right]
\end{array}
\end{empheq}\vspace{0.5em}
\par The non-zero covariant components of the Einstein tensor $G_{MN}$ are:\vspace{0.5em}
\begin{empheq}[box=\mymath]{equation}\label{loc.8}
\begin{array}{c}
G_{00}=\frac{-e^{2 A} \left[3 (r-2 m) \left(2 A'^2+A''\right)+4 A' \pa_ym+\pa_y^2m\right] r^2+2 r\pa_vm +2(r-2 m)\pa_r m}{r^3}\\ \\
G_{01}=G_{10}=\displaystyle{3 e^{2 A} \left(2 A'^2+A''\right)-\frac{2 \pa_rm}{r^2}}\\ \\
G_{04}=G_{40}=\displaystyle{\frac{\pa_ym+r\pa_r\pa_y m}{r^2}} \\ \\
G_{22}=\displaystyle{r \left[3 e^{2 A} r \left(2 A'^2+A''\right)-\pa_r^2m\right]} \\ \\
G_{33}=\displaystyle{r \sin ^2\theta \left[3 e^{2 A} r \left(2 A'^2+A''\right)-\pa_r^2m\right]} \\ \\
G_{44}=\displaystyle{6 A'^2-\frac{e^{-2 A} \left(2 \pa_rm+r\pa_r^2 m\right)}{r^2}}
\end{array}
\end{empheq}
while the mixed components of the Einstein tensor $G^M{}_N=g^{MK}G_{KN}$ are the following:
\begin{empheq}[box=\mymath]{equation}\label{loc.7}
\begin{array}{c}
G^0{}_0=\displaystyle{6 A'^2+3 A''-\frac{2 e^{-2 A} \pa_rm}{r^2}} \\ \\
G^1{}_0=\displaystyle{-\frac{4 r A' \pa_ym+r \pa_y^2m-2 e^{-2 A} \pa_vm}{r^2}}\\ \\
G^1{}_1=\displaystyle{6 A'^2+3 A''-\frac{2 e^{-2 A}\pa_r m}{r^2}} \\ \\
G^1{}_4=\displaystyle{\frac{e^{-2 A} \left(\pa_ym+r \pa_r\pa_ym\right)}{r^2}} \\ \\
G^2{}_2=\displaystyle{6 A'^2+3 A''-\frac{e^{-2 A}\pa_r^2 m}{r}}\\ \\
G^3{}_3=\displaystyle{6 A'^2+3 A''-\frac{e^{-2 A} \pa_r^2m}{r}}\\ \\
G^4{}_0=\displaystyle{\frac{\pa_ym+r \pa_r\pa_ym}{r^2}} \\ \\
G^4{}_4=\displaystyle{6 A'^2-\frac{e^{-2 A} \left(2 \pa_rm+r\pa_r^2 m\right)}{r^2}}
\end{array}
\end{empheq}
\par A detailed presentation of all the geometrical quantities i.e. Christoffel symbols, Riemann tensor, Ricci tensor and Einstein tensor, takes place in \hyperref[geom]{Appendix E}. Moreover, in the same Chapter of Appendices, one can find the proper mathematica commands in order to verify the validity of all the aforementioned quantities.
\par In the following section, we introduce the field theory model in the context of which the existence of a viable solution to the localization problem of a 5-dimensional black hole will be examined.

\section{Two Non-Minimally Coupled and Interacting Scalar Fields}

\par We consider the most general case of two interacting scalar fields $\phi,\chi$ which are not minimally coupled to gravity. These scalar fields can freely propagate into the bulk. Subsequently, the action of this model is the following:
\eq$\label{loc.9}
S=\int d^4xdy\sqrt{-\gfv}\left[\frac{f(\phi,\chi)}{2\kfv}R-\frac{1}{2}(\nabla\phi)^2-\frac{1}{2}(\nabla\chi)^2-V(\phi,\chi)-\Lambda_B\right]$
where $\gfv=\det[(g_{MN})]$. The scalar fields $\phi$, $\chi$ are both functions of the variables $(v,r,y)$ and $\Lambda_B$ denotes the cosmological constant of the bulk, it is the same constant as $\Lambda_5$ which was used in Chapter 2. The field equations of this model can be derived by varying the action with respect to the metric tensor and also with respect to the fields $\phi,\chi$. The steps that are going to be followed for the variation of the above action are the same as the steps which are followed in \hyperref[var]{Appendix F}, where it is thoroughly presented the variation of a single scalar field which is non-minimally coupled to gravity. Thus, we have:
\begin{align}\label{loc.10}
\del S=0&=\int d^4xdy\ \del\left(\sqrt{-\gfv}\right)\left[\frac{f(\phi,\chi)}{2\kfv}R-\frac{1}{2}(\nabla\phi)^2-\frac{1}{2}(\nabla\chi)^2-V(\phi,\chi)-\Lambda_B\right]\nonum \\
&\ \ \ \ +\int d^4xdy\sqrt{-\gfv}\left\{\frac{f(\phi,\chi)}{2\kfv}\del R-\frac{1}{2}\del[(\nabla\phi)^2]-\frac{1}{2}\del[(\nabla\chi)^2]\right\}
\end{align}
Moreover, it is
\eq$\label{loc.11}
\del[(\nabla\phi)^2]=\del(\nabla^M\phi\nabla_M\phi)=\del(g^{MN}\nabla_M\phi\nabla_N\phi)=\nabla_M\phi\nabla_N\phi\ \del g^{MN}$
Using equations \eqref{loc.11}, \eqref{C.13}, \eqref{C.28} and \eqref{C.36} into \eqref{loc.10} we get:
\begin{align}
0&=\int d^4xdy\sqrt{-\gfv}\ \del g^{MN}\left\{-\frac{1}{2}g_{MN}\left[\frac{f(\phi,\chi)}{2\kfv}R-\frac{1}{2}(\nabla\phi)^2-\frac{1}{2}(\nabla\chi)^2-V(\phi,\chi)-\Lambda_B\right]\right.\nonum\\
&\ \ \ \ \left.+\frac{1}{2\kfv}f(\phi,\chi)R_{MN}-\frac{1}{2}\nabla_M\nabla_Nf(\phi,\chi)+\frac{1}{2}g_{MN}\square f(\phi,\chi)-\frac{1}{2}\nabla_M\phi\nabla_N\phi-\frac{1}{2}\nabla_M\chi\nabla_N\chi\right\}\Ra\nonum \\ \nonum\\
0&=\frac{f(\phi,\chi)}{\kfv}\left(R_{MN}-\frac{1}{2}g_{MN}R\right)+g_{MN}\left[\frac{(\nabla\phi)^2}{2}+\frac{(\nabla\chi)^2}{2}+V(\phi,\chi)\right]+g_{MN}\Lambda_B-\nabla_M\phi\nabla_N\phi\nonum\\
&\ \ \ \ -\nabla_M\chi\nabla_N\chi-\nabla_M\nabla_Nf(\phi,\chi)+g_{MN}\square f(\phi,\chi)\Ra\nonum
\end{align}\\
\eq$\label{loc.12}
\boxed{f(\phi,\chi)\left(R_{MN}-\frac{1}{2}g_{MN}R\right)=f(\phi,\chi)\ G_{MN}=\kfv\left(T_{MN}-g_{MN}\Lambda_B\right)}$
where
\eq$\label{loc.13}
\boxed{T_{MN}=\nabla_M\phi\nabla_N\phi+\nabla_M\chi\nabla_N\chi-g_{MN}\left[\frac{(\nabla\phi)^2}{2}+\frac{(\nabla\chi)^2}{2}+V\right]+\nabla_M\nabla_Nf-g_{MN}\square f}$
\par The above field equations are more convenient for calculations when they are expressed in terms of mixed tensor components. Hence, we have
\begin{empheq}[box=\mymath]{equation}\label{loc.14}
f(\phi,\chi)\ G^M{}_N=T^M{}_N-\del^M{}_N\Lambda_B
\end{empheq}
where we have absorbed the constant $\kfv$ inside the function $f(\phi,\chi)$ and
\begin{empheq}[box=\mymath]{equation}\label{loc.15}
\begin{array}{l}
T^M{}_N=\nabla^M\phi\nabla_N\phi+\nabla^M\chi\nabla_N\chi-\del^M{}_N\left[\frac{(\nabla\phi)^2}{2}+\frac{(\nabla\chi)^2}{2}+V\right]+\nabla^M\nabla_Nf-\del^M{}_N\square f\\ \\
\hspace{2.55em}=\pa^M\phi\pa_N\phi+\pa^M\chi\pa_N\chi-\del^M{}_N\left[\frac{(\pa\phi)^2}{2}+\frac{(\pa\chi)^2}{2}+V\right]+\nabla^M\nabla_Nf-\del^M{}_N\square f
\end{array}
\end{empheq}
\par The variation of the action \eqref{loc.9} with respect to the fields $\phi, \chi$ (as it is done in the second section of \hyperref[var]{Appendix F}) provides us with two additional equations that should be satisfied in order to have an acceptable solution for the mass function and the scalar fields. In complete analogy to the variation method of section F.2, it is straightforward to derive the following expressions:
\begin{empheq}[box=\mymath]{equation}\label{loc20}
\sqrt{-\gfv}\left(\frac{1}{2}\frac{\pa f}{\pa\phi}R-\frac{\pa V}{\pa\phi}\right)=-\pa_M\left(\sqrt{-\gfv}\ g^{MN}\pa_N\phi\right)
\end{empheq}
\begin{empheq}[box=\mymath]{equation}\label{loc21}
\sqrt{-\gfv}\left(\frac{1}{2}\frac{\pa f}{\pa\chi}R-\frac{\pa V}{\pa\chi}\right)=-\pa_M\left(\sqrt{-\gfv}\ g^{MN}\pa_N\chi\right)
\end{empheq}

\par In order to derive the independent field equations that are resulting from relations \eqref{loc.14} and \eqref{loc.15} we firstly need to evaluate the mixed components of the energy-momentum tensor $T^M{}_N$. In \hyperref[enmom]{Appendix G} there is a comprehensive evaluation of the components of the energy-momentum tensor given by Eq.\eqref{loc.15}, thus, it is redundant to repeat any of these calculations. Subsequently, the non-zero components of the energy-momentum tensor are depicted in their compact form.
\begin{empheq}[box=\mymath]{align}\label{loc.93}
T^0{}_1&=e^{-2A}[(\pa_1\phi)^2+(\pa_1\chi)^2+\pa_1^2f]\nonum\\ \nonum\\
T^1{}_0&=e^{-2A}\left[(\pa_0\phi)^2+(\pa_0\chi)^2+\pa_0^2f+\left(1-\frac{2m}{r}\right)(\pa_1\pa_0\phi+\pa_1\chi\pa_0\chi+\pa_1\pa_0f)\right.\nonum\\
&\hspace{5em}\left. +\pa_1\left(\frac{m}{r}\right)\pa_0f-\pa_0\left(\frac{m}{r}\right)\pa_1f+e^{2A}\pa_4\left(\frac{m}{r}\right)\pa_4f\right]\nonum\\ \nonum\\
T^4{}_0&=\pa_4\phi\pa_0\phi+\pa_4\chi\pa_0\chi+\pa_4\pa_0f-A'\pa_0f-\frac{\pa_4m}{r}\pa_1f\nonum\\ \nonum\\
T^0{}_4&=e^{-2A}(\pa_1\phi\pa_4\phi+\pa_1\chi\pa_4\chi+\pa_1\pa_4f-A'\pa_1f)\nonum\\ \nonum\\
T^4{}_1&=e^{2A}T^0{}_4\nonum\\ \nonum\\
T^1{}_4&=e^{-2A}T^4{}_0+\left(1-\frac{2m}{r}\right)T^0{}_4\nonum\\ \nonum\\
T^0{}_0&=e^{-2A}\left[\pa_1\phi\pa_0\phi+\pa_1\chi\pa_0\chi+\pa_1\pa_0f-\pa_1\left(\frac{m}{r}\right)\pa_1f\right]+A'\pa_4f+\lagr-\square f\\ \nonum\\
T^1{}_1&=T^0{}_0+\left(1-\frac{2m}{r}\right)T^0{}_1\nonum\\ \nonum\\
T^2{}_2&=\frac{e^{-2A}}{r}\left[\pa_0f+\left(1-\frac{2m}{r}\right)\pa_1f\right]+A'\pa_4f+\lagr-\square f\nonum\\ \nonum\\
T^3{}_3&=T^2{}_2\nonum\\ \nonum\\
T^4{}_4&=(\pa_4\phi)^2+(\pa_4\chi)^2+\pa^2_4f+\lagr-\square f\nonum\\ \nonum\\
\lagr&=-\frac{e^{-2A}}{2}\left\{2(\pa_1\phi\pa_0\phi+\pa_1\chi\pa_0\chi)+\left(1-\frac{2m}{r}\right)[(\pa_1\phi)^2+(\pa_1\chi)^2]\right\}\nonum\\
&\hspace{5em}-\frac{1}{2}[(\pa_4\phi)^2+(\pa_4\chi)^2]-V(\phi,\chi)\nonum\\ \nonum\\
\square f&=e^{-2A}\pa_0\pa_1f+\frac{e^{-2A}}{r^2}\pa_1\left[r^2\pa_0f+r^2\left(1-\frac{2m}{r}\right)\pa_1f\right]+e^{-4A}\pa_4\left(e^{4A}\pa_4f\right)\nonum
\end{empheq}

\newpage
\section{The Field Equations of the Theory}
\par Having in our disposal both the components of the Einstein tensor from equation \eqref{loc.7} and the components of the energy-momentum tensor given by equation \eqref{loc.93}, the independent field equations can immediately ensue with the use of equation \eqref{loc.14}. It is clear that the independent field equations are the following.
\begin{center}
\underline{\large{Equation $(^0{}_1)$:}}\\
\end{center}
\eq$\label{loc.95}
\eqref{loc.14}\xRightarrow[\eqref{loc.93}]{\eqref{loc.7}}(\pa_1\phi)^2+(\pa_1\chi)^2+\pa_1^2f=0$
\par Substituting $\pa_1^2f$ from equation \eqref{loc.18} into \eqref{loc.95} we get\\
\eq$\label{loc.96}
(1+\pa_\phi^2f)(\pa_1\phi)^2+(1+\pa_\chi^2f)(\pa_1\chi)^2+2\pa_\chi\pa_\phi f\pa_1\chi\pa_1\phi+\pa_\phi f\pa_1^2\phi+\pa_\chi f\pa_1^2\chi=0$\\
\begin{center}
\underline{\large{Equation $(^1{}_0)$:}}\\
\end{center}
\begin{gather}
\eqref{loc.14}\xRightarrow[\eqref{loc.93}]{\eqref{loc.7}}(\pa_0\phi)^2+(\pa_0\chi)^2+\pa_0^2f+\left(1-\frac{2m}{r}\right)(\pa_1\phi\pa_0\phi+\pa_1\chi\pa_0\chi+\pa_1\pa_0f)+\pa_1\left(\frac{m}{r}\right)\pa_0f\nonum\\
\label{loc.97}
-\pa_0\left(\frac{m}{r}\right)\pa_1f+e^{2A}\pa_4\left(\frac{m}{r}\right)\pa_4f=f\left[\frac{2}{r^2}\pa_0m-\frac{e^{2A}}{r}(\pa_4^2m+4A'\pa_4m)\right]
\end{gather}
\par The substitution of quantities $\pa_0^2f$ and $\pa_1\pa_0f$ from equations \eqref{loc.24} and \eqref{loc.25} respectively into equation \eqref{loc.97} yields to
\begin{gather}
e^{-2A}\left\{(1+\pa_\phi^2f)(\pa_0\phi)^2+(1+\pa_\chi^2f)(\pa_0\chi)^2+2\pa_\chi\pa_\phi f\pa_0\phi\pa_0\chi+\pa_\phi f\pa_0^2\phi+\pa_\chi f\pa_0^2\chi \right.\nonum\\
\hspace{2em}+\left(1-\frac{2m}{r}\right)[(1+\pa_\phi^2f)\pa_1\phi\pa_0\phi+(1+\pa_\chi^2f)\pa_1\chi\pa_0\chi+\pa_\phi\pa_\chi f(\pa_1\phi\pa_0\chi+\pa_1\chi\pa_0\phi)+\pa_\phi f\pa_1\pa_0\phi+\pa_\chi f\pa_1\pa_0\chi]\nonum\\
\left.+\left(\frac{\pa_1m}{r}-\frac{m}{r^2}\right)(\pa_\phi f\pa_0\phi+\pa_\chi f\pa_0\chi)-\frac{\pa_0m}{r}(\pa_\phi f\pa_1\phi+\pa_\chi f\pa_1\chi)+e^{2A}\frac{\pa_4m}{r}(\pa_\phi f\pa_4\phi+\pa_\chi f\pa_4\chi)\right\}=\nonum\\
\label{loc.98}
=f\left[\frac{2}{r^2}\pa_0m-\frac{e^{2A}}{r}(\pa_4^2m+4A'\pa_4m)\right]
\end{gather}\\

\begin{center}
\underline{\large{Equation $(^4{}_0)$:}}\\
\end{center}
\eq$\eqref{loc.14}\xRightarrow[\eqref{loc.93}]{\eqref{loc.7}}\label{loc.99}
\pa_4\phi\pa_0\phi+\pa_4\chi\pa_0\chi+\pa_4\pa_0f-A'\pa_0f-\frac{\pa_4m}{r}\pa_1f=\frac{f}{r}\left(\frac{\pa_4m}{r}+\pa_1\pa_4m\right)$
The last equation in its extended form is:
\begin{gather}
(1+\pa_\phi^2f)\pa_4\phi\pa_0\phi+(1+\pa_\chi^2f)\pa_4\chi\pa_0\chi+\pa_\chi\pa_\phi f(\pa_4\chi\pa_0\phi+\pa_4\phi\pa_0\chi)+\pa_\phi f\pa_4\pa_0\phi+\pa_\chi f\pa_4\pa_0\chi\nonum\\
\label{loc.100}
-A'\pa_0f-\frac{\pa_4m}{r}(\pa_\phi f\pa_1\phi+\pa_\chi f\pa_1\chi)=\frac{f}{r}\left(\frac{\pa_4m}{r}+\pa_1\pa_4m\right)
\end{gather}\\

\begin{center}
\underline{\large{Equation $(^0{}_4)$:}}\\
\end{center}
\eq$\label{loc.101}
\eqref{loc.14}\xRightarrow[\eqref{loc.93}]{\eqref{loc.7}}\pa_1\phi\pa_4\phi+\pa_1\chi\pa_4\chi+\pa_1\pa_4f-A'\pa_1f=0$
\par Using equation \eqref{loc.35} into equation \eqref{loc.101} we obtain
\begin{gather}
(1+\pa_\phi^2f)\pa_1\phi\pa_4\phi+(1+\pa_\chi^2f)\pa_1\chi\pa_4\chi+\pa_\phi\pa_\chi f(\pa_1\chi\pa_4\phi+\pa_1\phi\pa_4\chi)+\pa_\phi f\pa_1\pa_4\phi\nonum\\
\label{loc.102}+\pa_\chi f\pa_1\pa_4\chi-A'(\pa_\phi f\pa_1\phi+\pa_\chi f\pa_1\chi)=0
\end{gather}\\
\begin{center}
\underline{\large{Equation $(^0{}_0)$:}}\\
\end{center}
\begin{gather}
\eqref{loc.14}\xRightarrow[\eqref{loc.93}]{\eqref{loc.7}}e^{-2A}\left[\pa_1\phi\pa_0\phi+\pa_1\chi\pa_0\chi+\pa_1\pa_0f-\pa_1\left(\frac{m}{r}\right)\pa_1f\right]+A'\pa_4f+\lagr-\square f-\Lambda_B=\nonum\\
\label{loc.103}
=f\left(6A'^2+3A''-\frac{2e^{-2A}}{r^2}\pa_1m\right)\end{gather}
\par Equation \eqref{loc.103} is equivalent to
\begin{gather}
e^{-2A}\left[(1+\pa_\phi^2f)\pa_1\phi\pa_0\phi+(1+\pa_\chi^2f)\pa_1\chi\pa_0\chi+\pa_\chi\pa_\phi f(\pa_1\chi\pa_0\phi+\pa_1\phi\pa_0\chi)+\pa_\phi f\pa_1\pa_0\phi+\pa_\chi f\pa_1\pa_0\chi\right.\nonum\\
\left.+\frac{\pa_\phi f\pa_1\phi+\pa_\chi f\pa_1\chi}{r}\left(\frac{m}{r}-\pa_1m\right)\right]+A'(\pa_\phi f\pa_4\phi+\pa_\chi f\pa_4\chi)+\lagr-\square f-\Lambda_B=\nonum\\
\label{loc.104}
=f\left(6A'^2+3A''-\frac{2e^{-2A}}{r^2}\pa_1m\right)
\end{gather}\\
\begin{center}
\underline{\large{Equation $(^2{}_2)$:}}\\
\end{center}
\begin{gather}
\eqref{loc.14}\xRightarrow[\eqref{loc.93}]{\eqref{loc.7}}\frac{e^{-2A}}{r}\left[\pa_0f+\left(1-\frac{2m}{r}\right)\pa_1f\right]+A'\pa_4f+\lagr-\square f-\Lambda_B=\nonum\\
\label{loc.105}
=f\left(6A'^2+3A''-\frac{e^{-2A}}{r}\pa_1^2m\right)
\end{gather}
\par Expanding the partial derivatives of the function $f=f(\phi,\chi)$ we get
\begin{gather}
\frac{e^{-2A}}{r}\left[(\pa_\phi f\pa_0\phi+\pa_\chi f\pa_0\chi)+\left(1-\frac{2m}{r}\right)(\pa_\phi f\pa_1\phi+\pa_\chi f\pa_1\chi)\right]+A'(\pa_\phi f\pa_4\phi+\pa_\chi f\pa_4\chi)\nonum\\
\label{loc.106}
+\lagr-\square f-\Lambda_B=f\left(6A'^2+3A''-\frac{e^{-2A}}{r}\pa_1^2m\right)\end{gather}\\
\begin{center}
\underline{\large{Equation $(^4{}_4)$:}}\\
\end{center}
\eq$\label{loc.107}
\eqref{loc.14}\xRightarrow[\eqref{loc.93}]{\eqref{loc.7}}(\pa_4\phi)^2+(\pa_4\chi)^2+\pa^2_4f+\lagr-\square f-\Lambda_B=f\left(6A'^2-\frac{e^{-2A}}{r}\pa_1^2m-\frac{2e^{-2A}}{r^2}\pa_1m\right)$\\
The combination of equations \eqref{loc.63} and \eqref{loc.107} gives
\begin{gather}
(1+\pa_\phi^2f)(\pa_4\phi)^2+(1+\pa_\chi^2f)(\pa_4\chi)^2+2\pa_\chi\pa_\phi f\pa_4\chi\pa_4\phi+\pa_\phi f\pa_4^2\phi+\pa_\chi f\pa_4^2\chi+\lagr-\square f-\Lambda_B=\nonum\\
\label{loc.108}
=f\left(6A'^2-\frac{e^{-2A}}{r}\pa_1^2m-\frac{2e^{-2A}}{r^2}\pa_1m\right)\end{gather}
\par Instead of using equations \eqref{loc.103}, \eqref{loc.105}, \eqref{loc.107} and the corresponding extended equations \eqref{loc.104}, \eqref{loc.106} and \eqref{loc.108}, which contain the terms $\lagr$ and $\square f$ that add extra complexity, we can be exempted from these terms by combining the aforementioned equations with each other.\\
\begin{center}
\underline{\large{Equation $(^0{}_0)-$Equation $(^2{}_2)$:}}\\
\end{center}
\par Subtracting equation \eqref{loc.105} from \eqref{loc.103} we get:\\
\eq$\label{loc.109}
r(\pa_1\phi\pa_0\phi+\pa_1\chi\pa_0\chi+\pa_1\pa_0f)-\pa_0f-\pa_1f\left(\pa_1m+1-\frac{3m}{r}\right)=f\left(\pa_1^2m-\frac{2}{r}\pa_1m\right)$
\par Correspondingly, from equations \eqref{loc.106} and \eqref{loc.104} one obtains
\begin{gather}
r[(1+\pa_\phi^2f)\pa_1\phi\pa_0\phi+(1+\pa_\chi^2f)\pa_1\chi\pa_0\chi+\pa_\phi\pa_\chi f(\pa_1\chi\pa_0\phi+\pa_1\phi\pa_0\chi)+\pa_\phi f\pa_1\pa_0\phi+\pa_\chi f\pa_1\pa_0\chi] \nonum\\
\label{loc.110}
-(\pa_\phi f\pa_0\phi+\pa_\chi f\pa_0\chi)-\left(\pa_1m+1-\frac{3m}{r}\right)(\pa_\phi f\pa_1\phi+\pa_\chi f\pa_1\chi)=f\left(\pa_1^2m-\frac{2}{r}\pa_1m\right)
\end{gather}\\
\begin{center}
\underline{\large{Equation $(^0{}_0)-$Equation $(^4{}_4)$:}}\\
\end{center}
\par In the same way, subtracting equation \eqref{loc.107} from \eqref{loc.103} and \eqref{loc.108} from \eqref{loc.104} we respectively have:
\begin{gather}
e^{-2A}\left[\pa_1\phi\pa_0\phi+\pa_1\chi\pa_0\chi+\pa_1\pa_0f-\pa_1\left(\frac{m}{r}\right)\pa_1f\right]+A'\pa_4f-(\pa_4\phi)^2-(\pa_4\chi)^2-\pa_4^2f=\nonum\\
\label{loc.111}
=f\left(3A''+\frac{e^{-2A}}{r}\pa_1^2m\right)
\end{gather}
\begin{gather}
e^{-2A}[(1+\pa_\phi^2f)\pa_1\phi\pa_0\phi+(1+\pa_\chi^2f)\pa_1\chi\pa_0\chi+\pa_\phi\pa_\chi f(\pa_1\chi\pa_0\phi+\pa_1\phi\pa_0\chi)+\pa_\phi f\pa_1\pa_0\phi+\pa_\chi f\pa_1\pa_0\chi\nonum\\
-\pa_1\left(\frac{m}{r}\right)(\pa_\phi f\pa_1\phi+\pa_\chi f\pa_1\chi)]+A'(\pa_\phi f\pa_4\phi+\pa_\chi f\pa_4\chi)-(1+\pa_\phi^2f)(\pa_4\phi)^2-(1+\pa_\chi^2f)(\pa_4\chi)^2\nonum\\
\label{loc.112}
-2\pa_\phi\pa_\chi f\pa_4\phi\pa_4\chi-\pa_\phi f\pa_4^2\phi-\pa_\chi f\pa_4^2\chi=f\left(3A''+\frac{e^{-2A}}{r}\pa_1^2m\right)
\end{gather}\\
\begin{center}
\underline{\large{Equation $(^2{}_2)-$Equation $(^4{}_4)$:}}\\
\end{center}
\par Finally, the last field equation that has no dependence on the quantities $\lagr$ and $\square f$ can be provided by equations \eqref{loc.105} and \eqref{loc.107} or equations \eqref{loc.106} and \eqref{loc.108}. We respectively obtain
\eq$\label{loc.113}
\frac{e^{-2A}}{r}\left[\pa_0f+\left(1-\frac{2m}{r}\right)\pa_1f\right]+A'\pa_4f-(\pa_4\phi)^2-(\pa_4\chi)^2-\pa_4^2f=f\left(3A''+\frac{2e^{-2A}}{r^2}\pa_1m\right)$\\
\begin{gather}
\frac{e^{-2A}}{r}\left[\pa_\phi f\pa_0\phi+\pa_\chi f\pa_0\chi+\left(1-\frac{2m}{r}\right)(\pa_\phi f\pa_1\phi+\pa_\chi f\pa_1\chi)\right]+A'(\pa_\phi f\pa_4\phi+\pa_\chi f\pa_4\chi)\nonum\\
-(1+\pa_\phi^2f)(\pa_4\phi)^2-(1+\pa_\chi^2f)(\pa_4\chi)^2-2\pa_\phi\pa_\chi f\pa_4\phi\pa_4\chi-\pa_\phi f\pa_4^2\phi\nonum\\
\label{loc.114}
-\pa_\chi f\pa_4^2\chi=f\left(3A''+\frac{e^{-2A}}{r}\pa_1^2m\right)\end{gather}\\
\par Of course, equation \eqref{loc.113} is not independent from equations \eqref{loc.109} and \eqref{loc.111}, the subtraction of Eq.\eqref{loc.109} from \eqref{loc.111} gives Eq.\eqref{loc.113}. Correspondingly, the subtraction of Eq.\eqref{loc.110} from Eq.\eqref{loc.112} gives Eq.\eqref{loc.114}. The reason that equations \eqref{loc.113} and \eqref{loc.114} were presented is that in some of the cases -which are going to be investigated in the next Chapter- these equations might be easier to be solved or they might provide more directly useful information about the mass function $m(v,r,y)$. In the context of the scalar field theory model, which was introduced at the beginning of this section, a solution to the localization problem -except from an appropriate expression for the mass function- requires also appropriate expressions for the scalar fields. A solution of the field equations that achieves to produce a mass function $m(v,r,y)$ which describes a 5-dimensional black hole and it is also localized close to the 3-brane at $y=0$, with the cost of producing functions for the scalar fields $\phi(v,r,y), \chi(v,r,y)$ that have an infinite value at $y\ra \infty$, cannot be accepted. Subsequently, we sum up the field equations. The field equations in their compact form are given by equations \eqref{loc.95}, \eqref{loc.97}, \eqref{loc.99}, \eqref{loc.101}, \eqref{loc.107}, \eqref{loc.109}, \eqref{loc.111} and \eqref{loc.113} (the last equation is not independent), while the field equations in their extended form are given by equations \eqref{loc.96}, \eqref{loc.98}, \eqref{loc.100}, \eqref{loc.102}, \eqref{loc.108}, \eqref{loc.110}, \eqref{loc.112} and \eqref{loc.114} (of course, the last equation is also not independent).
\begin{empheq}[box=\mymath]{gather}
\text{\textbf{\Large{Compact Form}}}\nonum\\ \nonum\\ 
\label{loc.115}
(\pa_r\phi)^2+(\pa_r\chi)^2+\pa_r^2f=0\\ \vspace{1em} \nonum\\
(\pa_v\phi)^2+(\pa_v\chi)^2+\pa_v^2f+\left(1-\frac{2m}{r}\right)(\pa_r\phi\pa_v\phi+\pa_r\chi\pa_v\chi+\pa_r\pa_vf)+\pa_r\left(\frac{m}{r}\right)\pa_vf\nonum\\
\label{loc.116}
-\pa_v\left(\frac{m}{r}\right)\pa_rf+e^{2A}\pa_y\left(\frac{m}{r}\right)\pa_yf=f\left[\frac{2}{r^2}\pa_vm-\frac{e^{2A}}{r}(\pa_y^2m+4A'\pa_ym)\right]\\ \vspace{1em} \nonum\\
\label{loc.117}
\pa_y\phi\pa_v\phi+\pa_y\chi\pa_v\chi+\pa_y\pa_vf-A'\pa_vf-\frac{\pa_ym}{r}\pa_rf=\frac{f}{r}\left(\frac{\pa_ym}{r}+\pa_r\pa_ym\right)\\ \vspace{1em}  \nonum\\
\label{loc.118}
\pa_r\phi\pa_y\phi+\pa_r\chi\pa_y\chi+\pa_r\pa_yf-A'\pa_rf=0\\ \vspace{1em}  \nonum\\
\label{loc118}
(\pa_y\phi)^2+(\pa_y\chi)^2+\pa^2_yf+\lagr-\square f-\Lambda_B=f\left(6A'^2-\frac{e^{-2A}}{r}\pa_r^2m-\frac{2e^{-2A}}{r^2}\pa_rm\right)\\ \vspace{1em}  \nonum\\
\label{loc.119}
r(\pa_r\phi\pa_v\phi+\pa_r\chi\pa_v\chi+\pa_r\pa_vf)-\pa_vf-\pa_rf\left(\pa_rm+1-\frac{3m}{r}\right)=f\left(\pa_r^2m-\frac{2}{r}\pa_rm\right)\\ \vspace{1em}  \nonum\\
e^{-2A}\left[\pa_r\phi\pa_v\phi+\pa_r\chi\pa_v\chi+\pa_r\pa_vf-\pa_r\left(\frac{m}{r}\right)\pa_rf\right]+A'\pa_yf-(\pa_y\phi)^2\nonum\\ 
\label{loc.120}
-(\pa_y\chi)^2-\pa_y^2f=f\left(3A''+\frac{e^{-2A}}{r}\pa_r^2m\right)\\ \vspace{1em}  \nonum\\
\frac{e^{-2A}}{r}\left[\pa_vf+\left(1-\frac{2m}{r}\right)\pa_rf\right]+A'\pa_yf-(\pa_y\phi)^2-(\pa_y\chi)^2\nonum\\
\label{loc.121}
-\pa_y^2f=f\left(3A''+\frac{2e^{-2A}}{r^2}\pa_rm\right)
\end{empheq}
Last but not least, there are also the following two equations for the scalar fields that are necessary to be satisfied:
\begin{empheq}[box=\mymath]{equation}\label{loc129}
\sqrt{-\gfv}\left(\frac{1}{2}\frac{\pa f}{\pa\phi}R-\frac{\pa V}{\pa\phi}\right)=-\pa_M\left(\sqrt{-\gfv}\ g^{MN}\pa_N\phi\right)
\end{empheq}
\begin{empheq}[box=\mymath]{equation}\label{loc130}
\sqrt{-\gfv}\left(\frac{1}{2}\frac{\pa f}{\pa\chi}R-\frac{\pa V}{\pa\chi}\right)=-\pa_M\left(\sqrt{-\gfv}\ g^{MN}\pa_N\chi\right)
\end{empheq}\\
The last two equations are going to bother us only in the case that we find appropriate functions for the quantities $A(y),\ m(v,r,y),\ \phi(v,r,y)$ and $\chi(v,r,y)$ which satisfy equations \eqref{loc.115}-\eqref{loc.121}. In this hypothetical scenario it would be necessary to verify if these two equations are satisfied as well.

\newpage
\blankpage

%% file: Chapters/Chapter_Sol_the_Eqs.tex
\chapter{Solving the Field Equations}

\par In this Chapter, we will try to solve the field equations, meaning that we will try to determine from the field equations the mass function $m=m(v,r,y)$ and then the functions of the scalar fields $\phi(v,r,y)$, $\chi(v,r,y)$. As we already mentioned in the previous Chapter, the mass function should have a suitable dependence on the extra dimension $y$ in order to be able to constitute a 5-dimensional black hole that is localized on our 4-dimensional 3-brane (namely our universe). In order to achieve that, we are going to use equations \eqref{loc.115}-\eqref{loc.121}.
\par The aforementioned field equations resulted from the assumption that the scalar fields $\phi$, $\chi$ and consequently the coupling function $f=f(\phi,\chi)$ depend on the coordinates $(v,r,y)$. However, it is possible to consider simpler cases where one or both of the scalar fields depend on just one or two of the $(v,r,y)$ coordinates. It is also possible to consider cases in which either $\pa_\phi f$ or $\pa_\chi f$ equals to zero but not both of them simultaneously, because it is important to preserve the non-minimal coupling. As will be clear from the next pages of this Chapter, it is extremely difficult to find a suitable solution of the field equations within the framework of our field theory which could yield to a mass function $m=m(v,r,y)$ that has the desirable dependence on the extra dimension $y$. On the contrary, in most of these cases the field equations are not consistent with our assumptions.

\section{All Possible Cases}
\par We now present all the possible cases that were mentioned previously, starting from the simplest cases and ending to the most complicated ones.
\begin{enumerate}[I)]
\item \begin{center}\textbf{Both scalar fields depend on one coordinate} \end{center}
$$\begin{array}{r}
1)\ \ \{\phi=\phi(v),\ \chi=\chi(v)\}\\ \\ 
2)\ \ \{\phi=\phi(v),\ \chi=\chi(r)\}\\ \\ 
3)\ \ \{\phi=\phi(v),\ \chi=\chi(y)\}\\ \\
4)\ \ \{\phi=\phi(r),\ \chi=\chi(r)\}\\ \\
5)\ \ \{\phi=\phi(r),\ \chi=\chi(y)\}\\ \\
6)\ \ \{\phi=\phi(y),\ \chi=\chi(y)\}\\ \\
\end{array}$$
The field equations are manifestly symmetrical under the exchange of $\phi$ and $\chi$. Therefore, the case $\{\phi=\phi(v),\ \chi=\chi(r)\}$ is the same as $\{\phi=\phi(r),\ \chi=\chi(v)\}$. This property reduces significantly the number of independent cases.

\newpage
\item \begin{center}\textbf{One scalar field depends on two coordinates and the other one depends on one}  \end{center}
$$\begin{array}{r}
7)\ \ \{\phi=\phi(v,r),\ \chi=\chi(v)\}\\ \\ 
8)\ \ \{\phi=\phi(v,r),\ \chi=\chi(r)\}\\ \\ 
9)\ \ \{\phi=\phi(v,r),\ \chi=\chi(y)\}\\ \\
10)\ \ \{\phi=\phi(v,y),\ \chi=\chi(v)\}\\ \\
11)\ \ \{\phi=\phi(v,y),\ \chi=\chi(r)\}\\ \\
12)\ \ \{\phi=\phi(v,y),\ \chi=\chi(y)\}\\ \\
13)\ \ \{\phi=\phi(r,y),\ \chi=\chi(v)\}\\ \\
14)\ \ \{\phi=\phi(r,y),\ \chi=\chi(r)\}\\ \\
15)\ \ \{\phi=\phi(r,y),\ \chi=\chi(y)\}\\ \\
\end{array}$$

\item \begin{center}\textbf{Both scalar fields depend on two coordinates}\end{center}
$$\begin{array}{r}
16)\ \ \{\phi=\phi(v,r),\ \chi=\chi(v,r)\}\\ \\
17)\ \ \{\phi=\phi(v,r),\ \chi=\chi(v,y)\}\\ \\
18)\ \ \{\phi=\phi(v,r),\ \chi=\chi(r,y)\}\\ \\
19)\ \ \{\phi=\phi(v,y),\ \chi=\chi(v,y)\}\\ \\
20)\ \ \{\phi=\phi(v,y),\ \chi=\chi(r,y)\}\\ \\
21)\ \ \{\phi=\phi(r,y),\ \chi=\chi(r,y)\}\\ \\
\end{array}$$
\item \begin{center}\textbf{One scalar field depends on all three coordinates and the other one depends on one}\end{center}
$$\begin{array}{r}
22)\ \ \{\phi=\phi(v,r,y),\ \chi=\chi(v)\}\\ \\
23)\ \ \{\phi=\phi(v,r,y),\ \chi=\chi(r)\}\\ \\
24)\ \ \{\phi=\phi(v,r,y),\ \chi=\chi(y)\}\\ \\
\end{array}$$
\item \begin{center}\textbf{One scalar field depends on all three coordinates and the other one depends on two}\end{center}
$$\begin{array}{r}
25)\ \ \{\phi=\phi(v,r,y),\ \chi=\chi(v,r)\}\\ \\
26)\ \ \{\phi=\phi(v,r,y),\ \chi=\chi(v,y)\}\\ \\
27)\ \ \{\phi=\phi(v,r,y),\ \chi=\chi(r,y)\}\\ \\
\end{array}$$

\newpage
\item \begin{center}\textbf{Both scalar fields depend on all three coordinates}\end{center}
$$28)\ \ \{\phi=\phi(v,r,y),\ \chi=\chi(v,r,y)\}$$
\end{enumerate}
\vspace{0.5em}
\par We henceforth start our quest for a valid solution of the field equations examining one by one the cases that were presented previously. The examination of the aforementioned cases will be split in two large categories. The first category includes the cases from 1 to 21 which can be studied (and excluded as it is shown    below) without further assumptions. The cases from 22 to 28 belong to the second category in which it is necessary to introduce an expression for the coupling function $f(\phi,\chi)$ in order to be able to proceed to the solution of the field equations. In these last cases, if we do not fix the coupling function $f(\phi,\chi)$ the field equations are unapproachable.

\section{Explicitly Rejected Cases}
\vspace*{1em}
\begin{center}
\underline{\textbf{1)}\hspace{1.5em}$\{\phi=\phi(v),\ \chi=\chi(v),\ f=f(\phi,\chi)=f(v)\}$}
\end{center}
$$\eqref{loc.120}\Ra 0=f\left(3A''+\frac{e^{-2A}}{r}\pa_r^2m\right)\Ra \pa_r^2m=-3A''e^{2A} r\Ra$$
\eq$\label{loc.129}
\boxed{\pa_rm=-\frac{3}{2}A''e^{2A}r^2+m_0(v,y)}$
\begin{align}
\eqref{loc.119}&\Ra -\pa_vf=f\left(\pa_r^2m-\frac{2}{r}\pa_rm\right)\xRightarrow{f\neq0} -\frac{\pa_vf}{f}=\pa_r^2m-\frac{2}{r}\pa_rm\xRightarrow{\eqref{loc.129}} \nonum\\
&\Ra -\frac{\pa_vf}{f}=-3A''e^{2A}r-\frac{2}{r}\left(-\frac{3}{2}A''e^{2A}r^2+m_0(v,y)\right)\Ra \nonum\\
&\Ra -\frac{\pa_vf}{f}=-\bcancel{3A''e^{2A}r}+\bcancel{3A''e^{2A}r}-\frac{2}{r}m_0(v,y)\Ra\nonum
\end{align}
\eq$\label{loc.130}
\boxed{\frac{\pa_vf}{f}=\frac{2}{r}m_0(v,y)}$\\
\par Obviously, the last equation is totally inconsistent, the left hand side (LHS) depends only on the $v-$coordinate while the right hand side (RHS) depends on the coordinates $(v,r,y)$. Even if we demand $\pa_ym_0=0$ in order to eliminate the $y-$dependence, we definitely cannot cancel the factor $1/r$ in the RHS. Therefore, this case is rejected.\\
\begin{center}
\underline{\textbf{2)}\hspace{1.5em}$\{\phi=\phi(v),\ \chi=\chi(r),\ f=f(\phi,\chi)=f(v,r)\}$}
\end{center}
\eq$\label{loc.131}
\eqref{loc.118}\Ra -A'\pa_rf=0\Ra A'\pa_rf=0\Ra\left\{\begin{array}{c} A'=\pa_yA=0\vspace{0.5em}\\ or\vspace{0.5em}\\ \pa_rf=\pa_\chi f\pa_r\chi=0\xRightarrow{\pa_r\chi\neq0} \pa_\chi f=0\end{array}\right\}$\\
\par From equation \eqref{loc.131} it is easily deducible that this case is rejected as well. The possibility $A'(y)=0$ is immediately rejected, while the constraint $\pa_\chi f=0$ leads to $\pa_r\chi=0$ if we use equation \eqref{loc.115}. Hence, none of the possibilities ($A'=0$ or $\pa_\chi f=0$) can be valid because they contradict with our assumptions about the functions $A(y)$ and $\chi(r)$.\\
\begin{center}
\underline{\textbf{3)}\hspace{1.5em}$\{\phi=\phi(v),\ \chi=\chi(y),\ f=f(\phi,\chi)=f(v,y)\}$}
\end{center}
\eq$\label{loc.132}
\eqref{loc.120}\Ra A'\pa_yf-(\pa_y\chi)^2-\pa_y^2f=f\left(3A''+\frac{e^{-2A}}{r}\pa_r^2m\right)$
\par The LHS of equation \eqref{loc.132} has $(v,y)-$dependence, while the RHS of the same equation is depended on $(v,r,y)$, hence, we are led to the following constraint:
$$\pa_r\left[f\left(3A''+\frac{e^{-2A}}{r}\pa_r^2m\right)\right]=0\Ra \underbrace{\pa_rf}_{0}\left(3A''+\frac{e^{-2A}}{r}\pa_r^2m\right)+fe^{-2A}\pa_r\left(\frac{\pa_r^2m}{r}\right)=0\xRightarrow{f\neq0}$$
\eq$\label{loc.133}
\Ra\pa_r\left(\frac{\pa_r^2m}{r}\right)=0\Ra \frac{\pa_r^3m}{r}-\frac{\pa_r^2m}{r^2}=0\Ra \boxed{\pa_r^3m-\frac{\pa_r^2m}{r}=0}$
\vspace{0.5em}
\par The differential equation \eqref{loc.133} can be solved easily assuming that the function of mass $m=m(v,r,y)$ can be written as a power series expansion with respect to $r-$coordinate, containing either positive or negative powers of $r$. The same assumption has been used as well in \cite{0264-9381-33-1-015003}. Thus, it is
\eq$\label{loc.134}
\boxed{m=m(v,r,y)=\sum_n a_n(v,y)r^n}$
\par The combination of equations \eqref{loc.133} and \eqref{loc.134} lead to
$$\sum_n n(n-1)(n-2)a_n(v,y)r^{n-3}-\sum_n n(n-1)a_n(v,y)r^{n-3}=0\Ra$$
$$\Ra\sum_nn(n-1)(n-3)a_n(v,y)r^{n-3}=0\Ra\left\{\begin{array}{c}a_n=0\ \ \forall n\neq\{0,1,3\}\\ \\ a_0,a_1,a_3\ \ arbitrary\ \ functions\end{array}\right\}\Ra$$\\
\eq$\label{loc.135}
\boxed{m(v,r,y)=a_0(v,y)+a_1(v,y)r+a_3(v,y)r^3}$\\
\begin{align}
\eqref{loc.119}&\Ra -\pa_vf=f\left(\pa_r^2m-\frac{2}{r}\pa_rm\right)\xRightarrow{f\neq0}-\frac{\pa_vf}{f}=\pa_r^2m-\frac{2}{r}\pa_rm\xRightarrow{\eqref{loc.135}}\nonum\\
&\Ra -\frac{\pa_vf}{f}=6a_3(v,y)r-\frac{2}{r}[a_1(v,y)+3a_3(v,y)r^2]=\bcancel{6a_3(v,y)r}-\bcancel{6a_3(v,y)r}-\frac{2}{r}a_1(v,y)\Ra\nonum
\end{align}
\eq$\label{loc.136}
\boxed{\frac{\pa_vf}{f}=\frac{2}{r}a_1(v,y)}$\\
\par Similarly to the case 1, the LHS of equation \eqref{loc.136} depends on $(v,y)-$coordinates while the RHS has the factor $1/r$ which cannot be cancelled, therefore this case is inconsistent as well.\\
\begin{center}
\underline{\textbf{4)}\hspace{1.5em}$\{\phi=\phi(r),\ \chi=\chi(r),\ f=f(\phi,\chi)=f(r)\}$}
\end{center}
$$\eqref{loc.118}\Ra-A'\pa_rf=0\Ra A'\pa_rf=0\Ra \left\{\begin{array}{c} A'=\pa_yA=0\vspace{0.5em}\\ or\vspace{0.5em}\\ \pa_rf=0\end{array}\right\}$$
\par In this case, we are also led to a contradiction to our primary assumption. Either $A'(y)=0$ or $\pa_r f=0$ cannot be true. Especially, the constraint $\pa_r f=0$ implies that the coupling function $f(\phi,\chi)$ is completely independent from both scalar fields $\phi$ and $\chi$, thus, the non-minimal coupling of the scalar fields to gravity is entirely vanished. Therefore, this case is also excluded from the list of possible solutions.\\
\begin{center}
\underline{\textbf{5)}\hspace{1.5em}$\{\phi=\phi(r),\ \chi=\chi(y),\ f=f(\phi,\chi)=f(r,y)\}$}
\end{center}
$$\eqref{loc.115}\Ra (\pa_r\phi)^2+\pa_r^2f=0\Ra \pa_r^2f=-(\pa_r\phi)^2\Ra \frac{\pa(\pa_rf)}{\pa r}=-(\pa_r\phi)^2\Ra$$
$$\Ra \int dr\frac{\pa(\pa_rf)}{\pa r}=-\int (\pa_r\phi)^2dr\Ra \pa_rf-f_0(y)=-\int (\pa_r\phi)^2dr\Ra$$
\eq$\label{loc.137}
\boxed{\pa_rf=-\int (\pa_r\phi)^2dr+f_0(y)}$\\
$$\eqref{loc.118}\Ra \pa_y(\pa_rf)-A'\pa_rf=0\xRightarrow{\eqref{loc.137}}\pa_y\left[-\int (\pa_r\phi)^2dr+f_0(y)\right]-A'\left[-\int (\pa_r\phi)^2dr+f_0(y)\right]=0\Ra$$
$$\Ra\underbrace{-\pa_y\left[\int(\pa_r\phi)^2dr\right]}_0+\pa_yf_0+A'\int(\pa_r\phi)^2dr-A'f_0=0\Ra$$
\eq$\label{loc.138}
A'f_0-\pa_yf_0=A'\int(\pa_r\phi)^2dr$
\par The LHS of equation \eqref{loc.138} depends only on $y-$coordinate, while the RHS depends on $(r,y)-$co\-or\-di\-nates. Hence, the constraint that is derived is the following:
\eq$\label{loc.139}
\pa_r\left[A'\int (\pa_r\phi)^2dr\right]=0\Ra A'(\pa_r\phi)^2=0\Ra\left\{\begin{array}{c}A'=\pa_yA=0\vspace{0.5em}\\ or\vspace{0.5em}\\ \pa_r\phi=0\end{array}\right\}$
\par According to our assumptions, both functions $A'$ and $\pa_r\phi$ cannot be zero, so none of the constraints which are demanded by equation \eqref{loc.139} is able to be fulfilled. Neither $\pa_\phi f=0$ nor $\pa_\chi f=0$ are able to be assumed in the context of this case. The constraint $\pa_\phi f=0$ leads to $\pa_r\phi=0$ through equation \eqref{loc.115} and $\pa_\chi f=0$ leads to the same analysis and the same negative result as the original case 5 with the only difference that in the case of $\pa_\chi f=0$ $f_0$ is a constant and not a function of $y$.\\
\begin{center}
\underline{\textbf{6)}\hspace{1.5em}$\{\phi=\phi(y),\ \chi=\chi(y),\ f=f(\phi,\chi)=f(y)\}$}
\end{center}
\begin{align}\eqref{loc.119}\Ra &\ \ 0=f\left(\pa_r^2m-\frac{2}{r}\pa_rm\right)\xRightarrow{f\neq0}\pa_r^2m-\frac{2}{r}\pa_rm=0\Ra\frac{\pa(\pa_rm)}{\pa r}=\frac{2}{r}\pa_rm\Ra\nonum\\
\Ra&\frac{1}{\pa_rm}\frac{\pa(\pa_rm)}{\pa r}=\frac{2}{r}\Ra\frac{\pa[\ln(\pa_rm)]}{\pa r}=\frac{2}{r}\Ra\int dr\frac{\pa[\ln(\pa_rm)]}{\pa r}=\int\frac{2}{r}dr\Ra\nonum\\
\Ra&\ln(\pa_rm)-m_0(v,y)=2\ln r\Ra\pa_rm=e^{m_0(v,y)+\ln r^2}\xRightarrow{e^{m_0(v,y)}\ra m_0(v,y)}\nonum\\
\Ra&\pa_rm=m_0(v,y)r^2\Ra \int dr\frac{\pa m}{\pa r}=m_0(v,y)\int r^2dr\Ra\nonum
\end{align}
\eq$\label{loc.140}
m(v,r,y)=m_0(v,y)\frac{r^3}{3}+m_1(v,y)$\\
\begin{align}\label{loc.141}
\eqref{loc.117}\Ra&\ \ 0=\frac{f}{r}\left(\frac{\pa_ym}{r}+\pa_r\pa_ym\right)\xRightarrow{f\neq0}\frac{\pa_ym}{r}+\pa_r\pa_ym=0\xRightarrow{\eqref{loc.140}}\pa_ym_0\frac{r^2}{3}+\pa_ym_1+\pa_ym_0r^2=0\Ra\nonum\\
\Ra& 2r^2\pa_ym_0+\pa_ym_1=0\Ra\left\{\begin{array}{c}\pa_ym_0=0\Ra m_0=m_0(v)\vspace{0.5em}\\and \vspace{0.5em}\\ \pa_ym_1=0\Ra m_1=m_1(v)\end{array}\right\}
\end{align}\\
\eq$\label{loc.142}
\eqref{loc.140}\xRightarrow{\eqref{loc.141}}\boxed{m=m(v,r)=m_0(v)\frac{r^3}{3}+m_1(v)}$\\
\par It is clear from equation \eqref{loc.142} that in order to satisfy simultaneously both field equations \eqref{loc.119} and \eqref{loc.117}, we are led to a mass function which does not depend on the $y-$coordinate. A mass function which is $y-$independent is impossible to describe a localized black hole. Therefore, this case is rejected without second thought.
\begin{center}
\underline{\textbf{7)}\hspace{1.5em}$\{\phi=\phi(v,r),\ \chi=\chi(v),\ f=f(\phi,\chi)=f(v,r)\}$}
\end{center}
$$\eqref{loc.118}\Ra-A'\pa_rf=0\Ra \left\{\begin{array}{c} A'=\pa_yA=0\vspace{0.5em}\\ or\vspace{0.5em}\\ \pa_rf=0\Ra \pa_\phi f\pa_r\phi=0\Ra \pa_\phi f=0\end{array}\right\}$$
\par Of course, $A'(y)\neq 0$. Thus, if $\pa_\phi f=0$ then equation \eqref{loc.115} leads to $\pa_r \phi=0$. Hence, we reject this case as well.\\
\begin{center}
\underline{\textbf{8)}\hspace{1.5em}$\{\phi=\phi(v,r),\ \chi=\chi(r),\ f=f(\phi,\chi)=f(v,r)\}$}
\end{center}
$$\eqref{loc.118}\Ra-A'\pa_rf=0\Ra \left\{\begin{array}{c} A'=\pa_yA=0\vspace{0.5em}\\ or\vspace{0.5em}\\ \pa_rf=0\Ra \pa_\phi f=\pa_\chi f=0 \end{array}\right\}$$
\par Clearly, this case is also excluded.\\
\begin{center}
\underline{\textbf{9)}\hspace{1.5em}$\{\phi=\phi(v,r),\ \chi=\chi(y),\ f=f(\phi,\chi)=f(v,r,y)\}$}
\end{center}
$$\eqref{loc.115}\Ra (\pa_r\phi)^2+\pa_r^2f=0\Ra \pa_r^2f=-(\pa_r\phi)^2\Ra$$
\eq$\label{loc.143}
\boxed{\pa_rf=-\int(\pa_r\phi)^2dr+f_0(v,y)}$\\
$$\eqref{loc.118}\Ra \pa_y\pa_rf-A'\pa_rf=0\xRightarrow{\eqref{loc.143}}\pa_y\left[-\int(\pa_r\phi)^2dr+f_0\right]-A'\left[-\int(\pa_r\phi)^2dr+f_0\right]=0\Ra$$
$$\Ra -\underbrace{\pa_y\int(\pa_r\phi)^2dr}_0+\pa_yf_0+A'\int(\pa_r\phi)^2dr-A'f_0=0\Ra$$
\eq$\label{loc.144}
A'f_0-\pa_yf_0=A'\int(\pa_r\phi)^2dr$
\par The LHS of equation \eqref{loc.144} depends on $(v,y)-$coordinates while the RHS depends on $(v,r,y)-$co\-or\-di\-nates. Hence, equation \eqref{loc.144} leads to the following constraint.
\eq$\label{loc.145}
\pa_r\left[A'\int(\pa_r\phi)^2dr\right]=0\Ra A'(\pa_r\phi)^2=0\Ra\left\{\begin{array}{c}A'=\pa_yA=0\vspace{0.5em}\\ or\vspace{0.5em}\\ \pa_r\phi=0\end{array}\right\}$
\par Equation \eqref{loc.145} is not possible to be satisfied because it contradicts with the original assumptions. We cannot assume $\pa_\phi f=0$ because equation \eqref{loc.115} results to $\pa_r\phi=0$ as well. Moreover, $\pa_\chi f=0$ is not helpful either, because the only difference with the previous analysis is that $f_0(v,y)\ra f_0(v)$. 
\begin{center}
\underline{\textbf{10)}\hspace{1.5em}$\{\phi=\phi(v,y),\ \chi=\chi(v),\ f=f(\phi,\chi)=f(v,y)\}$}
\end{center}
$$\eqref{loc.120}\Ra \underbrace{A'\pa_yf-(\pa_y\phi)^2-\pa_y^2f}_{(v,y)-\text{dependent}}=\underbrace{f\left(3A''+\frac{e^{-2A}}{r}\pa_r^2m\right)}_{(v,r,y)-\text{dependent}}\xRightarrow{\pa_r(RHS)=0}$$
$$\Ra\pa_r\left[f\left(3A''+\frac{e^{-2A}}{r}\pa_r^2m\right)\right]=0\Ra fe^{-2A}\pa_r\left(\frac{\pa_r^2m}{r}\right)=0\xRightarrow{f\neq0}\pa_r\left(\frac{\pa_r^2m}{r}\right)=0\Ra$$
\eq$\label{loc.146}
\pa_r^3m-\frac{\pa_r^2m}{r}=0$
\vspace{0.4em}
\par Equation \eqref{loc.146} is identical to equation \eqref{loc.133}. Hence, combining \eqref{loc.146} with \eqref{loc.134} as we did in case 3, we obtain the mass function which is given by equation \eqref{loc.135}. Subsequently, substituting equation \eqref{loc.135} into \eqref{loc.119} we are led to the same result as in case 3, namely equation \eqref{loc.136} which is inconsistent in this case as well.\\
\begin{center}
\underline{\textbf{11)}\hspace{1.5em}$\{\phi=\phi(v,y),\ \chi=\chi(r),\ f=f(\phi,\chi)=f(v,r,y)\}$}
\end{center}
$$\eqref{loc.115}\Ra (\pa_r\chi)^2+\pa_r^2f=0\Ra \pa_r^2f=-(\pa_r\chi)^2\Ra$$
\eq$\label{loc.147}
\boxed{\pa_rf=-\int (\pa_r\chi)^2dr+f_0(v,y)}$\\
$$\eqref{loc.118}\Ra \pa_r\pa_yf-A'\pa_rf=0\xRightarrow{\eqref{loc.147}}\pa_y\left[-\int (\pa_r\chi)^2dr+f_0(v,y)\right]-A'\left[-\int (\pa_r\chi)^2dr+f_0(v,y)\right]=0\Ra$$
$$\Ra\pa_yf_0+A'\int (\pa_r\chi)^2dr-A'f_0=0\Ra \underbrace{A'f_0-\pa_yf_0}_{(v,y)-\text{dependent}}=\underbrace{A'\int (\pa_r\chi)^2dr}_{(r,y)-\text{dependent}}\xRightarrow{\pa_r(RHS)=0}$$
\eq$\label{loc.148}
\pa_r\left[A'\int (\pa_r\chi)^2dr\right]=0\Ra A'(\pa_r\chi)^2=0\Ra \left\{\begin{array}{c}A'=\pa_yA=0\vspace{0.5em}\\ or\vspace{0.5em}\\ \pa_r\chi=0\end{array}\right\}$
\par None of the two choices of equation \eqref{loc.148} can be satisfied. In addition, demanding either the constraint $\pa_\phi f=0$ or $\pa_\chi f=0$, nothing changes. Both sub-cases lead to the inconsistent result $\pa_r\chi=0$ as well.\\
\begin{center}
\underline{\textbf{12)}\hspace{1.5em}$\{\phi=\phi(v,y),\ \chi=\chi(y),\ f=f(\phi,\chi)=f(v,y)\}$}
\end{center}
\vspace{1em}
\par Using equations \eqref{loc.119}, \eqref{loc.120} and the expansion of the mass function given by equation \eqref{loc.134} we obtain the same differential equations and constraints as in case 10. Therefore, this case results to an inconsistency as well.

\begin{center}
\underline{\textbf{13)}\hspace{1.5em}$\{\phi=\phi(r,y),\ \chi=\chi(v),\ f=f(\phi,\chi)=f(v,r,y)\}$}
\end{center}
$$\eqref{loc.115}\Ra (\pa_r\phi)^2+\pa_r^2f=0\Ra \pa_r^2f=-(\pa_r\phi)^2\Ra \pa_rf=-\int(\pa_r\phi)^2dr+B(v,y)\Ra$$
\eq$\label{loc.149}
\boxed{f(v,r,y)=-\int\left[\int(\pa_r\phi)^2dr\right]dr+B(v,y)r+C(v,y)}$\\
\begin{align}\label{loc.150}
\eqref{loc.118}\Ra&\ \pa_y\phi\pa_r\phi+\pa_r\pa_yf-A'\pa_rf=0\xRightarrow{\eqref{loc.149}}\nonum\\
\Ra&\ \pa_r\phi\pa_y\phi+\pa_y\left[-\int(\pa_r\phi)^2dr+B(v,y)\right]-A'\left[-\int(\pa_r\phi)^2dr+B(v,y)\right]=0\Ra\nonum\\
\Ra&\ \pa_r\phi\pa_y\phi-2\int\pa_r\phi\ \pa_y\pa_r\phi\ dr+\pa_yB+A'\int(\pa_r\phi)^2dr-A'B=0\Ra\nonum\\
\Ra&\ \underbrace{A'B-\pa_yB}_{(v,y)-\text{dependent}}=\underbrace{\pa_r\phi\pa_y\phi+\int\pa_r\phi(A'\pa_r\phi-2\pa_y\pa_r\phi)dr}_{(r,y)-\text{dependent}}
\end{align}\\
\par As it is indicated by the last equation \eqref{loc.150}, its LHS depends on $(v,y)-$coordinates, while its RHS depends on $(r,y)-$coordinates. Consequently, we can derive two new constrains by demanding the partial derivative of the LHS with respect to $v$ to be zero and likewise the partial derivative of the RHS with respect to $r$ to be zero. Thus, we have:
\eq$\label{loc.151}
\eqref{loc.150}\xRightarrow{\pa_v(LHS)=0}\pa_v(A'B-\pa_yB)=0\Ra \boxed{\pa_v\pa_yB=A'\pa_vB}$\\
\begin{align}\label{loc.152}
\eqref{loc.150}&\xRightarrow{\pa_r(RHS)=0}\pa_r\left[\pa_r\phi\pa_y\phi+\int\pa_r\phi(A'\pa_r\phi-2\pa_y\pa_r\phi)dr\right]=0\Ra\nonum\\
&\Ra\pa_r^2\phi\pa_y\phi+\pa_r\phi\pa_r\pa_y\phi+\pa_r\phi(A'\pa_r\phi-2\pa_r\pa_y\phi)=0\Ra\nonum\\
&\Ra\pa_r^2\phi\pa_y\phi-\pa_r\phi\pa_r\pa_y\phi+A'(\pa_r\phi)^2=0\Ra\nonum\\
&\Ra-(\pa_r\phi\pa_r\pa_y\phi-\pa_r^2\phi\pa_y\phi)+A'(\pa_r\phi)^2=0\xRightarrow{\pa_r\phi\neq0}\nonum\\
&\Ra-\frac{\pa_r\phi\pa_r\pa_y\phi-\pa_r^2\phi\pa_y\phi}{(\pa_r\phi)^2}(\pa_r\phi)^2+A'(\pa_r\phi)^2=0\Ra\nonum\\
&\Ra-\pa_r\left(\frac{\pa_y\phi}{\pa_r\phi}\right)(\pa_r\phi)^2+A'(\pa_r\phi)^2=0\Ra A'=\pa_r\left(\frac{\pa_y\phi}{\pa_r\phi}\right)\Ra\nonum\\
&\Ra \frac{\pa_y\phi}{\pa_r\phi}=A'r+F(y)\Ra \boxed{\pa_y\phi=\pa_r\phi[A'(y)r+F(y)]}
\end{align}\\
\begin{align}
\eqref{loc.117}&\Ra \pa_y\pa_vf-A'\pa_vf-\frac{\pa_ym}{r}\pa_rf=\frac{f}{r}\left(\frac{\pa_ym}{r}+\pa_r\pa_ym\right)\Ra\nonum\\
&\Ra \pa_y\pa_vf-A'\pa_vf=\frac{f\pa_ym}{r^2}+\frac{f\pa_r\pa_ym}{r}+\frac{\pa_rf\pa_ym}{r}\xRightarrow{\eqref{loc.149}}\nonum\\
&\Ra r\pa_y\pa_vB+\pa_y\pa_vC-A'(r\pa_vB+\pa_vC)=\frac{f\pa_ym}{r^2}+\frac{\pa_r(f\pa_ym)}{r}\xRightarrow{\eqref{loc.151}}\nonum\\
&\Ra \bcancel{rA'\pa_vB}+\pa_v\pa_yC-\bcancel{rA'\pa_vB}-A'\pa_vC=\frac{f\pa_ym}{r^2}+\frac{\pa_r(f\pa_ym)}{r}\Ra\nonum
\end{align}
\begin{align}\label{loc.153}
&\Ra\underbrace{\pa_v\pa_yC-A'\pa_vC}_{(v,y)-\text{dependent}}=\underbrace{\frac{f\pa_ym}{r^2}+\frac{\pa_r(f\pa_ym)}{r}}_{(v,r,y)-\text{dependent}}\xRightarrow{\pa_r(RHS)=0}\pa_r\left[\frac{f\pa_ym}{r^2}+\frac{\pa_r(f\pa_ym)}{r}\right]=0\Ra\nonum\\ \nonum\\
&\Ra \bcancel{\frac{\pa_r(f\pa_ym)}{r^2}}-2\frac{f\pa_ym}{r^3}+\frac{\pa_r^2(f\pa_ym)}{r}-\bcancel{\frac{\pa_r(f\pa_ym)}{r^2}}=0\Ra \boxed{\pa_r^2(f\pa_ym)-2\frac{f\pa_ym}{r^2}=0}
\end{align}\\
\par The last differential equation, with respect to the function $f\pa_ym$, leads us to seek out the solution in the following form:\\
\eq$\label{loc.154}
\boxed{f\pa_ym=\sum_n b_n(v,y)r^n}$
\par Substituting equation \eqref{loc.154} into equation \eqref{loc.153} we get\\
$$\sum_n n(n-1)b_n(v,y)r^{n-2}-2\sum_n b_n(v,y)r^{n-2}=0\Ra$$
$$\Ra\sum_n \underbrace{(n^2-n-2)}_{\text{roots:\ }\{-1,2\}}b_n(v,y)r^{n-2}=0\Ra\left\{\begin{array}{c}b_n=0\ \ \forall n\neq\{-1,2\}\\ \\ b_{-1},b_2\ \ arbitrary\ \ functions\end{array}\right\}\Ra$$
$$f\pa_ym=\frac{b_{-1}(v,y)}{r}+b_2(v,y)r^2\xRightarrow[b_{-1}(v,y)\ra D(v,y)]{b_2(v,y)\ra E(v,y)}f\pa_ym=\frac{D(v,y)}{r}+E(v,y)r^2\xRightarrow[\substack{\text{we need the dependence}\\ \text{on the extra dimension}}]{\pa_ym\neq 0}$$\\
\eq$\label{loc.155}
\boxed{f(v,r,y)=\frac{1}{\pa_y[m(v,r,y)]}\left[\frac{D(v,y)}{r}+E(v,y)r^2\right]}$\\
\eq$\label{loc.156}
\eqref{loc.155}\Ra \boxed{\pa_vf=-\frac{\pa_v\pa_ym}{(\pa_ym)^2}\left(\frac{D}{r}+E\ r^2\right)+\frac{1}{\pa_ym}\left(\frac{\pa_vD}{r}+\pa_vE\ r^2\right)}$\\
\eq$\label{loc.157}
\eqref{loc.149}\Ra \boxed{\pa_vf=\pa_vB\ r+\pa_vC}$\\
\begin{align}\label{loc.158}
\eqref{loc.156},\eqref{loc.157}&\Ra \pa_vB\ r+\pa_vC=-\frac{\pa_v\pa_ym}{(\pa_ym)^2}\left(\frac{D}{r}+E\ r^2\right)+\frac{1}{\pa_ym}\left(\frac{\pa_vD}{r}+\pa_vE\ r^2\right)\Ra\nonum\\ \nonum\\
&\Ra\boxed{r(\pa_ym)^2(r\pa_vB+\pa_vC)=-\pa_v\pa_ym(D+E\ r^3)+\pa_ym(\pa_vD+\pa_vE\ r^3)}
\end{align}
\par The combination of equations \eqref{loc.158} and \eqref{loc.134} leads to the relation\\
$$r\left(\sum_n\pa_ya_n\ r^n\right)^2(r\pa_vB+\pa_vC)=-\left(\sum_n\pa_v\pa_ya_n\ r^n\right)(D+E\ r^3)+\left(\sum_n\pa_ya_n\ r^n\right)(\pa_vD+\pa_vE\ r^3)\Ra$$
\eq$\label{loc.159}
\boxed{\left(\sum_n\pa_ya_n\ r^{n+1}\right)^2\pa_vB+r\left(\sum_n\pa_ya_n\ r^n\right)^2\pa_vC=\sum_n\pa_v\left(\frac{D}{\pa_ya_n}\right)(\pa_ya_n)^2r^n+\sum_n\pa_v\left(\frac{E}{\pa_ya_n}\right)(\pa_ya_n)^2r^{n+3}}$
\par Equation \eqref{loc.159} can be mathematically consistent if and only if each one of its terms vanish. Concentrating our attention on the LHS, we can immediately deduce that in order to nullify both terms, it should be either $\{\pa_ya_n=0\ \ \forall n\}$ or $\{\pa_vB=0$ and $\pa_vC=0\}$. If $\{\pa_ya_n=0\ \ \forall n\}$, then we lose the desirable dependence of the mass function on the extra dimension. On the other hand, if $\{\pa_vB=0$ and $\pa_vC=0\}$ then equation \eqref{loc.149} becomes\\
$$f=-\int\left[\int(\pa_r\phi)^2dr\right]dr+rB(y)+C(y)\Ra \pa_vf=0\Ra \underbrace{\pa_\chi f}_{\neq0}\pa_v\chi=0\Ra \pa_v\chi=0$$\\
which is inconsistent with our assumption about the field $\chi=\chi(v)$. Therefore, this case also fails to offer us a solution to the problem. If we consider $\pa_\phi f=0$ then equation \eqref{loc.115} gives $\pa_r\phi=0$, which is inconsistent with the assumption about $\phi=\phi(r,y)$. Finally, we consider the sub-case in which $\pa_\chi f=0$. Hence, we have:
\eq$\label{loc159a}
\pa_\chi f=0\Ra \boxed{\pa_vf=0\xRightarrow{\eqref{loc.149}}\pa_v B=\pa_C=0}$
Moreover, the substitution of equation \eqref{loc159a} into \eqref{loc.117} leads to
\eq$\label{loc159b}
0=\frac{f\pa_ym}{r^2}+\frac{\pa_r(f\pa_ym)}{r}\Ra \frac{\pa_r(f\pa_ym)}{f\pa_ym}=-\frac{1}{r}\Ra \boxed{f\pa_ym=\frac{D(v,y)}{r}}$
From equation \eqref{loc.119} and \eqref{loc159a} we get
\begin{gather}-\pa_rf\left(\pa_rm+1-\frac{3m}{r}\right)=f\left(\pa_r^2m-2\frac{\pa_rm}{r}\right)\xRightarrow{\eqref{loc159b}}\nonum\\
\frac{D(\pa_r\pa_ym\ r+\pa_ym)}{(r\pa_ym)^2}\left(\pa_rm+1-\frac{3m}{r}\right)=\frac{D}{r\pa_ym}\left(\pa_r^2m-2\frac{\pa_rm}{r}\right)\Ra\nonum\\
(r\pa_r\pa_ym+\pa_ym)\left(\pa_rm+1-\frac{3m}{r}\right)=r\pa_ym\left(\pa_r^2m-2\frac{\pa_rm}{r}\right)\xRightarrow{\eqref{loc.134}}\nonum\\
\sum_n(\pa_ya_n)(n+1)r^n\left(1+\sum_\ell a_\ell(\ell-3)r^{\ell-1}\right)=\sum_n(\pa_ya_n)r^n\sum_\ell a_\ell\underbrace{\left[\ell(\ell-1)-2\ell\right]}_{\ell(\ell-3)}r^{\ell-1}\Ra\nonum\\
\sum_n(\pa_ya_n)r^n\left[(n+1)+\sum_\ell a_\ell(n+1)(\ell-3)r^{\ell-1}-\sum_\ell a_\ell \ell(\ell-3)r^{\ell-1}\right]\Ra\nonum\\
\label{loc159c}
\boxed{\sum_n(\pa_ya_n)r^n\left[(n+1)+\sum_\ell a_\ell(\ell-3)(n+1-\ell)r^{\ell-1}\right]=0}
\end{gather}
We demand that $\pa_ya_n\neq0\ \forall n$, thus equation \eqref{loc159c} can only hold if $n,\ell$ take the values 1 and 3. Consequently, we obtain:
\begin{gather}
\sum_{n=1,3}(\pa_ya_n)r^n\left[(n+1)+a_1(-2)(n+2)\right]=0\Ra\nonum\\
(\pa_ya_1)r(2-6a_1)+(\pa_ya_3)r^3(4-10a_1)=0\Ra\left\{\begin{array}{c}
a_1=\frac{1}{3}\vspace{0.5em}\\ 
\text{and}\vspace{0.5em}\\  
a_1=\frac{5}{2} 
\end{array}\right\}\Ra\ \text{rejected}
\end{gather}
If on the other hand, we assume that $n,\ell$ take only the value 1. Then we have:
$$(\pa_ya_1)r(2-6a_1)=0\Ra a_1=\frac{1}{3}\Ra\pa_ya_1=0$$
which is also an undesirable result.

\begin{center}
\underline{\textbf{14)}\hspace{1.5em}$\{\phi=\phi(r,y),\ \chi=\chi(r),\ f=f(\phi,\chi)=f(r,y)\}$}
\end{center}
\begin{align}\label{loc.160}
\eqref{loc.117}&\Ra -\frac{\pa_ym}{r}\pa_rf=\frac{f}{r}\left(\frac{\pa_ym}{r}+\pa_r\pa_ym\right)\Ra \frac{f\pa_ym}{r^2}+\frac{\pa_r(f\pa_ym)}{r}=0\Ra\nonum\\
&\Ra \pa_r(f\pa_ym)=-\frac{f\pa_ym}{r}\xRightarrow[\pa_ym\neq 0]{f\neq0} \frac{1}{f\pa_ym}\frac{\pa(f\pa_ym)}{\pa r}=-\frac{1}{r}\Ra \nonum\\
&\Ra \frac{\pa[\ln(f\pa_ym)]}{\pa r}=-\frac{1}{r}\Ra \ln(f\pa_ym)=-\ln r+B(v,y)\xRightarrow{e^{B(v,y)}\ra C(v,y)}\nonum\\
&\Ra f\pa_ym=\frac{C(v,y)}{r}\xRightarrow{\pa_ym\neq 0} \boxed{f=f(r,y)=\frac{C(v,y)}{\pa_ym(v,r,y)\ r}}
\end{align}
\par Using now equation \eqref{loc.119} together with \eqref{loc.160} we are led to the same result as in the sub-case $\pa_\chi f=0$ of case 13, which was presented before.\\
\begin{enumerate}[(i)]
\item \begin{flushleft}
\underline{$\pa_\phi f=0$:}
\end{flushleft}
From equation \eqref{loc.118} we obtain:
$$-A'\pa_r f=0\Ra -A'\pa_\chi f\underbrace{\pa_r\chi}_{\neq 0}=0\Ra\left\{\begin{array}{c}A'=0\Ra \pa_yA=0\vspace{0.5em}\\ or \vspace{0.5em}\\ \pa_\chi f=0\Ra \pa_\chi f=\pa_\phi f=0\end{array}\right\}\Ra \text{rejected}$$
\item
\begin{flushleft}
\underline{$\pa_\chi f=0$:}
\end{flushleft}
Using the same equations and performing the same steps as in the original case, we can show that this sub-case provides us with exactly the same negative result as the initial case 14.\\
\end{enumerate}

\begin{center}
\underline{\textbf{15)}\hspace{1.5em}$\{\phi=\phi(r,y),\ \chi=\chi(y),\ f=f(\phi,\chi)=f(r,y)\}$}
\end{center}
\par In this case, similarly to the case 14, the coupling function $f$ depends on $(r,y)-$coordinates. Therefore, it is straightforward to verify that equations \eqref{loc.117} and \eqref{loc.119} result exactly to the same differential equations as in case 14 and subsequently to the same inappropriate form of the mass function. Consequently, case 15 is rejected as well.\\
\begin{center}
\underline{\textbf{16)}\hspace{1.5em}$\{\phi=\phi(v,r),\ \chi=\chi(v,r),\ f=f(\phi,\chi)=f(v,r)\}$}
\end{center}
$$\eqref{loc.118}\Ra -A'\pa_rf=0\Ra\left\{\begin{array}{c} A'=\pa_yA=0\vspace{0.5em}\\ or\vspace{0.5em}\\ \pa_rf=0\Ra \pa_\phi f=\pa_\chi f=0\end{array}\right\}$$
\par Obviously, this is also not a viable solution to the problem.\\
\begin{center}
\underline{\textbf{17)}\hspace{1.5em}$\{\phi=\phi(v,r),\ \chi=\chi(v,y),\ f=f(\phi,\chi)=f(v,r,y)\}$}
\end{center}
$$\eqref{loc.115}\Ra (\pa_r\phi)^2+\pa_r^2f=0\Ra \pa_rf=-\int(\pa_r\phi)^2\ dr+B(v,y)\Ra $$
\eq$\label{loc.171}
\boxed{f=-\int \left[\int(\pa_r\phi)^2\ dr\right]dr+r\ B(v,y)+C(v,y)}$
\begin{align}
\eqref{loc.118}&\Ra \pa_y\pa_rf-A'\pa_rf=0\xRightarrow{\eqref{loc.171}}\pa_y\left[-\int(\pa_r\phi)^2\ dr+B\right]-A'\left[-\int(\pa_r\phi)^2\ dr+B\right]=0\Ra\nonum\\
&\Ra -2\int \pa_r\phi\ \underbrace{\pa_y\pa_r\phi}_{0}\ dr+\pa_yB+A'\int(\pa_r\phi)^2\ dr-A'B=0\Ra\nonum\\
&\Ra \underbrace{A'B-\pa_yB}_{(v,y)-\text{dependent}}=\underbrace{A'\int(\pa_r\phi)^2\ dr}_{(v,r,y)-\text{dependent}}\xRightarrow{\pa_r(RHS)=0}\pa_r\left[A'\int(\pa_r\phi)^2\ dr\right]=0\Ra\nonum\\
&\Ra A'(\pa_r\phi)^2=0\Ra\left\{\begin{array}{c}A'=0\vspace{0.5em}\\ or\vspace{0.5em}\\ \pa_r\phi=0\end{array}\right\}\nonum
\end{align}
\par None of the choices which are depicted above is in agreement with our assumptions, hence, in this case the field equation \eqref{loc.118} is inconsistent.\\
\begin{enumerate}[(i)]
\item
\begin{flushleft}
\underline{$\pa_\chi f=0$:}
\end{flushleft}
$$\eqref{loc.118}\Ra -A'\pa_r f=0\Ra \left\{\begin{array}{c}A'=0\vspace{0.5em}\\ or\vspace{0.5em}\\ \pa_r f=0\Ra \pa_\phi f \underbrace{\pa_r\phi}_{\neq 0}=0\Ra \pa_\phi f=0=\pa_\chi f\end{array}\right\}\Ra \text{rejected}$$
\item
\begin{flushleft}
\underline{$\pa_\phi f=0$:}
\end{flushleft}
\begin{gather}
\eqref{loc.119}\Ra-\pa_v f=f\left(\pa_r^2m-\frac{2}{r}\pa_rm\right)\xRightarrow{f\neq 0}\underbrace{-\frac{\pa_vf}{f}}_{(v,y)-dependent}=\underbrace{\pa_r^2m-\frac{2}{r}\pa_rm}_{(v,r,y)-dependent}\xRightarrow{\pa_r(RHS)=0}\nonum\\
\pa_r\left(\pa_r^2m-\frac{2}{r}\pa_rm\right)=0\Ra \pa_r^3m-\frac{2}{r}\pa_r^2m+\frac{2}{r^2}\pa_rm=0\xRightarrow{\eqref{loc.134}}\nonum\\
\sum_n a_nn(n-1)(n-2)r^{n-3}-2\sum_na_nn(n-1)r^{n-3}+2\sum_na_nnr^{n-3}=0\Ra\nonum\\
\sum_na_n\underbrace{\left[n(n-1)(n-2)+2n-2n(n-1)\right]}_{n(n-2)(n-3)}r^{n-3}=0\Ra\sum_na_nn(n-2)(n-3)r^{n-3}=0\Ra\nonum\\
\label{loc171a}
\boxed{m(v,,r,y)=a_0(v,y)+a_2(v,y)r^2+a_3(v,y)r^3}
\end{gather}
Combining equations \eqref{loc.117} and \eqref{loc171a} we have:
\begin{gather}
\pa_y\chi\pa_v\chi+\pa_y\pa_vf-A'\pa_vf=\frac{f}{r}\left(\frac{\pa_ya_0}{r}+\pa_ya_2\ r+\pa_ya_3\ r^2+2r\pa_ya_2+3r^2\pa_ya_3\right)\Ra\nonum\\
\underbrace{\frac{\pa_y\chi\pa_v\chi+\pa_y\pa_vf-A'\pa_vf}{f}}_{(v,y)-dependent}=\underbrace{\frac{\pa_ya_0}{r^2}+3\pa_ya_2+4r\pa_ya_3}_{(v,r,y)-dependent}\xRightarrow[\text{consistency}]{\text{for}}\nonum\\
\left\{\begin{array}{c}\pa_ya_0=0\vspace{0.5em}\\ \pa_ya_2=0\vspace{0.5em}\\ \pa_ya_3=0\end{array}\right\}\Ra \boxed{\pa_ym=0}\nonum
\end{gather}
This sub-case should also be excluded, because with $\pa_ym=0$ the localization is not possible.
\end{enumerate}

\begin{center}
\underline{\textbf{18)}\hspace{1.5em}$\{\phi=\phi(v,r),\ \chi=\chi(r,y),\ f=f(\phi,\chi)=f(v,r,y)\}$}
\end{center}
$$\eqref{loc.115}\Ra (\pa_r\phi)^2+(\pa_r\chi)^2+\pa_r^2f=0\Ra \pa_rf=-\int\left[(\pa_r\phi)^2+(\pa_r\chi)^2\right]dr+B(v,y)\Ra$$
\eq$\label{loc.172}
\boxed{f=-\int\left\{\int\left[(\pa_r\phi)^2+(\pa_r\chi)^2\right]dr\right\}dr+rB(v,y)+C(v,y)}$\\
\begin{align}\label{loc.173}
\eqref{loc.118}&\Ra \pa_r\chi\pa_y\chi+\pa_y\pa_rf-A'\pa_rf=0\nonum\\
&\Ra \pa_r\chi\pa_y\chi+\pa_y\left\{-\int\left[(\pa_r\phi)^2+(\pa_r\chi)^2\right]dr+B\right\}-A'\left\{-\int\left[(\pa_r\phi)^2+(\pa_r\chi)^2\right]dr+B\right\}=0\nonum\\
&\Ra \pa_r\chi\pa_y\chi-\int\left(2\ \pa_r\phi\ \underbrace{\pa_y\pa_r\phi}_{0}+2\ \pa_r\chi\ \pa_y\pa_r\chi\right)dr+\pa_yB+A'\int\left[(\pa_r\phi)^2+(\pa_r\chi)^2\right]dr-A'B=0\nonum\\
&\Ra \underbrace{A'B-\pa_yB}_{(v,y)-\text{dependent}}=\underbrace{\pa_r\chi\pa_y\chi+\int\left[\pa_r\chi(A'\pa_r\chi-2\ \pa_y\pa_r\chi)+A'(\pa_r\phi)^2\right]dr}_{(v,r,y)-\text{dependent}}
\end{align}
\begin{align}\label{loc.174}
\eqref{loc.173}&\xRightarrow{\pa_r(RHS)=0}\pa_r\left\{\pa_r\chi\pa_y\chi+\int\left[\pa_r\chi(A'\pa_r\chi-2\ \pa_y\pa_r\chi)+A'(\pa_r\phi)^2\right]dr\right\}=0\nonum\\
&\Ra \pa_r^2\chi\ \pa_y\chi+\pa_r\chi\ \pa_r\pa_y\chi+A'(\pa_r\chi)^2-2\ \pa_r\chi\ \pa_y\pa_r\chi+A'(\pa_r\phi)^2=0\nonum\\
&\Ra \underbrace{\pa_r^2\chi\ \pa_y\chi-\pa_r\chi\ \pa_r\pa_y\chi+A'(\pa_r\chi)^2}_{(r,y)-\text{dependent}}=\underbrace{-A'(\pa_r\phi)^2}_{(v,r,y)-\text{dependent}}\xRightarrow{\pa_v(RHS)=0} \pa_v\left[A'(\pa_r\phi)^2\right]=0\nonum\\
&\Ra 2A'\pa_r\phi\ \pa_v\pa_r\phi=0\xRightarrow[\pa_r\phi\neq0]{A'\neq0}\pa_v\pa_r\phi=0\Ra \pa_v[\pa_r\phi(v,r)]=0\Ra\pa_r\phi(v,r)=C_1(r)\nonum\\
&\Ra \phi(v,r)=\int C_1(r)\ dr+C_2(v)\Ra\boxed{\phi(v,r)=\phi_1(r)+\phi_2(v)}
\end{align}\\
$$\eqref{loc.172}\Ra\pa_vf=\bigintsss \left[\bigintsss\left(2\ \pa_r\phi\ \underbrace{\pa_v\pa_r\phi}_{0}+2\ \pa_r\chi\ \underbrace{\pa_v\pa_r\chi}_{0}\right)dr\right]dr+r\pa_vB+\pa_vC\Ra$$
\eq$\label{loc.175}
\boxed{\pa_vf=r\pa_vB+\pa_vC}$\\
$$\eqref{loc.173}\xRightarrow{\eqref{loc.174}}\underbrace{A'B-\pa_yB}_{(v,y)-\text{dependent}}=\underbrace{\pa_r\chi\pa_y\chi+\bigintsss\left[\pa_r\chi(A'\pa_r\chi-2\ \pa_y\pa_r\chi)+A'\overbrace{(\pa_r\phi)^2}^{r-\text{dependent}}\right]dr}_{(r,y)-\text{dependent}}\xRightarrow{\pa_v(LHS)=0}$$
\eq$\label{loc.176}
\Ra\pa_v\left[A'B-\pa_yB\right]=0\Ra \boxed{A'\pa_vB=\pa_v\pa_yB}$\\
\begin{align}\label{loc.177}
\eqref{loc.117}&\Ra \pa_y\pa_vf-A'\pa_vf-\frac{\pa_ym}{r}\pa_rf=\frac{f}{r}\left(\frac{\pa_ym}{r}+\pa_r\pa_ym\right)\xRightarrow{\eqref{loc.175}}\nonum\\
&\Ra r\pa_y\pa_vB+\pa_y\pa_vC-A'(r\pa_vB+\pa_vC)=\frac{f\pa_ym}{r^2}+\frac{\pa_r(f\pa_ym)}{r}\xRightarrow{\eqref{loc.176}}\nonum\\
&\Ra\underbrace{\pa_y\pa_vC-A'\pa_vC}_{(v,y)-\text{dependent}}=\underbrace{\frac{f\pa_ym}{r^2}+\frac{\pa_r(f\pa_ym)}{r}}_{(v,r,y)-\text{dependent}}
\end{align}

$$\eqref{loc.177}\xRightarrow{\pa_r(RHS)=0}\pa_r\left[\frac{f\pa_ym}{r^2}+\frac{\pa_r(f\pa_ym)}{r}\right]=0\Ra \bcancel{\frac{\pa_r(f\pa_ym)}{r^2}}-\frac{2f\pa_ym}{r^3}+\frac{\pa_r^2(f\pa_ym)}{r}-\bcancel{\frac{\pa_r(f\pa_ym)}{r^2}}=0\Ra$$
\eq$\label{loc.178}
\boxed{\pa_r^2(f\pa_ym)-\frac{2}{r^2}(f\pa_ym)=0}$\\
\par Equation \eqref{loc.178} is identical to equation \eqref{loc.153} and equation \eqref{loc.175} is identical to \eqref{loc.157} of case 13; we can henceforth follow the same steps as in case 13 in order to show that this case should also be rejected.\\
\begin{enumerate}[(i)]
\item
\begin{flushleft}
\underline{$\pa_\phi f=0$:}
\end{flushleft}
$$\eqref{loc.115}\Ra(\pa_r\phi)^2+(\pa_r\chi)^2+\pa_r^2f=0\Ra \pa_rf=-\int\left[(\pa_r\phi)^2+(\pa_r\chi)^2\right]dr+B(y)\Ra$$
\eq$\label{loc181}
\boxed{f(r,y)=-\int\left\{\int\left[(\pa_r\phi)^2+(\pa_r\chi)^2\right]dr\right\}dr+r\ B(y)+C(y)}$\\
$$\eqref{loc181}\Ra\pa_vf=0\Ra-\int\left(\int 2\pa_r\phi\pa_v\pa_r\phi\ dr\right)dr=0\Ra \pa_v\pa_r\phi=0\Ra$$
\eq$\label{loc182}
\boxed{\phi(v,r)=\phi_1(v)+\phi_2(r)}$\\
\eq$\label{loc183}
\eqref{loc.117}\Ra 0=\frac{f\pa_ym}{r^2}+\frac{\pa_r(f\pa_ym)}{r}=0\xRightarrow{\pa_ym\neq0} \boxed{f(r,y)=\frac{D(v,y)}{r\ \pa_y[m(v,r,y)]}}$
where the following constraint is necessary to be satisfied:
\eq$\label{loc184}
\pa_vf=0\Ra\boxed{\pa_v\left(\frac{D}{\pa_ym}\right)=0}$\\
$$\eqref{loc.121}\Ra \frac{e^{-2A}}{r}\left(1-\frac{2m}{r}\right)\pa_rf+A'\pa_yf-(\pa_y\chi)^2-\pa_y^2f=f\left(3A''+\frac{2e^{-2A}}{r^2}\pa_rm\right)\xRightarrow{\pa_v(\ )}$$
\eq$\label{loc185}
-\frac{2e^{-2A}}{r^2}\pa_vm\ \pa_rf=f\frac{2e^{-2A}}{r^2}\pa_v\pa_rm\Ra \boxed{-\pa_vm\ \pa_rf=f\pa_v\pa_rm }$\\
\begin{gather}
\eqref{loc185}\xRightarrow[\eqref{loc183},\eqref{loc184}]{\eqref{loc.134}}-\pa_vm\left[-\frac{D(\pa_r\pa_ym\ r+\pa_ym)}{(r\pa_ym)^2}\right]=\frac{D}{r\pa_ym}\pa_v\pa_rm\Ra\nonum\\ \nonum\\
\pa_vm(r\pa_r\pa_ym+\pa_ym)=r\pa_ym\ \pa_v\pa_rm\Ra \nonum\\ \nonum\\
\left(\sum_n(\pa_va_n)r^n\right)\left(\sum_\ell(\pa_ya_\ell)(\ell+1)r^\ell\right)=\left(\sum_n(\pa_ya_n)r^n\right)\left(\sum_\ell(\pa_va_\ell) \ell\ r^\ell\right)\Ra\nonum\\ \nonum\\
\label{loc186}
\sum_{n,\ell}(\pa_va_n)(\pa_ya_\ell)(\ell-n+1)r^{n+\ell}=0\Ra\{\pa_va_n=0\ \forall n\}\Ra \boxed{\pa_vm=0}
\end{gather}
Thus, the combination of equation \eqref{loc184} and \eqref{loc186} results to
\eq$\label{loc187}
\boxed{\pa_vD=0}$
Hence, we obtain:
\eq$\label{loc188}
\eqref{loc183}\xRightarrow[\eqref{loc187}]{\eqref{loc186}}\boxed{f(r,y)=\frac{D(y)}{r\ \pa_y\left[m(r,y)\right]}}$\\
$$\eqref{loc.119}\Ra r\pa_r\phi\pa_v\phi-\pa_rf\left(\pa_rm-\frac{3m}{r}+1\right)=f\left(\pa_r^2m-\frac{2}{r}\pa_rm\right)\xRightarrow{\pa_v(\ )}r\underbrace{\pa_v\pa_r\phi}_0\pa_v\phi+r\pa_r\phi\pa_v^2\phi=0\Ra$$
\eq$\label{loc189}
\pa_v^2\phi=0\Ra\pa_v^2\phi_1=0\Ra\boxed{\phi_1(v)=\ome v+\xi}$
where $\ome$ and $\xi$ are constants.
\eq$\label{loc1810}
\pa_r\phi\pa_v\phi=\frac{\pa_rf}{r}\left(\pa_rm-\frac{3m}{r}+1\right)+\frac{f}{r}\left(\pa_r^2m-\frac{2}{r}\pa_rm\right)$
\begin{align}\label{loc1811}
\eqref{loc.116}\Ra(\pa_v\phi)^2+\left(1-\frac{2m}{r}\right)\pa_r\phi\pa_v\phi&=-\frac{e^{-2A}}{r}\left(\pa_ym\ \pa_yf+f\pa_y^2m+4A'\ f\pa_ym\right)\nonum\\
&=-\frac{e^{-2A}}{r}\left(e^{4A}\pa_ym\ \pa_yf+e^{4A}f\pa_y^2m+4A'e^{4A}f\pa_ym\right)\nonum\\
&=-\frac{e^{-2A}}{r}\pa_y\left(e^{4A}f\pa_ym\right)=-\frac{e^{-2A}}{r}\pa_y\left(e^{4A}\frac{D}{r}\right)\nonum\\
&=-\frac{e^{-2A}}{r^2}\pa_y\left(De^{4A}\right)
\end{align}
Substituting equation \eqref{loc1810} into \eqref{loc1811} and using also equations \eqref{loc.134}, \eqref{loc182} and \eqref{loc189}, we obtain:
\begin{gather}
\ome^2+\frac{e^{-2A}}{r^2}\pa_y\left(De^{4A}\right)+\left(1-\frac{2m}{r}\right)\left[-D\left(\frac{r\ \pa_r\pa_ym+\pa_ym}{(r\pa_ym)^2}\right)\left(\pa_rm+1-\frac{3m}{r}\right)\right.\nonum\\
\left.+\frac{D}{r^2\pa_ym}\left(\pa_r^2m-\frac{2}{r}\pa_rm\right)\right]=0\Ra
\end{gather}
\begin{gather}
\ome^2+\frac{e^{-2A}}{r^2}\pa_y\left(De^{4A}\right)+\left(1-\frac{2m}{r}\right)\frac{D}{(r\pa_ym)^2}\left[-\left(r\ \pa_r\pa_ym+\pa_ym\right)\left(\pa_rm+1-\frac{3m}{r}\right)\right.\nonum\\
\left.\pa_ym\left(\pa_r^2m-\frac{2}{r}\pa_rm\right)\right]=0\Ra\nonum\\ \nonum\\
\ome^2(r\pa_ym)^2+e^{-2A}\pa_y(De^{4A})(\pa_ym)^2-(\pa_r\pa_ym\ r+\pa_ym)\left(\pa_rm+1-\frac{3m}{r}\right)\nonum\\
+\frac{2m}{r}(\pa_r\pa_ym\ r+\pa_ym)\left(\pa_rm+1-\frac{3m}{r}\right)
\pa_ym\left(\pa_r^2m-\frac{2}{r}\pa_rm\right)
-\frac{2m}{r}\pa_ym\left(\pa_r^2m-\frac{2}{r}\pa_rm\right)=0\Ra\nonum
\end{gather}
\begin{gather}
\sum_{n,\ell,k}\pa_ya_n\left\{(n+1)+\left[(\pa_ya_\ell)\left(\ome^2\ r^2+e^{-2A}\pa_y(e^{4A}D)\right)+a_\ell\left((n+1)(5-\ell)r^{-1}+\ell(\ell-3)r^{-2}\right)\right.\right.\nonum\\
\label{loc1812}
\left.\left.+2a_\ell\ a_k(\ell-3)\left((n+1)r^{k-2}-\ell\ r^{k-3}\right)\right]r^\ell\right\}r^n=0
\end{gather}
Demanding $\pa_ya_n\neq0$ the above equation is impossible to be satisfied. Therefore, in order to have a consistent field equation we need to allow $\pa_ya_n=0$ at least for some values of $n$, but this is catastrophic for the localization of the 5-dimensional black hole.\\
\item
\begin{flushleft}
\underline{$\pa_\chi f=0$:}
\end{flushleft}
$$\eqref{loc.115}\Ra (\pa_r\phi)^2+(\pa_r\chi)^2+\pa_r^2f=0\xRightarrow{\pa_y(\ )}2\pa_r\phi\underbrace{\pa_y\pa_r\phi}_0+2\pa_r\chi\ \pa_y\pa_r\chi+\underbrace{\pa_y\pa_r^2f}_0=0\Ra\pa_y\pa_r\chi=0\Ra$$
\eq$\label{loc1813}
\boxed{\chi(r,y)=\chi_1(r)+\chi(y)}$\\
$$\eqref{loc.118}\Ra\pa_r\chi\pa_y\chi-A'\pa_rf=0\xRightarrow{\pa_v(\ )}-A'\pa_v\pa_rf=0\Ra$$
\eq$\label{loc1814}
\boxed{\pa_v\pa_rf=0\Ra f(v,r)=f_1(v)+f_2(r)}$\\
$$\eqref{loc.118}\Ra \pa_r\chi\pa_y\chi-A'\pa_rf=0\Ra\chi_1'(r)\chi_2'(y)-A'(y)f_2'(r)=0\Ra\frac{\chi_2'(y)}{A'(y)}=\frac{f_2'(r)}{\chi_1'(r)}=\lam\Ra$$
\begin{gather}\label{loc1815}
\boxed{\chi_2(y)=\lam\ A(y)+\kappa}\\ \nonum\\
\label{loc1816}
\boxed{f_2(r)=\lam\ \chi_1(r)+\tilde{\kappa}}
\end{gather}\\
\eq$\label{loc1817}
\eqref{loc.115}\xRightarrow[\eqref{loc1816}]{\pa_v(\ )}2\pa_r\phi\ \pa_v\pa_r\phi=0\Ra \boxed{\pa_v\pa_r\phi=0\Ra \phi(v,r)=\phi_1(v)+\phi_2(r)}$\\
$$\eqref{loc.117}\Ra \underbrace{-A'\pa_vf}_{(v,y)-dependent}=\underbrace{\frac{f\pa_ym}{r^2}+\frac{\pa_r(f\pa_ym)}{r}}_{(v,r,y)-dependent}\xRightarrow{\pa_r(RHS)=0}\pa_r^2(f\pa_ym)-\frac{2f\pa_ym}{r^2}=0\xRightarrow{\pa_ym\neq0}$$
\eq$\label{loc1818}
\boxed{f(v,r)=\frac{1}{\pa_y[m(v,r,y)]}\left[\frac{B(v,y)}{r}+C(v,y)\ r^2\right]}$\\
\eq$\label{loc1819}
\eqref{loc1818}\xRightarrow{\pa_yf=0}\left\{\begin{array}{c}\pa_y\left(\frac{B}{\pa_ym}\right)=0\\ \\
\pa_y\left(\frac{C}{\pa_ym}\right)=0\end{array}\right\}$\\
$$\pa_y\left(\frac{B}{\pa_ym}\right)=0\Ra\pa_B\pa_ym-B\pa_y^2m=0\Ra\sum_n\left(\pa_yB\pa_ya_n-B\pa_y^2a_n\right)r^n=0\Ra$$
\eq$\label{loc1820}
\frac{\pa_yB}{B}=\frac{\pa_y(\pa_ya_n)}{\pa_ya_n}\Ra\boxed{\frac{\pa_y[a_n(v,y)]}{B(v,y)}=b_n(v)}$
Similarly, we have:
\eq$\label{loc1821}
\boxed{\frac{\pa_y[a_n(v,y)]}{C(v,y)}=c_n(v)}$
Equation \eqref{loc1818} with the use of equations \eqref{loc1820} and \eqref{loc1821} gets the following form:
\eq$\label{l22}
\boxed{f(v,r)=\frac{1}{\sum_n d_n(v)r^{n+1}}+\frac{1}{\sum_n c_n(v)r^{n-2}}}$
Finally, from equation \eqref{loc1814} the relation $\pa_v\pa_rf=0$ should be satisfied. Thus, it is:
\begin{align}
\pa_v\pa_rf=0\xRightarrow{\eqref{l22}}&\ \frac{\sum_n d_n'(n+1)r^{n}}{\left(\sum_n d_nr^{n}\right)^2}-\frac{2\sum_n d_n(n+1)r^{n}\sum_\ell d_\ell'\ r^{\ell+1}}{\left(\sum_n d_nr^{n+1}\right)^3}\nonum\\
&+\frac{\sum_n c_n'(n-2)r^{n-3}}{\left(\sum_n c_nr^{n-2}\right)^2}-\frac{2\sum_n c_n(n-2)r^{n-3}\sum_\ell c_\ell'\ r^{\ell-2}}{\left(\sum_n c_nr^{n-2}\right)^3}=0\Ra\nonum
\end{align}\\
\eq$\label{l23}
\left\{\begin{array}{c}d_n'(v)=0\ \forall n\vspace{0.5em}\\and\vspace{0.5em} \\ c_n'(v)=0\ \forall n\end{array}\right\}\Ra \boxed{\pa_vf=0}$
Therefore, this sub-case should be rejected as well. In the previous analysis we silently assumed that both $B(v,y)$ and $C(v,y)$ were not equal to zero. In case that one of these functions is zero, we are led to exactly the same result, namely $\pa_vf=0$. This happens because even if we nullify one of the functions $B(v,y)$ or $C(v,y)$ in equation \eqref{loc1818}, none of the following steps is going to be changed.\\
\end{enumerate}
\begin{center}
\underline{\textbf{19)}\hspace{1.5em}$\{\phi=\phi(v,y),\ \chi=\chi(v,y),\ f=f(\phi,\chi)=f(v,y)\}$}
\end{center}
$$\eqref{loc.120}\Ra \underbrace{A'\pa_yf-(\pa_y\phi)^2-(\pa_y\chi)^2+\pa_y^2f}_{(u,y)-dependent}=\underbrace{f\left(3A''+\frac{e^{-2A}}{r}\pa_r^2m\right)}_{(v,r,y)-dependent}\xRightarrow{\pa_r(RHS)=0}fe^{-2A}\pa_r\left(\frac{\pa_r^2m}{r}\right)=0\Ra$$
\eq$\label{l24}
\pa_r^3m-\frac{\pa_r^2m}{r}=0\Ra \boxed{m(v,r,y)=a_0(v,y)+a_1(v,y)r+a_3(v,y)r^3}$\\
$$\eqref{loc.119}\Ra-\pa_vf=f\left(\pa_r^2m-\frac{2}{r}\pa_rm\right)\xRightarrow{\eqref{l24}}-\pa_vf=-f\frac{2a_1}{r}\Ra\boxed{\frac{\pa_v[f(v,y)]}{f(v,y)}=\frac{2a_1(v,y)}{r}}$$
\par It is clear that the last equation cannot be true, because of the factor $1/r$ in the RHS. For the sub-cases $\pa_\chi f=0$ and $\pa_\phi f=0$ nothing changes. The previous equations have exactly the same form, thus, the sub-cases are rejected as well.\\
\begin{center}
\underline{\textbf{20)}\hspace{1.5em}$\{\phi=\phi(v,y),\ \chi=\chi(r,y),\ f=f(\phi,\chi)=f(v,r,y)\}$}
\end{center}
$$\eqref{loc.115}\Ra (\pa_r\chi)^2+\pa_r^2f=0\Ra \pa_rf=-\int (\pa_r\chi)^2dr+B(v,y)\Ra$$
\eq$\label{loc.179}
\boxed{f=-\int \left[\int(\pa_r\chi)^2dr\right]dr+rB(v,y)+C(v,y)}$\\
$$\eqref{loc.118}\Ra \pa_r\chi\pa_y\chi+\pa_r\pa_yf-A'\pa_rf=0\xRightarrow{\eqref{loc.179}}$$
$$\pa_r\chi\pa_r\chi-2\int\pa_r\chi\ \pa_y\pa_r\chi\ dr+B-A'\left[-\int(\pa_r\chi)^2dr+\pa_yB\right]=0\Ra$$
\eq$\label{loc.180}
\underbrace{A'B-\pa_yB}_{(v,y)-\text{dependent}}=\underbrace{\pa_r\chi\pa_y\chi+\int\pa_r\chi(A'\pa_r\chi-2\ \pa_y\pa_r\chi)dr}_{(r,y)-\text{dependent}}$\\
\eq$\label{loc.181}
\eqref{loc.180}\xRightarrow{\pa_v(LHS)=0}\pa_v(A'B-\pa_yB)=0\Ra \boxed{A'\pa_vB=\pa_v\pa_yB}$\\
$$\eqref{loc.180}\xRightarrow{\pa_r(RHS)=0}\pa_r^2\chi\pa_y\chi+\pa_r\chi\ \pa_r\pa_y\chi+\pa_r\chi(A'\pa_r\chi-2\ \pa_r\pa_y\chi)=0\Ra$$
$$\pa_r^2\chi\pa_y\chi-\pa_r\chi\ \pa_r\pa_y\chi+A'(\pa_r\chi)^2=0\xRightarrow{\pa_r\phi\ne\chi0}-\pa_r\left(\frac{\pa_y\chi}{\pa_r\chi}\right)(\pa_r\chi)^2+A'(\pa_r\chi)^2=0\Ra$$
\eq$\label{loc.182}
\pa_r\left(\frac{\pa_y\chi}{\pa_r\chi}\right)=A'\Ra \boxed{\pa_y\chi=\pa_r\chi[A'\ r+F(y)]}$

\begin{align}\label{loc.183}
\eqref{loc.117}&\Ra \pa_v\phi\ \pa_y\phi+\pa_y\pa_vf-A'\pa_vf-\frac{\pa_ym\ \pa_rf}{r}=\frac{f}{r}\left(\frac{\pa_ym}{r}+\pa_r\pa_ym\right)\nonum\\
&\Ra\pa_v\phi\ \pa_y\phi+r\pa_y\pa_vB+\pa_y\pa_vC-A'(r\pa_vB+\pa_vC)=\frac{f\pa_ym}{r^2}+\frac{\pa_r(f\pa_ym)}{r}\xRightarrow{\eqref{loc.181}}\nonum\\
&\Ra \underbrace{\pa_v\phi\ \pa_y\phi+\pa_y\pa_vC-A'\pa_vC}_{(v,y)-\text{dependent}}=\underbrace{\frac{f\pa_ym}{r^2}+\frac{\pa_r(f\pa_ym)}{r}}_{(v,r,y)-\text{dependent}}\xRightarrow{\pa_r(RHS)=0}\nonum\\
&\Ra \pa_r\left[\frac{f\pa_ym}{r^2}+\frac{\pa_r(f\pa_ym)}{r}\right]=0\Ra \pa_r^2(f\pa_ym)-\frac{2}{r^2}(f\pa_ym)=0\nonum\\
&\Ra \boxed{f(v,r,y)=\frac{1}{\pa_y[m(v,r,y)]}\left[\frac{D(v,y)}{r}+E(v,y)r^2\right]}
\end{align}
\par We have already shown in previous cases (see case 13) that the last related combined with equation $\pa_vf=r\pa_vB+\pa_vC$ and \eqref{loc.134} leads to an inconsistency. If we assume $\pa_\phi f=0$ then equation \eqref{loc.115} gives $\pa_r\chi=0$ which is not accepted. If on the other hand assume $\pa_\chi f=0$ then using equations \eqref{loc.115}, \eqref{loc.118} and \eqref{loc.117} as before, we obtain again the relation \eqref{loc.183} (but now $\pa_vf=0$). Finally, with the use of \eqref{loc.119} it is:
$$-\pa_rf\left(\pa_rm+1-\frac{3m}{r}\right)=f\left(\pa_r^2m-\frac{2}{r}\pa_rm\right)$$
Using equations \eqref{loc.134}, \eqref{loc.183} into the previous expressions and after some algebra we get
$$\boxed{\sum_n(\pa_ya_n)\left\{\frac{D}{r^2}(n+1)+E\ r(n-2)+\sum_\ell a_\ell(\ell-3)\left[\frac{D}{r^3}r^\ell(\ell+n+1)+E\ r^\ell(\ell+n-2)\right]\right\}=0}$$
In the above equation, if we demand $\pa_ya_n\neq0$ then it is also restrictive to nullify both functions $D(y)$ and $E(y)$ in order for the relation to be consistent. However, there is the possibility to fix $n,\ell=1$ and $E(y)=0$. Then, we can calculate the value of the function $D(y)$ that manages to make the LHS of the above expression to vanish. The problem in this sub-sub-case is that for $n=1$ the metric tensor does not describe a black hole. Hence, this complete case does not give a viable solution to the problem.\\
\begin{center}
\underline{\textbf{21)}\hspace{1.5em}$\{\phi=\phi(r,y),\ \chi=\chi(r,y),\ f=f(\phi,\chi)=f(r,y)\}$}
\end{center}
\par The field equation \eqref{loc.117}, in this case, leads to the same expression for the coupling function $f(r,y)$ as in case 14, namely \eqref{loc.160}. Additionally, if we combine equations \eqref{loc.119}, \eqref{loc.134} and \eqref{loc.160}, as it has already been done in case 13, we are led to the same negative results. The sub-cases $\pa_\phi f=0$ and $\pa_\chi f=0$ do not solve the problem either. They produce the same results as the original case 21.\vspace{1em}

\section{$\mathbold{f(\phi,\chi)=a\ \phi+b\ \chi}$}
\par The rest of the cases are extremely difficult to be checked and even more difficult to be solved by using the field equations \eqref{loc.115}-\eqref{loc.121}; the complexity of the differential equations that emerge by the field equations increases dramatically with each variable that we add into the fields. Hence, we are going to examine the behaviour of some special coupling functions $f(\phi,\chi)=a\phi+b\chi$ in each one of the remaining cases (22-28).\\
\begin{center}
\underline{\textbf{22)}\hspace{1.5em}$\{\phi=\phi(v,r,y),\ \chi=\chi(v)\}$}
\end{center}\vspace{0.4em}
We consider the general case where $a,b\in \Re\wedge a\neq0\wedge b\neq0$.\vspace{0.6em}
\begin{align}\label{loc.184}
\eqref{loc.115}&\Ra (\pa_r\phi)^2+a\ \pa_r^2\phi=0\Ra -\frac{\pa(\pa_r\phi)}{(\pa_r\phi)^2}=\frac{dr}{a}\Ra (\pa_r\phi)^{-1}=\frac{r}{a}+B(v,y)\nonum\\
&\Ra\pa_r\phi=\frac{1}{\frac{r}{a}+B(v,y)}\Ra \pa_r\phi=\frac{a}{r+a B(v,y)}\xRightarrow{a B(v,y)\ra B(v,y)}\nonum\\
&\Ra \pa_r\phi=\frac{a}{r+B(v,y)}\Ra \boxed{\phi(v,r,y)=a\ln[r+B(v,y)]+C(v,y)}
\end{align}
\begin{align}\label{loc.185}
\eqref{loc.118}&\Ra \pa_r\phi\ \pa_y\phi+a\ \pa_r\pa_y\phi-A'a\ \pa_r\phi=0\xRightarrow{\eqref{loc.184}}\nonum\\
&\Ra\frac{a}{r+B}\left(\frac{\pa_yB}{r+B}+\pa_yC\right)-\frac{a^2\pa_yB}{(r+B)^2}-\frac{A'\ a^2}{r+B}=0\nonum\\
&\Ra \frac{a\pa_yC-A'\ a^2}{r+B}=0\Ra \pa_yC=a\ A'\Ra \boxed{C(v,y)=a A(y)+D(v)}
\end{align}\\
\eq$\label{loc.186}
\eqref{loc.184}\xRightarrow{\eqref{loc.185}}\boxed{\phi(v,r,y)=a\{\ln[r+B(v,y)]+A(y)\}+D(v)}$
\begin{align}
\eqref{loc.117}&\Ra \pa_y\phi\pa_v\phi+a\ \pa_y\pa_v\phi-A'( a\pa_v\phi+b\pa_v\chi)-\frac{\pa_ym}{r}a\ \pa_r\phi=\frac{f}{r}\left(\frac{\pa_ym}{r}+\pa_r\pa_ym\right)\xRightarrow{\eqref{loc.186}}\nonum\\
&\Ra\left(a\frac{\pa_yB}{r+B}+aA'\right)\left(a\ \frac{\pa_vB}{r+B}+\pa_vD\right)+a\left(\frac{\pa_y\pa_vB}{r+B}-\frac{a\pa_vB\pa_yB}{(r+B)^2}\right)\nonum\\
&\hspace{4em}-A'\left(\frac{a^2\pa_vB}{r+B}+aD'+b\chi'\right)-\frac{\pa_ym}{r}\frac{a^2}{r+B}=\frac{a\phi+b\chi}{r}\left(\frac{\pa_ym}{r}+\pa_r\pa_ym\right)\xRightarrow{\times(r+B)}\nonum
\end{align}
\eq$\label{sc1}
\underbrace{aD'\pa_yB+a^2\pa_v\pa_yB-bBA'\chi'}_{(v,y)-dependent}=\underbrace{(a\phi+b\chi)(r+B)\left(\frac{\pa_ym}{r^2}+\frac{\pa_r\pa_ym}{r}\right)+A'b\chi'r+\frac{a^2\pa_ym}{r}}_{(v,r,y)-dependent}$\\
Demanding the consistency of the previous equation, we are led to demand $\pa_r(RHS)=0$. Thus, it is:
\begin{gather}
\eqref{sc1}\xRightarrow{\pa_r(RHS)=0}A'b\chi'+a^2\left(\frac{\pa_r\pa_ym}{r}-\frac{\pa_ym}{r^2}\right)+a^2\left(\frac{\pa_ym}{r^2}+\frac{\pa_r\pa_ym}{r}\right)+(a\phi+b\chi)\left(\frac{\pa_ym}{r^2}+\frac{\pa_r\pa_ym}{r}\right)\nonum\\
+(a\phi+b\chi)(r+B)\left(-2\frac{\pa_ym}{r^3}+\frac{\pa_r^2\pa_ym}{r}\right)=0\Ra\nonum\\ \nonum\\
A'b\chi'+2a^2\frac{\pa_r\pa_ym}{r}+(a\phi+b\chi)\left(\frac{\pa_r\pa_ym}{r}-\frac{\pa_ym}{r^2}+\pa_r^2\pa_ym-2B\frac{\pa_ym}{r^3}+B\frac{\pa_r^2\pa_ym}{r}\right)=0\xRightarrow{\eqref{loc.134}}\nonum\\ \nonum\\
A'b\chi'+2a^2\sum_n(\pa_ya_n)nr^{n-2}+(a\phi+b\chi)\left[\sum_n(\pa_ya_n)nr^{n-2}-\sum_n(\pa_ya_n)r^{n-2}+\right.\nonum\\\left.+\sum_n(\pa_ya_n)n(n-1)r^{n-2}-2B\sum_n(\pa_ya_n)r^{n-3}+B\sum_n(\pa_ya_n)n(n-1)r^{n-3}\right]=0\Ra\nonum\\ \nonum\\
\label{sc2}
\boxed{A'b\chi'+\sum_n(\pa_ya_n)\left\{2a^2nr+(a\phi+b\chi)(n+1)\left[(n-1)r+B(n-2)\right]\right\}r^{n-3}}
\end{gather}
Both for $b=0$ and $b\neq0$ the above equation cannot have an overall value of zero, without allowing some of the function $a_n(v,y)$ to be $y-$independent. This is easy to be understood by writing the equation in the following expanded form.
\begin{gather}
A'b\chi'+\sum_n(\pa_ya_n)[2a^2n+b\chi(n+1)]r^{n-2}+Bb\chi\sum_n(\pa_ya_n)(n+1)(n-2)r^{n-3}+\nonum\\
\label{sc3}
+a\phi\sum_n(\pa_ya_n)(n+1)(n-1)r^{n-2}+Ba\phi\sum_n(\pa_ya_n)(n+1)(n-2)r^{n-3}=0
\end{gather}
Demanding finite number of values for $n$\footnote{The parameter $n$ is directly connected to the behaviour of the mass function $m(v,r,y)$ with respect to the variable $r$ via equation \eqref{loc.134}. Thus, the parameter $n$ should take some specific values in order to lead to a mass function that describes a black hole, or even a modified black hole. In any case, we cannot allow $n$ to take infinite number of values.}, it is impossible to find a set of values that could allow the last two terms of equation \eqref{sc3} to cancel each other (any other combination of terms in order to have a vanishing LHS either results to $\pa_r\phi=0$ or to an inconsistency in the expression of $\phi$ with respect to $r$). Hence, the consistency of the equation can only be achieved by setting $\pa_ya_n=0$ for some values of $n$. However, this would be devastating for the localization. Consequently, this case should also be excluded. Lastly, for $a=0$, equation \eqref{loc.115} results to $\pa_r\phi=0$, meaning that we need to investigate the case $\{\phi(v,y),\ \chi=\chi(v)\}$, but this case has already been excluded.\vspace{1em}
\begin{center}
\underline{\textbf{23)}\hspace{1.5em}$\{\phi=\phi(v,r,y),\ \chi=\chi(r)\}$}
\end{center}
$$\eqref{loc.115}\Ra(\pa_r\phi)^2+(\chi')^2+a\ \pa_r^2\phi+b\ \chi''=0\Ra (\pa_r\phi)^2+a\ \pa_r^2\phi=-(\chi')^2-b\ \chi''\Ra$$
\eq$\label{sc4}
\boxed{(\pa_r\phi)^2+a\ \pa_r^2\phi=h(r)=-(\chi')^2-b\ \chi''}$\\
With the use of mathematica or any other software that can solve differential equations, one can verify that the above differential equations (with respect to $\phi$ and $\chi$) have extremely complicated solutions even in case that $h(r)=r$. One can also examine different functions for $h(r)$ in order to be completely convinced. As an example we present the solutions of the differential equations with respect to the scalar fields $\phi, \chi$ that emanate from the relation \eqref{sc4} in case of $h(r)=r$.
\begin{align}
\phi(v,r,y)=&-\bigintss_1^r \frac{a \left\{\sqrt{-\frac{1}{a}}
   \sqrt{\frac{1}{a}} q^{3/2} \left[-c_1 J_{-\frac{4}{3}}\left(\frac{2}{3}
   \sqrt{-\frac{1}{a}} \sqrt{\frac{1}{a}} q^{3/2}\right)+c_1
   J_{\frac{2}{3}}\left(\frac{2}{3} \sqrt{-\frac{1}{a}} \sqrt{\frac{1}{a}}
   q^{3/2}\right)\right]\right\}}{2 q \left(c_1
   J_{-\frac{1}{3}}\left(\frac{2}{3} \sqrt{-\frac{1}{a}} \sqrt{\frac{1}{a}}
   q^{3/2}\right)+J_{\frac{1}{3}}\left(\frac{2}{3} \sqrt{-\frac{1}{a}} \sqrt{\frac{1}{a}}
   q^{3/2}\right)\right)} \, dq\nonum\\
   &+\bigintss_1^r \frac{a\left[2 \sqrt{-\frac{1}{a}}
   \sqrt{\frac{1}{a}} q^{3/2}J_\frac{2}{3}\left(\frac{2}{3} \sqrt{-\frac{1}{a}}
   \sqrt{\frac{1}{a}} q^{3/2}\right)+c_1 J_{-\frac{1}{3}}\left(\frac{2}{3}
   \sqrt{-\frac{1}{a}} \sqrt{\frac{1}{a}} q^{3/2}\right)\right]}{2 q \left(c_1
   J_{-\frac{1}{3}}\left(\frac{2}{3} \sqrt{-\frac{1}{a}} \sqrt{\frac{1}{a}}
   q^{3/2}\right)+J_{\frac{1}{3}}\left(\frac{2}{3} \sqrt{-\frac{1}{a}} \sqrt{\frac{1}{a}}
   q^{3/2}\right)\right)} \, dq+c_2\nonum
\end{align}\\
$$\chi(r)=\bigintss_1^r \frac{ q^{3/2} \left[-c_1
   J_{-\frac{4}{3}}\left(-\frac{2 q^{3/2}}{3 b}\right)+c_1 J_{\frac{2}{3}}\left(-\frac{2
   q^{3/2}}{3 b}\right)-2 J_{-\frac{2}{3}}\left(-\frac{2 q^{3/2}}{3
   b}\right)\right]+b\ c_1 J_{-\frac{1}{3}}\left(-\frac{2 q^{3/2}}{3 b}\right)}{2
   q \left(c_1 J_{-\frac{1}{3}}\left(-\frac{2 q^{3/2}}{3
   b}\right)+J_{\frac{1}{3}}\left(-\frac{2 q^{3/2}}{3 b}\right)\right)} \,
   dq+c_2$$
Thus, it would be practically impossible to use these complicated expression of the fields $\phi,\ \chi$ into the other field equations in order to examine their consistency. Consequently, we are essentially obliged to consider only the cases in which the function $h(r)$ is simply a constant (negative or positive) and zero.
\begin{enumerate}[(i)]
\item \underline{$h(r)=C=Q^2>0\wedge a\neq 0\wedge b\neq0$:}
\eq$\label{sc5}
\eqref{sc4}\Ra(\pa_r\phi)^2+a\ \pa_r^2\phi=Q^2=-(\chi')^2-b\ \chi''\Ra\left\{\begin{array}{c}(\pa_r\phi)^2+a\ \pa_r^2\phi=Q^2\\ \\ (\chi')^2+b\ \chi''=-Q^2\end{array}\right\}$
From equation \eqref{sc5} we obtain the following functions for the scalar fields:
\begin{gather}\label{sc6}
\boxed{\phi(v,r,y)=a\ \ln\left[\cosh\left(\frac{Q}{a}r+B(v,y)\right)\right]+D(v,y)}\\ \nonum\\
\label{sc7}
\boxed{\chi(r)=b \ln\left[\cos\left(\frac{Q}{b}r+E\right)\right]+F}
\end{gather}
We have assumed that $C>0$. In case that $C<0$, we would get the same expressions for the fields with the difference that now the hyperbolic cosine would describe the field $\chi$, while the cosine would correspond to the field $\phi$. The subsequent analysis in each case is the same.
$$\eqref{loc.118}\Ra\pa_r\phi\pa_y\phi+a\pa_r\pa_y\phi-A'(a\pa_r\phi+b\pa_r\chi)=0\xRightarrow[\eqref{sc7}]{\eqref{sc6}}$$
\begin{gather}
Q\tanh\left(\frac{Q}{a}r+B\right)\left[a\tanh\left(\frac{Q}{a}r+B\right)\pa_yB+\pa_yD\right]+aQ\frac{\pa_yB}{\cosh^2\left(\frac{Q}{a}r+B\right)}\nonum\\
-A'aQ\tanh\left(\frac{Q}{a}r+B\right)+A'bQ\tan\left(\frac{Q}{b}r+E\right)=0\Ra\nonum\\ \nonum\\
a\pa_yB+\tanh\left(\frac{Q}{a}r+B\right)(\pa_yD-aA')+A'b\tan\left(\frac{Q}{b}r+E\right)=0\Ra\left\{\begin{array}{c}
\pa_yB=0\vspace{0.5em}\\ and\vspace{0.5em}\\ 
\pa_yD-aA'=0\vspace{0.5em}\\ and\vspace{0.5em}\\
A'=0\end{array}\right\}\xRightarrow{A'\neq0}\text{rejected}\nonum
\end{gather}
\item \underline{$h(r)=C=-Q^2<0\wedge a\neq 0\wedge b=0$:}\\
\eq$\label{sc8}\eqref{sc4}\Ra(\pa_r \phi)^2+a\ \pa_r^2\phi=-Q^2=-(\chi')^2\Ra\left\{\begin{array}{c}(\pa_r\phi)^2+a\ \pa_r^2\phi=-Q^2\\ \\ (\chi')^2=Q^2\end{array}\right\}$
From equation \eqref{sc8} we get:
\eq$\label{sc9}
\boxed{\phi(v,r,y)=a\ \ln\left[\cos\left(\frac{Q}{a}r+B(v,y)\right)\right]+D(v,y)}$
\eq$\label{sc10}
\boxed{\chi(r)=\pm Qr+E}$
\begin{gather}
\eqref{loc.118}\Ra \pa_r\phi\pa_y\phi+a\pa_r\pa_y\phi-A'a\pa_r\phi=0\xRightarrow{\eqref{sc9}}\nonum\\ \nonum\\ 
-Q\tan\left(\frac{Q}{a}r+B\right)\left[-a\tan\left(\frac{Q}{a}r+B\right)\pa_yB+\pa_yD\right]-aQ\frac{\pa_yB}{\cos^2\left(\frac{Q}{a}r+B\right)}+A'aQ\tan\left(\frac{Q}{a}r+B\right)=0\Ra\nonum\\ \nonum\\
-aQ\pa_yB-Q\tan\left(\frac{Q}{a}r+B\right)(\pa_yD-aA')=0\Ra\nonum\\ \nonum\\
\label{sc11}
\left\{\begin{array}{c}
\pa_yB=0\Ra \boxed{B=B(v)}\vspace{0.5em}\\ and\vspace{0.5em}\\ \pa_yD=aA'\Ra \boxed{D(v,y)=aA(y)+F(v)}\end{array}\right\} 
\end{gather}
Thus, equations \eqref{sc9} and \eqref{sc11} lead to
\eq$\label{sc12}
\boxed{\phi(v,r,y)=a\ \ln\left[\cos\left(\frac{Q}{a}r+B(v)\right)\right]+a\ A(y)+F(v)}$
\begin{gather}
\eqref{loc.117}\Ra\pa_y\phi\pa_v\phi+a\underbrace{\pa_y\pa_v\phi}_0-A'a\pa_v\phi=\frac{a\phi}{r}\left(\frac{\pa_ym}{r}+\pa_r\pa_ym\right)+\frac{\pa_ym}{r}a\pa_r\phi\Ra\nonum\\ \nonum\\
aA'\left[-\tan\left(\frac{Q}{a}r+B\right)\pa_vB+\pa_vF\right]-A'a\left[-a\tan\left(\frac{Q}{a}r+B\right)\pa_vB+\pa_vF\right]=\frac{a\phi\pa_ym}{r^2}+\frac{\pa_r(a\phi\pa_ym)}{r}\Ra\nonum\\ \nonum\\
\label{sc13}
\frac{a\phi\pa_ym}{r^2}+\frac{\pa_r(a\phi\pa_ym)}{r}=0\Ra\frac{\pa(\phi\pa_ym)}{\phi\pa_ym}=-\frac{dr}{r}\Ra\boxed{\pa_ym=\frac{G(v,y)}{r\phi}}
\end{gather}
The substitution of equation \eqref{sc12} into \eqref{sc13} leads to
\eq$\label{sc14}
\boxed{\pa_ym=G(v,y)r^{-1}\left\{a\ \ln\left[\cos\left(\frac{Q}{a}r+B(v)\right)\right]a\ A(y)+F(v)\right\}^{-1}}$
An expansion of the form $m(v,r,y)=\sum_na_n(v,y)r^n$ and a finite number of values for the index $n$ cannot produce the above equation for the mass function. Hence, this sub-case does not result to a viable solution for the localization problem.\\
\item \underline{$h(r)=C=Q^2>0\wedge a=0\wedge b\neq0$:}\\
\par Using again equation \eqref{sc4} in the same way as in the previous sub-case, we obtain:
\eq$\label{sc15}
\boxed{\phi(v,r,y)=\pm Qr+B(v,y)}$
\eq$\label{sc16}
\boxed{\chi(r)=b\ \ln\left[\cos\left(\frac{Q}{b}r+D\right)\right]+E}$
\begin{gather}
\eqref{loc.118}\Ra\pa_r\phi\pa_y\phi-A'b\ \pa_r\chi=0\xRightarrow[\eqref{sc16}]{\eqref{sc15}}\pm Q\pa_yB+A'Q\tan\left(\frac{Q}{b}r+D\right)=0\Ra\nonum\\ \nonum\\ 
\pm \pa_yB+A'\tan\left(\frac{Q}{b}r+D\right)=0\Ra\left\{\begin{array}{c}\pa_yB=0\vspace{0.5em}\\ and\vspace{0.5em}\\ A'=0\end{array}\right\}\xRightarrow{A'\neq0}\text{rejected}\nonum
\end{gather}
\item \underline{$h(r)=0\wedge a\neq 0\wedge b\neq0$:}\\
\par In this sub-case, it is not possible to assume either $a=0$ or $b=0$, because from equation \eqref{sc4} we are led to $\pa_r\phi=0$ or $\pa_r\chi=0$ respectively, but these results are inconsistent with our primary assumption for the fields. Consequently, from equation \eqref{sc4} we get:
$$(\pa_r\phi)^2+a\ \pa_r^2\phi=0=(\chi')^2+b\ \chi''\Ra\left\{\begin{array}{c}(\pa_r\phi)^2+a\ \pa_r^2\phi=0\\ \\ (\chi')^2+b\ \chi''=0\end{array}\right\}$$
The solutions that emanate from the above differential equations are:
\eq$\label{sc17}
\boxed{\phi(v,r,y)=a\ \ln\left[r+B(v,y)\right]+C(v,y)}$
\eq$\label{sc18}
\boxed{\chi(r)=b\ \ln(r+D)+E}$
\begin{gather}
\eqref{loc.118}\Ra \pa_r\phi\pa_y\phi+a\ \pa_r\pa_y\phi-A'(a\ \pa_r\phi+b\ \pa_r\chi)=0\xRightarrow[\eqref{sc17}]{\eqref{sc18}}\nonum\\ \nonum\\
\frac{a}{r+B}\left(\frac{a\pa_yB}{r+B}+\pa_yC\right)-\frac{a^2\pa_yB}{(r+B)^2}-A'\left(\frac{a^2}{r+B}+\frac{b^2}{r+D}\right)=0\Ra\nonum\\ \nonum\\
\label{sc19}
\frac{a\pa_yC}{r+B}-\frac{A'a^2}{r+B}-\frac{b^2A'}{r+D}=0\Ra \underbrace{\pa_y[C(v,y)]-aA'(y)}_{(v,y)-dependent}=\underbrace{A'(y)\frac{b^2}{a}\ \frac{r+B(v,y)}{r+D}}_{(v,r,y)-dependent}
\end{gather}
Demanding the consistency of equation \eqref{sc19}, we get the following constraint.
\eq$\label{sc20}
\eqref{sc19}\xRightarrow{\pa_r(RHS)=0}\pa_r\left(\frac{r+B(v,y)}{r+D}\right)=0\Ra \boxed{B(v,y)=D}$
Substituting equation \eqref{sc20} into \eqref{sc19}, we obtain:
\eq$\label{sc21}
\pa_y[C(v,y)]-aA'(y)=A'(y)\frac{b^2}{a}\Ra \boxed{C(v,y)=\frac{a^2+b^2}{a}A(y)+F(v)}$
Hence, from equations \eqref{sc17}, \eqref{sc20} and \eqref{sc21} it is:
\eq$\label{sc22}
\boxed{\phi(v,r,y)=a\ \ln(r+D)+\frac{a^2+b^2}{a}A(y)+F(v)}$
\begin{gather}
\eqref{loc.117}\Ra\pa_y\phi\pa_v\phi+a\underbrace{ \pa_y\pa_v\phi}_0-A'a\ \pa_v\phi-\frac{\pa_ym}{r}\pa_rf=\frac{f}{r}\left(\frac{\pa_ym}{r}+\pa_r\pa_ym\right)\Ra\nonum\\ \nonum\\ 
\underbrace{\frac{a^2+b^2}{a}A'F'-aA'F'}_{(v,y)-dependent}=\underbrace{\frac{f\pa_ym}{r^2}+\frac{\pa_r(f\pa_ym)}{r}}_{(v,r,y)-dependent}\xRightarrow{\pa_r(RHS)=0}\nonum\\ \nonum\\ 
\label{sc23}
\pa_r^2(f\pa_ym)-\frac{2f\pa_ym}{r^2}=0\xRightarrow{\pa_ym\neq0}\boxed{f(v,r,y)=\frac{1}{\pa_y[m(v,r,y)]}\left[\frac{G(v,y)}{r}+H(v,y)r^2\right]}
\end{gather}
However, the coupling function $f(v,r,y)$ is also given by the primary expression $f=a\ \phi+b\ \chi$. Thus, the substitution of equations \eqref{sc18} and \eqref{sc22}
into the previous expression and then equating the result with the RHS of equation \eqref{sc23}, we obtain:
$$(a^2+b^2)\ln(r+D)+(a^2+b^2)A(y)+aF(v)+bE=\frac{1}{\pa_ym}\left[\frac{G(v,y)}{r}+H(v,y)r^2\right]\Ra$$
\eq$\label{sc24}
\boxed{\pa_ym=\frac{\frac{G(v,y)}{r}+H(v,y)r^2}{(a^2+b^2)\ln(r+D)+(a^2+b^2)A(y)+aF(v)+bE}}$
Similarly to sub-case (ii), the above equation cannot be consistent with the expansion of equation \eqref{loc.134} and the condition of having a finite number of values for the index $n$. Hence, this sub-case is also rejected.\\
\end{enumerate}

\newpage
\begin{center}
\underline{\textbf{24)}\hspace{1.5em}$\{\phi=\phi(v,r,y),\ \chi=\chi(y)\}$}
\end{center}
\eq$\label{sc25}
\eqref{loc.115}\Ra (\pa_r\phi)^2+a\ \pa_r^2\phi=0\Ra \boxed{\phi(v,r,y)=a\ \ln[r+B(v,y)]+C(v,y)}$\\
$$\eqref{loc.118}\Ra\pa_r\phi\pa_y\phi+a\ \pa_r\pa_y\phi-A'a\ \pa_r\phi=0\Ra\frac{a}{r+B}\left(\frac{a\pa_yB}{r+B}+\pa_yC\right)-\frac{a^2\pa_yB}{(r+B)^2}-\frac{a^2A'}{r+B}=0\Ra$$
\eq$\label{sc26}
\pa_yC=aA'\Ra \boxed{C(v,y)=aA(y)+D(v)}$\\
\eq$\label{sc27}
\eqref{sc25}\xRightarrow{\eqref{sc26}}\boxed{\phi(v,r,y)=a\ \ln[r+B(v,y)]+aA(y)+D(v)}$
\begin{gather}
\eqref{loc.117}\Ra \pa_y\phi\pa_v\phi+a\ \pa_v\pa_y\phi-A'a\ \pa_v\phi-\frac{\pa_ym}{r}a\ \pa_r\phi=\frac{a\phi+b\chi}{r}\left(\frac{\pa_ym}{r}+\pa_r\pa_ym\right)\xRightarrow{\eqref{sc27}}\nonum\\ \nonum\\
\label{sc28}
\underbrace{a\ \pa_yB\ D'+a^2\pa_v\pa_yB}_{(v,y)-dependent}=\underbrace{(a\phi+b\chi)(r+B)\left(\frac{\pa_ym}{r}+\pa_r\pa_ym\right)+a^2\frac{\pa_ym}{r}}_{(v,r,y)-dependent}
\end{gather}
The RHS of equation \eqref{sc28} is similar to the RHS of equation \eqref{sc1} except for the missing term $A'b\chi'r$. Thus, demanding $\pa_r(RHS)=0$ of equation \eqref{sc28} (as in case 22) and also using equation \eqref{loc.134}, we get:
\eq$\label{sc29}
\boxed{\sum_n(\pa_ya_n)\left\{2a^2nr+(a\phi+b\chi)(n+1)[(n-1)r+B(n-2)]\right\}r^{n-3}=0}$
The above equation cannot be consistent if we demand a finite number of values for $n$. The arguments here are the same as the arguments which were presented in case 22.\\

\begin{center}
\underline{\textbf{25)}\hspace{1.5em}$\{\phi=\phi(v,r,y),\ \chi=\chi(v,r)\}$}
\end{center}
$$\eqref{loc.115}\Ra (\pa_r\phi)^2+(\pa_r\chi)^2+a\ \pa_r^2\phi+b\ \pa_r^2\chi=0\Ra (\pa_r\phi)^2+a\ \pa_r^2\phi=-(\pa_r\chi)^2-b\ \pa_r^2\chi$$
\eq$\label{sc30}
\boxed{(\pa_r\phi)^2+a\ \pa_r^2\phi=h(v,r)=-(\pa_r\chi)^2-b\ \pa_r^2\chi}$\\
As it was justified in case 23, the field equation are essentially unapproachable if we do not demand $\pa_r[h(v,r)]=0$. Therefore, in this case, we are going to consider that $h=h(v)$. In complete analogy to the case 23, four different sub-cases are going to be investigated in the context of the original case.\\
\begin{enumerate}[(i)]
\item \underline{$h(v,r)=[C(v)]^2>0\wedge a\neq 0\wedge b\neq0$:}\\
\par From equation \eqref{sc30} we are led to the following functions:
\eq$\label{sc31}
\boxed{\phi(v,r,y)=a\ \ln\left[\cosh\left(\frac{C(v)}{a}r+B(v,y)\right)\right]+D(v,y)}$
\eq$\label{sc32}
\boxed{\chi(v,r)=b\ \ln\left[\cos\left(\frac{C(v)}{b}r+E(v)\right)\right]+F(v)}$\\
$$\eqref{loc.118}\Ra\pa_r\phi\pa_y\phi+a\ \pa_r\pa_y\phi-A'\left(a\ \pa_r\phi+b\ \pa_r\chi\right)=0\xRightarrow[\eqref{sc32}]{\eqref{sc31}}$$
$$a\pa_yB+\pa_yD\tanh\left(\frac{C}{a}r+B\right)-A'a\tanh\left(\frac{C}{a}r+B\right)+A'b\tan\left(\frac{C}{b}r+E\right)=0\Ra$$
$$\left\{\begin{array}{c}
\pa_yB=0\vspace{0.5em}\\ and\vspace{0.5em}\\ 
\pa_yD-aA'=0\vspace{0.5em}\\ and\vspace{0.5em}\\
A'=0\end{array}\right\}\xRightarrow{A'\neq0}\text{rejected}$$\\
\item \underline{$h(v,r)=-[C(v)]^2<0\wedge a\neq 0\wedge b=0$:}\\
\par In this sub-case equation \eqref{sc30} leads to
\eq$\label{sc33}
\boxed{\phi(v,r,y)=a\ \ln\left[\cos\left(\frac{C(v)}{a}r+B(v,y)\right)\right]+D(v,y)}$
\eq$\label{sc34}
\boxed{\chi(v,r)=\pm C(v)r+E(v)}$\\
$$\eqref{loc.118}\Ra \pa_r\phi\pa_y\phi+a\ \pa_r\pa_y\phi-A'a\pa_r\phi=0\xRightarrow{\eqref{sc33}}$$
$$-a\pa_yB-\pa_yD\tan\left(\frac{C}{a}r+B\right)+A'a\tan\left(\frac{C}{a}r+B\right)=0\Ra$$
\eq$\label{sc35}
\left\{\begin{array}{c}
\pa_yB=0\Ra \boxed{B=B(v)}\vspace{0.5em}\\ and\vspace{0.5em}\\ \pa_yD=aA'\Ra \boxed{D(v,y)=aA(y)+F(v)}\end{array}\right\}$\\
Thus, equations \eqref{sc33} and \eqref{sc35} give:
\eq$\label{sc36}
\boxed{\phi(v,r,y)=a\ \ln\left[\cos\left(\frac{C(v)}{a}r+B(v)\right)\right]+aA(y)+E(v)}$
Consequently, we have:
$$\eqref{loc.117}\Ra \pa_y\phi\pa_v\phi+a\ \underbrace{\pa_v\pa_y\phi}_0-A'a\pa_v\phi=\frac{a\phi\pa_ym}{r^2}+\frac{\pa_r(a\phi\pa_ym)}{r}\Ra$$
$$\frac{\phi\pa_ym}{r}+\pa_r(\phi\pa_ym)=0\Ra\phi\pa_ym=\frac{G(v,y)}{r}\xRightarrow{\phi\neq0}\pa_ym=\frac{G(v,y)}{r\phi}\Ra$$
\eq$\label{sc37}
\boxed{\pa_ym=G(v,y)r^{-1}\left\{a\ \ln\left[\cos\left(\frac{C(v)}{a}r+B(v)\right)\right]+aA(y)+E(v)\right\}^{-1}}$\\
The reason that we reject this sub-case is exactly the same as in sub-case (ii) of case 23.
\end{enumerate}
\par One can easily verify that the remaining two sub-cases of case 25, namely $h(v,r)=[C(v)]^2>0\wedge a=0\wedge b\neq0$ and $h(v,r)=0\wedge a\neq 0\wedge b\neq0$, can be excluded as well. The reason is that the complete analysis of case 25 can be simply deduced by the analysis of case 23 if we just replace the constant $Q$ with the function $C(v)$ (To clarify this statement, compare the analysis of the previous two sub-cases -which were performed extensively- to the corresponding analysis of case 23). This particular substitution does not change the negative results that was found in each sub-case of case 23. Therefore, we proceed to the following case.\\

\begin{center}
\underline{\textbf{26)}\hspace{1.5em}$\{\phi=\phi(v,r,y),\ \chi=\chi(v,y)\}$}
\end{center}
\eq$\label{sc38}
\eqref{loc.115}\Ra (\pa_r\phi)^2+a\ \pa_r^2\phi=0\Ra \boxed{\phi(v,r,y)=a\ \ln[r+B(v,y)]+C(v,y)}$\\
$$\eqref{loc.118}\Ra\pa_r\phi\pa_y\phi+a\ \pa_r\pa_y\phi-A'a\ \pa_r\phi=0\Ra\frac{a}{r+B}\left(\frac{a\pa_yB}{r+B}+\pa_yC\right)-\frac{a^2\pa_yB}{(r+B)^2}-\frac{a^2A'}{r+B}=0\Ra$$
\eq$\label{sc39}
\pa_yC=aA'\Ra \boxed{C(v,y)=aA(y)+D(v)}$\\
\eq$\label{sc40}
\eqref{sc38}\xRightarrow{\eqref{sc39}}\boxed{\phi(v,r,y)=a\ \ln[r+B(v,y)]+aA(y)+D(v)}$
\begin{gather}
\eqref{loc.117}\Ra \pa_y\phi\pa_v\phi+\pa_y\chi\pa_v\chi+a\ \pa_v\pa_y\phi+b\pa_v\pa_y\chi-A'(a\ \pa_v\phi+b\pa_v\chi)=\frac{f\pa_ym}{r^2}+\frac{\pa_r(f\pa_ym)}{r}\xRightarrow{\eqref{sc40}}\nonum\\ \nonum\\
\underbrace{a\ \pa_yB\ D'+\pa_y\chi\pa_v\chi B+a^2\pa_v\pa_yB-A'bB\pa_v\chi }_{(v,y)-dependent}=\nonum\\
\label{sc41}
=\underbrace{(a\phi+b\chi)(r+B)\left(\frac{\pa_ym}{r^2}+\frac{\pa_r\pa_ym}{r}\right)-\pa_y\chi\pa_v\chi\ r-b\pa_y\pa_v\chi\ r+A'b\pa_v\chi\ r+\frac{a^2\pa_ym}{r}}_{(v,r,y)-dependent}
\end{gather}
The RHS of equation \eqref{sc41} is similar to the RHS of equation \eqref{sc1} except from the additional terms $-\pa_y\chi\pa_v\chi\ r$ and $-b\pa_y\pa_v\chi\ r$. Thus, demanding $\pa_r(RHS)=0$ of equation \eqref{sc41} (as in case 22) and also using equation \eqref{loc.134}, we get:
\eq$\label{sc42}
\boxed{A'b\pa_v\chi-\pa_y\chi\pa_v\chi-b\pa_v\pa_y\chi+\sum_n(\pa_ya_n)\left\{2a^2nr+(a\phi+b\chi)(n+1)[(n-1)r+B(n-2)]\right\}r^{n-3}=0}$\\
The above equation cannot be consistent if we demand a finite number of values for $n$. The arguments here are the same as the arguments which were presented in case 22 and also used in case 24.\\

\begin{center}
\underline{\textbf{27)}\hspace{1.5em}$\{\phi=\phi(v,r,y),\ \chi=\chi(r,y)\}$}
\end{center}
$$\eqref{loc.115}\Ra (\pa_r\phi)^2+(\pa_r\chi)^2+a\ \pa_r^2\phi+b\ \pa_r^2\chi=0\Ra (\pa_r\phi)^2+a\ \pa_r^2\phi=-(\pa_r\chi)^2-b\ \pa_r^2\chi$$
\eq$\label{sc43}
\boxed{(\pa_r\phi)^2+a\ \pa_r^2\phi=h(r,y)=-(\pa_r\chi)^2-b\ \pa_r^2\chi}$\\
The same argument as it was presented in case 23 leads us to the following sub-cases:
\begin{enumerate}[(i)]
\item \underline{$h(r,y)=C(y)>0\wedge a\neq 0\wedge b\neq0$:}\\
\par From euqation \eqref{sc43} we obtain:
\eq$\label{sc44}
\boxed{\phi(v,r,y)=a\ \ln\left[\cosh\left(\frac{\sqrt{C(y)}}{a}r+B(v,y)\right)\right]+D(v,y)}$
\eq$\label{sc45}
\boxed{\chi(r,y)=b\ \ln\left[\cos\left(\frac{\sqrt{C(y)}}{b}r+E(y)\right)\right]+F(y)}$\\
$$\eqref{loc.118}\Ra\pa_r\phi\pa_y\phi+\pa_r\chi\pa_y\chi+a\pa_r\pa_y\phi+b\pa_r\pa_y\chi-A'(a\pa_r\phi+b\pa_r\chi)=0\xRightarrow[\eqref{sc45}]{\eqref{sc44}}$$
$$a\pa_yB-bE'+\tanh\left(\frac{\sqrt{C}}{a}r+B\right)\left(\pa_yD+\frac{aC'}{2C}-aA'\right)-\tan\left(\frac{\sqrt{C}}{b}r+E\right)\left(F'+\frac{bC'}{2C}-bA'\right)=0\Ra$$
\eq$\label{sc46}
\left\{\begin{array}{c}
\pa_yB=\frac{b}{a}E'\Ra \boxed{B(v,y)=\frac{b}{a}E(y)+G(v)}\vspace{0.5em}\\ and\vspace{0.5em}\\ 
\pa_yD-aA'+\frac{a}{2}\pa_y[\ln C(y)]=0\Ra \boxed{D(v,y)=aA(y)-\frac{a}{2}\ln[C(y)]+H(v)}\vspace{0.5em}\\ and\vspace{0.5em}\\
F'-bA'+\frac{b}{2}\pa_y[\ln C(y)]=0\Ra\boxed{F(y)=bA(y)-\frac{b}{2}\ln[C(y)]+\underbrace{J}_{const.}}\end{array}\right\}$\\
Thus, equations \eqref{sc44} and \eqref{sc45} take the following form:
\eq$\label{sc47}
\boxed{\phi(v,r,y)=a\ \ln\left[\cosh\left(\frac{\sqrt{C(y)}}{a}r+\frac{b}{a}E(y)+G(v)\right)\right]+aA(y)-\frac{a}{2}\ln[C(y)]+H(v)}$
\eq$\label{sc48}
\boxed{\chi(r,y)=b\ \ln\left[\cos\left(\frac{\sqrt{C(y)}}{b}r+E(y)\right)\right]+bA(y)-\frac{b}{2}\ln[C(y)]+J}$\\
$$\eqref{loc.117}\Ra \pa_y\phi\pa_v\phi+a\pa_v\pa_y\phi-A'a\pa_v\phi=\frac{f\pa_ym}{r^2}+\frac{\pa_r(f\pa_ym)}{r}\Ra$$
\eq$\label{sc49}
\boxed{\scriptstyle{\frac{f\pa_ym}{r^2}+\frac{\pa_r(f\pa_ym)}{r}=\frac{arC'G'}{2\sqrt{C}}+abE'G'+\tanh\left(\frac{\sqrt{C}}{a}r+\frac{b}{a}E+G\right)\left(\frac{C'H'r}{2\sqrt{C}}+bE'H'-\frac{aC'G'}{2C}\right)-\frac{aC'H'}{2C}}}$\\
Moreover, it is:
\eq$\label{sc50}
\boxed{\scriptstyle{f=a\phi+b\chi=a^2 \ln\left[\cosh\left(\frac{\sqrt{C(y)}}{a}r+\frac{b}{a}E(y)+G(v)\right)\right]+b^2 \ln\left[\cos\left(\frac{\sqrt{C(y)}}{b}r+E(y)\right)\right]+(a^2+b^2)A(y)-\frac{a^2+b^2}{2}\ln[C(y)]+aH(v)+bJ}}$\\
Obviously, the substitution of equation \eqref{sc50} into \eqref{sc49} leads to a very complicated equation which is not possible to be satisfied using equation \eqref{loc.134} for the mass function and demanding $n$ to take a finite number of values. The previous statement holds even in case of $C'=0$.\\

\item \underline{$h(r,y)=C(y)=-[K(y)]^2<0\wedge a\neq 0\wedge b=0$:}\\
\par Equation \eqref{sc43} results to the following functions for the fields:
\eq$\label{sc51}
\boxed{\phi(v,r,y)=a\ \ln\left[\cos\left(\frac{\sqrt{C(y)}}{a}r+B(v,y)\right)\right]+D(v,y)}$
\eq$\label{sc52}
\boxed{\chi(r,y)=\pm \sqrt{C(y)}\ r+E(y)}$
Using equation \eqref{loc.118} we obtain:
$$\pa_r\phi\pa_y\phi+\pa_r\chi\pa_y\chi+a\pa_r\pa_y\phi-A'a\pa_r\phi=0\Ra$$
$$-a\pa_yB-\tan\left(\frac{\sqrt{C}}{a}r+B\right)\left(\pa_yD+\frac{aC'}{2C}-aA'\right)\pm E'=0\Ra$$
\eq$\label{sc53}
\left\{\begin{array}{c}
\pa_yB=\pm\frac{E'}{a}\Ra \boxed{B(v,y)=\pm\frac{E(y)}{a}+F(v)}\vspace{0.5em}\\ and\vspace{0.5em}\\ 
\pa_yD-aA'+\frac{a}{2}\pa_y[\ln C(y)]=0\Ra \boxed{D(v,y)=aA(y)-\frac{a}{2}\ln[C(y)]+G(v)}\vspace{0.5em}\end{array}\right\}$\\
Hence, we have:
\eq$\label{sc54}
\boxed{\phi(v,r,y)=a\ \ln\left[\cos\left(\frac{\sqrt{C(y)}}{a}r\pm\frac{E(y)}{a}+F(v)\right)\right]+aA(y)-\frac{a}{2}\ln[C(y)]+G(v)}$\\
$$\eqref{loc.117}\Ra\pa_y\phi\pa_v\phi+a\pa_y\pa_v\phi-A'a\pa_v\phi=\frac{f\pa_ym}{r^2}+\frac{\pa_r(f\pa_ym)}{r}\Ra$$
\eq$\label{sc55}
\boxed{\scriptstyle{\frac{f\pa_ym}{r^2}+\frac{\pa_r(f\pa_ym)}{r}=-a^2E'\left(\frac{C'r}{2\sqrt{C}a}\pm\frac{E'}{a}\right)-a\tan\left(\frac{\sqrt{C}}{a}r\pm\frac{E}{a}+F\right)\left(\frac{G'C'r}{2\sqrt{C}a}\pm\frac{G'E'}{a}-\frac{aF'C'}{2C}\right)-\frac{aC'G'}{2C}}}$
For $\pa_ym\neq0$, it is necessary to nullify the quantity $\frac{\pa_ym}{r^2}+\frac{\pa_r\pa_ym}{r}$ which is the factor of the term which includes the coupling function $f$ in the LHS of the above equation. This term cannot be equated with another term, so it must vanish, otherwise the field equation is not consistent. Thus, we get:
\eq$\label{sc56}
\frac{\pa_ym}{r^2}+\frac{\pa_r\pa_ym}{r}=0\Ra \boxed{m(v,r,y)=\frac{\lam(v,y)}{r}}$
For the same reason the quantity $\pa_r^2m-\frac{2\pa_rm}{r}$ should also be zero in the RHS of equation \eqref{loc.119}. Hence, it is:
\eq$\label{sc57}
\pa_r^2m-\frac{2}{r}\pa_rm=0\Ra \boxed{m(v,r,y)=\kappa(v,y)r^3+\xi(v,y)}$
It is clear that equations \eqref{sc56} and \eqref{sc57} cannot be simultaneously true. Consequently, we should reject this sub-case as well, because it leads to an inconsistency.\\

\item \underline{$h(r,y)=C(y)>0\wedge a=0\wedge b\neq0$:}\\
\par This sub-case is almost identical to the previous one, but the functions of $\phi$ and $\chi$ are alternated. Here, equation \eqref{sc43} leads to
\eq$\label{sc58}
\boxed{\phi(v,r,y)=\pm \sqrt{C(y)}\ r+B(v,y)}$
\eq$\label{sc59}
\boxed{\chi(r,y)=b\ln\left[\cos\left(\frac{\sqrt{C(y)}}{b}r+D(y)\right)\right]+E(y)}$
Following the same steps as in sub-case (ii), we obtain again equations \eqref{sc56} and \eqref{sc57}. Therefore, this sub-case is also not capable of providing a viable solution to the localization problem.

\item \underline{$h(r,y)=0\wedge a\neq 0\wedge b\neq0$:}\\
\par From equation \eqref{sc43} we have:
\eq$\label{sc60}
\boxed{\phi(v,r,y)=a\ln[r+B(v,y)]+C(v,y)}$
\eq$\label{sc61}
\boxed{\chi(r,y)=b\ln[r+D(y)]+E(y)}$\\
\eq$\label{sc62}
\eqref{loc.118}\xRightarrow[[\eqref{sc61}]{\eqref{sc60}}
\underbrace{\pa_yC-aA'}_{(v,y)-dependent}=\underbrace{-\left(\frac{b}{a}E'-\frac{b^2}{a}A'\right)\frac{r+B}{r+D}}_{(v,r,y)-dependent}$
\eq$\label{sc63}
\eqref{sc62}\xRightarrow{\pa_r(RHS)=0}\boxed{B(v,y)=D(y)}$\\
Thus:
\eq$\label{sc64}
\eqref{sc62}\xRightarrow{\eqref{sc63}}
\boxed{C(v,y)=\frac{a^2+b^2}{a}A(y)-\frac{b}{a}E(y)+F(v)}\\
$\eq$\label{sc65}
\eqref{sc60}\xRightarrow[\eqref{sc64}]{\eqref{sc63}}\boxed{\phi(v,r,y)=a\ln\left[r+D(y)\right]+\frac{a^2+b^2}{a}A(y)-\frac{b}{a}E(y)+F(v)}$\\
Using equations \eqref{loc.117} and \eqref{loc.119} as they were used in sub-cases (ii) and (iii), we obtain again the equations \eqref{sc56} and \eqref{sc57}. Hence, the case 27 is complete excluded.\\

\begin{center}
\underline{\textbf{28)}\hspace{1.5em}$\{\phi=\phi(v,r,y),\ \chi=\chi(v,r,y)\}$}
\end{center}
\par Finally, the most general case 28 can be investigated and excluded in exactly the same way as case 27. The only difference is that in this case, we have to replace the functions $h(r,y)$ and $C(y)$ with $h(v,r,y)$ and $C(v,y)$ respectively.
\end{enumerate}

%% file: Chapters/Chapter_Conclusions.tex
\chapter{Conclusions and Discussion}

\par  In summary, the study that was preceded in the framework of this thesis had the following structure. In Chapter 1 (Introduction) took place a concise presentation of the biggest mysteries in physics that remain unsolved until these days i.e. the nature of Dark Matter and Dark Energy and the existence of a Unified Theory. The attempts that were made in order to resolve these problems led to various extra-dimensional theories, in which String Theory plays a starring role. The Hierarchy Problem led to the formulation of the Randall-Sundrum models (RS1 and RS2), which were analyzed in detail in Chapter 2. Briefly, RS1 model manages to resolve the Hierarchy Problem by assuming the existence of a compact extra dimension, which is finite and also bounded by two 3-branes. However, in the context of the RS2 model, the extra dimension is allowed to be infinite and therefore the second brane is essentially removed from the model. The astonishing result of the RS2 model is that although we have an infinite extra dimension, the gravity on the remaining 3-brane is effectively 4-dimensional. Although very popular, these models face the problem of the absence of an analytic solution describing a regular, 5-dimensional black hole solution localised close to the brane. Therefore, in Chapter 3, the geometrical framework and the scalar field theory -which were used in the context of the thesis in an attempt to solve black hole localization problem- were presented and the field equations of our theory were derived. Subsequently, these field equations were used in Chapter 4 in order to find a localized black hole solution in the context of a generalized RS2 brane-world model.
\par The main purpose of this thesis was to find a localized 5-dimensional black hole solution close to our 3-brane (or our 4-dimensional universe) by using an RS2-type geometrical background and a scalar field theory which consists of two scalar fields $\phi,\ \chi$ that interact with each other and they are also non-minimally coupled to gravity. This complicated choice of a scalar field theory was made because simpler types of scalar field theories failed to provide a localized black hole solution. Therefore, we assumed the existence of an extra scalar field $\chi$. Despite the additional degree of freedom that was provided by the second scalar field $\chi$, we were not able to find a viable configuration thus the analytical solution to the localization problem remains still an open problem. Although we were able to exclude mathematically all considered cases, the complexity of the equations did not allow us to formulate a no-go theorem that would exclude altogether the existence of a viable configuration. Our negative result is only one in a series of failed analytical attempts over a period of almost 20 years, and this certainly creates well-founded concerns about the compatibility of brane-world models with the predictions of General Relativity. It is also crucial to indicate that some of the solutions rejected may comprise instead novel black-string solutions, whose study we will undertake in the near future. If this proves to be true, it will refuel the existing debate in the literature regarding the question of why black-string solutions are so much easier to find in the context of brane-world models compared to black-hole solutions.
\par In the context of the scalar field theory that was considered in this thesis, there are some extra coupling functions that is possible to be examined, though they are more complicated:
\begin{center}
\large{$f(\phi,\chi)=a\ \phi^2+b\ \chi^2+c\ \phi\ \chi$}\\ \vspace{1em}
\large{$f(\phi,\chi)=a\ \phi^\kappa+b\ \chi^\lam$}\\\vspace{1em}
\large{$f(\phi,\chi)=e^{a \phi+b \chi}$}\\ \vspace{0.5em}
\end{center}
Except for the most complicated case where $a\neq0\wedge b\neq0\wedge c\neq0$, it is possible to set $a=0$ and/or $b=0$ and/or $c=0$. The previous analysis for the coupling function $f(\phi,\chi)=a\phi+b\chi$ have made clear that the most valuable information about the form of the functions of the scalar fields $\phi$ and $\chi$ emanates from equation \eqref{loc.115}. The reason is that compared to the other field equations, equation \eqref{loc.115} is the simplest one. If we therefore consider a coupling function, which is complicated enough to make the differential equation \eqref{loc.115} unsolvable, then we are doomed to give up the attempt to solve the localization problem. Consequently, one should be careful about the choice of the coupling function $f(\phi,\chi)$. If the coupling function is too simple, then probably the localization of a 5-dimensional black hole would not be possible, because the simple cases have already been investigated and excluded. On the other hand, if the coupling function is too complicated, then, even if there is a solution to the problem, we might not be able to find it.
\par After about 20 years of research in the direction of finding a closed-form, analytical 5-dimensional localized braneworld black hole solution, there is still work to do. Fortunately, there are numerical solutions to the problem, thus, the research for an analytical solution is not completely in vain.

%% file: Appendices/Appendix_BHS.tex
\chapter{Black Hole Solutions in General Relativity}
\label{bhs}

\par \textbf{Schwarzschild Solution} $(M\neq 0,\ Q=0,\ J=0)$:
\eq$\label{bhs1}
\boxed{ds^2=-\left(1-\frac{r_S}{r}\right)(cdt)^2+\left(1-\frac{r_S}{r}\right)^{-1}dr^2+r^2(d\theta^2+\sin^2\theta\ d\varphi^2)}$
where
$$r_S=\frac{2GM}{c^2}$$\vspace{1em}
\par \textbf{Reissner-Nordstr\"{o}m Solution} $(M\neq0,\ Q\neq0,\ J=0)$:
\eq$\label{bhs2}
\boxed{ds^2=-\Delta\ (cdt)^2+\Delta^{-1}\ dr^2+r^2(d\theta^2+\sin^2\theta\ \varphi^2)}$
where
$$\Delta=1-\frac{r_S}{r}+\frac{r_Q^2}{r^2}$$
and
$$r_Q^2=\frac{GQ^2}{4\pi\epsilon_0c^4}$$\vspace{1em}
\par \textbf{Kerr Solution} $(M\neq0,\ Q=0,\ J\neq0)$:
\eq$\label{bhs3}
\boxed{\begin{gathered}
ds^2=-\left(1-\frac{r_S\ r}{\rho^2}\right)(cdt)^2-\frac{2\ r_S\ r a \sin^2\theta}{\rho^2}cdtd\varphi+\frac{\rho^2}{\Delta}dr^2+\rho^2 d\theta^2\\
\hspace{4em}+\frac{\sin^2\theta}{\rho^2}\left[(r^2+a^2)^2-a^2\Delta\sin^2\theta\right]d\varphi^2
\end{gathered}}$
where
\begin{gather}
a=\frac{J}{Mc}\nonum\\ \nonum\\
\Delta=r^2-r_S\ r+a^2\nonum\\ \nonum\\
\rho^2=r^2+a^2\cos^2\theta\nonum
\end{gather}
\par The \textbf{Kerr-Newman} line element $(M\neq0,Q\neq0,J\neq0)$ results from Eq.\eqref{bhs3} by replacing\  $r_S\ r$\ with\ $r_S\ r-r_Q^2$.

\newpage
\blankpage

%% file: Appendices/Appendix_GNC.tex
\chapter{Gaussian Normal Coordinates}
\label{gnc}

\par Let us consider a Lorentzian \emph{(n+1)-manifold} $\mathcal{M}$ with coordinates $\{x^M\}=\{x^0,x^1\ldots, x^n\}$, where $M$ runs from $0$ to $n$. A \emph{(d+1)-submanifold} $\mathcal{S}$ of the original (n+1)-dimensional manifold $\mathcal{M}$ with intrinsic coordinates $\{y^\mu\}=\{y^0,y^1,\ldots,y^d\}$ ($d<n$) can be defined via the following set of $n+1$ parametric equations:
\eq$\label{gnc1}
x^M=x^M(y^0,y^1,\ldots,x^d)\Ra\left\{\begin{array}{c}
x^0=x^0(y^0,y^1,\ldots,x^d)\\ 
x^1=x^1(y^0,y^1,\ldots,x^d)\\ 
\vdots\\
x^n=x^n(y^0,y^1,\ldots,x^d)
\end{array}\right\}$
For $d=n-1$ we have a hypersurface. In this case, it is possible to define the hypersurface without using the above parametric equations. The way that this can be done, is the following:
\eq$\label{gnc2}
F(x^0,x^1,\ldots,x^n)=c$
where $c$ is a constant. In every point $p$ of the hypersurface $\mathcal{S}$  there is a tangent plane $\bar{T}_p$ which can be thought as an n-dimensional subspace of the tangent plane $T_p$ of the (n+1)-dimensional manifold $\mathcal{M}$. It is intuitively clear that there is always an \emph{(n+1)-vector} $\vec{n}\in T_p$ (unique up to scaling) which is orthogonal to all vectors in $\bar T_p$. This vector $\vec{n}$ is said to be normal to the hypersurface $\mathcal{S}$. The unit normal vector $\hat{n}$ in any point on $\mathcal{S}$ is given by
\eq$\label{gnc3}
\hat{n}=n^M\vec{e}_M=\frac{\pa^M F}{\sqrt{|g^{AB}\pa_AF\pa_BF|}} \vec{e}_M$
In order to clarify the unitarity of $\hat{n}$ one can perform the scalar product of $\hat{n}$ with itself and substitute $\vec{e}_M\cdot \vec{e}_N$\ for\ $g_{MN}$. We silently assumed that $g_{AB}\ n^An^B\neq 0$, in the opposite case we have a null-hypersurface and therefore is not possible to normalize the normal vector as we did in Eq.\eqref{gnc3}. Let us now consider two events on the hypersurface $\mathcal{S}$ which are separated by the \emph{n-vector} $d\vec{r}=dx^A\vec{e}_A$. These events are so close that we can safely assume that the vector $d\vec{r}$ lies on the tangent plane $\bar T_p$. Hence, the scalar product between $d\vec{r}$ and $\hat{n}$ is zero.
\eq$\label{gnc4}
\hat{n}\cdot d\vec{r}=0\Ra(n^A\vec{e}_A)\cdot(dx^B\vec{e}_B)=0\Ra n^Adx^Bg_{AB}=n^A\ dx_A=n_A\ dx^A=0$
\par The Gaussian Normal Coordinates for any non-null hypersurface can be constructed by using the following steps. Firstly, for each point $p\in \mathcal{S}$ we find the unique geodesic curve that passes by the point $p$ and its tangent vector is $\hat n$. Thereafter, we choose a coordinate system $\{x^0,x^1,\ldots,x^{n-1}\}$ on $\mathcal{S}$ and then we characterize each point in a neighborhood of the hypersurface $\mathcal{S}$ by using these coordinates and the parameter $y$ which is along the geodesic curve that emanates from the point $p\in\mathcal{S}$. Therefore, it is always possible to find a local coordinate system $\{x^0,x^1,\ldots,x^n,y\}$ in a neighborhood of a point $p\in\mathcal{S}$ where a vector along coordinate $y$ is perpendicular to the hypersurface $\mathcal{S}$.

\newpage
\blankpage

%% file: Appendices/Appendix_LG.tex
\chapter{Linearized Gravity}
\label{ln}

\par Linearized gravity is simply an approximation which is used to describe weak gravitational fields. The spacetime in the context of this approximation is considered nearly flat, thus, the metric tensor is expressed by the following relation:
\eq$\label{ln1}
g_{MN}=\eta_{MN}+h_{MN}\left(x^\lam,y\right)$
where
\eq$\label{ln2}
|h_{MN}|\ll 1$
The use of capital Latin characters indicates that we take into account extra spatial dimensions. In this chapter particularly, we consider a total of (4+1)-dimensions $\{x^0,x^1,\ldots,x^3,y\}$. In the calculations that follow we ignore as negligible any term that contains non-linear orders of $h_{MN}$. Consequently, the Christoffel symbols are calculated as follows:
\begin{gather}
\Gam^L{}_{MN}=\frac{1}{2}g^{LR}(g_{MR,N}+g_{NR,M}-g_{MN,R})\Ra\nonum\\ 
\Gam^L{}_{MN}=\frac{1}{2}(\eta^{LR}+h^{LR})\left[\pa_N(\eta_{MR}+h_{MR})+\pa_M(\eta_{NR}+h_{NR})-\pa_R(\eta_{MN}+h_{MN})\right]\Ra\nonum\\ 
\label{ln3}
\boxed{\Gam^L{}_{MN}=\frac{1}{2}\eta^{LR}(h_{MR,N}+h_{NR,M}-h_{MN,R})=\frac{1}{2}\left(h^L{}_{M,N}+h^L{}_{N,M}-h_{MN}{}^{,L}\right)}
\end{gather}
The components of the Riemann tensor can be computed very easily as well.
\begin{gather}
R^L{}_{MRN}=\Gam^L{}_{MN,R}-\Gam^L{}_{MR,N}+\underbrace{\Gam^L{}_{RK}\Gam^K{}_{MN}}_{h^2\ra0}-\underbrace{\Gam^L{}_{NK}\Gam^K{}_{MR}}_{h^2\ra0}\Ra\nonum\\
R^L{}_{MRN}=\frac{1}{2}(\cancel{h^L{}_{M,NR}}+h^L{}_{N,MR}-h_{MN,R}{}^L-\cancel{h^L{}_{M,RN}}-h^L{}_{R,MN}+h_{MR,N}{}^L)\Ra\nonum\\
\label{ln4}
\boxed{R^L{}_{MRN}=\frac{1}{2}\left(h^L{}_{N,MR}-h_{MN,R}{}^L-h^L{}_{R,MN}+h_{MR,N}{}^L\right)}
\end{gather}
Using Eq.\eqref{ln4} and contracting indices $L$ and $R$ we obtain the components of the Ricci tensor.
\begin{gather}
R_{MN}=\frac{1}{2}\left(h^L{}_{N,ML}-h_{MN,L}{}^L-h^L{}_{L,MN}+h_{ML,N}{}^L\right)\Ra\nonum\\
\label{ln5}
\boxed{R_{MN}=\frac{1}{2}\left(-\square h_{MN}-h_{,MN}+h_{ML,N}{}^L+h^L{}_{N,ML}\right)}
\end{gather}
where 
\eq$\label{ln6}
h\equiv h^L{}_L$
Subsequently, the Ricci scalar is going to be evaluated, which emanates directly from Eq.\eqref{ln5}.
\begin{gather}
R=g^{MN}R_{MN}=\eta^{MN}R_{MN}+\underbrace{h^{MN}R_{MN}}_{h^2\ra0}=\frac{1}{2}(-\square h-\square h+h^{ML}{}_{,ML}+ h^{ML}{}_{,ML})\Ra\nonum\\
\label{ln7}
\boxed{R=-\square h+h^{ML}{}_{,ML}}
\end{gather}
Finally, the components of the Einstein tensor are given by
$$G_{MN}=R_{MN}-\frac{1}{2}R\ g_{MN}=R_{MN}-\frac{1}{2}R\ \eta_{MN}\Ra$$
\eq$\label{ln8}
\boxed{G_{MN}=\frac{1}{2}\left(-\square h_{MN}-h_{,MN}+h_{ML,N}{}^L+h^L{}_{N,ML}-\eta_{MN}\square h+\eta_{MN}h^{RL}{}_{,RL}\right)}$

%% file: Appendices/Appendix_Vaidya.tex
\chapter{From Schwarzschild to Vaidya}
\label{vaid}

\par First of all, we clarify that for all the mathematical expressions that are depicted below, the Planck units have been used. The Schwarzschild metric as the static and spherically symmetric solution to \-Einstein's equations of gravity is given by the following line element
\eq$\label{A.1}
ds^2=-\left(1-\frac{2M}{r}\right)dt^2+\left(1-\frac{2M}{r}\right)^{-1}dr^2+r^2(d\theta^2+\sin^2\theta\ d\varphi^2)$
\par Someone could switch to \emph{Eddington-Finkelstein coordinates} in order to remove \-the \-undesirable \-coordinate singularity of the metric at $r=2M$. The new coordinates are using the \-null \-coordinate $v$ with the following definition.
\eq$\label{A.2}
t=v-r-2M\ln\left(\frac{r}{2M}-1\right)\ \Ra \ \boxed{dt=dv-\left(1-\frac{2M}{r}\right)^{-1}dr}$
Subsequently, using equation \eqref{B.2} into \eqref{B.1} the line element takes the form
\eq$\label{A.3}
ds^2=-\left(1-\frac{2M}{r}\right)dv^2+2dvdr+r^2(d\theta^2+\sin^2\theta\ d\varphi^2)$
\par If we now extend the mass parameter M from a constant to a function of $v$ we get the \emph{Vaidya metric}.
\eq$\label{A.4}
\boxed{ds^2=-\left(1-\frac{2M(v)}{r}\right)dv^2+2dvdr+r^2(d\theta^2+\sin^2\theta\ d\varphi^2)}$

\newpage
\blankpage

%% file: Appendices/Appendix_5Dgeometry.tex
\chapter{Complete Information about the 5-D Geometrical Background}
\label{geom}

\section{Christoffel Symbols}
\par Christoffel symbols can be calculated by equation
\eq$\label{B.1}
\Gam^K{}_{MN}=\frac{1}{2}g^{KL}\left(g_{ML,N}+g_{NL,M}-g_{MN,L}\right)$
Hence, combining equations \eqref{B.1}, \eqref{loc.2} and \eqref{loc.3}, we obtain

\begin{empheq}[box=\mymath]{equation}\label{B.2}
\begin{array}{ccc}
\Gam^0{}_{00}=\frac{m-r \pa_r m}{r^2} & \Gam^0{}_{04}=\Gam^0{}_{40}=A' & \Gam^0{}_{22}=-r \\ \\ \Gam^0{}_{33}=-r \sin ^2\theta & \Gam^1{}_{00}=\frac{r^2 \pa_v m-(r-2 m) \left(r \pa_r m-m\right)}{r^3} & \Gam^1{}_{01}=\Gam^1{}_{10}=\frac{r \pa_rm-m}{r^2}\\ \\
\Gam^1{}_{04}=\Gam^1{}_{40}=\frac{\pa_ym}{r} & \Gam^1{}_{14}=\Gam^1{}_{41}=A' & \Gam^1{}_{22}=2 m-r \\ \\
\Gam^1{}_{33}=(2 m-r)\sin ^2\theta & \Gam^2{}_{12}=\Gam^2{}_{21}=\frac{1}{r} & \Gam^2{}_{24}=\Gam^2{}_{42}=A'\\ \\
\Gam^2{}_{33}=-\sin \theta \cos \theta & \Gam^3{}_{13}=\Gam^3{}_{31}=\frac{1}{r} & \Gam^3{}_{23}=\Gam^3{}_{32}=\cot \theta \\ \\
\Gam^3{}_{34}=\Gam^3{}_{43}=A' & \Gam^4{}_{00}=\frac{e^{2 A} \left(A' (r-2 m)-\pa_ym\right)}{r} & \Gam^4{}_{01}=\Gam^4{}_{10}=-e^{2 A} A'\\ \\
\Gam^4{}_{22}=-r^2 e^{2 A} A' & \Gam^4{}_{33}=-r^2A' e^{2 A} \sin ^2\theta & \end{array}
\end{empheq}

\section{Riemann Tensor}
\par Riemann tensor's components are defined by
\eq$\label{B.3}
R^L{}_{KMN}=\Gam^L{}_{NK,M}-\Gam^L{}_{MK,N}+\Gam^L{}_{MJ}\Gam^J{}_{NK}-\Gam^L{}_{NJ}\Gam^J{}_{MK}$
Thus, having in our disposal the Christoffel symbols from equation \eqref{B.2} it is possible to compute the components of the Riemann tensor, but also one would need an eternity to make calculations. Therefore, it will be presented a less time-consuming method that provides the Riemann tensor's components.

\newpage
\par First and foremost, the main ingredients of this method are the matrices $\Gam_M$ and $B_{MN}$.
\begin{itemize}
\item The $\Gam_M$ matrices are resulting from Christoffel symbols $\Gam^K{}_{ML}$ as follows. \emph{The component $\Gam^K{}_{ML}$ is defined to be the element of the K-th row and L-th column of the matrix $\Gam_M$.} So, there are 5 matrices $\Gam_M$ of size $5\times 5$.
\item The $B_{MN}$ matrices are resulting from the components of Riemann tensor $R^L{}_{KMN}$. \emph{The component $R^L{}_{KMN}$ is defined to be the element of the L-th row and K-th column of the matrix $B_{MN}$.} Thus, there are 25 matrices $B_{MN}$ of size $5\times 5$.
\end{itemize}
\par It is now obvious that making use of the previous definitions, equation \eqref{B.3} can be written in the following form.
\eq$\label{B.4}
B_{MN}=\Gam_{N,M}-\Gam_{M,N}+\Gam_{M}\Gam_{N}-\Gam_{N}\Gam_{M}$
\par It is necessary to mention that the matrices $B_{MN}$ have the useful property $B_{MN}=-B_{NM}$. This property arise from the antisymmetry of Riemann tensor $R^L{}_{KMN}=-R^L{}_{KNM}$. Therefore, this property ensures us that $B_{MM}=0$.
\par We can easily extract the $\Gamma_M$ matrices from equation \eqref{B.2}. Hence, we have
\eq$\label{B.5}
\boxed{\Gam_0=\left(
\begin{array}{ccccc}
 \frac{m-r \pa_rm}{r^2} & 0 & 0 & 0 & A' \\
 \frac{r^2 \pa_vm-(r-2 m) \left(r \pa_rm-m\right)}{r^3} & \frac{r \pa_rm-m}{r^2} & 0 & 0 &
   \frac{\pa_ym}{r} \\
 0 & 0 & 0 & 0 & 0 \\
 0 & 0 & 0 & 0 & 0 \\
 \frac{e^{2 A} \left[(r-2 m) A'-\pa_ym\right]}{r} & -e^{2 A} A' & 0 & 0 & 0 \\
\end{array}\right)}$
\eq$\label{B.6}
\boxed{\Gam_1=\left(
\begin{array}{ccccc}
 0 & 0 & 0 & 0 & 0 \\
 \frac{r \pa_rm-m}{r^2} & 0 & 0 & 0 & A' \\
 0 & 0 & \frac{1}{r} & 0 & 0 \\
 0 & 0 & 0 & \frac{1}{r} & 0 \\
 -e^{2 A} A' & 0 & 0 & 0 & 0 \\
\end{array}
\right)}$
\eq$\label{B.7}
\boxed{\Gam_2=\left(
\begin{array}{ccccc}
 0 & 0 & -r & 0 & 0 \\
 0 & 0 & 2 m-r & 0 & 0 \\
 0 & \frac{1}{r} & 0 & 0 & A' \\
 0 & 0 & 0 & \cot \theta & 0 \\
 0 & 0 & -e^{2 A} r^2 A' & 0 & 0 \\
\end{array}
\right)}$
\eq$\label{B.8}
\boxed{\Gam_3=\left(
\begin{array}{ccccc}
 0 & 0 & 0 & -r \sin ^2\theta & 0 \\
 0 & 0 & 0 & (2 m-r) \sin ^2\theta  & 0 \\
 0 & 0 & 0 & -\cos \theta  \sin \theta  & 0 \\
 0 & \frac{1}{r} & \cot \theta  & 0 & A' \\
 0 & 0 & 0 & -e^{2 A} r^2 \sin ^2\theta  A' & 0 \\
\end{array}
\right)}$
\eq$\label{B.9}
\boxed{\Gam_4=\left(
\begin{array}{ccccc}
 A' & 0 & 0 & 0 & 0 \\
 \frac{\pa_ym}{r} & A' & 0 & 0 & 0 \\
 0 & 0 & A' & 0 & 0 \\
 0 & 0 & 0 & A' & 0 \\
 0 & 0 & 0 & 0 & 0 \\
\end{array}
\right)}$

\newpage
\par The combination of equations \eqref{B.4}-\eqref{B.9} leads to the following $B_{MN}$ matrices.
\begin{center}
\underline{$\mathbf{B_{01}=-B_{10}}$}
\end{center}\vspace{1em}
$$B_{01}=\Gam_{1,0}-\Gam_{0,1}+\Gam_0\Gam_1-\Gam_1\Gam_0\Ra$$
\eq$\label{B.10}
\boxed{B_{01}=\left(
\begin{array}{ccccc}
 -e^{2 A} A'^2+\frac{2 \left(m-r \pa_rm\right)}{r^3}+\frac{\pa_r^2m}{r} & 0 & 0 & 0 & 0 \\
 -\frac{(r-2 m) \left[r \left(e^{2 A} r^2 A'^2+2 \pa_rm-r \pa_r^2m\right)-2 m\right]}{r^4} & e^{2 A}
   A'^2-\frac{r^2\pa_r^2m -2 r\pa_rm +2 m}{r^3} & 0 & 0 & \frac{\pa_ym-r \pa_r\pa_ym}{r^2} \\
 0 & 0 & 0 & 0 & 0 \\
 0 & 0 & 0 & 0 & 0 \\
 \frac{e^{2 A} \left(r \pa_r\pa_ym-\pa_ym\right)}{r^2} & 0 & 0 & 0 & 0 \\
\end{array}
\right)}$
\begin{center}
\underline{$\mathbf{B_{02}=-B_{20}}$}
\end{center}\vspace{1em}
$$B_{02}=\Gam_{2,0}-\Gam_{0,2}+\Gam_0\Gam_2-\Gam_2\Gam_0\Ra$$
\eq$\label{B.11}
\boxed{B_{02}=\left(
\begin{array}{ccccc}
 0 & 0 & \pa_rm-\frac{e^{2 A} A'^2 r^3+m}{r} & 0 & 0 \\
 0 & 0 & \pa_vm-e^{2 A} r A' \pa_ym & 0 & 0 \\
 \frac{-e^{2 A} A' r^3\left[(r-2 m) A'-\pa_ym\right] -r^2\pa_vm +(r-2 m) \left(r
   \pa_rm-m\right)}{r^4} & \frac{e^{2 A} A'^2 r^3-r\pa_rm +m}{r^3} & 0 & 0 & -\frac{\pa_ym}{r^2}
   \\
 0 & 0 & 0 & 0 & 0 \\
 0 & 0 & e^{2 A} \pa_ym & 0 & 0 \\
\end{array}
\right)}$
\begin{center}
\underline{$\mathbf{B_{03}=-B_{30}}$}
\end{center}\vspace{1em}
$$B_{03}=\Gam_{3,0}-\Gam_{0,3}+\Gam_0\Gam_3-\Gam_3\Gam_0\Ra$$
\eq$\label{B.12}
\boxed{B_{03}=\left(
\begin{array}{ccccc}
 0 & 0 & 0 & -\frac{\sin ^2\theta  \left(e^{2 A} A'^2 r^3-r\pa_rm +m\right)}{r} & 0 \\
 0 & 0 & 0 & \scriptstyle{-\sin ^2\theta \left(e^{2 A} r A' \pa_ym-\pa_vm\right)} & 0 \\
 0 & 0 & 0 & 0 & 0 \\
 \scriptscriptstyle{\frac{-e^{2 A} A' r^3\left[(r-2 m) A'-\pa_ym\right] -r^2\pa_vm +(r-2 m) \left(r
   \pa_rm-m\right)}{r^4}} & \frac{e^{2 A} A'^2 r^3-r\pa_rm +m}{r^3} & 0 & 0 & -\frac{\pa_ym}{r^2}
   \\
 0 & 0 & 0 & e^{2 A} \sin ^2\theta \pa_ym & 0 \\
\end{array}
\right)}$
\begin{center}
\underline{$\mathbf{B_{04}=-B_{40}}$}
\end{center}\vspace{1em}
$$B_{04}=\Gam_{4,0}-\Gam_{0,4}+\Gam_0\Gam_4-\Gam_4\Gam_0\Ra$$
\eq$\label{B.13}
\boxed{B_{04}=\left(
\begin{array}{ccccc}
 \frac{r \pa_r\pa_ym-\pa_ym}{r^2} & 0 & 0 & 0 & -A'^2-A'' \\
 \frac{(r-2 m) \left(r \pa_r\pa_ym-\pa_ym\right)}{r^3} & \frac{\pa_ym-r \pa_r\pa_ym}{r^2} & 0 & 0 & -\frac{2A' \pa_ym+\pa_y^2m}{r} \\
 0 & 0 & 0 & 0 & 0 \\
 0 & 0 & 0 & 0 & 0 \\
 \frac{e^{2 A} \left[-(r-2 m) \left(A'^2+A''\right)+2 A' \pa_ym+\pa_y^2m\right]}{r} & e^{2 A}
   \left(A'^2+A''\right) & 0 & 0 & 0 \\
\end{array}
\right)}$

\newpage
\begin{center}
\underline{$\mathbf{B_{12}=-B_{21}}$}
\end{center}
\eq$\label{B.14}
\boxed{B_{12}=\Gam_{2,1}-\Gam_{1,2}+\Gam_1\Gam_2-\Gam_2\Gam_1=\left(
\begin{array}{ccccc}
 0 & 0 & 0 & 0 & 0 \\
 0 & 0 & \pa_rm-\frac{e^{2 A} A'^2 r^3+m}{r} & 0 & 0 \\
 \frac{e^{2 A} A'^2 r^3-r\pa_rm +m}{r^3} & 0 & 0 & 0 & 0 \\
 0 & 0 & 0 & 0 & 0 \\
 0 & 0 & 0 & 0 & 0 \\
\end{array}
\right)}$
\begin{center}
\underline{$\mathbf{B_{13}=-B_{31}}$}
\end{center}
\eq$\label{B.15}
\boxed{B_{13}=\Gam_{3,1}-\Gam_{1,3}+\Gam_1\Gam_3-\Gam_3\Gam_1=\left(
\begin{array}{ccccc}
 0 & 0 & 0 & 0 & 0 \\
 0 & 0 & 0 & -\frac{\sin ^2\theta \left(e^{2 A} A'^2 r^3-r\pa_rm +m\right)}{r} & 0 \\
 0 & 0 & 0 & 0 & 0 \\
 \frac{e^{2 A} A'^2 r^3-r\pa_rm +m}{r^3} & 0 & 0 & 0 & 0 \\
 0 & 0 & 0 & 0 & 0 \\
\end{array}
\right)}$
\begin{center}
\underline{$\mathbf{B_{14}=-B_{41}}$}
\end{center}
\eq$\label{B.16}
\boxed{B_{14}=\Gam_{4,1}-\Gam_{1,4}+\Gam_1\Gam_4-\Gam_4\Gam_1=\left(
\begin{array}{ccccc}
 0 & 0 & 0 & 0 & 0 \\
 0 & 0 & 0 & 0 & -A'^2-A'' \\
 0 & 0 & 0 & 0 & 0 \\
 0 & 0 & 0 & 0 & 0 \\
 e^{2 A} \left(A'^2+A''\right) & 0 & 0 & 0 & 0 \\
\end{array}
\right)}$
\begin{center}
\underline{$\mathbf{B_{23}=-B_{32}}$}
\end{center}
\eq$\label{B.17}
\boxed{B_{23}=\Gam_{3,2}-\Gam_{2,3}+\Gam_2\Gam_3-\Gam_3\Gam_2=\left(
\begin{array}{ccccc}
 0 & 0 & 0 & 0 & 0 \\
 0 & 0 & 0 & 0 & 0 \\
 0 & 0 & 0 & \frac{\sin ^2\theta\left(2 m-e^{2 A} r^3 A'^2\right)}{r} & 0 \\
 0 & 0 & e^{2 A} r^2 A'^2-\frac{2 m}{r} & 0 & 0 \\
 0 & 0 & 0 & 0 & 0 \\
\end{array}
\right)}$
\begin{center}
\underline{$\mathbf{B_{24}=-B_{42}}$}
\end{center}
\eq$\label{B.18}
\boxed{B_{24}=\Gam_{4,2}-\Gam_{2,4}+\Gam_2\Gam_4-\Gam_4\Gam_2=\left(
\begin{array}{ccccc}
 0 & 0 & 0 & 0 & 0 \\
 0 & 0 & -\pa_ym & 0 & 0 \\
 \frac{\pa_ym}{r^2} & 0 & 0 & 0 & -A'^2-A'' \\
 0 & 0 & 0 & 0 & 0 \\
 0 & 0 & e^{2 A} r^2 \left(A'^2+A''\right) & 0 & 0 \\
\end{array}
\right)}$
\begin{center}
\underline{$\mathbf{B_{34}=-B_{43}}$}
\end{center}
\eq$\label{B.19}
\boxed{B_{34}=\Gam_{4,3}-\Gam_{3,4}+\Gam_3\Gam_4-\Gam_4\Gam_3=\left(
\begin{array}{ccccc}
 0 & 0 & 0 & 0 & 0 \\
 0 & 0 & 0 & -\sin ^2\theta \pa_ym & 0 \\
 0 & 0 & 0 & 0 & 0 \\
 \frac{\pa_ym}{r^2} & 0 & 0 & 0 & -A'^2-A'' \\
 0 & 0 & 0 & e^{2 A} r^2 \sin ^2\theta \left(A'^2+A''\right) & 0 \\
\end{array}
\right)}$
\par Finally, knowing the matrices $B_{MN}$ from equations \eqref{B.10}-\eqref{B.19} and using the fact that the L-th row and K-th column of the matrix $B_{MN}$ gives the $R^L{}_{KMN}$ component of the Riemann tensor, it is quite easy to extract the components of the Riemann tensor, all non-zero components are depicted below. 
\begin{empheq}[box=\mymath]{equation}
\begin{array}{cc}
\scriptstyle{R^0{}_{001}=-R^0{}_{010}=\frac{2 \left(m-r \pa_rm\right)}{r^3}+\frac{\pa^2_rm}{r}-e^{2 A} A'^2}&\scriptstyle{R^0{}_{004}=-R^0{}_{040}=\frac{r \pa_r\pa_ym-\pa_ym}{r^2}}\\ \\
\scriptstyle{R^0{}_{202}=-R^0{}_{220}=\pa_rm-\frac{m+r^3 e^{2 A} A'^2}{r}}&\scriptstyle{R^0{}_{303}=-R^0{}_{330}=-\frac{\sin ^2\theta \left(-r \pa_rm+m+r^3 e^{2 A} A'^2\right)}{r}}\\ \\
\scriptstyle{R^0{}_{404}=-R^0{}_{440}=-A''(y)-A'(y)^2}&\scriptstyle{R^1{}_{004}=-R^1{}_{040}=\frac{(r-2 m) \left(r \pa_r\pa_ym-\pa_ym\right)}{r^3}}\\ \\
\scriptstyle{R^1{}_{001}=-R^1{}_{010}=-\frac{(r-2 m) \left[r \left(2 \pa_rm-r \pa^2_rm+r^2 e^{2 A} A'^2\right)-2 m\right]}{r^4}}&\scriptstyle{R^1{}_{101}=-R^1{}_{110}=e^{2 A} A'^2-\frac{r^2 \pa_r^2m-2 r \pa_rm+2 m}{r^3}}\\ \\
\scriptstyle{R^1{}_{104}=-R^1{}_{140}=\frac{\pa_ym-r \pa_r\pa_ym}{r^2}}&\scriptstyle{R^1{}_{202}=-R^1{}_{220}=\pa_vm-r e^{2 A} A' \pa_ym}\\ \\
\scriptstyle{R^1{}_{212}=-R^1{}_{221}=\pa_rm-\frac{m+r^3 e^{2 A} A'^2}{r}}&\scriptstyle{R^1{}_{224}=-R^1{}_{242}=-\pa_ym}\\ \\
\scriptstyle{R^1{}_{303}=-R^1{}_{330}=-\sin ^2\theta  \left(r e^{2 A} A' \pa_ym-\pa_vm\right)}&\scriptstyle{R^1{}_{313}=-R^1{}_{331}=-\frac{\sin ^2\theta \left(-r \pa_rm+m+r^3 e^{2 A} A'^2\right)}{r}}\\ \\
\scriptstyle{R^1{}_{334}=-R^1{}_{343}=-\sin ^2\theta  \pa_ym}&\scriptstyle{R^1{}_{401}=-R^1{}_{410}=\frac{\pa_ym-r \pa_r\pa_ym}{r^2}}\\ \\
\scriptstyle{R^1{}_{404}=-R^1{}_{440}=-\frac{2 A' \pa_ym+\pa^2_ym}{r}}&\scriptstyle{R^1{}_{414}=-R^1{}_{441}=-A''-A'^2}\\ \\
\scriptscriptstyle{R^2{}_{002}=-R^2{}_{020}=\frac{-r^3e^{2 A} A' \left[A' (r-2 m)-\pa_ym\right]-r^2 \pa_vm+(r-2 m) \left(r \pa_rm-m\right)}{r^4}}&\scriptstyle{R^2{}_{012}=-R^2{}_{021}=\frac{-r \pa_rm+m+r^3 e^{2 A} A'^2}{r^3}}\\ \\
\scriptstyle{R^2{}_{024}=-R^2{}_{042}=\frac{\pa_ym}{r^2}}&\scriptstyle{R^2{}_{102}=-R^2{}_{120}=\frac{-r \pa_rm+m+r^3 e^{2 A} A'^2}{r^3}}\\ \\
\scriptstyle{R^2{}_{323}=-R^2{}_{332}=\frac{\sin ^2\theta \left(2 m-r^3 e^{2 A} A'^2\right)}{r}}&\scriptstyle{R^2{}_{402}=-R^2{}_{420}=-\frac{\pa_ym}{r^2}}\\ \\
\scriptstyle{R^2{}_{424}=-R^2{}_{442}=-A''-A'^2}&\scriptstyle{R^3{}_{013}=-R^3{}_{031}=\frac{-r \pa_rm+m+r^3 e^{2 A} A'^2}{r^3}}\\ \\
\scriptscriptstyle{R^3{}_{003}=-R^3{}_{030}=\frac{-r^3e^{2 A} A' \left[A' (r-2m)-\pa_ym\right]-r^2 \pa_vm+(r-2 m) \left(r \pa_rm-m\right)}{r^4}}&\scriptstyle{R^3{}_{034}=-R^3{}_{043}=\frac{\pa_ym}{r^2}}\\ \\
\scriptstyle{R^3{}_{103}=-R^3{}_{130}=\frac{-r \pa_rm+m+r^3 e^{2 A} A'^2}{r^3}}&\scriptstyle{R^3{}_{223}=-R^3{}_{232}=r^2 e^{2 A} A'^2-\frac{2 m}{r}}\\ \\
\scriptstyle{R^3{}_{403}=-R^3{}_{430}=-\frac{\pa_ym}{r^2}}&\scriptstyle{R^3{}_{434}=-R^3{}_{443}=-A''-A'^2}\\ \\
\scriptstyle{R^4{}_{001}=-R^4{}_{010}=\frac{e^{2 A} \left(r \pa_r\pa_ym-\pa_ym\right)}{r^2}}&\scriptstyle{R^4{}_{014}=-R^4{}_{041}=e^{2 A} \left(A''+A'^2\right)}\\ \\
\scriptstyle{R^4{}_{004}=-R^4{}_{040}=\frac{e^{2 A} \left[2 A'\pa_ym+\left(A''+A'^2\right) (2m-r)+\pa_y^2m\right]}{r}}&\scriptstyle{R^4{}_{104}=-R^4{}_{140}=e^{2 A} \left(A''+A'^2\right)}\\ \\
\scriptstyle{R^4{}_{202}=-R^4{}_{220}=e^{2 A} \pa_ym}&\scriptstyle{R^4{}_{224}=-R^4{}_{242}=r^2 e^{2 A} \left(A''+A'^2\right)} \\ \\
\scriptstyle{R^4{}_{303}=-R^4{}_{330}=e^{2 A} \sin ^2\theta \pa_ym}&\scriptstyle{R^4{}_{334}=-R^4{}_{343}=r^2 e^{2 A} \sin ^2\theta  \left(A''+A'^2\right)}
\end{array}\nonum
\end{empheq}

\newpage
\section{Riemann Scalar}
\par Riemann scalar is defined as
\eq$\label{B.20}
R^{ABCD}R_{ABCD}$
Therefore, it is necessary to calculate the four-times contravariant and four-times covariant components of Riemann tensor in order to be able to execute the above contraction.
\par Contravariant and covariant components of Riemann tensor have some very useful properties that reduce the number of independent components. Hence, we present these properties below in order to use them afterwards.
\begin{itemize}
\item $R^{ABCD}=-R^{BACD}=-R^{ABDC}$
\item $R^{ABCD}=R^{CDAB}$
\item $R_{ABCD}=-R_{BACD}=-R_{ABDC}$
\item $R_{ABCD}=R_{CDAB}$
\end{itemize}
\par Knowing the components $R^L{}_{KMN}$ of the Riemann tensor, the four-times contravariant components are given by the equation
\eq$\label{B.21}
R^{ABCD}=R^{A}{}_{KMN}g^{KB}g^{MC}g^{ND}$
and the non-zero four-times contravariant components are
\eq$\label{B.22}
\boxed{\begin{array}{c}
\scriptstyle{R^{0101}=-R^{1001}=-R^{0110}=R^{1010}=e^{-6 A} \left(e^{2 A} A'^2-\frac{r^2 \pa_r^2m-2 r \pa_rm+2 m}{r^3}\right)}\\ \\
\scriptstyle{R^{0114}=-R^{1014}=-R^{0141}=R^{1401}=-R^{4101}=-R^{1410}=R^{1041}=R^{4110}=\frac{e^{-4 A} \left(r \pa_r\pa_ym-\pa_ym\right)}{r^2}}\\ \\
\scriptstyle{R^{0212}=-R^{2012}=-R^{0221}=R^{1202}=-R^{2102}=-R^{1220}=R^{2021}=R^{2120}=-\frac{e^{-6 A} \left(-r \pa_rm+m+r^3 e^{2 A} A'^2\right)}{r^5}}\\ \\
\scriptstyle{R^{0313}=-R^{3013}=-R^{0331}=R^{1303}=-R^{3103}=-R^{1330}=R^{3031}=R^{3130}=-\frac{e^{-6 A} \csc ^2\theta \left(-r \pa_rm+m+r^3 e^{2 A} A'^2\right)}{r^5}}\\ \\
\scriptstyle{R^{0414}=-R^{4014}=-R^{0441}=R^{1404}=-R^{4104}=-R^{1440}=R^{4041}=R^{4140}=-e^{-2 A} \left(A''+A'^2\right)}\\ \\
\scriptstyle{R^{1212}=-R^{2112}=-R^{1221}=R^{2121}=\frac{e^{-6 A} \left[(r-2 m) \left(\pa_rm-\frac{m+r^3 e^{2 A} A'^2}{r}\right)+r \left(\pa_vm-r e^{2 A} A'\pa_ym\right)\right]}{r^5}}\\ \\
\scriptstyle{R^{1224}=-R^{2124}=-R^{1242}=R^{2412}=-R^{4212}=-R^{2421}=R^{2142}=R^{4241}=-\frac{e^{-4 A} \pa_ym}{r^4}}\\ \\
\scriptstyle{R^{1313}=-R^{3113}=-R^{1331}=R^{3131}=\frac{e^{-6 A} \csc ^2\theta \left[r^2\left(\pa_vm-r e^{2 A} A' \pa_ym\right)-(r-2 m) \left(-r \pa_rm+m+r^3 e^{2A} A'^2\right)\right]}{r^6}}\\ \\
\scriptstyle{R^{1334}=-R^{3134}=-R^{1343}=R^{3413}=-R^{4313}=-R^{3431}=R^{3143}=R^{4331}=-\frac{e^{-4 A} \csc ^2\theta \pa_ym}{r^4}}\\ \\
\scriptstyle{R^{2323}=-R^{3223}=-R^{2332}=R^{3232}=\frac{e^{-6 A} \csc ^2\theta\left(2 m-r^3 e^{2 A}A'^2\right)}{r^7}}\\ \\
\scriptstyle{R^{2424}=-R^{4224}=-R^{2442}=R^{4242}=-\frac{e^{-2 A} \left(A''+A'^2\right)}{r^2}}\\ \\
\scriptstyle{R^{3434}=-R^{4334}=-R^{3443}=R^{4343}=-\frac{e^{-2 A} \csc ^2\theta \left(A''+A'^2\right)}{r^2}}\\ \\
\scriptstyle{R^{1414}=-R^{4114}=-R^{1441}=R^{4141}=-\frac{e^{-2 A} \left[2 A' \pa_ym+\left(A''+A'^2\right) (r-2m)+\pa_y^2m\right]}{r}}\\ \\
\end{array}}$
\par  The four-times covariant components are given by the equation
\eq$\label{B.23}
R_{ABCD}=g_{AL}R^L{}_{BCD}$
and the non-zero four-times covariant components are
\eq$\label{B.24}
\boxed{\begin{array}{c}
\scriptstyle{R_{0101}=-R_{1001}=-R_{0110}=R_{1010}=e^{4 A} A'^2-\frac{e^{2 A} \left[r \left(r\pa_r^2m-2 \pa_rm\right)+2 m\right]}{r^3}}\\ \\
\scriptstyle{R_{0104}=-R_{1004}=-R_{0140}=R_{0401}=-R_{4001}=-R_{0410}=R_{1040}=R_{4010}=\frac{e^{2 A} \left(\pa_ym-r \pa_r\pa_ym\right)}{r^2}}\\ \\
\scriptstyle{R_{0202}=-R_{2002}=-R_{0220}=R_{2020}=e^{2 A} \left[\frac{(r-2 m) \left(-r \pa_rm+m+r^3 e^{2 A} A'^2\right)}{r^2}-re^{2 A} A' \pa_ym+\pa_vm\right]}\\ \\
\scriptstyle{R_{0212}=-R_{2012}=-R_{0221}=R_{1202}=-R_{2102}=-R_{1220}=R_{2021}=R_{2120}=e^{2 A} \left(\pa_rm-\frac{m+r^3 e^{2 A} A'^2}{r}\right)}\\ \\
\scriptstyle{R_{0224}=-R_{2024}=-R_{0242}=R_{2402}=-R_{4202}=-R_{2420}=R_{2042}=R_{4220}=-e^{2 A} \pa_ym}\\ \\
\scriptstyle{R_{0303}=-R_{3003}=-R_{0330}=R_{3030}=e^{2 A} \sin ^2\theta\left[\frac{(r-2 m) \left(-r \pa_rm+m+r^3 e^{2 A} A'^2\right)}{r^2}-re^{2 A} A' \pa_ym+\pa_vm\right]}\\ \\
\scriptstyle{R_{0313}=-R_{3013}=-R_{0331}=R_{1303}=-R_{3103}=-R_{1330}=R_{3031}=R_{3130}=-\frac{e^{2 A} \sin ^2\theta \left(-r\pa_rm+m+r^3 e^{2 A} A'^2\right)}{r}}\\ \\
\scriptstyle{R_{0334}=-R_{3034}=-R_{0343}=R_{3403}=-R_{4303}=-R_{3430}=R_{3043}=R_{4330}=-e^{2 A} \sin ^2\theta \pa_ym}\\ \\
\scriptstyle{R_{0404}=-R_{4004}=-R_{0440}=R_{4040}=\frac{e^{2 A} \left[-2 A'\pa_ym+\left(A''+A'^2\right) (r-2 m)-\pa_y^2m\right]}{r}}\\ \\
\scriptstyle{R_{0414}=-R_{4014}=-R_{0441}=R_{1404}=-R_{4104}=-R_{1440}=R_{4041}=R_{4140}=-e^{2 A} \left(A''+A'^2\right)}\\ \\
\scriptstyle{R_{2323}=-R_{3223}=-R_{2332}=R_{3232}=r e^{2 A} \sin ^2\theta\left(2 m-r^3 e^{2 A} A'^2\
\right)}\\ \\
\scriptstyle{R_{2424}=-R_{4224}=-R_{2442}=R_{4242}=-r^2e^{2 A}\left(A''+A'^2\right)}\\ \\
\scriptstyle{R_{3434}=-R_{4334}=-R_{3443}=R_{4343}=-r^2e^{2 A} \sin ^2\theta\left(A''+A'^2\right)}\\ \\
\end{array}}$
\par Possessing all non-zero contravariant and covariant components of Riemann tensor from equations \eqref{B.22} and \eqref{B.24} respectively, we are now able to calculate the Riemann scalar.
\eqarr$\label{B.25}
R^{ABCD}R_{ABCD}&=&4R^{0101}R_{0101}+8R^{0212}R_{0212}+8R^{0313}R_{0313}+8R^{0414}R_{0414}\nonum\\
&&+4R^{2323}R_{2323}+4R^{2424}R_{2424}+4R^{3434}R_{3434}$
All the other combinations of contraction of the indices $ABCD$ nullify either $R^{ABCD}$ or $R_{ABCD}$.
\begin{empheq}[box=\mymath]{equation}
\begin{array}{c}\scriptstyle{R^{ABCD}R_{ABCD}}=\displaystyle{40A'^4+32A'^2A''+16A''^2+}\displaystyle{\frac{48e^{-4A}m^2}{r^6}}-\frac{8e^{-2A}A'^2}{r}\left(\pa_r^2m+\frac{2\pa_rm}{r}\right)\\ \\
\hspace{6em}+\displaystyle{\frac{4e^{-4A}}{r^2}}\left[(\pa_r^2m)^2+\frac{4m}{r^2}\left(\pa_r^2m-\frac{4\pa_rm}{r}\right)-\frac{4\pa_rm\pa_r^2m}{r}+\frac{8(\pa_r m)^2}{r}\right]
\end{array}\nonum
\end{empheq}

\section{Ricci Tensor and Ricci Scalars}
\par The two-times covariant components of Ricci tensor is defined by
\eq$\label{B.26}
R_{MN}=R^L{}_{MLN}=R^0{}_{M0N}+R^1{}_{M1N}+R^2{}_{M2N}+R^3{}_{M3N}+R^4{}_{M4N}$
\par Subsequently, the non-zero components of the Ricci tensor are
\begin{empheq}[box=\mymath]{equation}\label{B.27}
\begin{array}{c}
 R_{00}=\frac{e^{2 A} r \left[(r-2 m) \left(4 A'^2+A''\right)-4 A' \pa_ym-\pa_y^2m\right]-(r-2 m)\pa_r^2m+2 \pa_vm}{r^2}\\ \\
 R_{01}=R_{10}=\frac{\pa_r^2m}{r}-e^{2 A} \left(4 A'^2+A''\right)\\ \\
 R_{04}=R_{40}=\frac{\pa_ym+r \pa_r\pa_ym}{r^2} \\ \\ 
 R_{22}=2 \pa_rm-e^{2 A} r^2 \left(4 A'^2+A''\right)\\ \\
R_{33}= \sin ^2\theta \left[2 \pa_rm-e^{2 A} r^2 \left(4 A'^2+A''\right)\right]\\ \\
R_{44}=-4 \left(A'^2+A''\right)
\end{array}\end{empheq}
\par There are two Ricci scalars that one can evaluate in order to extract information for the curvature of spacetime. The first one is defined as $R=R_{MN}g^{MN}$ and the second one is defined as $R_{MN}R^{MN}$. It is obvious that for the evaluation of the latter scalar the two-times contravariant components of Ricci tensor are necessary. The following table depicts all the non-zero contravariant components of Ricci tensor. The evaluation is made by the following equation
\eq$\label{B.28}
R^{MN}=R_{AB}g^{AM}g^{BN}$
and the in between analytical actions are skipped here as well. Therefore, we have
\begin{empheq}[box=\mymath]{equation}\label{B.29}
\begin{array}{c}
R^{01}=R^{10}=e^{-4 A} \left[\frac{\pa_r^2m}{r}-e^{2 A} \left(4 A'^2+A''\right)\right]\\ \\
 R^{11}=\frac{e^{-4 A} \left[-e^{2 A} r \left((r-2m) \left(4 A'^2+A''\right)+4 A'\pa_ym+\pa_y^2m\right)+(r-2 m) \pa_r^2m+2\pa_vm\right]}{r^2}\\ \\
 R^{14}=R^{41}=\frac{e^{-2 A} \left(\pa_ym+r \pa_r\pa_ym\right)}{r^2} \\ \\ 
R^{22}=\frac{e^{-4 A} \left[2 \pa_rm-e^{2 A} r^2 \left(4 A'^2+A''\right)\right]}{r^4}\\ \\
R^{33}=\frac{e^{-4 A} \csc ^2\theta \left[2 \pa_rm-e^{2 A} r^2 \left(4 A'^2+A''\right)\right]}{r^4}\\ \\
R^{44}= -4 \left(A'^2+A''\right)
\end{array}\end{empheq}
\par We now have everything that is needed for the evaluation of the two aforementioned Ricci scalars.
\begin{empheq}[box=\mymath]{equation}\label{B.30}
R=R_{MN}g^{MN}=-20 A'^2-8 A''+\frac{2 e^{-2 A}}{r} \left(\pa_r^2m+\frac{2 \pa_rm}{r}\right)
\end{empheq}
\begin{empheq}[box=\mymath]{equation}\label{B.31}
\scriptstyle{R_{MN}R^{MN}=80A'^2+64A'^2 A''+20A''^2-\frac{4 e^{-2 A}}{r}\left(\pa_r^2m+\frac{2\pa_rm}{r}\right)(4A'^2+A'')+ \frac{2e^{-4A}}{r^2}\left[(\pa_r^2m)^2+\frac{4 (\pa_rm)^2}{r^2}\right]}
\end{empheq}

\section{Einstein tensor}
\par Einstein tensor $G_{MN}$ is defined through Ricci tensor $R_{MN}$ and Ricci scalar $R$ as follows
\eq$\label{B.32}
G_{MN}=R_{MN}-\frac{1}{2}g_{MN}R$ 
The combination of equations \eqref{loc.2},\eqref{B.30} and \eqref{B.32} can be done easily, thus, after a bit of algebra we are led to the following non-zero components of Einstein tensor.
\begin{empheq}[box=\mymath]{equation}\label{B.33}
\begin{array}{c}
G_{00}=\frac{-e^{2 A} \left[3 (r-2 m) \left(2 A'^2+A''\right)+4 A' \pa_ym+\pa_y^2m\right] r^2+2 r\pa_vm +2(r-2 m)\pa_r m}{r^3}\\ \\
G_{01}=G_{10}=3 e^{2 A} \left(2 A'^2+A''\right)-\frac{2 \pa_rm}{r^2}\\ \\
G_{04}=G_{40}=\frac{\pa_ym+r\pa_r\pa_y m}{r^2} \\ \\
G_{22}=r \left[3 e^{2 A} r \left(2 A'^2+A''\right)-\pa_r^2m\right] \\ \\
G_{33}=r \sin ^2\theta \left[3 e^{2 A} r \left(2 A'^2+A''\right)-\pa_r^2m\right] \\ \\
G_{44}=6 A'^2-\frac{e^{-2 A} \left(2 \pa_rm+r\pa_r^2 m\right)}{r^2}
\end{array}
\end{empheq}
\par The mixed components of Einstein tensor $G^M{}_{N}$ are given by
\eq$\label{B.34}
G^M{}_N=g^{MA}G_{AN}$
Equations \eqref{loc.3},\eqref{B.33} and \eqref{B.34} yield to
\begin{empheq}[box=\mymath]{equation}\label{B.35}
\begin{array}{c}
G^0{}_0=6 A'^2+3 A''-\frac{2 e^{-2 A} \pa_rm}{r^2} \\ \\
G^1{}_0=-\frac{4 r A' \pa_ym+r \pa_y^2m-2 e^{-2 A} \pa_vm}{r^2}\\ \\
G^1{}_1=6 A'^2+3 A''-\frac{2 e^{-2 A}\pa_r m}{r^2} \\ \\
G^1{}_4=\frac{e^{-2 A} \left(\pa_ym+r \pa_r\pa_ym\right)}{r^2} \\ \\
G^2{}_2=6 A'^2+3 A''-\frac{e^{-2 A}\pa_r^2 m}{r}\\ \\
G^3{}_3=6 A'^2+3 A''-\frac{e^{-2 A} \pa_r^2m}{r}\\ \\
G^4{}_0=\frac{\pa_ym+r \pa_r\pa_ym}{r^2} \\ \\
G^4{}_4=6 A'^2-\frac{e^{-2 A} \left(2 \pa_rm+r\pa_r^2 m\right)}{r^2}
\end{array}
\end{empheq}

\newpage
\section{Mathematica Code}
\par The mathematica code that is used to verify all the above calculations is illustrated below.\\
\begin{center}
\textcolor{purple}{\textbf{\underline{Definition of coordinates and metric tensor}}}
\end{center}
\begin{doublespace}
\noindent\(\pmb{n=5;}\\
\pmb{}\\
\pmb{\text{coord}=\{v,r,\theta ,\varphi ,y\};}\\
\pmb{}\\
\pmb{g=\{\{-\text{Exp}[2A[y]](1-2m[v,r,y]/r),\text{Exp}[2A[y]],0,0,0\},\{\text{Exp}[2A[y]],0,0,0,0\},}\\
\pmb{\{0,0,\text{Exp}[2A[y]]*r{}^{\wedge}2,0,0\},\{0,0,0,\text{Exp}[2A[y]]*r{}^{\wedge}2*\text{Sin}[\theta ]{}^{\wedge}2,0\},\{0,0,0,0,1\}\};}\\
\pmb{}\\
\pmb{\text{StringJoin}\left[\text{Characters}\left[\text{$\texttt{"}$(}g_{\text{MN}}\text{)=$\texttt{"}$}\right]\right]\text{MatrixForm}[g]}\\
\pmb{}\\
\pmb{\text{invg}=\text{Simplify}[\text{Inverse}[g]];}\\
\pmb{}\\
\pmb{\text{StringJoin}\left[\text{Characters}\left[\text{$\texttt{"}$(}g^{\text{MN}}\text{)=$\texttt{"}$}\right]\right]\text{MatrixForm}[\text{invg}]}\)
\end{doublespace}
\begin{center}
\textcolor{purple}{\textbf{\underline{Christoffel Symbols:\ \ \ $\mathbf{\Gam^L{}_{MN}}$}}}
\end{center}
\begin{doublespace}
\noindent\(\pmb{\text{christoffel}=\text{FullSimplify}[\text{Table}[(1/2)*\text{Sum}[(\text{invg}[[\lambda ,\rho ]])*(D[g[[\mu ,\rho ]],\text{coord}[[\nu
]]]}\\
\pmb{\text{       }+D[g[[\nu ,\rho ]],\text{coord}[[\mu ]]]-D[g[[\mu ,\nu ]],\text{coord}[[\rho ]]]),\{\rho ,1,n\}],\{\lambda ,1,n\},\{\mu ,1,n\},\{\nu
,1,n\}]];}\\
\pmb{}\\
\pmb{\text{chr}=\text{Table}[\text{If}[\text{UnsameQ}[\text{christoffel}[[\lambda ,\mu ,\nu ]],0],\{\text{ToString}[\Gamma [\lambda -1,\mu -1,\nu
-1]],}\\
\pmb{\text{christoffel}[[\lambda ,\mu ,\nu ]]\}],\{\lambda ,1,n\},\{\mu ,1,n\},\{\nu ,1,n\}];}\\
\pmb{}\\
\pmb{\text{TableForm}[\text{Partition}[\text{DeleteCases}[\text{Flatten}[\text{chr}],\text{Null}],2],\text{TableSpacing}\to \{2,2\}]}\)
\end{doublespace}
\begin{center}
\textcolor{purple}{\textbf{\underline{Riemann Tensor's Components:\ \ \  $\mathbf{R^K{}_{LMN}}$}}}
\end{center}
\begin{doublespace}
\noindent\(\pmb{\text{riemann}=\text{FullSimplify}[\text{Table}[D[\text{christoffel}[[\alpha ,\mu ,\sigma ]],\text{coord}[[\rho ]]]-D[\text{christoffel}[[\alpha
,\mu ,\rho ]],\text{coord}[[\sigma ]]]}\\
\pmb{\text{       }+\text{Sum}[\text{christoffel}[[\alpha ,\rho ,\lambda ]]*\text{christoffel}[[\lambda ,\mu ,\sigma ]],\{\lambda ,1,n\}]-\text{Sum}[\text{christoffel}[[\alpha
,\sigma ,\lambda ]]}\\
\pmb{\text{   }*\text{christoffel}[[\lambda ,\mu ,\rho ]],\{\lambda ,1,n\}],\{\alpha ,1,n\},\{\mu ,1,n\},\{\rho ,1,n\},\{\sigma ,1,n\}]];}\vspace{7em}
\pmb{}\\
\pmb{\text{rie}=\text{Table}[\text{If}[\text{UnsameQ}[\text{riemann}[[\alpha ,\mu ,\rho ,\sigma ]],0],\{\text{ToString}[R[\alpha -1,\mu -1,\rho -1,\sigma
-1]],}\\
\pmb{\text{riemann}[[\alpha ,\mu ,\rho ,\sigma ]]\}],}\pmb{\{\alpha ,1,n\},\{\mu ,1,n\},\{\rho ,1,n\},\{\sigma ,1,n\}];}\\
\pmb{}\\
\pmb{\text{TableForm}[\text{Partition}[\text{DeleteCases}[\text{Flatten}[\text{rie}],\text{Null}],2],\text{TableSpacing}\to \{2,2\}]}\)
\end{doublespace}
\begin{center}
\textcolor{purple}{\textbf{\underline{Contravariant Components of Riemann Tensor:\ \ \  $\mathbf{R^{ABCD}}$}}}
\end{center}
\begin{doublespace}
\noindent\(\pmb{\text{riemanncon}=\text{FullSimplify}[\text{Table}[\text{Sum}[\text{Sum}[\text{Sum}[\text{riemann}[[\mu ,\alpha ,\beta ,\gamma ]]*\text{invg}[[\alpha
,\nu ]],\{\alpha ,1,n\}]*}\\
\pmb{\text{        }\text{invg}[[\beta ,\kappa ]],\{\beta ,1,n\}]*\text{invg}[[\gamma ,\lambda ]],\{\gamma ,1,n\}],\{\mu ,1,n\},\{\nu ,1,n\},\{\kappa
,1,n\},\{\lambda ,1,n\}]];}\\
\pmb{}\\
\pmb{\text{riecon}=\text{Table}[\text{If}[\text{UnsameQ}[\text{riemanncon}[[\alpha ,\mu ,\rho ,\sigma ]],0],\{\text{ToString}[\text{Rcon}[\alpha
-1,\mu -1,\rho -1,\sigma -1]],}\\
\pmb{\text{       }\text{riemanncon}[[\alpha ,\mu ,\rho ,\sigma ]]\}],\{\alpha ,1,n\},\{\mu ,1,n\},\{\rho ,1,n\},\{\sigma ,1,n\}];}\\
\pmb{}\\
\pmb{\text{TableForm}[\text{Partition}[\text{DeleteCases}[\text{Flatten}[\text{riecon}],\text{Null}],2],\text{TableSpacing}\to \{2,2\}]}\)
\end{doublespace}
\begin{center}
\textcolor{purple}{\textbf{\underline{Covariant Components of Riemann Tensor:\ \ \  $\mathbf{R_{ABCD}}$}}}
\end{center}
\begin{doublespace}
\noindent\(\pmb{\text{riemanncov}=\text{FullSimplify}[\text{Table}[\text{Sum}[g[[\mu ,\delta ]]*\text{riemann}[[\delta ,\nu ,\kappa ,\lambda ]],}\\
\pmb{\{\delta ,1,n\}],\{\mu ,1,n\},\{\nu ,1,n\},\{\kappa ,1,n\},\{\lambda ,1,n\}]];}\\
\pmb{}\\
\pmb{\text{riecov}=\text{Table}[\text{If}[\text{UnsameQ}[\text{riemanncov}[[\alpha ,\mu ,\rho ,\sigma ]],0],\{\text{ToString}[\text{Rcov}[\alpha
-1,\mu -1,\rho -1,\sigma -1]],}\\
\pmb{\text{riemanncov}[[\alpha ,\mu ,\rho ,\sigma ]]\}],\{\alpha ,1,n\},\{\mu ,1,n\},\{\rho ,1,n\},\{\sigma ,1,n\}];}\\
\pmb{}\\
\pmb{\text{TableForm}[\text{Partition}[\text{DeleteCases}[\text{Flatten}[\text{riecov}],\text{Null}],2],\text{TableSpacing}\to \{2,2\}]}\)
\end{doublespace}
\begin{center}
\textcolor{purple}{\textbf{\underline{Covariant Components of Ricci Tensor:\ \ \  $\mathbf{R_{MN}}$}}}
\end{center}
\begin{doublespace}
\noindent\(\pmb{\text{ricci}=\text{FullSimplify}[\text{Table}[\text{Sum}[\text{riemann}[[\mu ,\alpha ,\mu ,\beta ]],\{\mu ,1,n\}],\{\alpha ,1,n\},\{\beta
,1,n\}]];}\\
\pmb{}\\
\pmb{\text{StringJoin}\left[\text{Characters}\left[\text{$\texttt{"}$(}R_{\text{MN}}\text{)=$\texttt{"}$}\right]\right]\text{MatrixForm}[\text{ricci}]}\)
\end{doublespace}
\begin{center}
\textcolor{purple}{\textbf{\underline{Contravariant Components of Ricci Tensor:\ \ \  $\mathbf{R^{MN}}$}}}
\end{center}
\begin{doublespace}
\noindent\(\pmb{\text{riccicon}=\text{FullSimplify}[\text{Table}[\text{Sum}[\text{ricci}[[\alpha ,\beta ]]*\text{invg}[[\alpha ,\mu ]]*\text{invg}[[\beta
,\nu ]],\{\alpha ,1,n\},\{\beta ,1,n\}],\{\mu ,1,n\},}\\
\pmb{\{\nu ,1,n\}]];}\\
\pmb{}\\
\pmb{\text{StringJoin}\left[\text{Characters}\left[\text{$\texttt{"}$(}R^{\text{MN}}\text{)=$\texttt{"}$}\right]\right]\text{MatrixForm}[\text{riccicon}]}\)
\end{doublespace}
\begin{center}
\textcolor{purple}{\textbf{\underline{Ricci and Riemann Scalars}}}
\end{center}
\begin{doublespace}
\noindent\(\pmb{\text{scalarricci}=\text{FullSimplify}[\text{Sum}[\text{invg}[[\mu ,\nu ]]*\text{ricci}[[\mu ,\nu ]],\{\mu ,1,n\},\{\nu ,1,n\}]];}\\
\pmb{}\\
\pmb{\text{StringJoin}\left[\text{Characters}\left[\texttt{"}R_{\text{MN}}g^{\text{MN}}\text{=$\texttt{"}$}\right]\right].\text{  }\text{scalarricci}}\\
\pmb{}\\
\pmb{\text{scalarricci2}=\text{FullSimplify}[\text{Sum}[\text{ricci}[[\mu ,\nu ]]*\text{riccicon}[[\mu ,\nu ]],\{\mu ,1,n\},\{\nu ,1,n\}]];}\\
\pmb{}\\
\pmb{\text{StringJoin}\left[\text{Characters}\left[\texttt{"}R_{\text{MN}}R^{\text{MN}}\text{=$\texttt{"}$}\right]\right] .\text{scalarricci2}}\\
\pmb{}\\
\pmb{\text{StringJoin}\left[\text{Characters}\left[\texttt{"}R_{\text{ABCD}}R^{\text{ABCD}}\text{=$\texttt{"}$}\right]\right] .\text{FullSimplify}[\text{Sum}[\text{riemanncov}[[\mu
,\nu ,\kappa ,\lambda ]]*}\\
\pmb{\text{riemanncon}[[\mu ,\nu ,\kappa ,\lambda ]],\{\mu ,1,n\},\{\nu ,1,n\},\{\kappa ,1,n\},\{\lambda ,1,n\}]]}\\
\pmb{}\)
\end{doublespace}
\begin{center}
\textcolor{purple}{\textbf{\underline{Covariant Components of Einstein Tensor:\ \ \ $\mathbf{G_{MN}}$}}}
\end{center}
\begin{doublespace}
\noindent\(\pmb{\text{einstein}=\text{FullSimplify}[\text{Table}[\text{ricci}[[\mu ,\nu ]]-(1/2)g[[\mu ,\nu ]] *\text{scalarricci},\{\mu ,1,n\},\{\nu
,1,n\}]];}\\
\pmb{}\\
\pmb{\text{StringJoin}\left[\text{Characters}\left[\text{$\texttt{"}$(}G_{\text{MN}}\text{)=$\texttt{"}$}\right]\right]\text{MatrixForm}[\text{einstein}]}\)
\end{doublespace}
\begin{center}
\textcolor{purple}{\textbf{\underline{Mixed Components of Einstein Tensor:\ \ \ $\mathbf{G^M{}_N}$}}}
\end{center}
\begin{doublespace}
\noindent\(\pmb{\text{einsteinUD}=\text{FullSimplify}[\text{Table}[\text{Sum}[\text{invg}[[\mu ,\lambda ]]*\text{einstein}[[\lambda ,\nu ]],\{\lambda
,1,n\}],\{\mu ,1,n\},\{\nu ,1,n\}]];}\\
\pmb{}\\
\pmb{\text{StringJoin}\left[\text{Characters}\left[\text{$\texttt{"}$(}G_{\text{    }N}^M\text{)=$\texttt{"}$}\right]\right]\text{MatrixForm}[\text{einsteinUD}]}\)
\end{doublespace}

%% file: Appendices/Appendix_Variation.tex
\chapter{Non-minimal coupling and Variation of the action}
\label{var}

\section{Variation with respect to the metric tensor}

\par We consider the following general action for a non-minimally coupled scalar field.
\eq$\label{C.1}
S=\int d^4x\sqrt{-g}\left[\frac{f(\Phi)}{2\kappa} R-\frac{1}{2}\pa_\mu \Phi\pa^\mu\Phi-V(\Phi)-\Lambda_B\right]+\int d^4x\sqrt{-g}\ \lagr_m$
where $\lagr_m$ describes matter and/or radiation. It is obvious from the above action that we are restricted in a 4-D spacetime. Every calculation and proof in this section is going to be done in 4-D spacetime. However, the generalization for extra spatial dimensions is instantaneous, the only thing that changes in the final field equations are the indices. For a 4-D spacetime the indices are Greek letters while for (4+n)-D spacetime the indices are capital Latin letters. 
\par In order to deduce the field equations of this general theory, we apply the principle of least action to Eq.\eqref{C.1}. We vary the action with respect to the metric tensor. Thus, we have:
\begin{align}\label{C.2}
\del S=0&=\int d^4x\left\{(\del\sqrt{-g})\left[\frac{f(\Phi)}{2\kappa} R-\frac{1}{2}\pa_\mu \Phi\pa^\mu\Phi-V(\Phi)-\Lambda_B\right]+\sqrt{-g}\left[\frac{f(\Phi)}{2\kappa} \del R-\frac{1}{2}\del(\pa_\mu \Phi\pa^\mu\Phi)\right]\right\}\nonum\\
&\hspace{24em}+\int d^4x\ \del(\sqrt{-g}\ \lagr_{m})\Ra\nonum\\
0&=\int d^4x\left\{(\del\sqrt{-g})\left[\frac{f(\Phi)}{2\kappa} R-\frac{1}{2}\pa_\mu \Phi\pa^\mu\Phi-V(\Phi)-\Lambda_B\right]+\sqrt{-g}\left[\frac{f(\Phi)}{2\kappa} \del R-\frac{1}{2}(\pa_\mu \Phi\pa_\nu\Phi)\del g^{\mu\nu}\right]\right\}\nonum\\
&\hspace{24em}+\int d^4x\ \del(\sqrt{-g}\ \lagr_{m})
\end{align}
\par We will calculate one by one the varying terms of the previous equation. First of all, we will prove the Jacobi's formula
\begin{equation}\label{C.3}
\frac{d}{dt}\{det[A(t)]\}=tr\left\{adj[A(t)]\frac{dA(t)}{dt}\right\}
\end{equation} 
which is necessary in order to derive the desirable field equations. We will prove Eq.\eqref{C.3} in two steps. At first, we prove a preliminary lemma.
\begin{lemma} Let A and B be a pair of square matrices of the same dimension n. Then
\begin{equation}\label{C.4}
\sum_{i=1}^n \sum_{j=1}^n A_{ij}B_{ij}=tr(A^{T}B)
\end{equation}
\end{lemma}
\begin{proof}[\textbf{Proof:}]
\begin{align}
(AB)_{jk}&=\sum_{i=1}^n A_{ji}B_{ik} \nonumber \\ \nonumber \\
(A^{T}B)_{jk}&=\sum_{i=1}^n A^{T}{}_{ji}B_{ik}=\sum_{i=1}^n A_{ij}B_{ik} \nonumber \\ \nonumber \\
tr(A^{T}B)&=\sum_{j=1}^n (A^{T}B)_{jj}=\sum_{j=1}^n\sum_{i=1}^n A_{ij}B_{ij}=\sum_{i=1}^n\sum_{j=1}^nA_{ij}B_{ij} \nonumber
\end{align}
\end{proof}
\begin{theorem}(Jacobi's formula):
\begin{equation}\label{C.5}
d[det(A)]=tr[adj(A)dA]
\end{equation}
\end{theorem}
\begin{proof}[\textbf{Proof:}]
Laplace's formula (or cofactor expansion) for the determinant of a matrix $A$ can be stated as
$$det(A)=\sum^{n}_{j=1} A_{ij}(-1)^{(i+j)}M_{ij}=\sum^{n}_{j=1} A_{ij}C_{ij}=\sum^{n}_{j=1} A_{ij}[adj(A)]_{ji}=\sum^{n}_{j=1} A_{ij}[adj^{T}(A)]_{ij}$$\\
where $M_{ij}$ is the i, j minor matrix of $A$, that is, the determinant of the $(n-1)\times(n-1)$ matrix that results from deleting the i-th row and j-th column of $A$. The summation is performed over some arbitrary row i of the matrix. The i, j cofactor of $A$ is the scalar $C_{ij}$ defined by $C_{ij}\equiv(-1)^{i+j}M_{ij}$. We note as well that the adjugate matrix of $A$ is the transpose of the cofactor matrix $C$ of $A$, which means that $adj(A)\equiv C^T\Ra [adj(A)]_{ij}=C_{ji}$.
\par The determinant of $A$ can also be considered to be a function of elements of $A$:
$$det(A)=F(A_{11},A_{12},...,A_{21},...,A_{nn})$$
so that, by the chain rule, its differential is
$$d[det(A)]=\sum_{i,j}\frac{\partial F}{\partial A_{ij}}dA_{ij}=\sum_{i,j}\frac{\partial det(A)}{\partial A_{ij}}dA_{ij}$$
\begin{align}
\frac{\partial det(A)}{\partial A_{ij}}&=\frac{\partial}{\partial A_{ij}}\left(\sum_{k} A_{ik}[adj^{T}(A)]_{ik}\right)=\sum_{k} \frac{\partial(A_{ik}[adj^{T}(A)]_{ik})}{\partial A_{ij}} \nonumber \\
&=\sum_{k} \underbrace{\frac{\partial A_{ik}}{\partial A_{ij}}}_{\del_{kj}}[adj^{T}(A)]_{ik}+\sum_{k} A_{ik}\underbrace{\frac{\partial [adj^{T}(A)]_{ik}}{\partial A_{ij}}}_{=0}=\sum_{k} \delta_{kj}[adj^{T}(A)]_{ik} \nonumber \\
&=[adj^{T}(A)]_{ij} \nonumber 
\end{align}
Hence, we have
$$d[det(A)]=\sum_{i,j}\frac{\partial det(A)}{\partial A_{ij}}dA_{ij}=\sum_{i,j} [adj^{T}(A)]_{ij}dA_{ij}=tr[adj^{T}dA]=tr[adj(A)dA]$$
Consequently, Eq.\eqref{C.3} is proved.
\end{proof}
\par We are now ready to calculate the variations of Eq.\eqref{C.2}.
\begin{itemize}
\item $\delta (\sqrt{-g})$:
\begin{equation}\label{C.6}
\delta (\sqrt{-g})=\frac{1}{2}\frac{1}{\sqrt{-g}}\delta (-g)=-\frac{1}{2}\frac{1}{\sqrt{-g}}\delta g
\end{equation}\\
It is important to explain the following notation
\eq$\label{C.7}
g\equiv det(g_{\mu\nu})$
where $(g_{\mu\nu})$ constitutes the matrix which depicts the components of metric tensor $g_{\mu\nu}$.
From equations \eqref{C.4} and \eqref{C.5} we obtain
\begin{equation}\label{C.8}
\delta g=\delta[det(g_{\mu \nu})]=\sum_{\alp,\beta}[adj^{T}(g_{\mu\nu})]_{\alp\beta}\ \delta g_{\alp\beta}
\end{equation}
In matrix calculus it is well known that $A^{-1}=\frac{1}{det(A)}adj(A)$, hence we have
$$(g_{\mu\nu})^{-1}=(g^{\mu\nu})=\frac{1}{det(g_{\mu \nu})}adj(g_{\mu\nu})\Ra$$
\eq$\label{C.9}
adj(g_{\mu\nu})=g(g^{\mu\nu})=adj^{T}(g_{\mu\nu})\Rightarrow [adj^{T}(g_{\mu\nu})]_{\alp \beta}=g\ g^{\alp \beta}$\\
Thus, from Eq.\eqref{C.8} and Eq.\eqref{C.9} we get
\begin{equation}\label{C.10}
\delta g=\sum_{\alp,\beta}g\ g^{\alp\beta}\delta g_{\alp\beta} \Rightarrow \delta g=g\ g^{\mu \nu}\delta g_{\mu \nu}
\end{equation}
It also obvious that
\eqarr$
\delta(g_{\mu \nu}g^{\nu \lambda})&=&\delta(\delta^{\lambda}{}_{\mu})=0 \Rightarrow \nonumber \\
\delta(g_{\mu \nu})g^{\nu \lambda}&=&-g_{\mu \nu}\delta (g^{\nu \lambda})\Rightarrow \nonumber$
\begin{equation}\label{C.11}
\boxed{\delta g^{\rho \sigma}=-g^{\rho \mu}g^{\sigma \nu}\delta g_{\mu \nu} \ , \ \delta g_{\alpha \beta}=-g_{\alpha \mu}g_{\beta \nu}\delta g^{\mu \nu}}
\end{equation}\\
Combining now equations \eqref{C.10} and \eqref{C.11} we have
\begin{equation}\label{C.12}
\boxed{ \delta g=g\ g^{\mu \nu}\delta g_{\mu \nu}=-g\ g_{\alpha \beta}\delta g^{\alpha \beta}}
\end{equation}
Therefore, we can now easily evaluate the quantity $\del(\sqrt{-g})$.
\eqarr$\label{C.13}
&&\delta (\sqrt{-g})=-\frac{1}{2}\frac{1}{\sqrt{-g}}g\ g^{\mu \nu}\delta g_{\mu \nu}=\frac{1}{2}\frac{1}{\sqrt{-g}}g\ g_{\alpha \beta}\delta g^{\alpha \beta} \nonumber \\ \nonumber \\
&&\boxed{\delta (\sqrt{-g})=\frac{\sqrt{-g}}{2} g^{\mu \nu}\delta g_{\mu \nu}=-\frac{\sqrt{-g}}{2} g_{\alpha \beta}\delta g^{\alpha \beta}}$

\item $\delta R$:
\\ \\In order to calculate $\delta R$ we need to prove first the Palatini's identity:
\begin{equation}\label{C.14}
\boxed{\delta R_{\mu \nu}=(\delta \Gamma^{\lambda}{}_{\mu \nu})_{;\lambda}-(\delta \Gamma^{\lambda}{}_{\lambda \mu})_{;\nu}=\nabla_\lam(\delta \Gamma^{\lambda}{}_{\mu \nu})-\nabla_\nu(\delta \Gamma^{\lambda}{}_{\lambda \mu})}
\end{equation}
\begin{proof}[\textbf{Proof:}]
$$R_{\mu \nu}=\Gamma^{\lambda}{}_{\mu \nu ,\lambda}-\Gamma^{\lambda}{}_{\mu \lambda ,\nu}+\Gamma^{\alpha}{}_{\mu \nu}\Gamma^{\lambda}{}_{\alpha \lambda}-\Gamma^{\lambda}{}_{\alpha \nu}\Gamma^{\alpha}{}_{\mu \lambda} \Rightarrow$$
\begin{align}\label{C.15}
\delta R_{\mu \nu}=\delta \Gamma^{\lambda}{}_{\mu \nu ,\lambda}&-\delta \Gamma^{\lambda}{}_{\mu \lambda,\nu}+\delta (\Gamma^{\alpha}{}_{\mu \nu})\Gamma^{\lambda}{}_{\alpha \lambda}+\Gamma^{\alpha}{}_{\mu \nu} \delta (\Gamma^{\lambda}{}_{\alpha \lambda})\nonum\\
&-\delta (\Gamma^{\lambda}{}_{\alpha \nu})\Gamma^{\alpha}{}_{\mu \lambda}-\Gamma^{\lambda}{}_{\alpha \nu}\delta (\Gamma^{\alpha}{}_{\mu \lambda}) 
\end{align}
Christoffel symbols of first kid $\Gam_{\mu\nu\sig}$ is defined by
\eqarr$\label{C.16}
\Gamma_{\mu \nu \sigma}&\equiv& g_{\mu \lambda}\Gamma^{\lambda}{}_{\nu \sigma}=\frac{1}{2}g_{\mu\lam}g^{\lam \rho}(g_{\nu\rho,\sig}+g_{\sig\rho,\nu}-g_{\nu\sig,\rho})=\frac{1}{2}\del^{\rho}{}_\mu(g_{\nu\rho,\sig}+g_{\sig\rho,\nu}-g_{\nu\sig,\rho})\nonum\\
&=&\frac{1}{2}(g_{\mu \nu ,\sigma}+g_{\mu \sigma ,\nu}-g_{\nu \sigma ,\mu})$
where we have used the expression of Christoffel symbols of second kid (or simply Christoffel symbols) which are well known. Using Eq.\eqref{C.16} we have
\begin{align}\label{C.17}
\delta(\Gamma^{\lambda}{}_{\nu \sigma})&=\delta(g^{\lambda \mu}\Gamma_{\mu \nu \sigma})=\delta(g^{\lambda \mu})\Gamma_{\mu \nu \sigma}+g^{\lambda \mu}\delta \Gamma_{\mu \nu \sigma}=-g^{\lambda \alpha}g^{\mu \beta}\delta (g_{\alpha \beta})\Gamma_{\mu \nu \sigma}+g^{\lambda \mu}\delta \Gamma_{\mu \nu \sigma} \nonumber \\
&=-g^{\lambda \alpha}\delta (g_{\alpha \beta})\Gamma^{\beta}{}_{\nu \sigma}+g^{\lambda \mu}\ \delta \left[\frac{1}{2}(g_{\mu \nu ,\sigma}+g_{\mu \sigma ,\nu}-g_{\nu \sigma ,\mu})\right] \nonumber \\
&=-g^{\lambda \mu}\delta (g_{\mu \beta})\Gamma^{\beta}{}_{\nu \sigma}+g^{\lambda \mu}\frac{1}{2}(\delta g_{\mu \nu ,\sigma}+\delta g_{\mu \sigma ,\nu}-\delta g_{\nu \sigma ,\mu}) \nonumber \\
&=\frac{1}{2}g^{\lambda \mu}(\delta g_{\mu \nu ,\sigma}+\delta g_{\mu \sigma ,\nu}-\delta g_{\nu \sigma ,\mu}-2\Gamma^{\beta}{}_{\nu \sigma}\delta g_{\mu \beta}) \nonumber \\
&=\frac{1}{2}g^{\lambda \mu}(\delta g_{\mu \nu ,\sigma}+\delta g_{\mu \sigma ,\nu}-\delta g_{\nu \sigma ,\mu}-2\Gamma^{\beta}{}_{\nu \sigma}\delta g_{\mu \beta}-\Gamma^{\beta}{}_{\nu \mu}\delta g_{\sigma \beta}\nonum\\
&\hspace{10em}+\Gamma^{\beta}{}_{\nu \mu}\delta g_{\sigma \beta}-\Gamma^{\beta}{}_{\sigma \mu}\delta g_{\nu \beta}+\Gamma^{\beta}{}_{\sigma \mu}\delta g_{\nu \beta}) \nonumber \\
&=\frac{1}{2}g^{\lambda \mu}[(\delta g_{\mu \nu ,\sigma}-\Gamma^{\beta}{}_{\sigma \mu}\delta g_{\nu \beta}-\Gamma^{\beta}{}_{\nu \sigma}\delta g_{\mu \beta})+(\delta g_{\mu \sigma ,\nu}-\Gamma^{\beta}{}_{\nu \mu}\delta g_{\sigma \beta}-\Gamma^{\beta}{}_{\nu \sigma}\delta g_{\mu \beta})\nonum\\
&\hspace{10em}-(\delta g_{\nu \sigma ,\mu}-\Gamma^{\beta}{}_{\nu \mu}\delta g_{\sigma \beta}-\Gamma^{\beta}{}_{\sigma \mu}\delta g_{\nu \beta})] \nonumber \\
&=\frac{1}{2}g^{\lam\mu}(\nabla_\sig\del g_{\mu\nu}+\nabla_\nu\del g_{\mu\sig}-\nabla_\mu\del g_{\nu\sig})
\end{align}\\
It is clear that elements like $\del g_{\mu\nu}$ and $\del g_{\mu\nu;\sig}$ are tensors. Thus, from Eq.\eqref{C.17} we can easily deduce that $\del\Gam^{\lam}{}_{\mu\nu}$ is also a tensor, because $\del\Gam^{\lam}{}_{\mu\nu}$ is expressed as a linear combination of tensors. Moreover, the quantity $(\del\Gam^{\lam}{}_{\mu\nu})_{;\lam}-(\del\Gam^{\lam}{}_{\lam\mu})_{;\nu}$ constitutes a tensor. Performing the expansion of the expression $(\del\Gam^{\lam}{}_{\mu\nu})_{;\lam}-(\del\Gam^{\lam}{}_{\lam\mu})_{;\nu}$ we obtain
\begin{align}\label{C.18}
(\delta \Gamma^{\lambda}{}_{\mu \nu})_{;\lambda}-(\delta \Gamma^{\lambda}{}_{\mu \lambda})_{;\nu}=&+(\delta \Gamma^{\lambda}{}_{\mu \nu ,\lambda}+\Gamma^{\lambda}{}_{\lambda \rho}\delta \Gamma^{\rho}{}_{\mu \nu}-\Gamma^{\rho}{}_{\mu \lambda}\delta \Gamma^{\lambda}{}_{\rho \nu}-\Gamma^{\rho}{}_{\nu \lambda}\delta \Gamma^{\lambda}{}_{\rho \mu}) \nonumber \\
&-(\delta \Gamma^{\lambda}{}_{\mu \lambda ,\nu}+\Gamma^{\lambda}{}_{\nu \rho}\delta \Gamma^{\rho}{}_{\mu \lambda}-\Gamma^{\rho}{}_{\mu \nu}\delta \Gamma^{\lambda}{}_{\rho \lambda}-\Gamma^{\rho}{}_{\nu \lambda}\delta \Gamma^{\lambda}{}_{\rho \mu}) \nonumber \\
=&+(\delta \Gamma^{\lambda}{}_{\mu \nu ,\lambda}+\Gamma^{\lambda}{}_{\lambda \rho}\delta \Gamma^{\rho}{}_{\mu \nu}-\Gamma^{\rho}{}_{\mu \lambda}\delta \Gamma^{\lambda}{}_{\rho \nu}) \nonumber \\
&-(\delta \Gamma^{\lambda}{}_{\mu \lambda ,\nu}+\Gamma^{\lambda}{}_{\nu \rho}\delta \Gamma^{\rho}{}_{\mu \lambda}-\Gamma^{\rho}{}_{\mu \nu}\delta \Gamma^{\lambda}{}_{\rho \lambda})
\end{align}
The right hand side of Eq.\eqref{C.15} is identical to the right hand side of Eq.\eqref{C.18}, so the left hand sides should be equal to each other as well. Hence, we derived the desirable Palatini's identity.
\end{proof}
\par The calculation of the quantity $\del R$ using the Eq.\eqref{C.14} is very simple.
\begin{align}\label{C.19}
\delta R&=\delta (g^{\mu \nu}R_{\mu \nu})=(\delta g^{\mu \nu})R_{\mu \nu}+g^{\mu \nu}\delta R_{\mu \nu} \nonumber \\
&=(\delta g^{\mu \nu})R_{\mu \nu}+g^{\mu \nu}\left(\nabla_\lam \delta \Gamma^{\lambda}{}_{\mu \nu}-\nabla_\nu \delta \Gamma^{\lambda}{}_{\mu \lambda}\right)
\end{align}
We have shown that 
\begin{align}\label{C.20}
\delta\Gamma^{\lambda}{}_{\nu \sigma}&=\frac{1}{2}g^{\lam\mu}(\nabla_\sig\del g_{\mu\nu}+\nabla_\nu\del g_{\mu\sig}-\nabla_\mu\del g_{\nu\sig})\nonum\\
&=\frac{1}{2}g^{\lam\rho}(\nabla_\sig\del g_{\rho\nu}+\nabla_\nu\del g_{\rho\sig}-\nabla_\rho\del g_{\nu\sig})
\end{align}
Using Eq.\eqref{C.20} we rename the index $\sig$ to $\mu$. Furthermore, using the symmetry $\del\Gam^{\lam}{}_{\nu\mu}=\del\Gam^{\lam}{}_{\mu\nu}$ we get
\eq$\label{C.21}
\delta\Gamma^{\lambda}{}_{\mu\nu}=\frac{1}{2}g^{\lam\rho}(\nabla_\mu\del g_{\rho\nu}+\nabla_\nu\del g_{\rho\mu}-\nabla_\rho\del g_{\nu\mu})$
Contracting the index $\lam$ with $\nu$ we have
\eqarr$\label{C.22}
\delta\Gamma^{\lambda}{}_{\mu\lam}=\delta\Gamma^{\lambda}{}_{\lam\mu}&=&\frac{1}{2}g^{\lam\rho}(\nabla_\mu\del g_{\rho\lam}+\nabla_\lam\del g_{\rho\mu}-\nabla_\rho\del g_{\lam\mu})\nonum\\
&=&\frac{1}{2}g^{\lam\rho}\nabla_\mu\del g_{\rho\lam}+\frac{1}{2}\left(g^{\lam\rho}\nabla_\lam\del g_{\rho\mu}-g^{\lam\rho}\nabla_\rho\del g_{\lam\mu}\right)\nonum\\
&=&\frac{1}{2}g^{\lam\rho}\nabla_\mu\del g_{\rho\lam}+\frac{1}{2}\left(\underbrace{g^{\rho\lam}\nabla_\rho\del g_{\lam\mu}-g^{\lam\rho}\nabla_\rho\del g_{\lam\mu}}_{0}\right)\nonum\\
&=&\frac{1}{2}g^{\lam\rho}\nabla_\mu\del g_{\rho\lam}$
We combine equations \eqref{C.14}, \eqref{C.21} and \eqref{C.22}. Thus, we get
\eqarr$\label{C.23}
\del R_{\mu\nu}&=&\nabla_\lam\left[\frac{1}{2}g^{\lam\rho}(\nabla_\mu\del g_{\rho\nu}+\nabla_\nu\del g_{\rho\mu}-\nabla_\rho\del g_{\nu\mu})\right]-\nabla_\nu\left(\frac{1}{2}g^{\lam\rho}\nabla_\mu\del g_{\rho\lam}\right)\nonum\\
&=&\frac{1}{2}g^{\lam\rho}(\nabla_\lam\nabla_\mu\del g_{\rho\nu}+\nabla_\lam\nabla_\nu\del g_{\rho\mu}-\nabla_\lam\nabla_\rho\del g_{\nu\mu})-\frac{1}{2}g^{\lam\rho}\nabla_\nu\nabla_\mu \del g_{\rho\lam}\nonum\\
&=&\frac{1}{2}g^{\lam\rho}(\nabla_\lam\nabla_\mu\del g_{\rho\nu}+\nabla_\lam\nabla_\nu\del g_{\rho\mu}-\nabla_\lam\nabla_\rho\del g_{\nu\mu}-\nabla_\nu\nabla_\mu \del g_{\rho\lam})$
Previously, we used the property $\nabla_\lam\ g^{\mu\nu}=0$. The proof of this property is presented below.\vspace{2em}
\begin{proof}[\textbf{Proof:}]
\begin{align}\label{C.24}
\nabla_\lam\ g^{\mu\nu}&=g^{\mu\nu}{}_{,\lam}+\Gam^\mu{}_{\lam\rho}\ g^{\rho\nu}+\Gam^{\nu}{}_{\lam\rho}\ g^{\rho\mu}\nonum\\
&=g^{\mu\nu}{}_{,\lam}+\frac{1}{2}g^{\mu\kappa}(g_{\lam\kappa,\rho}+g_{\rho\kappa,\lam}-g_{\lam\rho,\kappa})g^{\rho\nu}+\frac{1}{2}g^{\nu\kappa}(g_{\lam\kappa,\rho}+g_{\rho\kappa,\lam}-g_{\lam\rho,\kappa})g^{\rho\mu}\nonum\\
&=g^{\mu\nu}{}_{,\lam}+\frac{1}{2}g^{\mu\kappa}(g_{\lam\kappa,\rho}+g_{\rho\kappa,\lam}-g_{\lam\rho,\kappa})g^{\rho\nu}+\frac{1}{2}g^{\nu\rho}(g_{\lam\rho,\kappa}+g_{\kappa\rho,\lam}-g_{\lam\kappa,\rho})g^{\kappa\mu}\nonum\\
&=g^{\mu\nu}{}_{,\lam}+g^{\mu\kappa}\ g_{\kappa\rho,\lam}\ g^{\rho\nu}=g^{\mu\nu}{}_{,\lam}+\underbrace{(g^{\mu\kappa}\ g_{\kappa\rho})_{,\lam}}_{\del^\mu{}_{\rho,\lam}=0}\ g^{\rho\nu}-g^{\mu\kappa}{}_{,\lam}\ \underbrace{g_{\kappa\rho}\ g^{\rho\nu}}_{\del_\kappa{}^\nu}\nonum\\
&=g^{\mu\nu}{}_{,\lam}-g^{\mu\nu}{}_{,\lam}=0
\end{align}
\end{proof}
\eqarr$\label{C.25}
g^{\mu\nu}\del R_{\mu\nu}&=&\frac{1}{2}g^{\mu\nu}g^{\lam\rho}(\nabla_\lam\nabla_\mu\del g_{\rho\nu}+\nabla_\lam\nabla_\nu\del g_{\rho\mu}-\nabla_\lam\nabla_\rho\del g_{\nu\mu}-\nabla_\nu\nabla_\mu \del g_{\rho\lam})\nonum\\
&=&\frac{1}{2}\left(g^{\mu\nu}g^{\lam\rho}\nabla_\lam\nabla_\mu\del g_{\rho\nu}+g^{\mu\nu}g^{\lam\rho}\nabla_\lam\nabla_\nu\del g_{\rho\mu}-g^{\mu\nu}g^{\lam\rho}\nabla_\lam\nabla_\rho\del g_{\nu\mu}-g^{\mu\nu}g^{\lam\rho}\nabla_\nu\nabla_\mu \del g_{\rho\lam}\right)\nonum\\
&=&\frac{1}{2}\left(g^{\nu\mu}g^{\lam\rho}\nabla_\lam\nabla_\nu\del g_{\rho\mu}+g^{\mu\nu}g^{\lam\rho}\nabla_\lam\nabla_\nu\del g_{\rho\mu}-2g^{\mu\nu}g^{\lam\rho}\nabla_\lam\nabla_\rho\del g_{\nu\mu}\right)\nonum\\
&=&\frac{1}{2}\left(2g^{\lam\rho}g^{\nu\mu}\nabla_\lam\nabla_\nu\del g_{\rho\mu}-2g^{\mu\nu}g^{\lam\rho}\nabla_\lam\nabla_\rho\del g_{\nu\mu}\right)\nonum\\
&=&g^{\lam\rho}g^{\mu\nu}\nabla_\lam\nabla_\mu\del g_{\rho\nu}-g^{\mu\nu}\underbrace{g^{\lam\rho}\nabla_\lam\nabla_\rho}_{\square}\del g_{\mu\nu}=\nabla^\rho\nabla^\nu \del g_{\rho\nu}-g^{\mu\nu}\square \del g_{\mu\nu}$
Using Eq.\eqref{C.11} into Eq.\eqref{C.25} we obtain
\eqarr$\label{C.26}
g^{\mu\nu}\del R_{\mu\nu}&=&g^{\lam\rho}g^{\mu\nu}\nabla_\lam\nabla_\mu(-g_{\rho\alp}g_{\nu\beta}\del g^{\alp\beta})-g^{\mu\nu}g^{\lam\rho}\nabla_\lam\nabla_\rho(-g_{\mu\alp}g_{\nu\beta}\del g^{\alp\beta})\nonum\\
&=&-g^{\lam\rho}g^{\mu\nu}g_{\rho\alp}g_{\nu\beta}\nabla_\lam\nabla_\mu\del g^{\alp\beta}+g^{\mu\nu}g^{\lam\rho}g_{\mu\alp}g_{\nu\beta}\nabla_\lam\nabla_\rho\del g^{\alp\beta}\nonum\\
&=&-\del^{\lam}{}_{\alp}\del^{\mu}{}_{\beta}\nabla_\lam\nabla_\mu\del g^{\alp\beta}+g^{\lam\rho}\del^{\nu}{}_{\alp}g_{\nu\beta}\nabla_\lam\nabla_\rho\del g^{\alp\beta}\nonum\\
&=&-\nabla_\alp\nabla_\beta\del g^{\alp\beta}+g_{\alp\beta}\underbrace{g^{\lam\rho}\nabla_\lam\nabla_\rho}_{\square}\del g^{\alp\beta}\nonum\\
&=&-\nabla_\alp\nabla_\beta\del g^{\alp\beta}+g_{\alp\beta}\square\del g^{\alp\beta}$
Putting together Eq.\eqref{C.25} and Eq.\eqref{C.26} we have
\eq$\label{C.27}
\boxed{g^{\mu\nu}\del R_{\mu\nu}=\nabla^\mu\nabla^\nu \del g_{\mu\nu}-g^{\mu\nu}\square \del g_{\mu\nu}=-\nabla_\mu\nabla_\nu\del g^{\mu\nu}+g_{\mu\nu}\square\del g^{\mu\nu}}$
From Eq.\eqref{C.19} and Eq.\eqref{C.27} we conclude that
\begin{equation}\label{C.28}
\boxed{\delta R=\delta (g^{\mu \nu})R_{\mu \nu}-\nabla_\mu\nabla_\nu\del g^{\mu\nu}+g_{\mu\nu}\square\del g^{\mu\nu}}
\end{equation}
\item $\del(\sqrt{-g}\lagr_m)$:\\ \\
The energy-momentum tensor is defined by
\eq$\label{C.29}
T^{(m)}_{\mu\nu}\equiv\frac{-2}{\sqrt{-g}}\frac{\del(\sqrt{-g}\lagr_m)}{\del g^{\mu\nu}}$
Thus, we have
\eq$\label{C.30}
\boxed{\del(\sqrt{-g}\lagr_m)=-\frac{1}{2}\ T^{(m)}_{\mu\nu}\sqrt{-g}\ \del g^{\mu\nu}}$
\end{itemize}\vspace{1em}
\par Replacing the right hand sides from equations \eqref{C.13}, \eqref{C.28}, \eqref{C.30} into Eq.\eqref{C.2}. We get\vspace{1em}
\begin{align}
0=&\ -\frac{1}{2}\int d^4x \sqrt{-g}\ g_{\mu\nu}\ \del g^{\mu\nu}\left[\frac{f(\Phi)}{2\kappa} R-\frac{1}{2}\pa_\lam \Phi\pa^\lam\Phi-V(\Phi)-\Lambda_B\right]\nonum\\
&\ +\frac{1}{2}\int d^4x\sqrt{-g}\ \left\{\frac{1}{\kappa}\left[\del g^{\mu\nu} f(\Phi)R_{\mu\nu}-f(\Phi)\nabla_\mu \nabla_\nu\ \del g^{\mu\nu}+g_{\mu\nu}f(\Phi)\square\ \del g^{\mu\nu}\right]-\del g^{\mu\nu} \pa_\mu \Phi\pa_\nu\Phi\right\}\nonum\\
&\ -\frac{1}{2}\int d^4x\ T^{(m)}_{\mu\nu}\sqrt{-g}\ \del g^{\mu\nu}\Ra\nonum
\end{align}
\begin{align}\label{C.31}
0=& \int d^4x \sqrt{-g}\ \del g^{\mu\nu}\left\{g_{\mu\nu}\left[-\frac{f(\Phi)}{2\kappa} R+\frac{1}{2}\pa_\lam \Phi\pa^\lam\Phi+V(\Phi)+\Lambda_B\right]+\frac{f(\Phi)}{\kappa}R_{\mu\nu}-\pa_\mu\Phi\pa_\nu\Phi-T^{(m)}_{\mu\nu}\right\}\nonum\\
&\ +\frac{1}{\kappa}\int d^4x\sqrt{-g}\ \left(-f(\Phi)\nabla_\mu \nabla_\nu\ \del g^{\mu\nu}+g_{\mu\nu}f(\Phi)\square\ \del g^{\mu\nu}\right)
\end{align}
\par The integral in the second line of Eq.\eqref{C.31} can be modified to a more useful one, but first it is necessary to prove one general property that will help us doing the modification. \vspace{1em}
\begin{proof}[\textbf{Proof:}]
\eq$\nabla_\mu A^\mu=\pa_\mu A^\mu+\Gam^\mu{}_{\lam\mu}A^\lam
\label{C.32}$\\
\begin{align}
\Gam^\mu{}_{\lam\mu}&=\frac{1}{2}g^{\mu\rho}(g_{\lam\rho,\mu}+g_{\mu\rho,\lam}-g_{\lam\mu,\rho})=\frac{1}{2}g^{\mu\rho}g_{\lam\rho,\mu}+\frac{1}{2}g^{\mu\rho}g_{\mu\rho,\lam}-\frac{1}{2}g^{\mu\rho}g_{\lam\mu,\rho}\nonum\\
&=\frac{1}{2}g^{\mu\rho}g_{\lam\rho,\mu}+\frac{1}{2}g^{\mu\rho}g_{\mu\rho,\lam}-\frac{1}{2}g^{\rho\mu}g_{\lam\rho,\mu}=\frac{1}{2}g^{\mu\rho}g_{\mu\rho,\lam}
\label{C.33}
\end{align}\\
\begin{align}
\frac{1}{\sqrt{-g}}\frac{\pa\sqrt{-g}}{\pa x^\lam}&=\frac{1}{2g}\frac{\pa g}{\pa x^\lam}=\frac{1}{2}g^{-1}\frac{\pa g}{\pa x^\lam}=\frac{1}{2}\varepsilon_{\mu_0\ldots\mu_3}\ g^{0\mu_0}\cdots g^{3\mu_3}\frac{\pa}{\pa x^\lam}(\varepsilon^{\nu_0\ldots\nu_3}\ g_{0\nu_0}\cdots g_{3\nu_3})\nonum\\
&=\frac{1}{2}\varepsilon_{\mu_0\ldots\mu_3}\varepsilon^{\nu_0\ldots\nu_3}\ g^{0\mu_0}\cdots g^{3\mu_3}(g_{0\nu_0,\lam}\cdots g_{3\nu_3}+\cdots+g_{0\nu_0}\cdots g_{3\nu_3,\lam})\nonum\\
&=\frac{1}{2}\left|\begin{array}{cccc}
\del_{\mu_0}{}^{\nu_0} & \del_{\mu_0}{}^{\nu_1} & \del_{\mu_0}{}^{\nu_2} & \del_{\mu_0}{}^{\nu_3}\\
\del_{\mu_1}{}^{\nu_0} & \del_{\mu_1}{}^{\nu_1} & \del_{\mu_1}{}^{\nu_2} & \del_{\mu_1}{}^{\nu_3}\\
\del_{\mu_2}{}^{\nu_0} & \del_{\mu_2}{}^{\nu_1} & \del_{\mu_2}{}^{\nu_2} & \del_{\mu_2}{}^{\nu_3}\\
\del_{\mu_3}{}^{\nu_0} & \del_{\mu_3}{}^{\nu_1} & \del_{\mu_3}{}^{\nu_2} & \del_{\mu_3}{}^{\nu_3}
\end{array}\right|g^{0\mu_0}\cdots g^{3\mu_3}(g_{0\nu_0,\lam}\cdots g_{3\nu_3}+\cdots+g_{0\nu_0}\cdots g_{3\nu_3,\lam})\nonum\\
&=\frac{1}{2}\del_{\mu_0}{}^{\nu_0}\ \del_{\mu_1}{}^{\nu_1}\ \del_{\mu_2}{}^{\nu_2}\ \del_{\mu_3}{}^{\nu_3}\ g^{0\mu_0}\cdots g^{3\mu_3}(g_{0\nu_0,\lam}\cdots g_{3\nu_3}+\cdots+g_{0\nu_0}\cdots g_{3\nu_3,\lam})\nonum\\
&=\frac{1}{2}g^{0\nu_0}\ g^{1\nu_1}\ g^{2\nu_2}\ g^{3\nu_3}(g_{0\nu_0,\lam}\cdots g_{3\nu_3}+\cdots+g_{0\nu_0}\cdots g_{3\nu_3,\lam})\nonum\\
&=\frac{1}{2}(g^{0\nu_0}g_{0\nu_0,\lam}+g^{1\nu_1}g_{1\nu_1,\lam}+g^{2\nu_2}g_{2\nu_2,\lam}+g^{3\nu_3}g_{3\nu_3,\lam})\nonum\\
&=\frac{1}{2}(g^{0\rho}g_{0\rho,\lam}+g^{1\rho}g_{1\rho,\lam}+g^{2\rho}g_{2\rho,\lam}+g^{3\rho}g_{3\rho,\lam})=\frac{1}{2}g^{\mu\rho}g_{\mu\rho,\lam}
\label{C.34}
\end{align}\vspace{0.5em}
\par Combining equations \eqref{C.32}, \eqref{C.33}, \eqref{C.34} we get\vspace{1em}
\eq$\label{C.35}
\boxed{\nabla_\mu A^\mu=\pa_\mu A^\mu+\frac{1}{\sqrt{-g}}\frac{\pa\sqrt{-g}}{\pa x^\lam}A^\lam=\frac{1}{\sqrt{-g}}\pa_\mu(\sqrt{-g}\ A^\mu)}$
\end{proof}
\par Let us now modify the integral of the last term of Eq.\eqref{C.31} by using Eq.\eqref{C.24}, Eq.\eqref{C.35} and the fact that the variation of the metric $\del g^{\mu\nu}$ vanishes at infinity.
\begin{gather}
\int d^4x\sqrt{-g}\ \left(g_{\mu\nu}f(\Phi)\square\ \del g^{\mu\nu}-f(\Phi)\nabla_\mu \nabla_\nu\ \del g^{\mu\nu}\right)=\nonum\\
\int d^4x\sqrt{-g} \left(g_{\mu\nu} f g^{\rho\lam}\nabla_\rho\nabla_\lam \del g^{\mu\nu}-f \nabla_\mu \nabla_\nu \del g^{\mu\nu}\right)=\nonum\\
\int d^4x\sqrt{-g}\left[f \nabla_\rho\nabla_\lam(g_{\mu\nu} g^{\rho\lam} \del g^{\mu\nu})-\nabla_\mu(f \nabla_\nu\ \del g^{\mu\nu})+(\nabla_\mu f)\nabla_\nu \del g^{\mu\nu}\right]=\nonum\\
\int d^4x\sqrt{-g}\left[f \nabla_\rho\nabla_\lam(g_{\mu\nu} g^{\rho\lam} \del g^{\mu\nu})+(\nabla_\mu f)\nabla_\nu \del g^{\mu\nu}\right]-\underbrace{\int d^4x\ \pa_\mu(\sqrt{-g}\ f \nabla_\nu \del g^{\mu\nu})}_{0}=\nonum\\
\int d^4x\sqrt{-g}\left\{\nabla_\rho\left[f \nabla_\lam(g_{\mu\nu} g^{\rho\lam} \del g^{\mu\nu})\right]-(\nabla_\rho f)\nabla_\lam(g_{\mu\nu} g^{\rho\lam} \del g^{\mu\nu})+\nabla_\nu(\del g^{\mu\nu} \nabla_\mu f)-\del g^{\mu\nu} \nabla_\mu\nabla_\nu f \right\}=\nonum\\
\int d^4x\sqrt{-g}\left[-(\nabla_\rho f)\nabla_\lam(g_{\mu\nu} g^{\rho\lam} \del g^{\mu\nu})-\del g^{\mu\nu} \nabla_\mu\nabla_\nu f \right]+\underbrace{\scriptstyle{\int d^4x\ \pa_\rho\left[\sqrt{-g}\ f \nabla_\lam(g_{\mu\nu} g^{\rho\lam} \del g^{\mu\nu})\right]}}_{0}+\underbrace{\scriptstyle{\int d^4x\ \pa_\nu(\sqrt{-g}\ \del g^{\mu\nu}\nabla_\mu f)}}_0=\nonum\\
\int d^4x\sqrt{-g}\left[-\nabla_\lam(g_{\mu\nu} g^{\rho\lam} \del g^{\mu\nu}\nabla_\rho f)+g_{\mu\nu} g^{\rho\lam} \del g^{\mu\nu}\nabla_\lam\nabla_\lam f-\del g^{\mu\nu}\nabla_\mu\nabla_\nu f\right]=\nonum\\
\int d^4x\sqrt{-g}\ \del g^{\mu\nu}\left(g_{\mu\nu}\square f-\nabla_\mu\nabla_\nu f\right)-\underbrace{\int d^4x\ \pa_\lam(\sqrt{-g}\ g_{\mu\nu} g^{\rho\lam} \del g^{\mu\nu}\nabla_\rho f)}_0\Ra\nonum
\end{gather}
\eq$\label{C.36}
\int d^4x\sqrt{-g}\left(g_{\mu\nu}f(\Phi)\square\ \del g^{\mu\nu}-f(\Phi)\nabla_\mu \nabla_\nu\del g^{\mu\nu}\right)=\int d^4x\sqrt{-g}\ \del g^{\mu\nu}\left[g_{\mu\nu}\square f(\Phi)-\nabla_\mu\nabla_\nu f(\Phi)\right]$\vspace{0.5em}
\par The combination of equations \eqref{C.31} and \eqref{C.36} leads to\vspace{1em}
\begin{align}
0&=\int d^4x \sqrt{-g}\ \del g^{\mu\nu}\left\{g_{\mu\nu}\left[-\frac{f(\Phi)}{2\kappa} R+\frac{1}{2}\pa_\lam \Phi\pa^\lam\Phi+V(\Phi)+\Lambda_B\right]\right.\nonum\\
&\hspace{10em}\left.+\frac{1}{\kappa}\left[f(\Phi)R_{\mu\nu}-\nabla_\mu \nabla_\nu f(\Phi)+g_{\mu\nu}\square f(\Phi)\right]-\pa_\mu\Phi\pa_\nu\Phi-T^{(m)}_{\mu\nu}\right\}\Ra\nonum\\
0&=-g_{\mu\nu}\left[\frac{f(\Phi)}{2\kappa} R-\frac{1}{2}\pa_\lam \Phi\pa^\lam\Phi-V(\Phi)-\Lambda_B\right]+\frac{1}{\kappa}\left[f(\Phi)R_{\mu\nu}-\nabla_\mu \nabla_\nu f(\Phi)+g_{\mu\nu}\square f(\Phi)\right]\nonum\\
&\hspace{18.6em}-\pa_\mu \Phi\pa_\nu\Phi-T^{(m)}_{\mu\nu}\Ra\nonum
\end{align}
\eq$\label{C.37}
\boxed{\begin{gathered}(T^{(m)}_{\mu\nu}-g_{\mu\nu}\Lambda_B)=\frac{f(\Phi)}{\kappa}\left[R_{\mu\nu}-\frac{1}{2}g_{\mu\nu}R\right]+g_{\mu\nu}\left[\frac{\pa_\lam \Phi\pa^\lam\Phi}{2}+V(\Phi)\right]\\
\hspace{5.5em}+\frac{1}{\kappa}\left[-\nabla_\mu \nabla_\nu f(\Phi)+g_{\mu\nu}\square f(\Phi)\right]-\pa_\mu \Phi\pa_\nu\Phi\end{gathered}}$\\
We can define now the following tensor
$$-T^{(\Phi)}_{\mu\nu}=g_{\mu\nu}\left[\frac{\pa_\lam \Phi\pa^\lam\Phi}{2}+V(\Phi)\right]+\frac{1}{\kappa}\left[-\nabla_\mu \nabla_\nu f(\Phi)+g_{\mu\nu}\square f(\Phi)\right]-\pa_\mu \Phi\pa_\nu\Phi\Ra$$
\eq$\label{C.38}
\boxed{T^{(\Phi)}_{\mu\nu}=\pa_\mu \Phi\pa_\nu\Phi-g_{\mu\nu}\left[\frac{\pa_\lam \Phi\pa^\lam\Phi}{2}+V(\Phi)\right]+\frac{1}{\kappa}\left[\nabla_\mu \nabla_\nu f(\Phi)-g_{\mu\nu}\square f(\Phi)\right]}$
\par It is known that the covariant derivative of a scalar $\nabla_\mu \Phi$ equals to the simple derivative $\pa_\mu\Phi$. Thus, in the above equation $\square\equiv g_{\alp\beta}\nabla^\alp\nabla^\beta=\nabla_\alp\nabla^\alp=\nabla^2$ and $\pa_\lam\Phi\pa^\lam\Phi=\nabla_\lam\Phi\nabla^\lam\Phi=(\nabla\Phi)^2$. Hence, we can write Eq.\eqref{C.38} as
\eq$\label{C.39}
\boxed{T^{(\Phi)}_{\mu\nu}=\nabla_\mu \Phi\nabla_\nu\Phi-g_{\mu\nu}\left[\frac{(\nabla\Phi)^2}{2}+V(\Phi)\right]+\frac{1}{\kappa}\left[\nabla_\mu \nabla_\nu f(\Phi)-g_{\mu\nu}\nabla^2 f(\Phi)\right]}$
\par The combination of equations \eqref{C.37} and \eqref{C.38} yields to
$$(T^{(m)}_{\mu\nu}-g_{\mu\nu}\Lambda_B)=\frac{f(\Phi)}{\kappa}\left[R_{\mu\nu}-\frac{1}{2}g_{\mu\nu}R\right]-T^{(\Phi)}_{\mu\nu}\Ra$$
\eq$\label{C.40}
\boxed{\kappa\left(T^{(m)}_{\mu\nu}+T^{(\Phi)}_{\mu\nu}-g_{\mu\nu}\Lambda_B\right)=f(\Phi)\left(R_{\mu\nu}-\frac{1}{2}g_{\mu\nu}R\right)}$\vspace{1em}
\par Finally, we present the above equations in the case of (4+n)-dimensional spacetime. As we already mentioned the only difference is that Greek indices become capital Latin indices. Moreover, we should have in mind that in this case $\kappa\ra\kappa_{(4+n)}$ and $\Lambda_B$ is the higher dimensional cosmological constant. Therefore, it is
\begin{empheq}[box=\mymath]{equation}\label{C.41}
T_{MN}^{(\Phi)}=\nabla_M \Phi\nabla_N\Phi-g_{MN}\left[\frac{(\nabla\Phi)^2}{2}+V(\Phi)\right]+\frac{1}{\kappa_{(4+n)}}\left[\nabla_M \nabla_N f(\Phi)-g_{MN}\nabla^2 f(\Phi)\right]
\end{empheq}
where now it is\vspace{0.5em}
\begin{empheq}[box=\mymath]{equation}\label{C.42}
(\nabla\Phi)^2=\nabla_K\Phi\nabla^K\Phi,\hspace{2em} \square=\nabla^2=\nabla_K\nabla^K
\end{empheq}\\
The field equations are
\begin{empheq}[box=\mymath]{equation}\label{C.43}
\kappa_{(4+n)}\left(T^{(m)}_{MN}+T^{(\Phi)}_{MN}-g_{MN}\Lambda_B\right)=f(\Phi)\left(R_{MN}-\frac{1}{2}g_{MN}R\right)
\end{empheq}

\section{Variation with respect to the scalar field}
\par By varying the action of equation \eqref{C.1} with respect to the scalar field $\Phi$ and $\pa_\mu\Phi$, we obtain a new equation that relates functions $f(\Phi)$ and $V(\Phi)$ to the field $\Phi$ and its derivatives. The procedure that one should follow is depicted below. Firstly, we write the action \eqref{C.1} in the following form in order to be more convenient for evaluation.
\eq$\label{C.44}
S=\int d^4x\sqrt{-g}\ \lagr_{tot}$
where
\eq$\label{C.45}
\lagr_{tot}=\frac{f(\Phi)}{2\kappa}R-\frac{1}{2}\pa_\mu\Phi\pa^\mu\Phi-V(\Phi)-\Lambda_B+\lagr_m$
The variation of the action \eqref{C.44} with respect to $\Phi$ and $\pa_\mu\Phi$ will provide us the Euler-Lagrange equation.
$$\del S=0=\int d^4x\ \del(\sqrt{-g}\ \lagr_{tot})=\int d^4x\left[\frac{\pa(\sqrt{-g}\ \lagr_{tot})}{\pa\Phi}\del\Phi+\frac{\pa(\sqrt{-g}\ \lagr_{tot})}{\pa(\pa_\mu \Phi)}\del(\pa_\mu \Phi)\right]\xRightarrow{\del(\pa_\mu \Phi)=\pa_\mu(\del\Phi)}$$
\begin{align}\label{C.46}
0&=\int d^4x \left[\frac{\pa(\sqrt{-g}\ \lagr_{tot})}{\pa \Phi}\del\Phi+\frac{\pa(\sqrt{-g}\ \lagr_{tot})}{\pa(\pa_\mu \Phi)}\pa_\mu(\del\Phi)\right]\nonum\\ \nonum\\
&=\int d^4x \left[\frac{\pa(\sqrt{-g}\ \lagr_{tot})}{\pa \Phi}\del\Phi+\pa_\mu\left(\frac{\pa(\sqrt{-g}\ \lagr_{tot})}{\pa(\pa_\mu \Phi)}\del\Phi\right)-\pa_\mu\left(\frac{\pa(\sqrt{-g}\ \lagr_{tot})}{\pa(\pa_\mu \Phi)}\right)\del\Phi\right]
\end{align}\\
However
\eq$\label{C.47}
\int d^4x\ \pa_\mu\left(\frac{\pa(\sqrt{-g}\ \lagr_{tot})}{\pa(\pa_\mu \Phi)}\del\Phi\right)=0$\\
because at the limits of the integration $\del\Phi=0$. Thus, combining equations \eqref{C.46} and \eqref{C.47} we get
$$0=\int d^4x \left[\frac{\pa(\sqrt{-g}\ \lagr_{tot})}{\pa \Phi}-\pa_\mu\left(\frac{\pa(\sqrt{-g}\ \lagr_{tot})}{\pa(\pa_\mu \Phi)}\right)\right]\del\Phi\Ra$$
\eq$\label{C.48}
\boxed{\frac{\pa(\sqrt{-g}\ \lagr_{tot})}{\pa \Phi}=\pa_\mu\left[\frac{\pa(\sqrt{-g}\ \lagr_{tot})}{\pa(\pa_\mu \Phi)}\right]}$\\
Eq.\eqref{C.48} constitutes the \emph{Euler-Lagrange equation}. Substituting now the quantity $\lagr_{tot}$ from Eq.\eqref{C.45} into Eq.\eqref{C.48} we obtain
$$\sqrt{-g}\ \frac{\pa}{\pa\Phi}\left[\frac{f(\Phi)}{2\kappa}R-V(\Phi)\right]=\pa_\mu\left[\sqrt{-g}\ \frac{\pa}{\pa(\pa_\mu\Phi)}\left(-\frac{1}{2}\pa_\nu\Phi\pa^\nu\Phi\right)\right]\Ra$$
$$\sqrt{-g}\left(\frac{1}{2\kappa}\frac{df}{d\Phi}R-\frac{dV}{d\Phi}\right)=-\frac{1}{2}\pa_\mu\left[\sqrt{-g}\ \frac{\pa}{\pa(\pa_\mu\Phi)}(g^{\rho\nu}\pa_\nu\Phi\pa_\rho\Phi)\right]\Ra$$
$$\sqrt{-g}\left(\frac{1}{2\kappa}\frac{df}{d\Phi}R-\frac{dV}{d\Phi}\right)=-\frac{1}{2}\pa_\mu\left[\sqrt{-g}\left( g^{\rho\nu}\del^{\mu}{}_\nu\pa_\rho\Phi+g^{\rho\nu}\pa_\nu\Phi\del^\mu{}_\rho\right)\right]\Ra$$
$$\sqrt{-g}\left(\frac{1}{2\kappa}\frac{df}{d\Phi}R-\frac{dV}{d\Phi}\right)=-\frac{1}{2}\pa_\mu\left[\sqrt{-g}\underbrace{\left( g^{\rho\mu}\pa_\rho\Phi+g^{\mu\nu}\pa_\nu\Phi\right)}_{2g^{\mu\rho}\pa_\rho\Phi}\right]\Ra$$
\begin{empheq}[box=\mymath]{equation}\label{C.49}
\sqrt{-g}\left(\frac{1}{2\kappa}\frac{df}{d\Phi}R-\frac{dV}{d\Phi}\right)=-\pa_\mu\left(\sqrt{-g}\ g^{\mu\rho}\pa_\rho\Phi\right)
\end{empheq}

%% file: Appendices/Appendix_ener_mom.tex
\chapter{Energy-Momentum Tensor's Components}
\label{enmom}

\par Firstly, the non-diagonal and non-zero components of energy-momentum tensor $T^M{}_N$ are going to be calculated. Subsequently, the diagonal components will be calculated and finally the zero-valued  components will be presented.
\begin{center}
\larger{\underline{$\mathbf{T^0{}_1:}$}}
\end{center}
\begin{align}\label{loc.16}
\eqref{loc.15}\xRightarrow[N=1]{M=0} T^0{}_1&=\pa^0\phi\pa_1\phi+\pa^0\chi\pa_1\chi+\nabla^0\nabla_1f=g^{0K}\pa_K\phi\pa_1\phi+g^{0K}\pa_K\chi\pa_1\chi+g^{0K}\nabla_K\nabla_1f\nonum\\
&=g^{01}(\pa_1\phi)^2+g^{01}(\pa_1\chi)^2+g^{01}\nabla_1^2f=\underbrace{g^{01}}_{e^{-2A}}\left[(\pa_1\phi)^2+(\pa_1\chi)^2+\nabla_1^2f\right]\nonum\\
&=e^{-2A}\left[(\pa_1\phi)^2+(\pa_1\chi)^2+\nabla_1^2f\right]
\end{align}\\
\eq$\label{loc.17}
\nabla_1^2f=\nabla_1(\nabla_1f)=\nabla_1(\pa_1f)=\pa_1(\pa_1f)-\underbrace{\Gam^L{}_{11}}_{0}\pa_Lf=\pa_1^2f$
\begin{align}\label{loc.18}
\pa_1^2f&=\pa_1(\pa_\phi f\pa_1\phi+\pa_\chi f\pa_1\chi)\nonum\\
&=\pa_\phi^2f(\pa_1\phi)^2+\pa_\chi\pa_\phi f\pa_1\chi\pa_1\phi+\pa_\phi f\pa_1^2\phi+\pa_\chi^2f(\pa_1\chi)^2+\pa_\phi\pa_\chi f\pa_1\phi\pa_1\chi+\pa_\chi f\pa_1^2\chi\nonum\\
&=\pa_\phi^2f(\pa_1\phi)^2+2\pa_\chi\pa_\phi f\pa_1\chi\pa_1\phi+\pa_\phi f\pa_1^2\phi+\pa_\chi^2f(\pa_1\chi)^2+\pa_\chi f\pa_1^2\chi
\end{align}
Combining equations \eqref{loc.16}-\eqref{loc.18}, we obtain
\eq$\label{loc.19}
T^0{}_1=e^{-2A}[(\pa_1\phi)^2+(\pa_1\chi)^2+\pa_1^2f]$
or
\eq$\label{loc.20}
T^0{}_1=e^{-2A}[(1+\pa_\phi^2f)(\pa_1\phi)^2+(1+\pa_\chi^2f)(\pa_1\chi)^2+2\pa_\chi\pa_\phi f\pa_1\chi\pa_1\phi+\pa_\phi f\pa_1^2\phi+\pa_\chi f\pa_1^2\chi]$
\begin{center}\vspace{1em}
\larger{\underline{$\mathbf{T^1{}_0:}$}}
\end{center}
\begin{align}
\eqref{loc.15}\xRightarrow[N=0]{M=1} T^1{}_0&=\pa^1\phi\pa_0\phi+\pa^1\chi\pa_0\chi+\nabla^1\nabla_0f=g^{1K}\pa_K\phi\pa_0\phi+g^{1K}\pa_K\chi\pa_0\chi+g^{1K}\nabla_K\nabla_0f\nonum\\
&=g^{10}(\pa_0\phi)^2+g^{11}\pa_1\phi\pa_0\phi+g^{10}(\pa_0\chi)^2+g^{11}\pa_1\chi\pa_0\chi+g^{10}\nabla_0^2f+g^{11}\nabla_1\nabla_0f\nonum\\
&=\underbrace{g^{10}}_{e^{-2A}}[(\pa_0\phi)^2+(\pa_0\chi)^2+\nabla_0^2f]+\underbrace{g^{11}}_{e^{-2A}\left(1-\frac{2m}{r}\right)}[\pa_1\phi\pa_0\phi+\pa_1\chi\pa_0\chi+\nabla_1\nabla_0f]\Ra\nonum
\end{align}
\eq$\label{loc.21}
T^1{}_0=e^{-2A}\left\{[(\pa_0\phi)^2+(\pa_0\chi)^2+\nabla_0^2f]+\left(1-\frac{2m}{r}\right)[\pa_1\phi\pa_0\phi+\pa_1\chi\pa_0\chi+\nabla_1\nabla_0f]\right\}$
\begin{align}\label{loc.22}
\nabla_0\nabla_0f&=\pa_0^2f-\Gam^L{}_{00}\pa_Lf=\pa_0^2f-\Gam^0{}_{00}\pa_0f-\Gam^1{}_{00}\pa_1f-\Gam^4{}_{00}\pa_4f\nonum\\
&=\scriptstyle{\pa_0^2f-\left(\frac{m}{r^2}-\frac{\pa_1m}{r}\right)\pa_0f-\left(\frac{\pa_0m}{r}-\frac{\pa_1m}{r}+\frac{m}{r^2}+\frac{2m\pa_1m}{r^2}-\frac{2m^2}{r^3}\right)\pa_1f-e^{2A}\left(A'-\frac{2mA'}{r}-\frac{\pa_4m}{r}\right)\pa_4f}
\end{align}
\begin{align}\label{loc.23}
\nabla_1\nabla_0f&=\pa_1\pa_0f-\Gam^L{}_{10}\pa_Lf=\pa_1\pa_0f-\Gam^1{}_{10}\pa_1f-\Gam^4{}_{10}\pa_4f\nonum\\
&=\pa_1\pa_0f-\left(\frac{\pa_1m}{r}-\frac{m}{r^2}\right)\pa_1f+e^{2A}A'\pa_4f	
\end{align}
\begin{align}\label{loc.24}
\pa_0^2f&=\pa_0(\pa_\phi f\pa_0\phi+\pa_\chi f\pa_0\chi)\nonum\\
&=\pa_\phi^2f(\pa_0\phi)^2+\pa_\chi\pa_\phi f\pa_0\chi\pa_0\phi+\pa_\phi f\pa_0^2\phi+\pa_\chi^2f(\pa_0\chi)^2+\pa_\phi\pa_\chi f\pa_0\phi\pa_0\chi+\pa_\chi f\pa_0^2\chi\nonum\\
&=\pa_\phi^2f(\pa_0\phi)^2+2\pa_\chi\pa_\phi f\pa_0\chi\pa_0\phi+\pa_\chi^2f(\pa_0\chi)^2+\pa_\phi f\pa_0^2\phi+\pa_\chi f\pa_0^2\chi
\end{align}
\begin{align}\label{loc.25}
\pa_1\pa_0f&=\pa_1(\pa_\phi f\pa_0\phi+\pa_\chi f\pa_0\chi)\nonum\\
&=\pa_\phi^2f\pa_1\phi\pa_0\phi+\pa_\chi\pa_\phi f\pa_1\chi\pa_0\phi+\pa_\phi f\pa_1\pa_0\phi+\pa_\chi^2f\pa_1\chi\pa_0\chi+\pa_\phi\pa_\chi f\pa_1\phi\pa_0\chi+\pa_\chi f\pa_1\pa_0\chi\nonum\\
&=\pa_\phi^2f\pa_1\phi\pa_0\phi+\pa_\chi^2f\pa_1\chi\pa_0\chi+\pa_\chi\pa_\phi f(\pa_1\chi\pa_0\phi+\pa_1\phi\pa_0\chi)+\pa_\phi f\pa_1\pa_0\phi+\pa_\chi f\pa_1\pa_0\chi
\end{align}
From equations \eqref{loc.21}-\eqref{loc.25} we have
\begin{align}\label{loc.26}
T^1{}_0&=e^{-2A}\left[(\pa_0\phi)^2+(\pa_0\chi)^2+\pa_0^2f+\left(1-\frac{2m}{r}\right)(\pa_1\pa_0\phi+\pa_1\chi\pa_0\chi+\pa_1\pa_0f)\right.\nonum\\
&\hspace{5em}\left. +\pa_1\left(\frac{m}{r}\right)\pa_0f-\pa_0\left(\frac{m}{r}\right)\pa_1f+e^{2A}\pa_4\left(\frac{m}{r}\right)\pa_4f\right]
\end{align}
or
\begin{align}\label{loc.27}
T^1{}_0&=e^{-2A}\left\{(1+\pa_\phi^2f)(\pa_0\phi)^2+(1+\pa_\chi^2f)(\pa_0\chi)^2+2\pa_\chi\pa_\phi f\pa_0\phi\pa_0\chi+\pa_\phi f\pa_0^2\phi+\pa_\chi f\pa_0^2\chi \right.\nonum\\
&\hspace{4em}\scriptstyle{+\left(1-\frac{2m}{r}\right)[(1+\pa_\phi^2f)\pa_1\phi\pa_0\phi+(1+\pa_\chi^2f)\pa_1\chi\pa_0\chi+\pa_\phi\pa_\chi f(\pa_1\phi\pa_0\chi+\pa_1\chi\pa_0\phi)+\pa_\phi f\pa_1\pa_0\phi+\pa_\chi f\pa_1\pa_0\chi]}\nonum\\
&\hspace{4em}\left. \scriptstyle{+\left(\frac{\pa_1m}{r}-\frac{m}{r^2}\right)(\pa_\phi f\pa_0\phi+\pa_\chi f\pa_0\chi)-\frac{\pa_0m}{r}(\pa_\phi f\pa_1\phi+\pa_\chi f\pa_1\chi)+e^{2A}\frac{\pa_4m}{r}(\pa_\phi f\pa_4\phi+\pa_\chi f\pa_4\chi)}\right\}
\end{align}
\thispagestyle{plain}
\begin{center}\vspace{1em}
\larger{\underline{$\mathbf{T^4{}_0:}$}}
\end{center}
\begin{align}\label{loc.28}
\eqref{loc.15}\xRightarrow[N=0]{M=4}T^4{}_0&=\pa^4\phi\pa_0\phi+\pa^4\chi\pa_0\chi+\nabla^4\nabla_0f=g^{4K}(\pa_K\phi\pa_0\phi+\pa_K\chi\pa_0\chi+\nabla_K\nabla_0f)\nonum\\
&=\underbrace{g^{44}}_{1}(\pa_4\phi\pa_0\phi+\pa_4\chi\pa_0\chi+\nabla_4\nabla_0f)=\pa_4\phi\pa_0\phi+\pa_4\chi\pa_0\chi+\nabla_4\nabla_0f
\end{align}
\begin{align}\label{loc.29}
\nabla_4\nabla_0f&=\pa_4\pa_0f-\Gam^L{}_{40}\pa_Lf=\pa_4\pa_0f-\Gam^0{}_{40}\pa_0f-\Gam^1{}_{40}\pa_1f\nonum\\
&=\pa_4\pa_0f-A'\pa_0f-\frac{\pa_4m}{r}\pa_1f
\end{align}
\begin{align}\label{loc.30}
\pa_4\pa_0f&=\pa_4(\pa_\phi f\pa_0\phi+\pa_\chi f\pa_0\chi)\nonum\\
&=\pa_\phi^2f\pa_4\phi\pa_0\phi+\pa_\chi\pa_\phi f\pa_4\chi\pa_0\phi+\pa_\phi f\pa_4\pa_0\phi+\pa_\chi^2f\pa_4\chi\pa_0\chi+\pa_\phi\pa_\chi f\pa_4\phi\pa_0\chi+\pa_\chi f\pa_4\pa_0\chi\nonum\\
&=\pa_\phi^2f\pa_4\phi\pa_0\phi+\pa_\chi^2f\pa_4\chi\pa_0\chi+\pa_\phi\pa_\chi f(\pa_4\chi\pa_0\phi+\pa_4\phi\pa_0\chi)+\pa_\phi f\pa_4\pa_0\phi+\pa_\chi f\pa_4\pa_0\chi
\end{align}
Equations \eqref{loc.28}-\eqref{loc.30} yield to
\eq$\label{loc.31}
T^4{}_0=\pa_4\phi\pa_0\phi+\pa_4\chi\pa_0\chi+\pa_4\pa_0f-A'\pa_0f-\frac{\pa_4m}{r}\pa_1f$
or
\begin{align}\label{loc.32}
T^4{}_0=(1+\pa_\phi^2f)\pa_4\phi\pa_0\phi&+(1+\pa_\chi^2f)\pa_4\chi\pa_0\chi+\pa_\chi\pa_\phi f(\pa_4\chi\pa_0\phi+\pa_4\phi\pa_0\chi)\nonum\\
&+\pa_\phi f\pa_4\pa_0\phi+\pa_\chi f\pa_4\pa_0\chi-A'\pa_0f-\frac{\pa_4m}{r}\pa_1f
\end{align}
\thispagestyle{plain}
\begin{center}
\larger{\underline{$\mathbf{T^0{}_4:}$}}
\end{center}
\begin{align}\label{loc.33}
\eqref{loc.15}\xRightarrow[N=4]{M=0}T^0{}_4&=\pa^0\phi\pa_4\phi+\pa^0\chi\pa_4\chi+\nabla^0\nabla_4f=g^{0K}(\pa_K\phi\pa_4\phi+\pa_K\chi\pa_4\chi+\nabla_K\nabla_4f)\nonum\\
&=\underbrace{g^{01}}_{e^{-2A}}(\pa_1\phi\pa_4\phi+\pa_1\chi\pa_4\chi+\nabla_1\nabla_4f)=e^{-2A}(\pa_1\phi\pa_4\phi+\pa_1\chi\pa_4\chi+\nabla_1\nabla_4f)
\end{align}
\begin{align}\label{loc.34}
\nabla_1\nabla_4f&=\pa_1\pa_4f-\Gam^L{}_{14}\pa_Lf=\pa_1\pa_4f-\Gam^1{}_{14}\pa_1f=\pa_1\pa_4f-A'\pa_1f
\end{align}
\begin{align}\label{loc.35}
\pa_1\pa_4f=&=\pa_1(\pa_\phi f\pa_4\phi+\pa_\chi f\pa_4\chi)\nonum\\
&=\pa_\phi^2f\pa_1\phi\pa_4\phi+\pa_\chi\pa_\phi f\pa_1\chi\pa_4\phi+\pa_\phi f\pa_1\pa_4\phi+\pa_\chi^2f\pa_1\chi\pa_4\chi+\pa_\phi\pa_\chi f\pa_1\phi\pa_4\chi+\pa_\chi f\pa_1\pa_4\chi\nonum\\
&=\pa_\phi^2f\pa_1\phi\pa_4\phi+\pa_\chi^2f\pa_1\chi\pa_4\chi+\pa_\phi\pa_\chi f(\pa_1\chi\pa_4\phi+\pa_1\phi\pa_4\chi)+\pa_\phi f\pa_1\pa_4\phi+\pa_\chi f\pa_1\pa_4\chi
\end{align}\\
Inserting the expressions of the quantities $\nabla_1\nabla_4f$ and $\pa_1\pa_4f$ from equations \eqref{loc.34} and \eqref{loc.35} respectively, into equation \eqref{loc.33} we get
\eq$\label{loc.36}
T^0{}_4=e^{-2A}(\pa_1\phi\pa_4\phi+\pa_1\chi\pa_4\chi+\pa_1\pa_4f-A'\pa_1f)$
or
\begin{align}\label{loc.37}
T^0{}_4=e^{-2A}[(1+\pa_\phi^2f)\pa_1\phi\pa_4\phi&+(1+\pa_\chi^2f)\pa_1\chi\pa_4\chi+\pa_\phi\pa_\chi f(\pa_1\chi\pa_4\phi+\pa_1\phi\pa_4\chi) \nonum\\
&+\pa_\phi f\pa_1\pa_4\phi+\pa_\chi f\pa_1\pa_4\chi-A'(\pa_\phi f\pa_1\phi+\pa_\chi f\pa_1\chi)]
\end{align}
\begin{center}
\larger{\underline{$\mathbf{T^4{}_1:}$}}
\end{center}
\begin{align}\label{loc.38}
\eqref{loc.15}\xRightarrow[N=1]{M=4}T^4{}_1&=\pa^4\phi\pa_1\phi+\pa^4\chi\pa_1\chi+\nabla^4\nabla_1f=g^{4K}(\pa_K\phi\pa_1\phi+\pa_K\chi\pa_1\chi+\nabla_K\nabla_1f)\nonum\\
&=\underbrace{g^{44}}_{1}(\pa_4\phi\pa_1\phi+\pa_4\chi\pa_1\chi+\nabla_4\nabla_1f)=\pa_4\phi\pa_1\phi+\pa_4\chi\pa_1\chi+\nabla_4\nabla_1f
\end{align}
It is obvious from equations \eqref{loc.33} and \eqref{loc.38} that\\
\eq$\label{loc.39}
T^4{}_1=e^{2A}T^0{}_4$

\newpage
\begin{center}
\larger{\underline{$\mathbf{T^1{}_4:}$}}
\end{center}
\begin{align}\label{loc.40}
\eqref{loc.15}\xRightarrow[N=4]{M=1}T^1{}_4&=\pa^1\phi\pa_4\phi+\pa^1\chi\pa_4\chi+\nabla^1\nabla_4f=g^{1K}(\pa_K\phi\pa_4\phi+\pa_K\chi\pa_4\chi+\nabla_4\nabla_1f)\nonum\\
&=\underbrace{g^{10}}_{e^{-2A}}(\pa_0\phi\pa_4\phi+\pa_0\chi\pa_4\chi+\nabla_0\nabla_4f)+\underbrace{g^{11}}_{e^{-2A}\left(1-\frac{2m}{r}\right)}(\pa_1\phi\pa_4\phi+\pa_1\chi\pa_4\chi+\nabla_1\nabla_4f)\nonum\\
&=e^{-2A}(\pa_0\phi\pa_4\phi+\pa_0\chi\pa_4\chi+\nabla_0\nabla_4f)+e^{-2A}\left(1-\frac{2m}{r}\right)(\pa_1\phi\pa_4\phi+\pa_1\chi\pa_4\chi+\nabla_1\nabla_4f)
\end{align}
From equations \eqref{loc.28}, \eqref{loc.33} and \eqref{loc.40} we deduce that
\eq$\label{loc.41}
T^1{}_4=e^{-2A}T^4{}_0+\left(1-\frac{2m}{r}\right)T^0{}_4$
\thispagestyle{plain}
\begin{center}
\larger{\underline{$\mathbf{T^0{}_0:}$}}
\end{center}
\begin{align}\label{loc.42}
\eqref{loc.15}\xRightarrow[N=0]{M=0}T^0{}_0&=\pa^0\phi\pa_0\phi+\pa^0\chi\pa_0\chi+\nabla^0\nabla_0f+\lagr-\square f\nonum\\
&=g^{01}(\pa_1\phi\pa_0\phi+\pa_1\chi\pa_0\chi+\nabla_1\nabla_0f)+\lagr-\square f\nonum\\
&=e^{-2A}(\pa_1\phi\pa_0\phi+\pa_1\chi\pa_0\chi+\nabla_1\nabla_0f)+\lagr-\square f
\end{align}
where\\
\eq$\label{loc.43}
\lagr\equiv -\frac{(\pa\phi)^2+(\pa\chi)^2}{2}-V(\phi,\chi)=-\frac{\pa^L\phi\pa_L\phi+\pa^L\chi\pa_L\chi}{2}-V(\phi,\chi)$\\
\eq$\label{loc.44}
\square f\equiv \nabla^2f=\nabla^K\nabla_Kf$\\
We are going to calculate firstly the quantities $\lagr$ and $\square f$.
\begin{align}
\lagr&=-\frac{1}{2}(\pa^L\phi\pa_L\phi+\pa^L\chi\pa_L\chi)-V(\phi,\chi)\nonum\\
&=-\frac{1}{2}(\pa^0\phi\pa_0\phi+\pa^1\phi\pa_1\phi+\pa^4\phi\pa_4\phi+\pa^0\chi\pa_0\chi+\pa^1\chi\pa_1\chi+\pa^4\chi\pa_4\chi)-V(\phi,\chi)\nonum\\
&=-\frac{1}{2}(g^{0K}\pa_K\phi\pa_0\phi+g^{1K}\pa_K\phi\pa_1\phi+g^{4K}\pa_K\phi\pa_4\phi+g^{0K}\pa_K\chi\pa_0\chi+g^{1K}\pa_K\chi\pa_1\chi+g^{4K}\pa_K\chi\pa_4\chi)-V(\phi,\chi)\nonum\\
&=-\frac{1}{2}\left\{2\underbrace{g^{01}}_{e^{-2A}}(\pa_1\phi\pa_0\phi+\pa_1\chi\pa_0\chi)+\underbrace{g^{11}}_{e^{-2A}\left(1-\frac{2m}{r}\right)}[(\pa_1\phi)^2+(\pa_1\chi)^2]+\underbrace{g^{44}}_{1}[(\pa_4\phi)^2+(\pa_4\chi)^2]\right\}-V(\phi,\chi)\nonum\\
&=-\frac{e^{-2A}}{2}\left\{2(\pa_1\phi\pa_0\phi+\pa_1\chi\pa_0\chi)+\left(1-\frac{2m}{r}\right)[(\pa_1\phi)^2+(\pa_1\chi)^2]\right\}-\frac{1}{2}[(\pa_4\phi)^2+(\pa_4\chi)^2]-V(\phi,\chi)\Ra\nonum
\end{align}\\
\eq$\label{loc.45}
\lagr=-\frac{e^{-2A}}{2}\left\{2(\pa_1\phi\pa_0\phi+\pa_1\chi\pa_0\chi)+\left(1-\frac{2m}{r}\right)[(\pa_1\phi)^2+(\pa_1\chi)^2]\right\}-\frac{1}{2}[(\pa_4\phi)^2+(\pa_4\chi)^2]-V(\phi,\chi)$
\begin{align}\label{loc.46}
\square f&=\nabla^K\nabla_Kf=g^{KL}\nabla_L\nabla_Kf\nonum\\
&=2g^{01}\nabla_0\nabla_1f+g^{11}\nabla_1\nabla_1f+g^{22}\nabla_2\nabla_2f+g^{33}\nabla_3\nabla_3f+g^{44}\nabla_4\nabla_4f
\end{align}\\
$\nabla_1\nabla_1f$ and $\nabla_0\nabla_1f$ are known from equations \eqref{loc.17} and \eqref{loc.23} respectively.\\
\begin{align}\label{loc.47}
\nabla_2\nabla_2f&=\underbrace{\pa_2^2f}_{0}-\Gam^{L}{}_{22}\pa_Lf=-\Gam^{0}{}_{22}\pa_0f-\Gam^{1}{}_{22}\pa_1f-\Gam^{4}{}_{22}\pa_4f\nonum\\
&=r\pa_0f-(2m-r)\pa_1f+r^2e^{2A}A'\pa_4f
\end{align}
\begin{align}\label{loc.48}
\nabla_3\nabla_3f&=\underbrace{\pa_3^2f}_{0}-\Gam^{L}{}_{33}\pa_Lf=-\Gam^{0}{}_{33}\pa_0f-\Gam^{1}{}_{33}\pa_1f-\Gam^{4}{}_{33}\pa_4f\nonum\\
&=\left[r\pa_0f-(2m-r)\pa_1f+r^2e^{2A}A'\pa_4f\right]\sin^2\theta=(\nabla_2\nabla_2f)\sin^2\theta
\end{align}
\begin{align}\label{loc.49}
\nabla_4\nabla_4f&=\pa_4\pa_4f-\underbrace{\Gam^L{}_{44}}_{0}\pa_Lf=\pa_4^2f	
\end{align}
Combining equations \eqref{loc.17}, \eqref{loc.23}, \eqref{loc.46}-\eqref{loc.49} we obtain
\begin{align}
\square f&=2g^{01}\nabla_0\nabla_1f+g^{11}\nabla_1\nabla_1f+g^{22}\nabla_2\nabla_2f+g^{33}\nabla_3\nabla_3f+g^{44}\nabla_4\nabla_4f\nonum\\
&=2e^{-2A}\left[\pa_0\pa_1f+\left(\frac{m}{r^2}-\frac{\pa_1m}{r}\right)\pa_1f+e^{2A}A'\pa_4f\right]+e^{-2A}\left(1-\frac{2m}{r}\right)\pa_1^2f\nonum\\
&\hspace{1.2em}+\underbrace{g^{22}\nabla_2\nabla_2f+\frac{g^{22}}{\sin^2\theta}(\nabla_2\nabla_2f)\sin^2\theta}_{2g^{22}\nabla_2\nabla_2f}+\pa_4^2f\nonum\\
&=2e^{-2A}\pa_0\pa_1f+e^{-2A}\underbrace{2\left(\frac{m}{r^2}-\frac{\pa_1m}{r}\right)}_{\pa_1\left(1-\frac{2m}{r}\right)}\pa_1f+2A'\pa_4f+e^{-2A}\left(1-\frac{2m}{r}\right)\pa_1^2f\nonum\\
&\hspace{1.2em}+2\frac{e^{-2A}}{r^2}\left[r\pa_0f-(2m-r)\pa_1f+r^2e^{2A}A'\pa_4f\right]+\pa_4^2f\nonum\\
&=e^{-2A}\pa_0\pa_1f+e^{-2A}\pa_1(\pa_0f)+e^{-2A}\pa_1\left[\left(1-\frac{2m}{r}\right)\pa_1f\right]+2A'\pa_4f\nonum\\
&\hspace{1.2em}+\frac{e^{-2A}}{r^2}\underbrace{2r}_{\pa_1(r^2)}\pa_0f+\underbrace{2\frac{e^{-2A}}{r}}_{\frac{e^{-2A}}{r^2}\pa_1(r^2)}\left(1-\frac{2m}{r}\right)\pa_1f+2A'\pa_4f+\pa_4^2f\nonum\\
&=e^{-2A}\pa_0\pa_1f+\frac{e^{-2A}}{r^2}r^2\pa_1(\pa_0f)+\frac{e^{-2A}}{r^2}r^2\pa_1\left[\left(1-\frac{2m}{r}\right)\pa_1f\right]+4A'\pa_4f\nonum\\
&\hspace{1.2em}+\frac{e^{-2A}}{r^2}\pa_1(r^2)\pa_0f+\frac{e^{-2A}}{r^2}\pa_1(r^2)\left[\left(1-\frac{2m}{r}\right)\pa_1f\right]+\pa_4^2f\nonum\\
&=e^{-2A}\pa_0\pa_1f+\frac{e^{-2A}}{r^2}\pa_1\left[r^2\pa_0f+r^2\left(1-\frac{2m}{r}\right)\pa_1f\right]+e^{-4A}\pa_4\left(e^{4A}\pa_4f\right)\Ra\nonum
\end{align}\\
\eq$\label{loc.50}
\square f=e^{-2A}\pa_0\pa_1f+\frac{e^{-2A}}{r^2}\pa_1\left[r^2\pa_0f+r^2\left(1-\frac{2m}{r}\right)\pa_1f\right]+e^{-4A}\pa_4\left(e^{4A}\pa_4f\right)$\\
Substituting $\nabla_1\nabla_0f$ from equation \eqref{loc.23} into \eqref{loc.42} we get
\thispagestyle{plain}
$$T^0{}_0=e^{-2A}\left[\pa_1\phi\pa_0\phi+\pa_1\chi\pa_0\chi+\pa_1\pa_0f+\frac{\pa_1f}{r}\left(\frac{m}{r}-\pa_1m\right)\right]+A'\pa_4f+\lagr-\square f\Ra$$
\eq$\label{loc.51}
T^0{}_0=e^{-2A}\left[\pa_1\phi\pa_0\phi+\pa_1\chi\pa_0\chi+\pa_1\pa_0f-\pa_1\left(\frac{m}{r}\right)\pa_1f\right]+A'\pa_4f+\lagr-\square f$\\
Using equation \eqref{loc.25} as well, we obtain
\begin{align}\label{loc.52}
T^0{}_0&=e^{-2A}\left[(1+\pa_\phi^2f)\pa_1\phi\pa_0\phi+(1+\pa_\chi^2f)\pa_1\chi\pa_0\chi+\pa_\chi\pa_\phi f(\pa_1\chi\pa_0\phi+\pa_1\phi\pa_0\chi)+\pa_\phi f\pa_1\pa_0\phi\right. \nonum\\
&\hspace{4.1em}+\left.\pa_\chi f\pa_1\pa_0\chi+\frac{\pa_\phi f\pa_1\phi+\pa_\chi f\pa_1\chi}{r}\left(\frac{m}{r}-\pa_1m\right)\right]+A'(\pa_\phi f\pa_4\phi+\pa_\chi f\pa_4\chi)+\lagr-\square f
\end{align}
\thispagestyle{plain}
\begin{center}
\larger{\underline{$\mathbf{T^1{}_1:}$}}
\end{center}
\begin{align}\label{loc.53}
\eqref{loc.15}\xRightarrow[N=1]{M=1}T^1{}_1&=\pa^1\phi\pa_1\phi+\pa^1\chi\pa_1\chi+\nabla^1\nabla_1f+\lagr-\square f\nonum\\
&=g^{1K}(\pa_K\phi\pa_1\phi+\pa_K\chi\pa_1\chi+\nabla_K\nabla_1f)+\lagr-\square f\nonum\\
&=\underbrace{g^{10}}_{e^{-2A}}(\pa_0\phi\pa_1\phi+\pa_0\chi\pa_1\chi+\nabla_0\nabla_1f)+\underbrace{g^{11}}_{e^{-2A}\left(1-\frac{2m}{r}\right)}[(\pa_1\phi)^2+(\pa_1\chi)^2+\nabla_1^2f]+\lagr-\square f
\end{align}
Equations \eqref{loc.16} and \eqref{loc.42} combined with \eqref{loc.53} give
\eq$\label{loc.54}
T^1{}_1=T^0{}_0+\left(1-\frac{2m}{r}\right)T^0{}_1$
\begin{center}
\larger{\underline{$\mathbf{T^2{}_2:}$}}
\end{center}
\begin{align}\label{loc.55}
\eqref{loc.15}\xRightarrow[N=2]{M=2}T^2{}_2&=\pa^2\phi\pa_2\phi+\pa^2\chi\pa_2\chi+\nabla^2\nabla_2f+\lagr-\square f\nonum\\
&=g^{2K}(\pa_K\phi\pa_2\phi+\pa_K\chi\pa_2\chi+\nabla_K\nabla_2f)+\lagr-\square f\nonum\\
&=\underbrace{g^{22}}_{\frac{e^{-2A}}{r^2}}(\underbrace{\pa_2\phi\pa_2\phi}_{0}+\underbrace{\pa_2\chi\pa_2\chi}_{0}+\nabla_2\nabla_2f)+\lagr-\square f\nonum\\
&=\frac{e^{-2A}}{r^2}\nabla_2\nabla_2f+\lagr-\square f
\end{align}
\begin{align}\label{loc.56}
\eqref{loc.55}\xRightarrow{\eqref{loc.47}}T^2{}_2&=\frac{e^{-2A}}{r^2}[r\pa_0f-(2m-r)\pa_1f+r^2e^{2A}A'\pa_4f]+\lagr-\square f\nonum\\
&=\frac{e^{-2A}}{r}\left[\pa_0f+\left(1-\frac{2m}{r}\right)\pa_1f\right]+A'\pa_4f+\lagr-\square f
\end{align}
Thus, we have
\eq$\label{loc.57}
T^2{}_2=\frac{e^{-2A}}{r}\left[\pa_0f+\left(1-\frac{2m}{r}\right)\pa_1f\right]+A'\pa_4f+\lagr-\square f$
or
\begin{align}\label{loc.58}
T^2{}_2&=\frac{e^{-2A}}{r}\left[(\pa_\phi f\pa_0\phi+\pa_\chi f\pa_0\chi)+\left(1-\frac{2m}{r}\right)(\pa_\phi f\pa_1\phi+\pa_\chi f\pa_1\chi)\right]\nonum\\
&\hspace{4em}+A'(\pa_\phi f\pa_4\phi+\pa_\chi f\pa_4\chi)+\lagr-\square f
\end{align}
\begin{center}
\larger{\underline{$\mathbf{T^3{}_3:}$}}
\end{center}
\begin{align}\label{loc.59}
\eqref{loc.15}\xRightarrow[N=3]{M=3}T^3{}_3&=\pa^3\phi\pa_3\phi+\pa^3\chi\pa_3\chi+\nabla^3\nabla_3f-\lagr-\square f\nonum\\
&=g^{3K}(\pa_K\phi\pa_3\phi+\pa_K\chi\pa_3\chi+\nabla_K\nabla_3f)-\lagr-\square f\nonum\\
&=\underbrace{g^{33}}_{\frac{e^{-2A}}{r^2\sin^2\theta}}(\underbrace{\pa_3\phi\pa_3\phi}_{0}+\underbrace{\pa_3\chi\pa_3\chi}_{0}+\nabla_3\nabla_3f)-\lagr-\square f\nonum\\
&=\frac{e^{-2A}}{r^2\sin^2\theta}\nabla_3\nabla_3f+\lagr-\square f
\end{align}
\begin{align}
\eqref{loc.59}\xRightarrow{\eqref{loc.48}}T^3{}_3&=\frac{e^{-2A}}{r^2\sin^2\theta}(\nabla_2\nabla_2f)\sin^2\theta-\lagr-\square f=\frac{e^{-2A}}{r^2}\nabla_2\nabla_2f+\lagr-\square f\Ra\nonum
\end{align}
\eq$\label{loc.60}
T^3{}_3=T^2{}_2$
\thispagestyle{plain}
\begin{center}
\larger{\underline{$\mathbf{T^4{}_4:}$}}
\end{center}
\begin{align}\label{loc.61}
\eqref{loc.15}\xRightarrow[N=4]{M=4}T^4{}_4&=\pa^4\phi\pa_4\phi+\pa^4\chi\pa_4\chi+\nabla^4\nabla_4f+\lagr-\square f\nonum\\
&=g^{4K}(\pa_K\phi\pa_4\phi+\pa_K\chi\pa_4\chi+\nabla_K\nabla_4f)+\lagr-\square f\nonum\\
&=\underbrace{g^{44}}_{1}[(\pa_4\phi)^2+(\pa_4\chi)^2+\nabla_4\nabla_4f]+\lagr-\square f\nonum\\
&=(\pa_4\phi)^2+(\pa_4\chi)^2+\nabla_4\nabla_4f+\lagr-\square f
\end{align}\\
The combination of equations \eqref{loc.49} and \eqref{loc.61} gives
\eq$\label{loc.62}
T^4{}_4=(\pa_4\phi)^2+(\pa_4\chi)^2+\pa^2_4f+\lagr-\square f$
Moreover, it is
\begin{align}\label{loc.63}
\pa_4^2f&=\pa_4(\pa_\phi f\pa_4\phi+\pa_\chi f\pa_4\chi)\nonum\\
&=\pa_\phi^2f(\pa_4\phi)^2+\pa_\chi\pa_\phi f\pa_4\chi\pa_4\phi+\pa_\phi f\pa_4^2\phi+\pa_\chi^2f(\pa_4\chi)^2+\pa_\phi\pa_\chi f\pa_4\phi\pa_4\chi+\pa_\chi f\pa_4^2\chi\nonum\\
&=\pa_\phi^2f(\pa_4\phi)^2+\pa_\chi^2f(\pa_4\chi)^2+2\pa_\chi\pa_\phi f\pa_4\chi\pa_4\phi+\pa_\phi f\pa_4^2\phi+\pa_\chi f\pa_4^2\chi
\end{align}
Therefore, equations \eqref{loc.62} and \eqref{loc.63} are combined to give
\eq$\label{loc.64}
T^4{}_4=(1+\pa_\phi^2f)(\pa_4\phi)^2+(1+\pa_\chi^2f)(\pa_4\chi)^2+2\pa_\chi\pa_\phi f\pa_4\chi\pa_4\phi+\pa_\phi f\pa_4^2\phi+\pa_\chi f\pa_4^2\chi+\lagr-\square f$

\begin{center}
\larger{\underline{$\mathbf{T^0{}_2:}$}}
\end{center}
\begin{align}\label{loc.65}
\eqref{loc.15}\xRightarrow[N=2]{M=0} T^0{}_2&=\pa^0\phi\pa_2\phi+\pa^0\chi\pa_2\chi+\nabla^0\nabla_2f=g^{0K}\pa_K\phi\underbrace{\pa_2\phi}_{0}+g^{0K}\pa_K\chi\underbrace{\pa_2\chi}_{0}+g^{0K}\nabla_K\nabla_2f\nonum\\
&=\underbrace{g^{01}}_{e^{-2A}}\nabla_1\nabla_2f=e^{-2A}\nabla_1\nabla_2f
\end{align}\\
\eq$\label{loc.66}
\nabla_1\nabla_2f=\left(\underbrace{\pa_1\pa_2f}_{0}-\Gam^{L}{}_{12}\pa_Lf\right)=-\Gam^{L}{}_{12}\pa_Lf=-\Gam^2{}_{12}\underbrace{\pa_2f}_{0}=0$\\
Thus, equations \eqref{loc.65} and \eqref{loc.66} give\\
\eq$\label{loc.67}
T^0{}_2=0$
\begin{center}
\larger{\underline{$\mathbf{T^0{}_3:}$}}
\end{center}
\begin{align}\label{loc.68}
\eqref{loc.15}\xRightarrow[N=3]{M=0} T^0{}_3&=\pa^0\phi\pa_3\phi+\pa^0\chi\pa_3\chi+\nabla^0\nabla_3f=g^{0K}\pa_K\phi\underbrace{\pa_3\phi}_{0}+g^{0K}\pa_K\chi\underbrace{\pa_3\chi}_{0}+g^{0K}\nabla_K\nabla_3f\nonum\\
&=\underbrace{g^{01}}_{e^{-2A}}\nabla_1\nabla_3f=e^{-2A}\nabla_1\nabla_3f
\end{align}\\
\eq$\label{loc.69}
\nabla_1\nabla_3f=\left(\underbrace{\pa_1\pa_3f}_{0}-\Gam^{L}{}_{13}\pa_Lf\right)=-\Gam^{L}{}_{13}\pa_Lf=-\Gam^3{}_{13}\underbrace{\pa_3f}_{0}=0$\\
Equations \eqref{loc.68} and \eqref{loc.69} yield to\\
\eq$\label{loc.70}
T^0{}_3=0$
\thispagestyle{plain}
\begin{center}
\larger{\underline{$\mathbf{T^1{}_2:}$}}
\end{center}
\begin{align}\label{loc.71}
\eqref{loc.15}\xRightarrow[N=2]{M=1}T^1{}_2&=\pa^1\phi\underbrace{\pa_2\phi}_{0}+\pa^1\chi\underbrace{\pa_2\chi}_{0}+\nabla^1\nabla_2f=g^{1K}\nabla_K\nabla_2f=g^{10}\nabla_0\nabla_2f+g^{11}\nabla_1\nabla_2f\nonum\\
&=e^{-2A}\nabla_0\nabla_2f+e^{-2A}\left(1-\frac{2m}{r}\right)\nabla_1\nabla_2f
\end{align}\\
\eq$\label{loc.72}
\nabla_0\nabla_2f=\underbrace{\pa_0\pa_2f}_{0}-\underbrace{\Gam^L{}_{02}}_{0}\pa_Lf=0$\\
Combining equations \eqref{loc.66}, \eqref{loc.71}, \eqref{loc.72} we get\\
\eq$\label{loc.73}
T^{1}{}_2=0$
\begin{center}
\larger{\underline{$\mathbf{T^1{}_3:}$}}
\end{center}
\begin{align}\label{loc.74}
\eqref{loc.15}\xRightarrow[N=3]{M=1}T^1{}_3&=\pa^1\phi\underbrace{\pa_3\phi}_{0}+\pa^1\chi\underbrace{\pa_3\chi}_{0}+\nabla^1\nabla_3f=g^{1K}\nabla_K\nabla_3f=g^{10}\nabla_0\nabla_3f+g^{11}\nabla_1\nabla_3f\nonum\\
&=e^{-2A}\nabla_0\nabla_3f+e^{-2A}\left(1-\frac{2m}{r}\right)\nabla_1\nabla_3f
\end{align}\\
\eq$\label{loc.75}
\nabla_0\nabla_3f=\underbrace{\pa_0\pa_3f}_{0}-\underbrace{\Gam^L{}_{03}}_{0}\pa_Lf=0$\\
From equations \eqref{loc.69} and \eqref{loc.75} into \eqref{loc.74} we obtain\\
\eq$\label{loc.76}
T^{1}{}_3=0$
\begin{center}
\larger{\underline{$\mathbf{T^2{}_0:}$}}\vspace{1em}
\end{center}
\eq$\eqref{loc.15}\xRightarrow[N=0]{M=2}T^2{}_0=\underbrace{\pa^2\phi}_{0}\pa_0\phi+\underbrace{\pa^2\chi}_{0}\pa_0\chi+\nabla^2\nabla_0f=g^{2K}\nabla_K\nabla_0f=g^{22}\underbrace{\nabla_2\nabla_0f}_{0\ \eqref{loc.72}}=0\Ra\nonum$
\eq$\label{loc.77}
T^2{}_0=0$
\begin{center}
\larger{\underline{$\mathbf{T^2{}_1:}$}}\vspace{1em}
\end{center}
\eq$\eqref{loc.15}\xRightarrow[N=1]{M=2}T^2{}_1=\underbrace{\pa^2\phi}_{0}\pa_1\phi+\underbrace{\pa^2\chi}_{0}\pa_1\chi+\nabla^2\nabla_1f=g^{2K}\nabla_K\nabla_1f=g^{22}\underbrace{\nabla_2\nabla_1f}_{0\ \eqref{loc.66}}=0\Ra\nonum$
\eq$\label{loc.78}
T^2{}_1=0$
\thispagestyle{plain}
\begin{center}
\larger{\underline{$\mathbf{T^2{}_3:}$}}\vspace{1em}
\end{center}
\eq$\label{loc.79}
\eqref{loc.15}\xRightarrow[N=3]{M=2}T^2{}_3=\underbrace{\pa^2\phi}_{0}\pa_3\phi+\underbrace{\pa^2\chi}_{0}\pa_3\chi+\nabla^2\nabla_3f=g^{2K}\nabla_K\nabla_3f=g^{22}\nabla_2\nabla_3f$\\
\eq$\label{loc.80}
\nabla_2\nabla_3f=\underbrace{\pa_2\pa_3f}_{0}-\Gam^L{}_{23}\pa_Lf=-\Gam^3{}_{23}\underbrace{\pa_3f}_{0}=0$\\
From equations \eqref{loc.79} and \eqref{loc.80} it is obvious that\\
\eq$\label{loc.81}
T^2{}_3=0$
\begin{center}
\larger{\underline{$\mathbf{T^2{}_4:}$}}\vspace{1em}
\end{center}
\eq$\label{loc.82}
\eqref{loc.15}\xRightarrow[N=4]{M=2}T^2{}_4=\underbrace{\pa^2\phi}_{0}\pa_4\phi+\underbrace{\pa^2\chi}_{0}\pa_4\chi+\nabla^2\nabla_4f=g^{2K}\nabla_K\nabla_4f=g^{22}\nabla_2\nabla_4f$\\
\eq$\label{loc.83}
\nabla_2\nabla_4f=\underbrace{\pa_2\pa_4f}_{0}-\Gam^L{}_{24}\pa_Lf=-\Gam^2{}_{24}\underbrace{\pa_2f}_{0}=0$\\
Equations \eqref{loc.82} and \eqref{loc.83} are combined to give\\
\eq$\label{loc.84}
T^2{}_4=0$
\begin{center}
\larger{\underline{$\mathbf{T^3{}_0:}$}}\vspace{1em}
\end{center}
\eq$
\eqref{loc.15}\xRightarrow[N=0]{M=3}T^3{}_0=\underbrace{\pa^3\phi}_{0}\pa_0\phi+\underbrace{\pa^3\chi}_{0}\pa_0\chi+\nabla^3\nabla_0f=g^{3K}\nabla_K\nabla_0f=g^{33}\underbrace{\nabla_3\nabla_0f}_{0\ \eqref{loc.75}}\Ra\nonum$
\eq$\label{loc.85}
T^3{}_0=0$
\begin{center}
\larger{\underline{$\mathbf{T^3{}_1:}$}}\vspace{1em}
\end{center}
\eq$
\eqref{loc.15}\xRightarrow[N=1]{M=3}T^3{}_1=\underbrace{\pa^3\phi}_{0}\pa_1\phi+\underbrace{\pa^3\chi}_{0}\pa_1\chi+\nabla^3\nabla_1f=g^{3K}\nabla_K\nabla_1f=g^{33}\underbrace{\nabla_3\nabla_1f}_{0\ \eqref{loc.69}}\Ra\nonum$
\eq$\label{loc.86}
T^3{}_1=0$

\begin{center}
\larger{\underline{$\mathbf{T^3{}_2:}$}}\vspace{1em}
\end{center}
\eq$
\eqref{loc.15}\xRightarrow[N=2]{M=3}T^3{}_2=\underbrace{\pa^3\phi}_{0}\pa_2\phi+\underbrace{\pa^3\chi}_{0}\pa_2\chi+\nabla^3\nabla_2f=g^{3K}\nabla_K\nabla_2f=g^{33}\underbrace{\nabla_3\nabla_2f}_{0\ \eqref{loc.80}}\Ra\nonum$
\eq$\label{loc.87}
T^3{}_2=0$
\begin{center}
\larger{\underline{$\mathbf{T^3{}_4:}$}}\vspace{1em}
\end{center}
\eq$\label{loc.88}
\eqref{loc.15}\xRightarrow[N=4]{M=3}T^3{}_4=\underbrace{\pa^3\phi}_{0}\pa_4\phi+\underbrace{\pa^3\chi}_{0}\pa_4\chi+\nabla^3\nabla_4f=g^{3K}\nabla_K\nabla_4f=g^{33}\nabla_3\nabla_4f$\\
\eq$\label{loc.89}
\nabla_3\nabla_4f=\underbrace{\pa_3\pa_4f}_{0}-\Gam^L{}_{34}\pa_Lf=-\Gam^3{}_{34}\underbrace{\pa_3f}_{0}=0$\\
The combination of equations \eqref{loc.88} and \eqref{loc.89} yield to
\eq$\label{loc.90}
T^3{}_4=0$
\thispagestyle{plain}
\begin{center}
\larger{\underline{$\mathbf{T^4{}_2:}$}}\vspace{1em}
\end{center}
\eq$
\eqref{loc.15}\xRightarrow[N=2]{M=4}T^4{}_2=\underbrace{\pa^4\phi}_{0}\pa_2\phi+\underbrace{\pa^4\chi}_{0}\pa_2\chi+\nabla^4\nabla_2f=g^{4K}\nabla_K\nabla_2f=g^{44}\underbrace{\nabla_4\nabla_2f}_{0\ \eqref{loc.83}}\Ra\nonum$
\eq$\label{loc.91}
T^4{}_2=0$
\begin{center}
\larger{\underline{$\mathbf{T^4{}_3:}$}}\vspace{1em}
\end{center}
\eq$
\eqref{loc.15}\xRightarrow[N=3]{M=4}T^4{}_2=\underbrace{\pa^4\phi}_{0}\pa_3\phi+\underbrace{\pa^4\chi}_{0}\pa_3\chi+\nabla^4\nabla_3f=g^{4K}\nabla_K\nabla_3f=g^{44}\underbrace{\nabla_4\nabla_3f}_{0\ \eqref{loc.89}}\Ra\nonum$\\
\eq$\label{loc.92}
T^4{}_3=0$